%% file: paper.tex
\documentclass{jair-arxiv}

\usepackage{ebproof} %
\usepackage{stmaryrd} %
\usepackage{thmtools,thm-restate} %
\usepackage{enumitem} %

\newtheorem{theorem}{Theorem}[section]

\theoremstyle{remark}
\newtheorem{remark}[theorem]{Remark} %

\input{macros}

\setcopyright{cc}
\copyrightyear{2026}

\begin{document}

\title[A Sequent Calculus for Non-Monotone Inductive Definitions]{A Sequent Calculus for Non-Monotone Inductive Definitions}

\author{Robbe Van den Eede}
\authornote{Corresponding Author.}
\email{robbe.vandeneede@kuleuven.be}
\orcid{0000-0002-2579-9053}
\affiliation{%
  \institution{KU Leuven, Department of Computer Science}
  \city{B-3001, Leuven}
  \country{Belgium}
}
\affiliation{%
\institution{Vrije Universiteit Brussel, Department of Computer Science}
\city{B-1050, Brussel}
\country{Belgium}
}

\author{Marc Denecker}
\orcid{0000-0002-0422-7339}
\email{marc.denecker@kuleuven.be}
\affiliation{%
	\institution{KU Leuven, Department of Computer Science}
	\city{B-3001, Leuven}
	\country{Belgium}
}

\renewcommand{\shortauthors}{Van den Eede and Denecker}

\begin{abstract}
{\bf Background:} 
    Inductive definitions constitute an important form of knowledge.
    The logic \foid is an extension of classical first-order logic (\fo) with a language construct to express non-monotone inductive definitions. %
    In \foid, non-monotone inductive definitions are assigned meaning through the well-founded semantics, which is strongly based on the homonymous semantics from Logic Programming.
    Most existing proof systems for inductive definitions impose syntactic constraints on their definitions such as positivity and stratification, thereby excluding many useful and natural definitions.
    
    {\bf Objectives:}
    We aim to obtain a sequent calculus for \foid that supports non-monotone inductive definitions. 
    
    {\bf Methods:}
    We extend an existing sequent calculus \lkid by Brotherston and Simpson, formalizing the principle of mathematical induction, to a sequent calculus \scfoid for \foid.
    To accommodate for non-monotone inductive definitions, we introduce an induction rule with an
    asymmetry between positive and negative occurrences of defined atoms.
    This asymmetric rule has a theoretical basis in  
    the stable (model) semantics, a commonly used semantics in Logic Programming that is closely related to the well-founded semantics.
    
    {\bf Results:}
    \scfoid covers the non-monotone inductive definitions of \foid.
    \scfoid is sound, both w.r.t.\ the well-founded and the stable semantics.
    As a consequence, \scfoid is also suitable to prove theorems about (fragments of) logic programs under the well-founded and the stable semantics.
    As a consequence of G\"odel's first incompleteness theorem, \scfoid (like \lkid) cannot be complete, neither w.r.t.~the well-founded nor the stable semantics.
    However, \scfoid is complete on a propositional fragment of \foid w.r.t.\ the stable semantics, and on a first-order fragment of \foid w.r.t. a weaker Henkin semantics.
    Cut-elimination, which intuitively says that all theorems can be proven without detours, does not hold in general for \scfoid.
    Nevertheless, cut-elimination holds for a fragment of \foid consisting of positive definitions, and a weakened version of cut-elimination holds for a fragment of \foid consisting of stratified definitions.
    Furthermore, we establish a correspondence between \scfoid-proofs and \scfo-proofs (where \scfo is a sequent calculus for \fo) in which each inductive definition is replaced by a natural \fo induction scheme.  
    We illustrate the versatility of \scfoid by proving various statements about non-monotone definitions.
    Incidentally, \scfoid allows for proofs of non-totality of definitions, offering a proof-theoretic analysis of paradoxes in ill-formed non-monotone definitions.
    
    {\bf Conclusions:}
    \scfoid is a theoretically solid sequent calculus for \foid, enabling formal proofs of theorems involving non-monotone inductive definitions.
\end{abstract}

\maketitle

\input{introduction}
\input{foid}

\input{lfoid}
\input{correspondence}
\input{alternative-semantics}

\input{soundness}

\input{completeness}
\input{cut-elimination}

\input{conclusion}

\begin{acks}
	This work was supported by Fonds Wetenschappelijk Onderzoek Vlaanderen (FWO Flanders) under the projects %
	G0B2221N,
	G070521N,
	and 11A2R26N.
	
	We thank Maurice Bruynooghe for his thorough proof-reading of an earlier draft of this paper and the many constructive comments that resulted from this.
	We also thank Bart Bogaerts for his insightful questions, remarks and suggestions regarding the content and the structure of this paper.
	The feedback of Maurice Bruynooghe and Bart Bogaerts significantly improved the correctness, readability and overall quality of this paper.
\end{acks}

\bibliographystyle{ACM-Reference-Format}
\bibliography{krrlib.bib,refs_not_in_krrlib.bib}

\appendix

\input{proof-induction-scheme-wf}
\input{proofs-section-scfoid}
\input{proofs-section-alternative-semantics}
\input{proofs-section-cut-elimination}

\end{document}

%% file: macros.tex
\newcommand{\vsp}{\vspace{2mm}}

\newcommand{\N}{\mathbb{N}}
\newcommand{\true}{\mathbf{t}}
\newcommand{\false}{\mathbf{f}}
\newcommand{\unknown}{\mathbf{u}}
\newcommand{\inconsistent}{\mathbf{i}}
\newcommand{\mci}{\mathcal{I}}
\newcommand{\mcj}{\mathcal{J}}
\newcommand{\mco}{\mathcal{O}}

\newcommand{\mch}{\mathcal{H}}

\newcommand{\mfa}{\mathfrak{A}}
\newcommand{\mfb}{\mathfrak{B}}

\newcommand{\arity}[1]{\mathsf{ar}(#1)}

\newcommand{\normalization}[1]{#1^{\mathrm{N}}}
\newcommand{\voc}[1]{\mathsf{voc}(#1)}
\newcommand{\dom}[1]{\mathsf{dom}(#1)}
\newcommand{\Pow}[1]{\mathbb{P}(#1)}
\newcommand{\defi}[1]{\mathrm{def}(#1)}
\newcommand{\pars}[1]{\mathrm{pars}(#1)}
\newcommand{\sym}[1]{\mathrm{sym}(#1)}

\newcommand{\seqcal}{$\mathcal{SC}$\xspace}
\newcommand{\logic}{\mathcal{L}\xspace}

\newcommand{\fo}{{\normalfont \textrm{FO}}\xspace}
\newcommand{\foid}{{\normalfont \textrm{FO(ID)}}\xspace}

\newcommand{\lk}{{\normalfont \textrm{LK}}\xspace}
\newcommand{\lkid}{{\normalfont \textrm{LKID}}\xspace}
\newcommand{\lkidinf}{{\normalfont \textrm{LKID}$^\omega$}\xspace}
\newcommand{\clkid}{{\normalfont \textrm{CLKID}$^\omega$}\xspace}

\newcommand{\scfo}{{\normalfont\textrm{SC}$_{\fo}$}\xspace}
\newcommand{\scfoid}{{\normalfont\textrm{SC}$_{\foid}$}\xspace}
\newcommand{\scfoidsmd}{{\normalfont\textrm{SC}$_{\foid}^{\subseteq\textrm{MD}}$}\xspace}
\newcommand{\scfoidmd}{{\normalfont\textrm{SC}$_{\foid}^{\textrm{MD}}$}\xspace}

\newcommand{\scfoiddefl}{{\normalfont\textrm{SC}$_{\foid}^{\textrm{\deflalt}}$}\xspace}

\newcommand{\matimps}[1]{\normalfont\textrm{MI}(#1)}
\newcommand{\indscheme}[1]{\normalfont\textrm{IS}(#1)}
\newcommand{\indax}[1]{\normalfont\textrm{Ind}(#1)}
\newcommand{\foapp}[1]{\normalfont\textrm{FO}(#1)}

\newcommand{\vocab}{\Sigma}
\newcommand{\formula}{\varphi}
\newcommand{\formulatwo}{\psi}
\newcommand{\formulathree}{\chi}
\newcommand{\expr}{\epsilon}
\newcommand{\defn}{\Phi}
\newcommand{\defns}{\Phi_1, \dots, \Phi_n}
\newcommand{\defrul}{\forall \bar{x}: P(\bar{t}) \rul \formula}
\newcommand{\ssubstpos}[2]{[#1 {\sslash^+}  #2]}
\newcommand{\ssubstposg}{\ssubstpos{\formulatwo}{\formula}}
\newcommand{\ssubstneg}[2]{[#1 {\sslash^-} #2]}
\newcommand{\ssubstnegg}{\ssubstneg{\formulatwo}{\formula}}
\newcommand{\subst}[2]{[#1 \, {\slash} \, #2]}
\newcommand{\substih}[1]{\subst{\ihs{#1}}{#1}}
\newcommand{\substihpi}{\substih{\Pi}}
\newcommand{\substpos}[2]{[#1 \, {\slash^+} \, #2]}
\newcommand{\substposih}[1]{\substpos{\ihs{#1}}{#1}}
\newcommand{\substposihpi}{\substposih{\Pi}}
\newcommand{\setexp}[2]{\{ #1 \mid #2 \}}
\newcommand{\setexpg}{\setexp{\bar{z}}{\formula}}
\newcommand{\mb}[1]{B_{#1}} %
\newcommand{\mbd}[2]{\mb{#1}^{#2}}

\newcommand{\leqt}{\leq_{\mathrm{t}}}
\newcommand{\lesst}{<_{\mathrm{t}}}
\newcommand{\leqp}{\leq_{\mathrm{p}}}

\newcommand{\modelswf}{\models_{\mathrm{wf}}}
\newcommand{\modelsst}{\models_{\mathrm{st}}}
\newcommand{\modelsh}{\models_{\mathrm{H}}}

\newcommand{\structure}{\mci}
\newcommand{\structuretwo}{\mcj}
\newcommand{\natstr}{\mathcal{N}}
\newcommand{\natth}{T_{\natstr}}
\newcommand{\natvoc}{\vocab_{\natstr}}
\newcommand{\tvs}{\mfa}
\newcommand{\tvstwo}{\mfb}
\newcommand{\structurespaceg}[2]{\mathfrak{I}^{#1}_{#2}}
\newcommand{\structurespace}{\structurespaceg{\defn}{\context}}
\newcommand{\consg}[2]{\mathcal{C}^{#1}_{#2}}
\newcommand{\cons}{\consg{\defn}{\context}}
\newcommand{\stabopg}[2]{\mathrm{ST}^{#1}_{#2}}
\newcommand{\stabop}{\stabopg{\defn}{\context}}
\newcommand{\henkinclass}{\mch}
\newcommand{\henkinstructure}{(\structure, \henkinclass)}
\newcommand{\henkincountermodel}{(\structure_\omega, \henkinclass_\omega)}
\newcommand{\context}{\mco}
\newcommand{\wellfoundedinduction}{\langle \mfa_i \rangle_{0 \leq i \leq \beta}}

\newcommand{\ih}{\normalfont\textrm{IH}}
\newcommand{\ihs}[1]{\ih_{#1}}
\newcommand{\ihf}[2]{\ih_{#1}(#2)}
\newcommand{\ihstwo}[1]{\ih'_{#1}}
\newcommand{\ihftwo}[2]{\ih'_{#1}(#2)}
\newcommand{\deriv}{\mathcal{D}}
\newcommand{\prf}{\mathcal{P}}
\newcommand{\seq}[2]{#1 \vdash #2}
\newcommand{\sequent}{\seq{\Gamma}{\Delta}}
\newcommand{\seqshort}{\mathcal{S}}
\newcommand{\limitsequent}{\defns, \Gamma_\omega \vdash \Delta_\omega}
\newcommand{\regularsequent}{\defns, \Gamma \vdash \Delta}
\newcommand{\compl}[1]{#1^{\mathrm{c}}}
\newcommand{\equivs}{\Omega}
\newcommand{\positiverewriting}{\compl{\defn_1}, \dots, \compl{\defn_n}, \equivs, \Gamma \vdash \Delta}
\newcommand{\peax}{\normalfont\textrm{PA}}
\newcommand{\coax}{\normalfont\textrm{CA}}
\newcommand{\clarkrul}{(def L\textsuperscript{C})\xspace}
\newcommand{\deflalt}{(def L')\xspace}

\newcommand{\ddep}[1]{\normalfont\textrm{DDep}_{#1}}
\newcommand{\dep}[1]{\normalfont\textrm{Dep}_{#1}}
\newcommand{\md}[2]{\normalfont\textrm{MD}_{#1}(#2)}

\newcommand{\fixedcanseq}{\defns, \Gamma_0 \vdash \Delta_0}
\newcommand{\schedelemg}[1]{\xi_{#1}}
\newcommand{\schedelem}{\schedelemg{i}}
\newcommand{\sched}{(\schedelem)_{i \in \N}}
\renewcommand{\terms}{\normalfont\textrm{Terms}(\vocab)}
\newcommand{\eqrel}{\sim}

\newcommand{\srsubst}{\subst{\lnot S}{R}}%

\DeclareMathOperator{\free}{free}
\DeclareMathOperator{\lfp}{lfp}

\def\qed{\hfill{\qedboxempty}      %
	\ifdim\lastskip<\medskipamount \removelastskip\penalty55\medskip\fi}

\def\qedboxempty{\vbox{\hrule\hbox{\vrule\kern3pt
			\vbox{\kern3pt\kern3pt}\kern3pt\vrule}\hrule}}

\def\qedfull{\hfill{\qedboxfull}   %
	\ifdim\lastskip<\medskipamount \removelastskip\penalty55\medskip\fi}

\def\qedboxfull{\vrule height 4pt width 4pt depth 0pt}
\newcommand{\markfull}{}

\newcommand{\defin}[1]{\left\{ 
	\begin{array}{l} #1 \end{array} 
	\right\}}
\newcommand{\defintight}[1]{\left\{ 
	\hspace{-3pt} 
	\begin{array}{l} #1 \end{array} 
	\hspace{-3pt} 
	\right\}}
\newcommand{\rul}{\leftarrow}

\newcommand{\ignore}[1]{}

%% file: introduction.tex
\section{Introduction} \label{sec:introduction}

\subsection{Non-Monotone Inductive Definitions}

Inductive definitions are an important form of knowledge, specifying a range of useful notions in mathematics and computer science.
A prototypical example is the definition of the natural numbers via the following rules:
\begin{itemize}
	\item $0$ is a natural number.
	\item If $n$ is a natural number, then so is its successor $s(n)$.
\end{itemize}
An inductive definition defines a set by specifying how to construct it. 
The \emph{construction process}, often called the induction process, is naturally conceived as a process of iterated rule application, starting from the empty set, and terminating when the set is saturated under application of the rules. 

The natural number definition is a \emph{monotone} definition, meaning that the addition of objects to the defined set depends only on the presence of other objects, not on their absence. 
Some inductive definitions are non-monotone. 
Consider, for instance, the \emph{satisfaction relation} $\models$ in propositional logic, which is a binary relation between structures $\structure$ (seen as sets of propositional symbols) and propositional formulas $\formula$ (written in terms of $\land$ and $\lnot$). 
It is commonly defined by \emph{structural induction} on $\formula$, via the following rules:
\begin{itemize}
	\item $\structure \models p$ if $p \in \structure$.
	\item $\structure \models \formula \land \formulatwo$ if $\structure \models \formula$ and $\structure \models \formulatwo$.
	\item $\structure \models \neg \formula$ if not $\structure \models \formula$.
\end{itemize}
The third rule is non-monotone, as it derives the presence of pairs $(\structure, \lnot \formula)$ in the defined relation $\models$ from the absence of pairs $(\structure, \formula)$. 
Non-monotone rules complicate the process of iterated rule application. 
For example, every instance of the third rule is trivially applicable at the start of the construction process, when the set under construction is still empty and hence, ``not $\structure \models \formula$'' is trivially satisfied.
Applying the rule at this point would be {\em unsafe}, however, as the construction process may later derive $\structure \models \formula$. 
As argued in \cite{KR/DeneckerV14,ai/BogaertsVD18}, application of a non-monotone rule is to be delayed until the truth of its condition is {\em safely} established and cannot be falsified later in the process. 
The role of the \emph{well-founded induction order} associated to this type of definition (here the subformula order on formulas $\formula$) is to ensure safety; applying rules along this order, i.e., deriving smaller elements  before larger ones, ensures that each rule application is safe. 

\foid is an extension of classical first-order logic with a language construct to represent a wide range of inductive definitions \cite{tocl/DeneckerT08}.
More concretely, inductive definitions are represented in \foid as sets of \emph{definitional rules}, which are expressions of the form $\forall \bar{x}: P(\bar{t}) \rul \formula$.
Here $P(\bar{t})$ is an atom called the \emph{head} of the rule, and $\formula$ a first-order formula called the \emph{body} of the rule.
For instance, the natural number definition from above can be formally represented in \foid as 
\begin{equation*}
	\defin{
		\mathit{Nat}(\mathit{zero}) \rul \top \\
		\forall n: \mathit{Nat}(\mathit{succ}(n)) \rul \mathit{Nat}(n)
	} .
\end{equation*}

The logic \foid has a somewhat convoluted history, with roots outside mathematical logic.
Its origins lie in the semantic troubles surrounding the logic programming language Prolog in the 1970s. 
Prolog programs consist of \emph{general program clauses}.
These are expressions of the form \texttt{Head} $\colon$\!\!- \texttt{Body}, in which \texttt{Head} consists of an atom, and \texttt{Body} of any (finite) number of literals.
For some time, the logic programming community struggled to provide a coherent logical declarative interpretation for Prolog programs.
Denecker et al.~proposed interpreting Prolog programs as non-monotone inductive definitions \cite{tocl/DeneckerBM01}. 
Denecker and Vennekens showed that the \emph{well-founded semantics}, a formal semantics developed for logic programs with negation in the body \cite{GelderRS91}, correctly computes the defined relation of non-monotone inductive definitions via three-valued evaluation, without requiring an induction order to ensure the safe application of rules \cite{KR/DeneckerV14,Denecker98}. %
Later, Denecker and Ternovska proposed extending classical first-order logic with a language construct to express non-monotone inductive definitions %
\cite{Denecker:CL2000,tocl/DeneckerT08}. 
This resulted in the logic \foid, which serves as an expressive knowledge representation language\footnote{
	In the field of Knowledge Representation, first-order logic is commonly used as a modeling language.
	However, first-order logic has a rather limited expressivity, as it cannot represent certain elementary concepts such the set of natural numbers and the transitive closure of a graph.
	Adding inductive definitions to first-order logic helps overcome these representational weaknesses.
} 
and a formal scientific study of general inductive definitions.

Other logics of inductive definitions impose syntactic constraints on their definitions such as \emph{positivity} \cite{ajm/Post43,Spector61,Moschovakis74,Aczel77,jlc/BrotherstonS11} and \emph{stratification} \cite{Kreisel63,Feferman70,MartinLoef71,BuchholzFPS81,tcs/McDowellM00,tpp/MomiglianoT04}. 
Positivity forbids negation in bodies of definitional rules, thereby only admitting monotone definitions.
Stratification does permit negation in rule bodies, but only such that the definition can be decomposed into a hierarchy of positive definitions, each depending positively or negatively on concepts defined strictly lower in the hierarchy. 
Although less restrictive than positivity, stratification still excludes many natural and useful definitions.
By not imposing these constraints, \foid offers a general account of inductive definitions.
In particular, \foid covers monotone definitions, definitions over a well-founded order, and \emph{iterated inductive definitions} \cite{Kreisel63,Feferman70,MartinLoef71,BuchholzFPS81}, as shown in \cite{tocl/DeneckerT08}.

By discarding the stratification restriction, \foid admits \emph{non-total} definitions, for which the construction process terminates in a state where at least one rule is applicable, but no rule is safely applicable. 
A simple example is $\{P \rul \neg P\}$, which defines the propositional symbol $P$ as its own negation. 
In the initial state, where $P$ is false, the single rule of this definition is applicable but not safely, since its application falsifies the condition of the rule.  
The well-founded model assigns {\em undefined} to all atoms that are derivable but not safely derivable. 

Non-total definitions are related to self-referential paradoxes.
For instance, $\{P \rul \neg P\}$ can be seen as a formalization of the \emph{liar paradox}: ``this sentence is false'', where $P$ denotes the truth value of the sentence.
Alternatively, consider a version of the \emph{barber paradox} in which a man named Jim shaves himself, and a barber named Bob shaves all men who do not shave themselves.
This scenario can be formalized with the following definition in \foid:
\[
	\defin{
		\mathit{Shaves}(\mathit{Jim}, \mathit{Jim}) \rul \top\\
		\forall m: \mathit{Shaves}(\mathit{Bob}, m) \rul \lnot \mathit{Shaves}(m, m)
	}
\]
This definition is non-total, as it contains the paradoxical instance $\mathit{Shaves}(\mathit{Bob}, \mathit{Bob}) \rul \lnot \mathit{Shaves}(\mathit{Bob}, \mathit{Bob})$.
Paradoxical ``definitions'' are anomalies, to be avoided in practice. 
But they are interesting anomalies, especially from the perspective of \foid as a formal scientific study of inductive definitions.
Paradoxical definitions can be stated in natural language, and the logic \foid allows us to study such definitions, and to distinguish them from sensible, i.e., total definitions.  
Nevertheless, when desired, one can always impose stratification on definitions in \foid to rule out non-total definitions, or other more general conditions as investigated in \cite{tocl/DeneckerT08}. 

\subsection{A Sequent Calculus for Non-Monotone Inductive Definitions}

Deductive reasoning or theorem proving is an important and useful form of reasoning that occupies a central position in logic.
In this paper, we investigate deductive reasoning for non-monotone inductive definitions by introducing a sequent calculus \scfoid for the logic \foid.
A \emph{sequent calculus} is a style of proof system due to Gentzen \cite{Gentzen35}, which is well-known for its theoretical elegance and its goal-directed approach to theorem proving.
A proof system for \foid enables formally proving theorems about non-monotone inductive definitions.
Proof systems have applications in \emph{formal verification}, in which crucial properties of systems are logically derived from their implementation using proof assistants.
A more recent application of proof systems is in \emph{proof logging} \cite{Heule21unsatisfiability,acm/BarbosaBCDKLNNOPRTZ23proofcertificaties,jair/BogaertsGMN23}, in which combinatorial search algorithms are made \emph{certifying} \cite{csr/McConnellMNS11} by generating, alongside the output to a problem instance, a proof of its correctness, which is consequently verified by an independent proof checker.

Our sequent calculus \scfoid is based on the \emph{principle of mathematical induction}, which is a technique to prove theorems about inductively defined relations by using \emph{induction hypotheses}.
\scfoid extends the sequent calculus \lkid by Brotherston and Simpson \cite{jlc/BrotherstonS11}, which formalizes the principle of mathematical induction for positive (hence, monotone) definitions.
Our extension of \lkid to non-monotone definitions is surprisingly simple: in the induction rule, we replace only the positive occurrences of defined predicates by an induction hypothesis.
This refined rule has a theoretical basis in the \emph{stable (model) semantics} from Logic Programming \cite{iclp/GelfondL88}.
While the stable semantics is not about inductive definitions, it can be seen as a good approximation of the well-founded semantics, which coincides with the well-founded semantics on every total (hence, on every `sensible') definition.

\subsection{Related Work} \label{sec:related-work}

Several proof systems have been developed for inductive definitions, supporting them in different degrees of generality.
Brotherston and Simpson introduced sequent calculi for classical first-order logic with positive inductive definitions \cite{jlc/BrotherstonS11}.
More specifically, they introduced the sequent calculus \lkid, as well as the \emph{infinitary} sequent calculus \lkidinf and the \emph{cyclic} sequent calculus \clkid, in which proofs take the form of possibly infinite trees and finite graphs, respectively.
\lkid formalizes the principle of mathematical induction, whereas \lkidinf and \clkid formalize the \emph{principle of infinite descent}, another technique to prove theorems about inductive definitions, relying on the well-foundedness of sets.

In the realm of intuitionistic type theory, Martin-L\"of \cite{MartinLoef71} introduced a natural deduction system (another style of proof system due to Gentzen \cite{Gentzen35}) for iterated inductive definitions, which are stratified inductive definitions.
McDowell and Miller developed a sequent calculus for iterated inductive definitions \cite{tcs/McDowellM00}, which can be seen as a counterpart of Martin-L\"of's natural deduction system.
Momigliano and Tiu provided an extension of McDowell and Miller's sequent calculus, supporting coinduction \cite{tpp/MomiglianoT04}.

Very recently, Guermond and Nadathur developed an intuitionistic sequent calculus for \emph{ground stratified inductive definitions} \cite{fscd26/GuermondN}.
These inductive definitions relax the traditional stratification condition to \emph{ground stratification}, also known as \emph{local stratification}, refining the hierarchical order on predicates to the level of atoms.
According to this hierarchical order, defined atoms are (strictly) greater than the defined atoms on which they depend (negatively).
As such, ground stratification imposes an \emph{induction order} \cite{KR/DeneckerV14} on definitions, thereby guaranteeing totality.
Notably, the induction rule proposed by Guermond and Nadathur is similar to our induction rule, in that they also introduce an asymmetry between positive and negative occurrences of defined atoms.

Earlier, Halln\"as and Schroeder-Heister introduced intuitionistic sequent calculi for inductive definitions based on 
\emph{definitional reflection} \cite{jlc90/HallnasS,ieee93/SchroederH,elp94/SchroederHeister,lu07/SchroederHeister,topoi12/SchroederHeister,apal13/SchroederHeister}.
Definitional reflection is a proof-theoretic counterpart to \emph{Clark's completion} \cite{elp94/SchroederHeister,ieee93/SchroederH}, which says that a defined atom is true iff (at least) one of its defining conditions is true \cite{adbt/Clark78}.
Interestingly, Halln\"as and Schroeder-Heister do not impose syntactic constraints on their definitions, and Schroeder-Heister explicitly discusses (the intuitionistic variant of) the paradoxical definition $\defintight{P \rul \lnot P}$ \cite{ieee93/SchroederH,elp94/SchroederHeister,topoi12/SchroederHeister}.

A fundamental difference between our framework and the frameworks mentioned above is that we develop a single proof \scfoid system for the single logic \foid, in which definitions appear as formulas, whereas the frameworks mentioned above define a family of proof systems, parameterized by inductive definitions.
In the work of Schroeder-Heister, for instance, the paradoxical definition $\defintight{P \rul \lnot P}$ corresponds to a sequent calculus in which both $P$ and $\lnot P$ are provable.
Applying the cut rule to these formulas, falsity becomes provable, meaning that the calculus is inconsistent.
In \scfoid, on the other hand, we can only prove $P$ and $\lnot P$ \emph{conditional on} the definition $\defintight{P \rul \lnot P}$, and derive from this that $\defintight{P \rul \lnot P}$ has no models.
Another significant difference between our framework and the intuitionistic proof systems mentioned above is the pivotal role of model theory in our framework (since the well-founded semantics is inextricably connected to \foid, the stable semantics provides the motivation for our induction rule, and the Henkin semantics is instrumental in establishing proof-theoretic properties), whereas the intuitionistic systems are entirely proof-theoretic in nature.

In the past, proof systems for \foid have been developed by Hou Ping, Wittocx and Denecker, supporting non-stratified inductive definitions.
More specifically, Hou Ping, Wittocx and Denecker devised a classical sequent calculus for the propositional fragment of \foid \cite{lpnmr/HouWD07}, and Hou Ping and Denecker devised a sequent calculus for an extension of \foid with least fixpoint expressions \cite{lpnmr/HouD09}.
The latter calculus, which we here refer to as HD, differs from our calculus \scfoid in the following ways:
\begin{itemize}
	\item HD is based more directly on the well-founded semantics, as it uses the notion of \emph{unfounded set} \cite{GelderRS91}.
	\scfoid, on the other hand, is based on the principle of mathematical induction, extending the sequent calculus \lkid by Brotherston and Simpson \cite{jlc/BrotherstonS11}.
	\item HD supports an extension of \foid with \emph{least fixpoint expressions} \cite{tcs/Compton94stratified}. %
	Adding least fixpoint expressions to the language of \foid enlarges the class of expressible sets.
	However, one inference rule for these least fixpoint expressions is \emph{infinitary}, meaning that it has infinitely many premises.
	\scfoid does not support least fixpoint expressions, but only has finitary inference rules.
	\item Contrary to HD, \scfoid is sound w.r.t. both~the well-founded and the stable semantics.
	Therefore, \scfoid can additionally be used to prove theorems about (fragments of) logic programs under either semantics.\footnote{
		Besides entire logic programs, \scfoid can can also capture \emph{fragments} of logic programs, leaving the remainder of the program open.
		For example, given a logic program that contains facts about an edge relation $E/2$ and a definition of the transitive closure $T/2$ of $E/2$, \scfoid can capture the definition of $T/2$ independently of the facts about $E/2$.
	}
	The fact that \scfoid is sound w.r.t.~the two semantics demarcates its deductive strength, as it can only prove statements that are valid under both semantics.
	When used to reason about inductive definitions, this constraint is rather minor, however, since the two semantics correspond on all total %
	definitions.
\end{itemize}

\subsection{Contributions}

In previous work, we introduced the sequent calculus \scfoid and proved its soundness w.r.t.~the well-founded semantics \cite{lpnmr/VandenEedeVD24}.
This paper extends our previous work by adding:
\begin{itemize}
	\item a more detailed discussion on the intuitions behind \scfoid and additional examples;
	\item an explicit link between \scfoid and the stable semantics;
	\item several proof-theoretic results, including:
	\begin{itemize}
		\item a correspondence between theorems in \scfoid and theorems in a sequent calculus \scfo for classical first-order logic \fo, in which definitions are replaced by an (infinite) approximating \fo-schema;
		\item soundness w.r.t.~the stable semantics;
		\item completeness w.r.t.~the stable semantics for a propositional fragment of \foid and completeness w.r.t.~the weaker \emph{Henkin semantics} for a first-order fragment of \foid;
		\item cut-elimination for a fragment of \foid consisting of positive definitions, and a weakened version of cut-elimination for a fragment consisting of stratified definitions. 
		As shown by multiple counterexamples, cut-elimination does not hold in general for \foid.
	\end{itemize}
\end{itemize}

\subsection{Overview}

The paper is structured as follows.
In Section \ref{sec:foid}, we introduce the syntax and the well-founded semantics for \foid.
In Section \ref{sec:scfoid}, we present the sequent calculus \scfoid and prove a number of propositions about its inference rules for inductive definitions.
In Section \ref{sec:correspondence}, we establish the aforementioned correspondence between theorems of \scfoid and theorems of \scfo.
In Section \ref{sec:alternative-semantics}, we define the stable semantics and the Henkin semantics for \foid.
In Sections \ref{sec:soundness}, \ref{sec:completeness} and \ref{sec:cut-elimination}, we prove the aforementioned theoretical results on soundness, completeness and cut-elimination, respectively.
Finally, we provide a summary and possible directions for future research in Section \ref{sec:conclusion}.
Technical proofs of several lemmas and propositions are included in the appendix.

%% file: foid.tex
\section{\foid} \label{sec:foid}

Large parts of this section appear in earlier papers \cite{KR/DeneckerV14,tocl/DeneckerT08,tocl/DeneckerBM01,lpnmr/VandenEedeVD24}.

\subsection{Syntax} \label{sec:syntax}

The syntax of \foid extends the syntax of classical first-order logic (\fo), which is defined as usual. 
A \emph{vocabulary} is a set of non-logical symbols, subdivided into predicate symbols and function symbols. 
We write $\sigma/n$ to indicate that a predicate or function symbol $\sigma$ has arity $n$, and $\arity{\sigma}$ to refer to the arity of $\sigma$. 
A \emph{propositional symbol} is a $0$-ary predicate symbol, and an \emph{object symbol} is a $0$-ary function symbol.
We will often refer to an object symbol as a \emph{variable} or a \emph{constant symbol} when appropriate (e.g., an object symbol bound by a quantifier classifies as a variable).
\emph{Terms} are built from function symbols.
We write $\bar{t}$ for a tuple of terms $(t_1, \dots, t_n)$.
\emph{Atoms} are expressions of the form $P(\bar{t})$, $t=s$, $\top$ or $\bot$, where $P/n$ is a predicate symbol, $\bar{t}$ an $n$-tuple of terms, $t$ and $s$ terms, and $\top$ and $\bot$ are the Boolean constants representing truth and falsity, respectively.
(\fo-)\emph{formulas} are built from atoms with propositional connectives $\lnot, \land, \lor, \Rightarrow, \Leftrightarrow$ and first-order quantifiers $\forall, \exists$. 
An (\fo-)\emph{expression} is a term or an \fo-formula.
If all non-logical symbols in an expression $\expr$ occur in a vocabulary $\vocab$, we call $\expr$ an expression (term resp.~formula) \emph{over} $\vocab$.
We write $t \neq s$ for $\lnot (t=s)$.
Given tuples of terms $\bar{t} = (t_1, \dots, t_n)$ and $\bar{s} = (s_1, \dots, s_n)$, we write $\bar{t} = \bar{s}$ for $t_1 = s_1 \land \dots \land t_n = s_n$.
Given a tuple of object symbols $\bar{x} = (x_1, \dots, x_n)$, we write $\forall \bar{x}: \formula$ for $\forall x_1: \forall x_2: \dots \forall x_n: \formula$ and $\exists \bar{x}: \formula$ for $\exists x_1: \forall x_2: \dots \exists x_n: \formula$.
An occurrence of a object symbol $x$ in a formula $\formula$ is said to be \emph{bound} if it is inside a quantified subformula $\exists x: \formulatwo$ or $\forall x: \formulatwo$ of $\formula$, and is said to be \emph{free} otherwise.
Note that we do not explicitly distinguish between variables and constants.
Such distinction can be made implicitly by viewing free occurrences of objects symbols in a formula as constants and bound occurrences as variables.
The set of freely occurring object symbols in a formula $\formula$ is denoted by $\free(\formula)$. 
For a set of formulas $\Gamma$, we write $\free(\Gamma)$ for $\cup_{\formula \in \Gamma} \free(\formula)$. 
An occurrence of a subformula $\formulatwo$ in a formula $\formula$ is said to be \emph{positive} if it occurs in the scope of an even number of negation symbols $\lnot$ (after unwinding all biconditionals $\Leftrightarrow$ and material implications $\Rightarrow$ in terms of $\land$, $\lor$ and $\lnot$), and is said to be \emph{negative} otherwise.
We use the symbol $\doteq$ for syntactic equality.

Given terms $t$, $s$ and an object symbol $x$, we denote by $s[t/x]$ the term obtained from $s$ by replacing all occurrences of $x$ in $s$ with $t$.
Given a formula $\formula$, a term $t$ and an object symbol $x$, we denote by $\formula[t/x]$ the formula obtained from $\formula$ by replacing all free occurrences of $x$ in $\formula$ with $t$.
We refer to $[t/x]$ as a \emph{substitution}.
Given an expression $\expr$, a tuple of object symbols $\bar{x} = (x_1, \dots, x_n)$ and a tuple of terms $\bar{t} = (t_1, \dots, t_n)$, we write $\expr[\bar{t}/\bar{x}]$ and $\expr[t_1/x_1, \dots, t_n/x_n]$ for $\expr[t_1/x_1] \dots [t_n/x_n]$.
Given a set of formulas $\Gamma$, we write $\Gamma[t/x]$ for the set $\cup_{\formula \in \Gamma} \formula[t/x]$.
Substitution in formulas is a subtle matter, as it does not always preserve truth.
For instance, let $\formula$ be the formula $\exists y: x \neq y$.
This formula clearly has a model, %
whereas the formula $\formula[y/x] \doteq \exists y: y \neq y$ does not.
The problem with this substitution is that the substituted term $y$ gets bound by the quantification.
As such, the substitution introduces a \emph{co-reference} between occurrences of terms that was not present prior to the substitution.
Indeed, the two occurrences of $y$ in the subformula $y \neq y$ of $\formula[y/x]$ are \emph{co-referent}, in the sense that they always refer to the same object, whereas the occurrence of $x$ and the occurrence of $y$ are not co-referent in the subformula $x \neq y$ of $\formula$.
The workaround to this issue is simple: it suffices to rename the variable $y$ before the substitution.
We will silently assume that variables are renamed prior to substitution whenever necessary to avoid the creation of new bound occurrences of variables.

\begin{definition} \label{def:definitions}
	A \emph{definitional rule} (or \emph{rule} for short) is an expression of the form %
	\begin{equation*}
		\forall \bar{x} : P(\bar{t}) \rul \formula
	\end{equation*}
	where $\bar{x}$ is a tuple of object symbols, $P$ a predicate symbol, $\bar{t}$ an $\arity{P}$-tuple of terms, and $\formula$ an \fo-formula. 
	The atom $P(\bar{t})$ is called the \emph{head} and the formula $\formula$ the \emph{body} of the definitional rule. 
	An object symbol $y$ is said to \emph{occur freely} in a definitional rule $\defrul$ if it occurs freely in the head $P(\bar{t})$ or the body $\formula$ of the rule, and is not among $\bar{x}$.
	An (\foid-)\emph{definition} is a finite set $\defn$ of definitional rules such that for any definitional rule $\defrul$ in $\defn$, no object symbol in $\bar{x}$ occurs freely in any (other) definitional rule in $\defn$.
	A predicate symbol $P$ that occurs in the head of a definitional rule of a definition $\defn$ is said to be a \emph{defined predicate} of $\defn$.
	A \emph{defined atom} is an atom $P(\bar{t})$ where $P$ is a defined predicate of $\defn$.
	All other non-logical symbols occurring in $\defn$ are called \emph{parameters} of $\defn$. 
	The set of defined predicates of $\defn$ is denoted by $\defi{\defn}$, the set of parameters of $\defn$ by $\pars{\defn}$, and their union by $\sym{\defn}$. 
	We say that $\defn$ defines its defined predicates in terms of its parameters. 
	An object symbol $y$ is said to \emph{occur freely} in $\defn$ if it occurs freely in a definitional rule of $\defn$.
	We denote the set of all freely occurring object symbols in $\defn$ by $\free(\defn)$.
	\hfill\markfull
\end{definition}

Inductive definitions define a concepts by specifying how to construct their values from the values of their parameter concepts. 
Intuitively, the construction proceeds by iteratively applying the rules of the definition.\footnote{
	Martin-L\"of refers to definitional rules as \emph{productions}, emphasizing the constructive nature of inductive definitions.
} 
Concretely, an instance of a definitional rule $\defrul$ becomes applicable during the construction process when its body is derived to (safely) hold in the relation being constructed. 
Applying the rule amounts to extending the value of $P$ with the instance of $\bar{t}$.  

The connective $\rul$ is referred to as \emph{definitional implication}.
We will see that the presence of a definitional rule in a definition logically entails the corresponding material implication. 
Yet, the two connectives $\rul$ and $\Rightarrow$ should not be confused.  
An important difference is that definitional rules are not \emph{truth-conditional}, meaning that they are not assigned truth values in structures. 
Definitions, on the other hand, \emph{are} truth-conditional in \foid, as will be clarified in Section \ref{sec:well-founded-semantics}.
Intuitively, a definition $\defn$ is satisfied in a structure $\structure$ if $\structure$ interprets the defined predicates of $\defn$ as specified by the rules of $\defn$. 

\emph{\foid-formulas} built from atoms and definitions with propositional connectives and first-order quantifiers.
Hence, definitions are an explicit part of the syntax of \foid, and they can occur in complex formulas.

In the remainder of this section, we illustrate the use of \foid's definition construct to express a representative set of mathematical inductive definitions. 
The discussed examples show the structural correspondence between informal (natural language) definitions and their formalizations in \foid. %
Subsequently, in Section \ref{sec:well-founded-semantics}, we recall the formal semantics that formalize the meaning of these mathematical definitions. 

\begin{example} \label{ex:def-reachable}
	Let there be a universe of nodes, let $G$ be a graph on them and let $s$ be a particular node.  
	A node $t$ is said to be \emph{reachable} (from $s$, through $G$) if there exists 
	a path from $s$ to $t$ along the edges of $G$.
	We can define the set of reachable nodes from $s$ by means of an \foid-definition $\defn$ of a predicate $\mathit{Reachable}/1$, in terms of a predicate $\mathit{Edge}/2$ representing the edges of $G$, and a constant $s$ representing the homonymous node:
	\[
	\defin{
		\mathit{Reachable}(s) \rul \top \\
		\forall x,y: \mathit{Reachable}(y) \rul \mathit{Reachable}(x) \land \mathit{Edge}(x,y)
	}
	\]
	The \foid-formula $\exists s: \defn \land \forall t: \mathit{Reachable}(t)$ expresses that $G$ has a node with a path to every other node in $G$, and $\forall s: \defn \land \forall t: \mathit{Reachable}(t)$ expresses that every node in $G$ has a path to every other node in $G$, i.e., that $G$ is \emph{strongly connected}.
	\hfill\markfull
\end{example}

Below is an example of a definition by \emph{simultaneous} or \emph{mutual induction}.  

\begin{example} \label{ex:def-even-odd}
	The sets of even and odd numbers can be mutually defined by an \foid-definition of a predicate $\mathit{Even}/2$ and a predicate $\mathit{Odd}/2$, in terms of the successor function $\mathit{succ}/1$ and the constant $\mathit{zero}$:
	\[
	\defin{
	\mathit{Even}(\mathit{zero}) \rul \top \\
	\forall n: \mathit{Even}(\mathit{succ}(n)) \rul \mathit{Odd}(n) \\
	\forall n: \mathit{Odd}(\mathit{succ}(n)) \rul \mathit{Even}(n) 
	}
	\]
\end{example}

Several logics \cite{ajm/Post43,Spector61,Moschovakis74,Aczel77,jlc/BrotherstonS11} restrict to \emph{positive} inductive definitions, in which the bodies of definitional rules contain no negation $\lnot$ (and no implications $\Rightarrow$, $\Leftrightarrow$).
Other logics \cite{Kreisel63,Feferman70,MartinLoef71,BuchholzFPS81,tcs/McDowellM00,tpp/MomiglianoT04} support \emph{stratified} inductive definitions, which allow negation, but only in a certain well-behaved manner.

\begin{definition} \label{def:stratified}
	Let $\defn$ be a definition.
	We say that $\defn$ is \emph{positive} if the body $\formula$ of every definitional rule $\defrul$ in $\defn$ is composed of atoms, connected only by conjunctions and disjunctions. 
	A \emph{stratification} of $\defn$ is a function $\ell: \defi{\defn} \to \N$ such that for every definitional rule $\forall \bar{x} : P(\bar{t}) \rul \formula$ of $\defn$:
	\begin{enumerate}[label=(\roman*)]
		\item $\ell(Q) \leq \ell(P)$ for every $Q \in \defi{\defn}$ that occurs in $\formula$;
		\item $\ell(Q) < \ell(P)$ for every $Q \in \defi{\defn}$ that occurs negatively in $\formula$.
	\end{enumerate}
	If a stratification of $\defn$ exists, we say that $\defn$ is \emph{stratified}.
\end{definition}

\begin{example} \label{ex:def-unreachable}
	The following definition is stratified:
	\[
	\defin{
		\mathit{Reachable}(s) \rul \top \\
		\forall x,y: \mathit{Reachable}(y) \rul \mathit{Reachable}(x) \land \mathit{Edge}(x,y) \\
		\forall x: \mathit{Unreachable}(x) \rul \lnot \mathit{Reachable}(x)
	}
	\]
	Indeed, $\ell(\mathit{Reachable}) = 0$ and $\ell(\mathit{Unreachable})=1$ specify a stratification for this definition.
\end{example}

A stratified definition can be seen as one that can be split up in a hierarchy of monotone subdefinitions.\footnote{
	The conditions under which a general definition can be split up in subdefinitions were investigated in \cite{tocl/DeneckerT08}.
}
For the definition in Example \ref{ex:def-unreachable}, for instance, this hierarchy consists of the monotone subdefinitions 
\[
\defin{
	\mathit{Reachable}(s) \rul \top \\
	\forall x,y: \mathit{Reachable}(y) \rul \mathit{Reachable}(x) \land \mathit{Edge}(x,y)
}
\]
and 
\[
\defin{
	\forall x: \mathit{Unreachable}(x) \rul \lnot \mathit{Reachable}(x)
} .
\]

Even though stratified definitions are more general than positive definitions, there are many examples of inductive definitions that are not stratified:

\begin{example} \label{ex:def_even}
	While commonly defined through a monotone inductive definition, the set of even natural numbers can also be defined through a non-monotone inductive definition, by structural induction.
	In \foid, this non-monotone inductive definition can be expressed as follows:
	\[ 
	\defin{
		\mathit{Even}(\mathit{zero}) \rul \top \\
		\forall n : \mathit{Even}(\mathit{succ}(n)) \rul \mathit{Nat}(n) \land \lnot \mathit{Even}(n)
	}
	\]
	This definition defines the predicate $\mathit{Even}/1$ in terms of the constant $\mathit{zero}$, the function $\mathit{succ}/1$ and the natural number predicate $\mathit{Nat}/1$.
	There exists no stratification $\ell$ of this definition, as the second definitional rule would impose $\ell(\mathit{Even}) < \ell(\mathit{Even})$.
\end{example}

As will be explained in Section \ref{sec:well-founded-semantics}, an \foid-definition $\defn$ determines an interpretation of the defined predicates $\defi{\defn}$ of $\defn$ based on a \emph{$\defn$-context}, which is a structure interpreting the parameters $\pars{\defn}$ of $\defn$.
For the definition $\defn$ from Example \ref{ex:def_even}, the intended $\defn$-context $\context$ has the standard interpretation for $\mathit{Nat}$, $\mathit{zero}$ and $\mathit{succ}$, i.e., $\context$ interprets $\mathit{Nat}$ as the set of natural numbers $\N = \{0, 1, \dots\}$, $\mathit{zero}$ as the number $0$ and $\mathit{succ}$ as the function $\N \to \N: n \mapsto n+1$.
In this context, $\defn$ assigns the set of even natural numbers $2 \N$ to $\mathit{Even}$, which is its intended interpretation.
In other, non-standard contexts $\context$ (e.g., in which $\context$ interprets $\mathit{succ}$ as a non-injective function), 
$\defn$ does not always provide a well-defined interpretation for $\mathit{Even}$.
In such case, we will say that $\defn$ is \emph{non-total} in $\context$.
In (good) mathematical texts, definitions are only introduced in contexts in which they are total.

\begin{example} \label{ex:sat}
	The satisfaction relation $\models$ for propositional logic is a well-known definition by induction on the structure of formulas. 
	It is non-monotone due to the rule for negation, which is typically formulated as ``$\structure \models \neg \formula$ if $\structure \not\models \formula$''.
	The condition ``$\structure \not\models \formula$'' of this rule is negative, as it says that ``it is not the case that  $\structure \models \formula$''.
	This definition is naturally expressed in a simple multi-sorted variant of the logic \foid, using the sort $\mathit{Structure}$ for the structures and the sort $\mathit{Sentence}$ for sentences:\footnote{
		The presentation of this example is similar to the one in \cite{KR/DeneckerV14}.
	}\footnote{
	We restrict to the connectives $\land$ and $\neg$, to which the other propositional connectives can be reformulated.
	}

	\[ 
	\defin{
		\forall i,p: \mathit{Sat}(i,p) \rul \mathit{In}(p,i)\\
		\forall i,f,g: \mathit{Sat}(i, and(f,g)) \rul \mathit{Sat}(i,f) \land \mathit{Sat}(i,g)\\
		\forall i,f: \mathit{Sat}(i,not(f)) \rul \lnot \mathit{Sat}(i,f)
	}
	\]
	We use the predicate symbol $\mathit{Sat} : \mathit{Structure} \times \mathit{Sentence}$ to represent the satisfaction relation $\models$, and the predicate symbol $\mathit{In}: \mathit{Sentence} \times \mathit{Structure}$ to represent the set of pairs $(p,i)$ such that $p$ is a propositional atom that is true in the structure $i$.
	We also use functions: $\mathit{and} : \mathit{Sentence} \times \mathit{Sentence} \to \mathit{Sentence}$ and $\mathit{not}: \mathit{Sentence} \to \mathit{Sentence}$ to represent the connectives $\land$ and $\neg$.
	This formal definition in \foid preserves the syntactic structure of the mathematical definition of satisfaction, including the monotone recursion in the second rule, and the non-monotone recursion in the third rule.
	Note that this \foid-definition admits no stratification $\ell$, as the third definitional rule would impose that $\ell(\mathit{Sat}) < \ell(\mathit{Sat})$.
	
	A suitable context for this definition is a structure $\context$ where, for some given propositional vocabulary $\tau$, $\mathit{Structure}^\structure$ contains all $\tau$-structures and $\mathit{Sentence}^\structure$ all sentences over $\tau$. 
	Additionally, the predicate $\mathit{In}$ and the functions $\mathit{and}$ and $\mathit{not}$ have appropriate values, e.g., $(p,i) \in \mathit{In}^\context$ iff $p$ is a propositional atom that is true in the structure $i$. 
	Given such a suitable context $\context$, this \foid-definition has a unique model that expands $\context$ in the formal semantics of \foid, which will be recalled in Section \ref{sec:well-founded-semantics}. 
	In this unique model, the value of $\mathit{Sat}/1$ is exactly the satisfaction relation for propositional logic, restricted to $\tau$-structures and formulas over $\tau$.
\end{example}

\begin{example} \label{ex:distance}
	The distance from a node $n$ to a node $m$ in a graph is the length of the shortest path from $n$ to $m$ if such path exists, and is undefined otherwise. 
	This notion can be defined by an \foid-definition of a predicate $\mathit{Distance}/3$, containing all triples $(x,y,n)$ for which $x$ and $y$ are nodes in the graph with distance $n$ from $x$ to $y$, in terms of the constant $\mathit{zero}$, function $\mathit{succ}/1$, predicate $\mathit{Edge}/2$, and order relation $\leq \! \! /2$ (used with infix notation):
	\begin{equation*}
	\defintight{ 
		\forall x: \mathit{Distance}(x,x,\mathit{zero}) \rul \top \\
		\forall x,y,n: \mathit{Distance}(x,y,succ(n)) \rul (\exists z : \mathit{Distance}(x,z,n) \land \mathit{Edge}(z,y)) \land \lnot (\exists m : m \leq n \land \mathit{Distance}(x,y,m)) 
	}
	\end{equation*}
	Since $\mathit{Distance}$ occurs negatively in the body of the second rule, this definition is not stratified.

	A context of this definition should have as domain the union of a set of nodes $N$ and the natural numbers $\N$.
	The relations $\mathit{Edge}$ and $\leq$ should be interpreted as subsets of $N \times N$ and $\N \times \N$, respectively, and $\mathit{zero}$, $\mathit{succ}$ and $\leq$ should have the standard interpretations on $\N$.
\end{example}

\begin{example} \label{ex:temporal_graph}
	In this example, we model an abstract temporal reasoning environment.
	Within a universe of nodes, an action $\mathit{MakeActive}(n,m,t)$ activates an edge from node $n$ to node $m$ at time point $t$. 
	Once activated, the edge remains active as long as there is no path of active edges from a given node $a$ to a given node $b$. 
	If such a path emerges, all edges become inactive again at the next time point. 
	The state of the active edges can be modeled through an \foid-definition of a predicate $\mathit{Active}/3$, together with an auxiliary predicate $\mathit{ActivePath}/3$, in terms of the function $\mathit{succ}/1$ and predicate $\mathit{MakeActive}/3$:
	\[
	\defin{ 
		\forall n,m,t: \mathit{Active}(n,m,\mathit{succ}(t)) \rul \mathit{MakeActive}(n,m,t) \land \lnot \mathit{ActivePath}(a,b,t)\\
		\forall n,m,t: \mathit{Active}(n,m,\mathit{succ}(t)) \rul \mathit{Active}(n,m,t) \land \lnot \mathit{ActivePath}(a,b,t)\\
		\forall n,m,t: \mathit{ActivePath}(n,m,t) \rul \mathit{Active}(n,m,t)\\
		\forall n,m,p,t: \mathit{ActivePath}(n,p,t) \rul \mathit{ActivePath}(n,m,t) \land \mathit{ActivePath}(m,p,t)
	}
	\]
	Time corresponds to the natural numbers.
	The atom $\mathit{Active}(n,m,t)$ says that there is an active edge from node $n$ to node $m$ at time point $t$, and $\mathit{ActivePath}(n,m,t)$ says that there is a path of active edges from node $n$ to node $m$ at time point $t$.
	This definition does not admit a stratification $\ell$, as the first and the third rule would impose $\ell(\mathit{Active}) \leq \ell(\mathit{ActivePath}) < \ell(\mathit{Active})$.
	
	A context of this definition should have as domain the union of a set of nodes $N$ and the natural number $\N$.
	The relation $\mathit{MakeActive}$ should be interpreted as a subsets of $N \times N \times \N$, and $\mathit{succ}$ should have the standard interpretation on $\N$.
	
	Through variations, this abstract example can be converted into real-world examples, such as the creation of electric circuits that burn through once a short circuit has been formed, or the construction of towers that collapse after a threshold height/weight has been exceeded.
	
	This small modeling problem combines two well-known modeling problems in the field of Knowledge Representation: the frame problem and the ramification problem \cite{McCHay69}. 
	The frame problem is to express that a temporal fact retains its truth value until some action changes its state, while the ramification problem is to model how effects propagate through a dynamic system; in this case the effect of creating an active path from $a$ to $b$.  
	Equally interesting is the mathematical modeling technique used here. 
	The above definition merges a non-monotone inductive definition over the well-founded temporal order with a monotone inductive definition of the active path relation, resulting in an iterated inductive definition. 
	The use of such definitions to solve temporal reasoning problems was explored in \cite{DeneckerT07}. 
\end{example}

\begin{example} \label{ex:access_control}
	Consider a system with a multiple users and a file.
	A user has access to the file if they are the owner, or if they are granted access by a user with access while not blocked by a user with access. 
	The delegation of access can be modeled by an \foid-definition of a predicate $\mathit{Access}/1$ in terms of a constant $\mathit{owner}$ and predicates $\mathit{Grants}/2$ and $\mathit{Blocks}/2$:
	\[ 
	\defin{
		\mathit{Access}(\mathit{owner}) \rul \top \\
		\forall u: \mathit{Access}(u) \rul (\exists v : \mathit{Access}(v) \land \mathit{Grants}(v,u)) \land \lnot (\exists v : \mathit{Access}(v) \land \mathit{Blocks}(v,u))
	}
	\]
	This definition is not stratified, as the defined predicate $\mathit{Access}$ occurs negatively in the body of the second rule.
	
	A context of this definition should have as domain a set of users, and should interpret $\mathit{Grants}$ and $\mathit{Blocks}$ as relations that respect a certain hierarchy among the users.
	In other contexts, e.g., in which a user blocks a user who grants them access, the definition may be non-total.
\end{example}

In general, a definition $\defn$ contains multiple rules $\forall \bar{x}: P(\bar{t}) \leftarrow \formula$ per defined predicate $P$. 
These rules can be merged into a single rule $\forall \bar{y}: P(\bar{y}) \leftarrow \mbd{P}{\defn}(\bar{y})$ without altering the (informal) meaning of the definition. 

\begin{definition} \label{def:merged-body}
	A definition $\defn$ is said to be \emph{normal} if, for every defined predicate $P \in \defi{\defn}$, there is precisely one definitional rule in $\defn$ with $P$ in its head, which is of the form $\forall \bar{y} : P(\bar{y}) \rul \formula$.
	The \emph{normalization} $\normalization{\defn}$ of a definition $\defn$ is the definition consisting of one rule $\forall \bar{y} : P(\bar{y}) \rul\mbd{P}{\defn}(\bar{y})$ per defined predicate $P$ of $\defn$, where $\bar{y}$ is a tuple of variables not occurring in $\defn$, and
	$\mbd{P}{\defn}(\bar{y})$ is the disjunction of all formulas $\exists \bar{x} : \bar{y} = \bar{t} \land \formula$ such that $\defrul$ is a rule in $\defn$.\footnote{
		Since the tuples $\bar{y}$ are not unique, \emph{the} normalization of $\defn$ is technically speaking ill-defined.
		However, since the choice of $\bar{y}$ is semantically insignificant, we nevertheless say \emph{the} normalization.
	}
	We refer to $\mbd{P}{\defn}(\bar{y})$ as the \emph{merged body} of $P$ in $\defn$ w.r.t.~$\bar{y}$.
	Given an $\arity{P}$-tuple of terms $\bar{t}$, we write $\mbd{P}{\defn}(\bar{t})$ for $\mbd{P}{\defn}(\bar{y})[\bar{t}/\bar{y}]$, which we call the \emph{merged body} of $P(\bar{t})$.
	
\end{definition}

\begin{example}
	Let $\defn$ be the definition of the even natural numbers from Example \ref{ex:def_even}. 
	The merged body $\mbd{\mathit{Even}}{\defn}(y)$ of $\mathit{Even}$ w.r.t.~$y$ is $y=\mathit{zero} \lor \exists n: y = \mathit{succ}(n) \land Nat(n) \land \lnot \mathit{Even}(n)$.\footnote{
		Strictly speaking, $\mbd{\mathit{Even}}{\defn}(y)$ is equal to $(y=\mathit{zero} \land \top) \lor \exists n: y = \mathit{succ}(n) \land \lnot \mathit{Even}(n)$. However, we sometimes remove $\top$ for simplicity.
	}
	Hence, the normalization $\normalization{\defn}$ of $\defn$ is
	\begin{equation*}
		\defin{\forall y: \mathit{Even}(y) \rul y=\mathit{zero} \lor \exists n : y = \mathit{succ}(n) \land \mathit{Nat}(n) \land \lnot \mathit{Even}(n)}.
	\end{equation*}
\end{example}

\subsection{Well-Founded Semantics} \label{sec:well-founded-semantics}

For monotone inductive definitions, the standard construction process starts from the assumption that all defined facts are false, and then gradually revises some of these assumptions through iterated rule application.
Once saturation is reached, i.e., once there are no more applicable rules, it concludes that all defined atoms that have not been derived to be true must be false.
At each point of the construction process, partial information about the defined atoms is available: atoms that are true at an intermediate stage have reached their defined value, but atoms that are false at an intermediate stage not necessarily; these may still be derived true at a later stage.

The construction process for monotone definitions does not work for non-monotone definitions, since non-monotone rules can only be applied safely when it is certain that the negative subformulas $\neg \formulatwo$ in their bodies are true, i.e., that $\formulatwo$ will not be derived later in the construction process.
Based on the well-founded semantics for logic programs \cite{GelderRS91}, Denecker and Vennekens introduced a new formalization of the construction process \cite{KR/DeneckerV14}.
By re-formalizing the construction process as a sequence of three-valued structures of increasing precision, it was made explicit at each stage of the process whether a defined fact is derived to be true ($\true$), derived to be false ($\false$), or not yet derived, i.e., still unknown ($\unknown$).
As in the monotone case, a fact is derived to be true if the body of one of its rules is true in the current three-valued structure.
Deriving an unknown fact to be (certainly) false is more subtle.
It uses the notion of \emph{unfounded set} \cite{GelderRS91}, which is a set of unknown atoms such that, if all of these atoms are set to false, then the bodies of all the rules deriving these facts become false as well.
Intuitively, the atoms in an unfounded set can only make each other true, and therefore, they can be safely derived to be false.
We we will formalize the aforementioned notions in the remainder of this section.

(Two-valued) structures $\structure$ are defined as usual, consisting of a set $D$ called the \emph{domain} of $\structure$, also denoted by $\dom{\structure}$, and a mapping from non-logical symbols $\sigma$ to appropriate values $\sigma^{\structure}$ in $D$, called the \emph{interpretation} of $\sigma$ in $\structure$. 
In particular, $\sigma^{\structure} \subseteq D^n$ if $\sigma$ is an $n$-ary predicate symbol, and $\sigma^{\structure} : D^n \to D$ if $\sigma$ is an $n$-ary function symbol.
By abuse of notation, we will view the interpretation of a propositional symbol as a Boolean constant, $\false$ or $\true$, %
and the interpretation of an object symbol as an element of $D$.
The set of interpreted symbols of $\structure$ is called the \emph{vocabulary} $\voc{\structure}$ of $\structure$.
If $\voc{\structure} = \vocab$, we cal $\structure$ a \emph{$\vocab$-structure}.
If a vocabulary $\vocab$ contains all non-logical symbols that occur freely in an expression (i.e., a term or a formula) $\epsilon$, we call $\epsilon$ an expression \emph{over} $\vocab$.
We say that $\structure$ \emph{interprets} an expression $\epsilon$ if $\epsilon$ an expression over $\voc{\structure}$.
Given a tuple of terms $\bar{t} = (t_1, \dots, t_n)$ interpreted by a structure $\structure$, we write $\bar{t}^{\structure}$ for $(t_1^{\structure}, \dots, t_n^{\structure})$.
We say that a structure $\structure$  \emph{expands} a structure $\mcj$ if $\dom{\structure} = \dom{\mcj}$, $\voc{\mcj} \subseteq \voc{\structure}$ and $\sigma^{\structure}$ = $\sigma^{\mcj}$ for all $\sigma \in \voc{\mcj}$. 
Given a structure $\structure$ with domain $D$, a tuple of non-logical symbols $\bar{\sigma}$, and a tuple of corresponding values $\bar{\alpha}$ in $D$, we denote by $\structure[\bar{\sigma} : \bar{\alpha}]$ the structure obtained from $\structure$ by interpreting $\bar{\sigma}$ as $\bar{\alpha}$. 
Given a structure $\structure$ and a vocabulary $\vocab \subseteq \voc{\structure}$, the \emph{$\vocab$-reduct} $\structure|_{\Sigma}$ of $\structure$ is the $\vocab$-structure obtained from $\structure$ by restricting its interpretation mapping to $\vocab$.

Denecker and Vennekens formalize the construction processes behind inductive definitions as \emph{well-founded inductions}, which are (possibly transfinite) sequences of increasingly precise \emph{three-valued structures} \cite{KR/DeneckerV14}.
These structures have three different \textit{truth values}: $\false$, $\unknown$ and $\true$, where $\unknown$ stands for `unknown'.
We consider two orders on the set of truth values: the \emph{truth order} $\leqt$ and the \emph{precision order} $\leqp$, 
defined by $\false \leqt \unknown \leqt \true$, and $\unknown \leqp \false$, $\unknown \leqp \true$ respectively. 
Abusing notation, we write $\lnot \false \coloneq \true$, $\lnot \unknown \coloneq \unknown$ and $\lnot \true \coloneq \false$.
A three-valued structure $\tvs$ is defined similarly as a two-valued structure, except that $n$-ary predicate symbols are interpreted as functions from $D^n$ to $\{ \false, \unknown, \true\}$, where $D$ is the domain of $\tvs$. 
We extend the truth and precision order pointwisely to functions with codomain $\{ \false, \unknown, \true\}$, i.e., given $f,g : X \to \{ \false, \unknown, \true\}$, we let $f \leqt g$ iff $f(x) \leqt g(x)$ for all $x \in X$, and similarly for $\leqp$.
Furthermore, we extend these orders to three-valued structures $\tvs, \tvstwo$, by letting $\tvs \leqt \tvstwo$ iff $\tvs$ and $\tvstwo$ have the same domain and vocabulary, the same interpretation for non-predicate symbols, and $P^{\tvs} \leqt P^{\tvstwo}$ for every predicate symbol $P$, and similarly for $\leqp$.
Note that a two-valued structure can be seen as a special case of a three-valued structure, by identifying a subset $S$ of $D^n$ with its characteristic function $D^n \to \{\false, \true\}$, sending every element of $S$ to $\true$ and every non-element to $\false$. 
As such, two-valued structures correspond to maximally precise three-valued structures. 
In the sequel, we will frequently abuse this correspondence, for instance by writing $P^\structure(\bar{a}) = \true$ instead of $\bar{a} \in P^\structure$ for two-valued structures $\structure$.

Given an \fo-formula $\formula$ and a three-valued structure $\tvs$, we define the \emph{truth value} or \emph{truth assignment} $\formula^{\tvs}$ of $\formula$ in $\tvs$ by structural induction on $\formula$, via the following rules:\footnote{
	The defined truth assignment is known as Kleene's truth assignment.
	In \cite{KR/DeneckerV14}, Denecker and Vennekens also allow other truth assignments.
	The only conditions that such an assignment must satisfy are $\leqp$-monotonicity (i.e., $\tvs \leqp \tvstwo$ implies $\formula^{\tvs} \leqp \formula^{\tvstwo}$) and restriction to the standard truth assignment on two-valued structures.
	Another sensible truth assignment, for instance, is \emph{supervaluation}, which defines $\formula^{\tvs}$ as $\mathsf{glb}_{\leqp} \{ \formula^{\structure} \mid \structure \text{ is two-valued and } \tvs \leqp \structure\}$ \cite{jp/vanFraassen66}.
	Supervaluation is more precise than Kleene's truth assignment.
	For instance, for $P^{\tvs} = \unknown$, Kleene's truth assignment of $P \lor \lnot P$ is $\unknown$, while the supervaluation of $P \lor \lnot P$ is $\true$.
	Different truth assignments lead to different notions of \emph{well-founded model}.
	Kripke's truth assignments leads to the \emph{standard} well-founded model, while supervaluation leads to the \emph{ultimate} well-founded model \cite{DeneckerMT04}. 
	We restrict to Kleene's truth assignment for simplicity, but since our proof of soundness w.r.t.\ the well-founded semantics does not rely on the specifics of Kleene's truth assignment, \scfoid is also sound w.r.t.\ other notions of the well-founded semantics corresponding to different truth assignments.
}
\begin{itemize}
	\item $(t=s)^{\tvs} = \true$ if $t^{\tvs} = s^{\tvs}$, and $(t=s)^{\tvs} = \false$ if $t^{\tvs} \neq s^{\tvs}$; %
	\item $P(\bar{t})^{\tvs} = P^{\structure}\!(\bar{t}^{\tvs})$; %
	\item $(\lnot \formula)^{\tvs} = \lnot \formula^{\tvs}$;
	\item $(\formula \land \formulatwo)^{\tvs} = \min_{\leqt}\{\formula^{\tvs}, \formulatwo^{\tvs}\}$;
	\item $(\forall x: \formula)^{\tvs} = \min_{\leqt} \{ \formula^{\tvs[x:a]} \mid a \in \dom{\tvs} \}$.
\end{itemize}
The rules for $\lor$, $\Rightarrow$, $\Leftrightarrow$ and $\exists$ can de derived from the above rules through reformulation in terms of $\lnot$, $\land$ and $\forall$.
The value $t^\tvs$ of a term $t$ in a three-valued structure $\tvs$ is defined as for two-valued structures, via the rule $f(t_1, \dots, t_n)^{\tvs} = f^{\tvs}(t_1^{\tvs}, \dots, t_n^{\tvs})$.
We extend the satisfaction relation $\models$ for \fo to three-valued structures by letting $\tvs \models \formula$ iff $\formula^{\tvs} = \true$. 
It is straightforward to verify that this truth assignment is $\leqp$-monotone, i.e., $\tvs \leqp \tvstwo$ implies $\formula^{\tvs} \leqp \formula^{\tvstwo}$; %
and that if $\tvs$ is two-valued, then $\formula^{\tvs}$ is the standard truth value of $\formula$ in $\tvs$.

Given a definition $\defn$, the interpretation of the defined predicates of $\defn$ generally depends on the interpretation of the parameters of $\defn$.
A structure with vocabulary $\pars{\defn}$ is called a $\defn$-\emph{context}.
For the remainder of the section, we fix a definition $\defn$ and a $\defn$-context $\context$, and come to the notion of \emph{well-founded model} of $\defn$ in $\context$.

A \emph{domain atom} is a pair $(P, \bar{a})$, where $P$ is a predicate symbol and $\bar{a} \in D^{\arity{P}}$. 
With an abuse of notation, we will often write $P(\bar{a})$ instead of $(P, \bar{a})$. 
Furthermore, we write $P(\bar{a})^{\tvs}$ for $P^{\tvs}\!(\bar{a})$, and $\mb{P(\bar{a})}^{\tvs}$ for $\mbd{P}{\defn}(\bar{y})^{\tvs[\bar{y} : \bar{a}]}$ (given that the definition $\defn$ is clear from the context). 
Here $\bar{y}$ is an $\arity{P}$-tuple of variables not occurring %
in $\defn$. 
Given a set of predicate symbols $\Pi$, we denote by $At_D^\Pi$ the set of domain atoms $(P, \bar{a})$ for which $P \in \Pi$.

\begin{definition} \label{def:refinement}
	Let $\tvs$ be a three-valued $\sym{\defn}$-structure with domain $D$. 
	A three-valued $\sym{\defn}$-structure $\tvs'$ is said to be a \emph{$\defn$-refinement} of $\tvs$ if there exists a non-empty set $U \subseteq At_{D}^{\defi{\defn}}$ such that $A^{\tvs} = \unknown$ for all $A \in U$, and either
	\begin{itemize}
		\item $\tvs' = \tvs[U : \true]$ and for all $A \in U$, $\mb{A}^{\tvs} = \true$, or
		\item $\tvs' = \tvs[U : \false]$ and for all $A \in U$, $\mb{A}^{\tvs'} = \false$.
	\end{itemize}
	Here $\tvs[U : \true]$ refers to the structure identical to $\tvs$, except that its interpretation of $P$ maps $\bar{a}$ to $\true$ for all $(P,\bar{a}) \in U$, and similarly for $\tvs[U : \false]$.
	\hfill\markfull
\end{definition}

Note the asymmetry between deriving truth and falsity of defined atoms: the truth of defined atoms can only be derived if the corresponding bodies are true in $\tvs$, while the falsity of defined atoms can be derived if this implies the falsity of the corresponding bodies in $\tvs'$. 
The second kind of refinement is based on the notion of \emph{unfounded set} in the well-founded semantics for logic programs \cite{GelderRS91}.

\begin{definition} \label{def:wf-model}
	A \emph{well-founded induction} of $\defn$ in $\context$ is a sequence $\langle \tvs_i \rangle_{0 \leq i \leq \beta}$ of three-valued $\sym{\defn}$-structures extending $\context$, such that:
	\begin{itemize}
		\item $P^{\tvs_0}(\bar{a}) = \unknown$ for all $P \in \defi{\defn}$ and all $\bar{a} \in \dom{\tvs_0}^{\arity{P}}$;
		\item $\tvs_{i+1}$ is a $\defn$-refinement of $\tvs_i$ for every ordinal $0 \leq i < \beta$; and 
		\item $\tvs_\lambda$ is the $\leqp$-limit of $\langle \tvs_i \rangle_{i < \lambda}$ for every limit ordinal $0 < \lambda \leq \beta$. 
	\end{itemize}
	A well-founded induction $\langle \tvs_i \rangle_{0 \leq i \leq \beta}$ is called \emph{terminal} if its \emph{limit} $\tvs_\beta$ has no $\defn$-refinement. 
	The \emph{well-founded model} of $\defn$ in $\context$ is the limit of any terminal well-founded induction of $\defn$ in $\context$. 
	If this well-founded model is (not) two-valued, we say that $\defn$ is \emph{(non-)total} in $\context$. 
	If $\defn$ is total in any $\defn$-context, we say that $\defn$ is \emph{total}.
	\hfill\markfull
\end{definition}

\begin{remark} \label{rem:props-wf-inductions}
	Note that well-founded inductions are $\leqp$-increasing, which guarantees that the $\leqp$-limits in the third rule of Definition \ref{def:wf-model} are well-defined. 
	Since the truth function is $\leqp$-monotone, well-founded inductions $\langle \tvs_i \rangle_{0 \leq i \leq \beta}$ have the property that once $\formula^{\tvs_i} \in \{\false, \true\}$ for some $i$, then $\formula^{\tvs_i} = \formula^{\tvs_j}$ for all $j \geq i$.
\end{remark}

The following result by Denecker and Vennekens \cite{lpnmr/DeneckerV07} guarantees that the notion of well-founded model of a definition in a context is well-defined.

\begin{theorem}%
	Any definition has a terminal well-founded induction in any context. 
	Furthermore, all terminal well-founded inductions of a given definition in a given context have the same limit.
	\hfill\markfull
\end{theorem}

Non-totality of $\defn$ in $\context$ means that the well-founded inductions of $\defn$ in $\context$ leave at least one defined atom unknown. 
This usually indicates some sort of flaw in the definition or in the interpretation of its parameters.

\begin{example}
	Let $\defn$ be the definition of the even natural numbers from Example \ref{ex:def_even}, and let $\context$ be the $\defn$-context with $ \dom{\defn} = \N$, and with the standard interpretations for $\mathit{zero}$, $\mathit{succ}$ and $\mathit{Nat}$. 
	Recall that $\mbd{\mathit{Even}}{\defn}(y) \doteq (y = \mathit{zero}) \lor \exists n: y = s(n) \land \mathit{Nat}(n) \land \lnot \mathit{Even}(n)$.
	As can be checked, the assignments 
	\begin{align*}
		E^{\tvs_0}(x) &= \unknown \, \text{  for all } x \in \N;\\
		E^{\tvs_1}(0) &= \true, \quad E^{\tvs_1}(x) = \unknown  \, \text{  for  all } x \in \N \setminus \{0\};\\
		E^{\tvs_2}(0) &= \true, \quad E^{\tvs_2}(1) = \false, \quad E^{\tvs_2}(x) = \unknown \, \text{  for all } x \in \N \setminus \{0,1\};\\
		& \; \; \vdots \\
		E^{\tvs_\omega}(x) &= \true \, \text{  for all } x \in 2\N ,\quad E^{\tvs_\omega}(x) = \false \, \text{  for all } x \in 2\N+1
	\end{align*}
	specify a terminal well-founded induction $\langle \tvs_i \rangle_{0 \leq i \leq \omega}$ of $\defn$ in $\context$ (in fact, the only one). 
	Here $\omega$ denotes the first infinite ordinal. Since the well-founded model $\tvs_\omega$ of $\defn$ in $\context$ is two-valued, $\defn$ is total in $\context$. 
	
	The definition $\defn$ is not total in general. 
	Indeed, let $\context'$ be the $\defn$-context defined by $ \dom{\defn} = \{0,1\}$, $ \mathit{Nat}^{\context'} = \{0,1\}$, $\mathit{zero}^{\context'} = 0$, $\mathit{succ}^{\context'}\!(0) = 1$ and $\mathit{succ}^{\context'}\!(1) = 1$.
	As can be checked, the assignments 
	\begin{align*}
		E^{\tvs_0}(0) &= \unknown, \quad E^{\tvs_0}(1) = \unknown;\\
		E^{\tvs_1}(0) &= \true, \quad E^{\tvs_1}(1) = \unknown
	\end{align*}
	specify a terminal well-founded induction $\langle \tvs_i \rangle_{0 \leq i \leq 1}$ of $\defn$ in $\context'$ (in fact, the only one). 
	Indeed, $\tvs_1$ has no $\defn$-refinements, since $\mbd{\mathit{Even}}{\defn}(1)^{\tvs_1} = \false$ and $\mbd{\mathit{Even}}{\defn}(1)^{\tvs_1[\mathit{Even}(1):\false]} = \true$.
	Hence, the well-founded model of $\defn$ in $\context'$ is strictly three-valued, and therefore, $\defn$ is non-total in $\context'$.
\end{example}

The \emph{well-founded satisfaction relation} $\modelswf$ is a binary relation between two-valued structures $\structure$ and \foid-formulas $\formula$ interpreted by $\structure$, defined via the following rules:
\begin{itemize}
	\item $\structure \modelswf t=s$ if $t^{\structure} = s^{\structure}$;
	\item $\structure \modelswf P(\bar{t})$ if $\bar{t} \in P^{\structure}$;
	\item $\structure \modelswf \neg \formula$ if not $\structure \modelswf \formula$;
	\item $\structure \modelswf \formula \land \formulatwo$ if $\structure \modelswf \formula$ and $\structure \modelswf \formulatwo$;
	\item $\structure \modelswf \forall x: \formula$ if $\structure[x:a] \modelswf \formula$ for all $a \in \dom{\structure}$;
	\item $\structure \modelswf \defn$ if $\structure|_{\sym{\defn}}$ is the well-founded model of $\defn$ in $\structure|_{\pars{\defn}}$.
\end{itemize}
The rules for $\lor$, $\Rightarrow$, $\Leftrightarrow$ and $\exists$ can de derived from the above rules through reformulation in terms of $\lnot$, $\land$ and $\forall$.
An \foid-formula $\formula$ is said to be \emph{valid} (under the well-founded semantics) if $\structure \modelswf \formula$ for all structures $\structure$ interpreting the non-logical symbols of $\formula$.
Two \foid-formulas $\formula$ and $\formulatwo$ are said to be \emph{equivalent} (under the well-founded semantics) if $\formula \Leftrightarrow \formulatwo$ is valid under the well-founded semantics.

\begin{remark} \label{rem:equivalence-normalization-wf}
	Under the well-founded semantics, every \foid-definition $\defn$ is equivalent to its normalization $\normalization{\defn}$, i.e., $\structure \modelswf \defn$ iff $\structure \modelswf \normalization{\defn}$ for all structures $\structure$.
	This is because the well-founded model of a definition in a context is not defined in terms of the bodies of the individual rules, only in terms of the merged bodies $\mbd{P}{\defn}(\bar{y})$.

	\hfill\markfull
\end{remark}

A distinctive feature of \foid is that it provides a declarative account of inductive definitions while acknowledging their constructive nature.
\foid differs from the convention \cite{ajm/Post43,Spector61,Moschovakis74,Aczel77,Kreisel63,Feferman70,MartinLoef71,BuchholzFPS81,jlc/BrotherstonS11,tcs/McDowellM00,tpp/MomiglianoT04} in that it incorporates inductive definitions in its syntax and gives them a truth-conditional treatment.
We claim that the truth-conditional account of inductive definitions in \foid is natural, as it corresponds to the truth-conditional account of non-inductive definitions in \fo by means of biconditionals.
For instance, the concept of sibling can be defined in \fo via the formula $\forall x,y: \mathit{SiblingOf}(x,y) \Leftrightarrow \exists z: \mathit{ParentOf}(z,x) \land \mathit{ParentOf}(z,y)$. %
In \foid, this concept can additionally be defined via the \foid-definition $\{\forall x,y: \mathit{SiblingOf}(x,y) \rul \exists z: \mathit{ParentOf}(z,x) \land \mathit{ParentOf}(z,y)\}$.
In fact, any \fo-formula of the form $\forall \bar{x}: P(\bar{x}) \Leftrightarrow \formula$ where $\formula$ does not contain $P$ is semantically equivalent to the \foid-definition $\{\forall \bar{x}: P(\bar{x}) \rul \formula\}$, as can easily be checked.
An advantage of the explicit treatment of definitions in \foid, is that it allows for the expression of statements about definitions \emph{themselves}, instead of only about defined predicates.
For instance, we can express the equivalence of two inductive definitions $\defn_1$ and $\defn_2$ in \foid as $\defn_1 \Leftrightarrow \defn_2$.
The negation $\neg \defn$ of a definition $\defn$ is sensible as well, as it relates to non-totality.
More concretely, if $\lnot \defn$ is a valid, then $\defn$ has no two-valued models. 
Hence, its well-founded model in any $\defn$-context is strictly three-valued, meaning that $\defn$ is non-total in any $\defn$-context.

We conclude this section by establishing two properties of the well-founded semantics that will lay the basis of the sequent calculus \scfoid for \foid.
The first property states that any model of a definition $\defn$ under the well-founded semantics is also a model of its \emph{Clark completion} \cite{adbt/Clark78}, which is an \fo-theory containing a formula $\forall \bar{y} : P(\bar{y}) \Leftrightarrow \mbd{P}{\defn}(\bar{y})$ for every $P \in \defi{\defn}$.

\begin{proposition} \label{prop:completion-wf}
	Let $\defn$ be a definition, $P$ a defined predicate of $\defn$, and $\bar{y}$ an $\arity{P}$-tuple of variables not occurring %
	in $\defn$.
	Let $\structure$ a (two-valued) structure such that $\structure \modelswf \defn$. 
	Then $\structure \modelswf \forall \bar{y} : P(\bar{y}) \Leftrightarrow \mbd{P}{\defn}(\bar{y})$.
\end{proposition}

\begin{proof}
	Let $\wellfoundedinduction$ be a terminal well-founded induction of $\defn$ in $\context \coloneq \structure|_{\pars{\defn}}$.
	Since $\structure \modelswf \defn$, $\structure$ expands the well-founded model $\tvs_\beta$ of $\defn$ in $\context$.
	Take $\bar{a} \in \dom{\structure}^{\arity{P}}$, and let $A$ denote the domain atom $P(\bar{a})$.
	Suppose that $\structure[\bar{y} : \bar{a}] \modelswf P(\bar{y})$.
	Then the truth of $\mb{A}$ must have been derived in $\wellfoundedinduction$ at a point $i+1$ for some $0 \leq i < \beta$.
	By Definition \ref{def:refinement}, this is only possible if $\mb{A}^{\tvs_i} = \true$.
	By Remark \ref{rem:props-wf-inductions}, $\mb{A}^{\tvs_\beta} = \true$, and hence, $\structure[\bar{y} : \bar{a}] \modelswf \mbd{P}{\defn}(\bar{y})$.
	
	Suppose conversely that $\structure[\bar{y} : \bar{a}] \modelswf \mbd{P}{\defn}(\bar{y})$.
	Then the truth $\mb{A}$ must have been derived at a point $0 \leq i \leq \beta$ in $\wellfoundedinduction$.
	Since $\wellfoundedinduction$ is terminal, the truth $A$ must have been derived at a point $i < j \leq \beta$ in $\wellfoundedinduction$.
	Note that $A$ cannot have been derived false at an earlier point, for that would imply that $\mb{A}$ were derived false at an earlier point, and by Remark \ref{rem:props-wf-inductions}, $\mb{A}^{\tvs_i} = \false$, which is a contradiction.
	Thus, $\structure[\bar{y} : \bar{a}] \modelswf P(\bar{y})$.
	This shows that $\structure \modelswf \forall \bar{y} : P(\bar{y}) \Leftrightarrow \mbd{P}{\defn}(\bar{y})$.
\end{proof}

For monotone definitions, the defined predicates $P$ of $\defn$ can be characterized as the smallest relations that respect the definitional rules of $\defn$, i.e., that satisfy the material implications $\forall \bar{x} :\formula \Rightarrow P(\bar{t})$ corresponding to the definitional rules $\forall \bar{x}: P(\bar{t}) \leftarrow \formula$ of $\defn$.
Formally, any model of a monotone definition $\defn$ satisfies the following second-order axiom:
\begin{equation*} 
	\forall G_{\defi{\defn}}: \Big(  \bigwedge \forall \bar{y} : \formulatwo\subst{G_{\defi{\defn}}}{\defi{\defn}}	\Rightarrow G_Q(\bar{s}) \Big) 
	\Rightarrow \Big( \bigwedge \forall \bar{x} : P(\bar{x}) \Rightarrow G_P(\bar{x}) \Big)
\end{equation*}
Here $G_{\defi{\defn}}$ is a tuple of second-order variables, containing an $\arity{P}$-ary relation variable for every defined predicate $P$ of $\defn$.
The first conjunction ranges over all definitional rules $\forall \bar{y} : Q(\bar{s}) \rul \formulatwo$ in $\defn$, and the second conjunction over all defined predicates $P$ of $\defn$.
We denote by $\formulatwo\subst{G_{\defi{\defn}}}{\defi{\defn}}$ the formula obtained from $\formulatwo$ by replacing every occurrence of a defined atom $R(\bar{u})$ of $\defn$ with $G_R(\bar{u})$.
The above formula is known as the \emph{induction axiom}. %
In words, it says that every tuple of relations $G_{\defi{\defn}}$ that respect the definitional rules of $\defn$ must contain the defined predicates $\defi{\defn}$.
In general, the induction axiom does not hold for non-monotone definitions.
One counterexample is given by the non-monotone definition of $\mathit{Even}$ from Example \ref{ex:def_even}.
For this definition $\defn$, the axiom takes the form
\begin{equation*} 
	\forall G: \Big( G(\mathit{zero}) \land \forall n: \mathit{Nat}(n) \land \neg G(n) \Rightarrow G(\mathit{succ}(n)) \Big) \Rightarrow \Big(  \forall n: \mathit{Even}(n) \Rightarrow G(n) \Big).
\end{equation*}
This formula is not satisfied by the well-founded model $\structure$ of $\defn$ with the standard interpretation for $\mathit{zero}$, $\mathit{succ}$ and $\mathit{Nat}$ (and, hence, for $\mathit{Even}$), as can be seen by letting $G$ be the set $ \{0\} \cup (2 \N +1)$.
Indeed, this set respects the definitional rules of $\defn$, but does not contain the set of even numbers $2 \N$.
Nevertheless, we can extend the induction axiom to non-monotone definitions by replacing only the positive occurrences of defined predicates $R$ with $G_R$ in the bodies $\psi$.
Generalizing further, we can replace $\defi{\defn}$ with any subset $\Pi$ of $\defi{\defn}$.
We also refer to this extended version as the \emph{induction axiom}.

\begin{definition} \label{def:induction-axiom}
	Let $\defn$ be a definition and $\Pi$ a subset of $\defi{\defn}$.
	The \emph{induction axiom} $\indax{\defn, \Pi}$ for $\defn$ and $\Pi$ is the second-order axiom 
	\begin{equation*} %
		\forall G_{\Pi}: \Big(  \bigwedge_{\Pi} \forall \bar{y} : \formulatwo\substpos{G_{\Pi}}{\Pi} \Rightarrow G_Q(\bar{s}) \Big) \Rightarrow \Big( \bigwedge_{P\in\Pi} \forall \bar{x} : P(\bar{x}) \Rightarrow G_P(\bar{x}) \Big) .
	\end{equation*}
	Here $G_{\Pi}$ is a tuple of second-order variables, containing an $\arity{P}$-ary relation variable $G_P$ for every predicate $P \in \Pi$.\footnote{
		Since the second-order variables $G_{\Pi}$ are not fixed, \emph{the} induction axiom for $\defn$ and $\Pi$ is strictly speaking ill-defined.
		However, since the exact choice of variables is semantically irrelevant, we nevertheless speak of \emph{the} induction axiom. %
	}
	The first conjunction ranges over all definitional rules $\forall \bar{y} : Q(\bar{s}) \rul \formulatwo$ in $\defn$ with $Q \in \Pi$. %
	We denote by $\formulatwo\substpos{G_{\Pi}}{\Pi}$ the formula obtained from $\formulatwo$ by replacing every positive occurrence of a defined atom $R(\bar{u})$ of $\defn$ for which $R \in \Pi$ with $G_R(\bar{u})$.
\end{definition}

\begin{example}
	Let $\defn$ be the definition of $\mathit{Even}$ %
	from Example \ref{ex:def_even} and let $\Pi$ be $\{ \mathit{Even} \}$.
	Then $\indax{\defn, \Pi}$ is equal to
	\begin{equation*}
		\forall G: \Big( G(\mathit{zero}) \land \forall n: \mathit{Nat}(n) \land \neg \mathit{Even}(n) \Rightarrow G(\mathit{succ}(n)) \Big) \Rightarrow \Big(  \forall n: \mathit{Even}(n) \Rightarrow G(n) \Big).
	\end{equation*}
\end{example}

\begin{restatable}{proposition}{inductionschemewf}\label{prop:induction-scheme-wf}
	Let $\defn$ be a definition, $\Pi \subseteq \defi{\defn}$, and $\structure$ a structure.
	If $\structure \modelswf \defn$, then $\structure \models \indax{\defn, \Pi}$. 
\end{restatable}

\begin{proof}
	See Appendix \ref{app:induction-scheme-wf}.
\end{proof}

%% file: lfoid.tex
\section{A Sequent Calculus for \foid} \label{sec:scfoid}

To realize formal proofs for non-monotone inductive definitions, we introduce the \emph{sequent calculus} \scfoid for \foid.
A sequent calculus \seqcal for a logic $\logic$ consists of a number of inference rules, operating on \emph{sequents}.
An ($\logic$-)\emph{sequent} is an expression of the form $\sequent$ where $\Gamma$ and $\Delta$ are sets of $\logic$-formulas.\footnote{
	In other presentations, $\Gamma$ and $\Delta$ often take the form of multisets or lists of $\logic$-formulas.
}
Unless stated otherwise, we assume $\Gamma$ and $\Delta$ to be finite.
By convention, an expression of the form $\Gamma, \formula$ in a sequent refers to $\Gamma \cup \{\formula\}$.
A sequent $\sequent$ says that if all formulas in $\Gamma$ are true, then some formula in $\Delta$ is true.
Hence, a sequent $\sequent$ is equivalent a formula $\bigwedge \Gamma \Rightarrow \bigvee \Delta$ (if $\Gamma$ and $\Delta$ are finite).
We extend the satisfaction relation $\models_{\logic}$ of $\logic$ to sequents $\sequent$ by letting $\structure \models_{\logic} \sequent$ if $\structure \models_{\logic} \formula$ for a formula $\formula \in \Delta$ or $\structure \not\models_{\logic} \formula$ for a formula $\formula \in \Gamma$. 

An \emph{\seqcal-derivation} of a sequent $\sequent$ from sequents $\seq{\Gamma_1}{\Delta_1}, \dots, \seq{\Gamma_n}{\Delta_n}$ is a finite tree that \emph{derives} $\sequent$ from $\seq{\Gamma_1}{\Delta_1}, \dots, \seq{\Gamma_n}{\Delta_n}$ using the inference rules of \seqcal.
More concretely, it is a tuple $\deriv = (N, E, r, s, i)$ where $(N, E)$ is a finite tree with $N$ is the set of \emph{nodes} and $E \subseteq N \times N$ the set of \emph{edges}, $r \in N$ the \emph{root} of the tree $(N, E)$, $s$ a \emph{labeling function} on the nodes $N$, sending every node to an $\logic$-sequent, and $i$ a partial function on $N$, sending nodes to inferences rules in \seqcal.
For every node $n \in N$ in the domain of $i$, there must exist an instance of the rule $i(n)$ with conclusion $s(n)$ and premises $s(m_1), \dots, s(m_k)$, where $m_1, \dots, m_k$ are the children of $n$ in $(N, E)$.
The nodes $n \in N$ on which $i$ is not defined are called the \emph{buds} of $\deriv$.
If the root $r$ has label $\sequent$ and the buds of $\deriv$ have labels $\seq{\Gamma_1}{\Delta_1}, \dots, \seq{\Gamma_n}{\Delta_n}$, we call $\deriv$ an \emph{\seqcal-derivation} of $\sequent$ from $\seq{\Gamma_1}{\Delta_1}, \dots, \seq{\Gamma_n}{\Delta_n}$.
The \emph{height} of an \seqcal-derivation $\deriv$ is the maximal length of a path in $\deriv$.
An \emph{\seqcal-proof} of a sequent $\sequent$ is an \seqcal-derivation of $\sequent$ from no sequent, i.e., without buds.
We say that a sequent $\sequent$ is \emph{\seqcal-provable} or an \emph{\seqcal-theorem} if there exists an \seqcal-proof of $\sequent$.
By convention, we write \seqcal-derivations in the typical sequent calculus style with horizontal bars representing applications of inference rules, which we label with the name of the rule:
\[
\begin{prooftree}
	\hypo{\seq{\Gamma_1}{\Delta_1}}
	\hypo{\dots}
	\hypo{\seq{\Gamma_n}{\Delta_n}}
	\infer3[(rule)]{\sequent}
\end{prooftree}
\]

\subsection{A Sequent Calculus for \fo}

Figure \ref{fig:scfo} contains the inference rules of the sequent calculus \scfo for \fo.
The symbols $\Gamma$ and $\Delta$ denote sets of \foid-formulas, $\formula$ and $\formulatwo$ denote \foid-formulas, $x$, $y$ and $z$ denote variables, and $t$ and $s$ denote terms.
Some rules are only applicable under certain conditions, written next to them.
The rules of \scfoid are subdivided into \emph{structural} and \emph{logical} rules.
The logical rules in \scfoid consist of \emph{left} and \emph{right introduction rules} for the connectives $\lnot$, $\lor$, $\land$, $\Rightarrow$, the quantifiers $\forall$ and $\exists$, and for equality $=$, introducing these logical symbols in the left resp.~right side of the conclusion.
The left introduction rule for equality ($=$L) intuitively says that given an equality of terms $t=s$, occurrences of $t$ and $s$ may be freely interchanged.
Recall from Section \ref{sec:syntax} that substitutions $[t/x]$ in formulas $\formula[t/x]$ must not introduce new bound occurrences of terms, to avoid the creation of new co-references.
By a slight extension of the rules, we will sometimes apply (subst), ($\forall$L), ($\forall$R), ($\exists$L), ($\exists$R), ($=$L) and ($=$R) directly to tuples of terms $\bar{t}$ and variables $\bar{x}$ instead of to individual terms $t$ and variables $x$.

The structural rules (wk) and (subst) are \emph{admissible} in \scfo \cite{jlc/BrotherstonS11}, meaning that their removal does not affect the set of theorems of the calculus.
In other words, every sequent that is \scfo-provable with (wk) and (subst) is also provable without (wk) and (subst).
More substantial is the admissibility of the cut rule (cut) in \scfo.
This was originally shown by Gentzen for \lk,\footnote{
	Gentzen's proof of cut-elimination for \lk extends to \scfo.
} a result known as his \emph{Hauptsatz} or \emph{cut-elimination theorem} \cite{Gentzen35}.
The cut rule formalizes the introduction of a lemma in a proof, with the \emph{cut formula} $\formula$ formalizing the lemma.
Note that the cut formula $\formula$ appears only in the premises of (cut), not in the conclusion.
This is an obstacle for proof search, as $\formula$ can be any \fo-formula.
However, the cut-elimination theorem guarantees that any sequent that is provable with the cut rule can also be proven without it.\footnote{
	A noteworthy side remark is that eliminating cuts can result in {substantially} larger proofs \cite{Boolos1984}.
}

\begin{figure}[t]
	\setlength{\tabcolsep}{42pt}
	\renewcommand{\arraystretch}{2}
	\centering
	
	\noindent \textbf{Structural rules}
	
	\vsp
	
		\hspace{3mm}
		$
		\begin{prooftree}[center=false]
			\hypo{ }
			\infer1[$\Gamma \cap \Delta \neq \emptyset$, or $\bot \in \Gamma$ or $\top \in \Delta$\, (ax)]{\Gamma \vdash \Delta}
		\end{prooftree}
		$ \hfill $
		\begin{prooftree}[center=false]
			\hypo{\Gamma' \vdash \Delta'}
			\infer1[$\Gamma' \subseteq \Gamma$, $\Delta' \subseteq \Delta$ (wk)]{\Gamma \vdash \Delta}
		\end{prooftree}
		$
		\hspace{3mm}
		
		\vsp
		
		\hspace{3mm}
		$
		\begin{prooftree}[center=false]
			\hypo{\Gamma \vdash \Delta}
			\infer1[(subst)]{\Gamma[t / x] \vdash \Delta[t/x]}
		\end{prooftree}
		$ \hfill $
		\begin{prooftree}[center=false]
			\hypo{\Gamma \vdash \formula, \Delta}
			\hypo{\Gamma, \formula \vdash \Delta}
			\infer2[(cut)]{\Gamma \vdash \Delta}
		\end{prooftree}
		$
		\hspace{3mm}
	
	\vsp
	
	\noindent \textbf{Logical rules}
	
	\vsp
		
		\hspace{3mm}
		$
		\begin{prooftree}[center=false]
			\hypo{\Gamma \vdash \formula, \Delta}
			\infer1[($\lnot$L)]{\Gamma, \lnot \formula \vdash \Delta}
		\end{prooftree}
		$ \hfill $ 
		\begin{prooftree}[center=false]
			\hypo{\Gamma, \formula \vdash \Delta}
			\infer1[($\lnot$R)]{\Gamma \vdash \lnot \formula, \Delta}
		\end{prooftree}
		$
		\hspace{3mm}
		
		\vsp
		
		\hspace{3mm}
		$
		\begin{prooftree}[center=false]
			\hypo{\Gamma, \formula \vdash \Delta}
			\hypo{\Gamma, \formulatwo \vdash \Delta}
			\infer2[($\lor$L)]{\Gamma, \formula \lor \formulatwo \vdash \Delta}
		\end{prooftree}
		$ \hfill $ 
		\begin{prooftree}[center=false]
			\hypo{\Gamma \vdash \formula, \formulatwo, \Delta}
			\infer1[($\lor$R)]{\Gamma \vdash \formula \lor \formulatwo, \Delta}
		\end{prooftree}
		$
		\hspace{3mm}
		
		\vsp
		
		\hspace{3mm}
		$
		\begin{prooftree}[center=false]
			\hypo{\Gamma, \formula, \formulatwo \vdash \Delta}
			\infer1[($\land$L)]{\Gamma, \formula \land \formulatwo \vdash \Delta}
		\end{prooftree}
		$ \hfill $ 
		\begin{prooftree}[center=false]
			\hypo{\Gamma \vdash \formula, \Delta}
			\hypo{\Gamma \vdash \formulatwo, \Delta}
			\infer2[($\land$R)]{\Gamma \vdash \formula \land \formulatwo, \Delta}
		\end{prooftree}
		$
		\hspace{3mm}
		
		\vsp 
		
		\hspace{3mm}
		$
		\begin{prooftree}[center=false]
			\hypo{\Gamma \vdash \formula, \Delta}
			\hypo{\Gamma, \formulatwo \vdash \Delta}
			\infer2[($\Rightarrow$L)]{\Gamma, \formula \Rightarrow \formulatwo \vdash \Delta}
		\end{prooftree}
		$ \hfill $ 
		\begin{prooftree}[center=false]
			\hypo{\Gamma, \formula \vdash \formulatwo, \Delta}
			\infer1[($\Rightarrow$R)]{\Gamma \vdash \formula \Rightarrow \formulatwo, \Delta}
		\end{prooftree}
		$
		\hspace{3mm}
		
		\vsp 
		
		\hspace{3mm}
		$
		\begin{prooftree}[center=false]
			\hypo{\Gamma \vdash \formula, \formulatwo, \Delta}
			\hypo{\Gamma, \formula, \formulatwo \vdash \Delta}
			\infer2[($\Leftrightarrow$L)]{\Gamma, \formula \Leftrightarrow \formulatwo \vdash \Delta}
		\end{prooftree}
		$ \hfill $ 
		\begin{prooftree}[center=false]
			\hypo{\Gamma, \formula \vdash \formulatwo, \Delta}
			\hypo{\Gamma, \formulatwo \vdash \formula, \Delta}
			\infer2[($\Leftrightarrow$R)]{\Gamma \vdash \formula \Leftrightarrow \formulatwo, \Delta}
		\end{prooftree}
		$
		\hspace{3mm}
		
		\vsp 
		
		\hspace{3mm}
		$
		\begin{prooftree}[center=false]
			\hypo{\Gamma, \formula[t/x] \vdash \Delta}
			\infer1[($\forall$L)]{\Gamma, \forall x : \formula \vdash \Delta}
		\end{prooftree}
		$ \hfill $ 
		\begin{prooftree}[center=false]
			\hypo{\Gamma \vdash \formula[z/x], \Delta}
			\infer1[$z \not\in \free(\Gamma \cup \Delta)$\, ($\forall$R)]{\Gamma \vdash \forall x: \formula, \Delta}
		\end{prooftree}
		$
		\hspace{3mm}
		
		\vsp 
		
		\hspace{3mm}
		$
		\begin{prooftree}[center=false]
			\hypo{\Gamma, \formula[z/x] \vdash \Delta}
			\infer1[$z \not\in \free(\Gamma \cup \Delta)$\, ($\exists$L)]{\Gamma, \exists x:  \formula \vdash \Delta}
		\end{prooftree}
		$ \hfill $ 
		\begin{prooftree}[center=false]
			\hypo{\Gamma \vdash \formula[t/x], \Delta}
			\infer1[($\exists$R)]{\Gamma \vdash \exists x: \formula, \Delta}
		\end{prooftree}
		$
		\hspace{3mm}
		
		\vsp
		
		\hspace{3mm}
		$
		\begin{prooftree}[center=false]
			\hypo{\Gamma[s/x, t/y] \vdash \Delta[s/x, t/y]}
			\infer1[($=$L)]{\Gamma[t/x, s/y], t=s \vdash \Delta[t/x, s/y]}
		\end{prooftree}
		$ \hfill $ 
		\begin{prooftree}[center=false]
			\hypo{\phantom{\Gamma \vdash \Delta}}
			\infer1[($=$R)]{\Gamma \vdash t=t, \Delta}
		\end{prooftree}
		$
		\hspace{3mm}

	\caption{
		Inference rules of \scfo, based on \cite{jlc/BrotherstonS11}. \scfo can be seen as an extension of Gentzen’s
		sequent calculus \lk for classical first-order logic \cite{Gentzen35} with introduction rules for equality.
		}
	\label{fig:scfo}
\end{figure}

\subsection{Extension to \foid} %

We extend \scfo to a sequent calculus \scfoid that accommodates \foid-definitions, by providing a left and a right introduction rule for defined atoms. 
The \textbf{right introduction rule for defined atoms} is
\begin{equation*} 
	\begin{prooftree}
		\hypo{\defn, \Gamma \vdash \formula[\bar{s}/\bar{x}], \Delta}
		\infer1[(def R)]{\defn, \Gamma \vdash P(\bar{t}[\bar{s}/\bar{x}]), \Delta}
	\end{prooftree}
\end{equation*}
Here $\defn$ is a definition with a definitional rule $\forall \bar{x}: P(\bar{t}) \rul \formula$ and $\bar{s}$ is a tuple of terms with the same length as $\bar{x}$.
Intuitively, (def R) allows us to derive heads of definitional rules from bodies.

The \textbf{left introduction rule for defined atoms} is
\begin{equation*}
	\begin{prooftree}
		\hypo{\text{minor premises}}
		\hypo{\defn, \Gamma, \ihf{P}{\bar{t}} \vdash \Delta}
		\infer2[(def L)]{\defn, \Gamma, P(\bar{t}) \vdash \Delta}
	\end{prooftree}
\end{equation*}
It involves the selection of a \emph{predicate set} $\Pi$, which is a subset of $\defi{\defn}$ containing $P$, and for every $Q \in \Pi$ an \emph{induction hypothesis} $\ih_Q$ associated with $Q$.
An induction hypothesis $\ih_Q$ associated with $Q$ is a \emph{set expression} $\setexpg$, where $\bar{z}$ is an $\arity{Q}$-tuple of variables and $\formula$ an \fo-formula.
Informally, a set expression $\setexpg$ can be read as the set of tuples $\bar{z}$ satisfying $\formula$.
In an application of (def L), an induction hypothesis $\ih_Q \doteq \setexpg$ should be chosen such that $\setexpg$ describes an upper bound on $Q$ under the conditions $\defn$ and $\Gamma$.
Given an induction hypothesis $\ih_Q \doteq \setexpg$ and an $\arity{Q}$-tuple of terms $\bar{s}$, we write $\ihf{Q}{\bar{s}}$ to denote the formula $\formula[\bar{s}/\bar{z}]$.
For every definitional rule $\forall \bar{y}: Q(\bar{s}) \rul \formulatwo$ in $\defn$ with $Q \in \Pi$, the inference rule (def L) has a \emph{minor premise}
\[\defn, \Gamma, \formulatwo\substposihpi \vdash \ihf{Q}{\bar{s}}, \Delta . \]
Here $\formulatwo\substposihpi$ denotes the formula obtained from $\formulatwo$ by replacing all positive occurrences of atoms $R(\bar{u})$ with $R \in \Pi$ by $\ihf{R}{\bar{u}}$.
This inference rule requires that for all definitional rules $\forall \bar{y}: Q(\bar{s}) \rul \formulatwo$ in $\defn$ with $Q \in \Pi$, no variable in $\bar{y}$ occurs freely in $\defn$, $\Gamma$ or $\Delta$. 
The premise $\defn, \Gamma, \ihf{P}{\bar{t}} \vdash \Delta$ is called the \emph{major premise} of the rule.\footnote{
	We adopted the terms \emph{minor} and \emph{major premise} from Brotherston and Simpson \cite{jlc/BrotherstonS11}.
	The terms originate from Aristotle's syllogisms \cite{Smith89aristotle}.
}

We also call (def L) the \emph{induction rule}, as it is based on the \emph{principle of mathematical induction}, in which the elements of inductively defined sets are shown to satisfy a certain property by finding a suitable induction hypothesis that entails the property and that is preserved under the definitional rules.
The preservation under definitional rules is formalized by the minor premises and the entailment of the property by the major premise.

The introduction rules for defined atoms can be understood as proof-theoretic implementations of semantic properties of the predicates defined by a (non-monotone) definition $\defn$. 
The right introduction rule implements the idea that a defined predicate $P$ is closed under forward application of any of its rules $\defrul$.
Hence, if $\formula[\bar{s}/\bar{x}]$ is entailed, then so is $P(\bar{t}[\bar{s}/\bar{x}])$. 
The left introduction rule is more involved.
In the special case where $\defn$ defines a single predicate $P$ with a single rule $\forall \bar{x}: P(\bar{s}) \rul \formula$, it simplifies to
\[
\begin{prooftree}
	\hypo{\defn, \Gamma, \formula\substposih{P} \vdash \ihf{P}{\bar{s}}, \Delta}
	\hypo{\defn, \Gamma, \ihf{P}{\bar{t}} \vdash \Delta}
	\infer2[(def L)]{\defn, \Gamma, P(\bar{t}) \vdash \Delta}
\end{prooftree}
\]
This rule implements the semantic idea that for any $\arity{P}$-ary relation $S_P$ (represented by a set expression $\ihs{P}$), 
if interpreting the positive occurrences of $P$ in the body $\formula$ by $S_P$ 
yields a subset of $S_P$, then $P$ is a subset of $S_P$ (and hence, $P(\bar{t})$ implies $\ihf{P}{\bar{t}}$).
In fact, this is the defining principle behind the \emph{stable (model) semantics} that we will discuss in Section \ref{sec:stable-semantics}.
The left introduction rule can be used to infer falsity of a defined atom by setting $\Delta = \emptyset$.
The general form of (def L) takes into account that a definition may define multiple predicates with multiple rules.

\begin{example}
	Let $\defn$ be the natural number definition from Section \ref{sec:introduction}.
	Suppose we want to prove that a generic natural number $x$ satisfies a property expressed by an \fo-formula $\formulatwo$. %
	Then we can do so with the left introduction rule for defined atoms, picking $\Gamma = \emptyset$ and $\Delta = \{ \formulatwo \}$:\footnote{
		Strictly speaking, the first premise should be $\defn, \top \vdash \ihf{\mathit{Nat}}{\mathit{zero}}, \formulatwo$.
	}
	\[
	\begin{prooftree}
		\hypo{\defn \vdash \ihf{\mathit{Nat}}{\mathit{zero}}, \formulatwo}
		\hypo{\defn, \ihf{\mathit{Nat}}{n} \vdash \ihf{\mathit{Nat}}{\mathit{succ}(n)}, \formulatwo}
		\hypo{\defn, \ihf{\mathit{Nat}}{x} \vdash \formulatwo}
		\infer3[(def L)]{\defn, \mathit{Nat}(x) \vdash \formulatwo}
	\end{prooftree}
	\]
	Here $\Pi = \{\mathit{Nat}\}$, since $\mathit{Nat}$ is the only defined predicate of $\defn$.
	We need to find a suitable induction hypothesis $\ihs{\mathit{Nat}}$.
	To satisfy the major premise, $\ihf{\mathit{Nat}}{x}$ must entail $\formulatwo$ in all models of $\defn$.
	To satisfy the minor premises, the set represented by $\ihs{\mathit{Nat}}$ must be closed under the definitional rules of $\defn$ (or $\formulatwo$ must hold, but that is precisely what we want to prove).
	In general, $\Gamma$ allows adding any number of assumptions, and $\Delta$ allows replacing $\formulatwo$ by any disjunction of formulas.
	In particular, we can prove that an object is \emph{not} a natural number by picking $\Delta = \emptyset$.
\end{example}

Induction hypotheses $\ih_Q \doteq \setexpg$ in (def L) can be defined by any \fo-formula $\formula$. 
However, only certain choices are suitable to derive the premises of (def L).
Finding such induction hypotheses generally requires creativity, as is typical in proofs by mathematical induction.

In the minor premises of (def L), only the positive occurrences of defined predicates are replaced by induction hypotheses.
Replacing the negative occurrences as well would generate provable minor premises involving induction hypotheses that do not describe upper bounds on the corresponding defined predicates.
This would result in an unsound rule, as illustrated by Example \ref{ex:naive-def-L-counterexample}.

\begin{example} \label{ex:naive-def-L-counterexample}
	Let (naive def L) be the alternative to (def L) that does not discriminate between positive and negative occurrences of defined predicates, i.e., the minor premises are of the form $\defn, \Gamma, \formulatwo\substihpi\vdash \ihf{Q}{\bar{s}}, \Delta$ instead of $\defn, \Gamma, \formulatwo\substposihpi \vdash \ihf{Q}{\bar{s}}, \Delta$.
	Consider the definition
	\[
		\defn = \defin{
			Q \rul \bot \\
			P \rul \lnot Q
		} .
	\]
	Clearly, $\defn$ defines $Q$ to be false and $P$ to be true.
	However, (naive def L) allows us to derive that $P$ is false under $\defn$:
	\[
		\begin{prooftree}
			\infer0[(ax)]{\defn, \bot \vdash \top}
			
			\infer0[(ax)]{\defn \vdash \top, \bot}
			\infer1[($\lnot$L)]{\defn, \lnot \top \vdash \bot}
			
			\infer0[(ax)]{\defn, \bot \vdash}
			
			\infer3[(naive def L)]{\defn, P \vdash}
			\infer1[($\lnot$R)]{\defn \vdash \lnot P}
		\end{prooftree}
	\]
	Here we applied (naive def L) with $\Pi = \{P, Q\}$, $\ihs{P} \doteq \bot$ and $\ihs{Q} \doteq \top$.
	This shows that (naive def L) is not sound.
\end{example}

The sequent calculus \scfoid differs from \lkid in the following aspects:
\begin{itemize}
	\item \lkid does not (need to) differentiate between positive and negative occurrences of defined predicates, as Brotherston and Simpson only consider positive definitions.
	\item Sequents in \lkid do not contain definitions, since definitions are not seen as formulas in the logic of Brotherston and Simpson.
	Instead, definitions are specified separately in their logic and sequent calculus \lkid, and \lkid-provable sequents are valid only in models of these definitions.
	\item In \lkid, the predicate set $\Pi$ in (def L) is fixed as the set of all defined predicates of $\defn$ that are \emph{mutually dependent} with $P$, as will be explained in Section \ref{sec:pi}.
\end{itemize}

The following proposition establishes a deductive equivalence between definitions and their normalizations. %

\begin{restatable}{proposition}{normalizationequiderivable} \label{prop:normalization-equiderivable}
	Let $\Gamma$ and $\Delta$ be sets of \foid-formulas and let $\defn$ be a definition. 
	Let $\normalization{\defn}$ be the normalization of\/ $\defn$.
	Then $\defn, \Gamma \vdash \Delta$ is \scfoid-provable iff\/ $\Gamma, \normalization{\defn} \vdash \Delta$ is \scfoid-provable.
	\hfill\markfull
\end{restatable}

\begin{proof}
	See Appendix \ref{app:proofs-section-scfoid}.
\end{proof}

\subsection{Alternative Inference Rules} \label{sec:alternative-inference-rules}

Recall from Proposition \ref{prop:completion-wf} that under the well-founded semantics, any definition $\defn$ logically entails its Clark completion, saying that $\defn \vdash \forall \bar{y} : P(\bar{y}) \Leftrightarrow \mbd{P}{\defn}(\bar{y})$ for every defined predicate $P$ of $\defn$.
Based on this property, one may have expected the following left introduction rule for defined atoms:
\begin{equation*} %
	\begin{prooftree}
		\hypo{\defn, \Gamma, \mbd{P}{\defn}(\bar{t}) \vdash \Delta}
		\infer1[\clarkrul]{\defn, \Gamma, P(\bar{t}) \vdash \Delta}
	\end{prooftree}
\end{equation*}
This rule is \emph{derivable} in \scfoid, meaning that for any instance of the rule, there exists an \scfoid-derivation of the conclusion from the premises of the instance.\footnote{
	Clearly, any derivable rule is admissible, but the converse does not hold in general.
	The rules (wk) and (subst), for instance, are admissible but not derivable in \scfo.
}
The fact that \clarkrul is derivable in \scfoid follows from the fact that $\defn \vdash \forall \bar{y} : P(\bar{y}) \Leftrightarrow \mbd{P}{\defn}(\bar{y})$ is provable in \scfoid for any definition $\defn$ and any defined predicate $P$ of $\defn$.

\begin{restatable}{lemma}{safereplacementofpositiveandnegativeoccurrences} \label{lem:safe-replacement-of-positive-and-negative-occurrences}
	Let $\formula$, $\formulatwo$ and $\formulathree$ be \fo-formulas and $\Gamma$ and $\Delta$ finite sets of \foid-formulas. 
	Let $\formulathree\ssubstposg$ denote a formula obtained from $\formulathree$ by replacing positive occurrences of $\formula$ (not necessarily all or any) with $\formulatwo$, and let $\formulathree\ssubstnegg$ denote a formula obtained from $\formulathree$ by replacing negative occurrences of $\formula$ (again, not necessarily all or any) with $\formulatwo$.
	If $\Gamma, \formula \vdash \formulatwo, \Delta$ is \scfoid-provable, then so are $\Gamma, \formulathree \vdash \formulathree\ssubstposg, \Delta$ and $\Gamma, \formulathree\ssubstnegg \vdash \formulathree, \Delta$.
\end{restatable}

\begin{proof}
	See Appendix \ref{app:proofs-section-scfoid}.
\end{proof}

\begin{lemma} \label{lem:proof-completion-sequent}
	Let $\defn$ be a definition, $P \in \defi{\defn}$ and $\bar{y}$ an $\arity{P}$-tuple of variables, none of which occurs in $\defn$.
	There exists an \scfoid-proof of\/ $\defn \vdash \forall \bar{y} : P(\bar{y}) \Leftrightarrow \mbd{P}{\defn}(\bar{y})$.
\end{lemma}

\begin{proof}
	Consider the following \scfoid-derivation of $\defn \vdash \forall \bar{y} : P(\bar{y}) \Leftrightarrow \mbd{P}{\defn}(\bar{y})$. %
	\begin{equation*}
		\begin{prooftree}
			\hypo{\defn, P(\bar{y}) \vdash \mbd{P}{\defn}(\bar{y})}
			\hypo{\defn, \mbd{P}{\defn}(\bar{y}) \vdash P(\bar{y})}
			\infer2[($\Leftrightarrow$R)]{\defn \vdash P(\bar{y}) \Leftrightarrow \mbd{P}{\defn}(\bar{y})}
			\infer1[($\forall$R)]{\defn \vdash \forall \bar{y} : P(\bar{y}) \Leftrightarrow \mbd{P}{\defn}(\bar{y})}
		\end{prooftree}
	\end{equation*}
	It shows that it suffices to construct \scfoid-proofs of the sequents $\defn, \mbd{P}{\defn}(\bar{y}) \vdash P(\bar{y})$ and $\defn, P(\bar{y}) \vdash \mbd{P}{\defn}(\bar{y})$.
	
	Recall that $\mbd{P}{\defn}(\bar{y})$ is the disjunction of formulas $\exists \bar{x}:  \bar{y} = \bar{t} \land \formula$ such that $\forall \bar{x} : P(\bar{t}) \rul \formula$ is a rule in $\defn$.
	Fix such a rule $\forall \bar{x} : P(\bar{t}) \rul \formula$ and consider the following \scfoid-proof:
	\begin{equation*}
		\begin{prooftree}
			\infer0[(ax)]{\defn, \formula \vdash \formula}
			\infer1[(def R)]{ \defn, \formula \vdash P(\bar{t})}
			\infer1[($=$L)]{ \defn, \bar{y} = \bar{t}, \formula \vdash P(\bar{y})}
			\infer1[($\land$L)]{\defn, \bar{y} = \bar{t} \land \formula \vdash P(\bar{y})}
			\infer1[($\exists$L)]{ \defn, \exists \bar{x}:  \bar{y} = \bar{t} \land \formula \vdash P(\bar{y})}
		\end{prooftree}
	\end{equation*}
	In the application of ($\exists$L), we used that none of the variables $\bar{x}$ occurs freely in $\defn$ or $P(\bar{y})$, which follows from Definition \ref{def:definitions} resp.~from the assumption on $\bar{y}$.
	Combining these proofs for all rules $\forall \bar{x} : P(\bar{t}) \rul \formula$ in $\defn$ by means of ($\lor$L), we obtain a proof of $\defn, \mbd{P}{\defn}(\bar{y}) \vdash P(\bar{y})$.

	To construct an \scfoid-proof of $\defn, P(\bar{y}) \vdash \mbd{P}{\defn}(\bar{y})$, we apply (def L) with $\Pi = \{P\}$ and $\ihs{P} \doteq \setexp{\bar{z}}{P(\bar{z}) \land \mbd{P}{\defn}(\bar{z})}$: 
	\begin{equation*}
		\begin{prooftree}
			\hypo{\text{minor premises}}
			\infer0[(ax)]{\defn, P(\bar{y}), \mbd{P}{\defn}(\bar{y}) \vdash \mbd{P}{\defn}(\bar{y})}
			\infer1[($\land$L)]{\defn, P(\bar{y}) \land \mbd{P}{\defn}(\bar{y}) \vdash \mbd{P}{\defn}(\bar{y})}
			\infer2[(def L)]{\defn, P(\bar{y}) \vdash \mbd{P}{\defn}(\bar{y})}
		\end{prooftree}
	\end{equation*}
	For every rule $\defrul$ in $\defn$ defining $P$, the application of (def L) has a minor premise $\defn, \formula\substposih{P} \vdash \ihf{P}{\bar{t}}$.
	Fix a rule $\defrul$ in $\defn$ and consider the following \scfoid-derivation of $\defn, \formula\substposih{P} \vdash \ihf{P}{\bar{t}}$:
	\begin{equation*}
		\begin{prooftree}
			\hypo{\defn, \formula\substposih{P} \vdash \mbd{P}{\defn}(\bar{t})}
			\infer1[(wk)]{\defn, \formula\substposih{P} \vdash \mbd{P}{\defn}(\bar{t}), P(\bar{t})}
			\hypo{\defn, \mbd{P}{\defn}(\bar{t}) \vdash P(\bar{t})}
			\infer1[(wk)]{\defn, \formula\substposih{P}, \mbd{P}{\defn}(\bar{t}) \vdash P(\bar{t})}
			\infer2[(cut)]{\defn, \formula\substposih{P} \vdash P(\bar{t})}
			\hypo{\defn, \formula\substposih{P} \vdash \mbd{P}{\defn}(\bar{t})}
			\infer2[($\land$R)]{\defn, \formula\substposih{P} \vdash \ihf{P}{\bar{t}}}
		\end{prooftree}
	\end{equation*}
	Since we already constructed an \scfoid-proof of $\defn, \mbd{P}{\defn}(\bar{t}) \vdash P(\bar{t})$, it suffices to construct an \scfoid-proof of $\defn, \formula\substposih{P} \vdash \mbd{P}{\defn}(\bar{t})$.
	The following \scfoid-derivation of $\defn, \formula\substposih{P} \vdash \mbd{P}{\defn}(\bar{t})$ shows that it suffices to construct an \scfoid-proof of $\defn, \formula\substposih{P} \vdash \formula$:
	\begin{equation*}
		\begin{prooftree}
			\hypo{\defn, \formula\substposih{P} \vdash \formula}
			\infer1[(wk)]{\defn, \formula\substposih{P} \vdash \formula, \mbd{P}{\defn}(\bar{t})}
			\infer0[($=$R)]{\defn, \formula \vdash \bar{t} = \bar{t}}
			\infer0[(ax)]{\defn, \formula \vdash \formula}
			\infer2[($\land$R)]{\defn, \formula \vdash \bar{t} = \bar{t} \land \formula}
			\infer1[($\exists$R)]{\defn, \formula \vdash \exists \bar{x} : \bar{t} = \bar{t} \land \formula}
			\infer1[($\lor$R), (wk)]{\defn, \formula \vdash \mbd{P}{\defn}(\bar{t})}
			\infer1[(wk)]{\defn, \formula\substposih{P}, \formula \vdash \mbd{P}{\defn}(\bar{t})}
			\infer2[(cut)]{\defn, \formula\substposih{P} \vdash \mbd{P}{\defn}(\bar{t})}
		\end{prooftree}
	\end{equation*}
	In the (repeated) application of ($\lor$R) and (wk), we used that $\exists \bar{x} : \bar{t} = \bar{t} \land \formula$ is one of the disjuncts in $\mbd{P}{\defn}(\bar{t})$.
	Note that $\formula$ can be obtained from $\formula\substposih{P}$ by replacing all (positive) occurrences of formulas of the form $\ihf{P}{\bar{s}}$ that have been introduced by the substitution $\substposih{P}$, back with $P(\bar{s})$.
	Note furthermore that every sequent of the form $\defn, \ihf{P}{\bar{s}} \vdash P(\bar{s})$ admits an \scfoid-proof:
	\begin{equation*}
		\begin{prooftree}
			\infer0[(ax)]{\defn, P(\bar{s}), \mbd{P}{\defn}(\bar{s}) \vdash P(\bar{s})}
			\infer1[($\land$L)]{\defn, \ihf{P}{\bar{s}} \vdash P(\bar{s})}
		\end{prooftree}
	\end{equation*}
	Hence, Lemma \ref{lem:safe-replacement-of-positive-and-negative-occurrences} guarantees the existence of an \scfoid-proof of $\defn, \formula\substposih{P} \vdash \formula$, which finishes the proof.
\end{proof}

\begin{proposition}
	The inference rule \emph{\clarkrul} is derivable in \scfoid.
\end{proposition}

\begin{proof}
	Consider an arbitrary instance of \clarkrul:
	\begin{equation*} 
		\begin{prooftree}
			\hypo{\defn, \Gamma, \mbd{P}{\defn}(\bar{t}) \vdash \Delta}
			\infer1[\clarkrul]{\defn, \Gamma, P(\bar{t}) \vdash \Delta}
		\end{prooftree}
	\end{equation*}
	Let $\bar{y}$ be an $\arity{P}$-tuple of variables that do not occur in $\defn$.
	Consider the following \scfoid-derivation of $\defn, \Gamma, P(\bar{t}) \vdash \Delta$ from $\defn, \Gamma, \mbd{P}{\defn}(\bar{t}) \vdash \Delta$ and $\defn \vdash \forall \bar{y} : P(\bar{y}) \Leftrightarrow \mbd{P}{\defn}(\bar{y})$:
	\begin{equation*}
		\begin{prooftree}
			\hypo{\defn \vdash \forall \bar{y} : P(\bar{y}) \Leftrightarrow \mbd{P}{\defn}(\bar{y})}
			\infer1[(wk)]{\defn \vdash \forall \bar{y} : P(\bar{y}) \Leftrightarrow \mbd{P}{\defn}(\bar{y}), P(\bar{t})}
			
			\infer0[(ax)]{\defn, P(\bar{t}), \mbd{P}{\defn}(\bar{t}) \vdash P(\bar{t})}
			
			\hypo{\defn \vdash \mbd{P}{\defn}(\bar{t})}
			\infer1[(wk)]{\defn \vdash P(\bar{t}), \mbd{P}{\defn}(\bar{t})}
			
			\infer2[($\Leftrightarrow$L)]{\defn, P(\bar{t}) \Leftrightarrow \mbd{P}{\defn}(\bar{t}) \vdash P(\bar{t})}
			\infer1[($\forall$L)]{\defn, \forall \bar{y}: P(\bar{y}) \Leftrightarrow \mbd{P}{\defn}(\bar{y}) \vdash P(\bar{t})}
			
			\infer2[(cut)]{\defn \vdash P(\bar{t})}
		\end{prooftree}
	\end{equation*}
	By Proposition \ref{lem:proof-completion-sequent}, there exists an \scfoid-proof of $\defn \vdash \forall \bar{y} : P(\bar{y}) \Leftrightarrow \mbd{P}{\defn}(\bar{y})$. 
	Replacing the left bud of the above derivation with such a proof, we obtain an \scfoid-derivation of $\defn, \Gamma, P(\bar{t}) \vdash \Delta$ from $\defn, \Gamma, \mbd{P}{\defn}(\bar{t}) \vdash \Delta$.
	This shows that \clarkrul is derivable in \scfoid.
\end{proof}

While sound, \clarkrul is rather weak.
The issue is that $\mbd{P}{\defn}(\bar{t})$ generally contains defined atoms, which need to be introduced by the same rule, often resulting in a circular reasoning:
\begin{example}
	Let $\defn$ be the definition of $\mathit{Reachable}$ from Example \ref{ex:def-reachable}.
	Suppose we want to prove that a node $t$ is not reachable from the starting node $s$ under some conditions $\Gamma$ (entailing that $t \neq s$, in particular).
	Then we generally need to complete the following \scfoid-derivation: %
	\begin{equation*}
		\begin{prooftree}
			\hypo{\defn, \Gamma, t=s \vdash}
			
			\hypo{\defn, \Gamma, \mathit{Reachable}(x), \mathit{Edge}(x,t) \vdash}
			\infer1[($\land$R)]{\defn, \Gamma, \mathit{Reachable}(x) \land \mathit{Edge}(x,t) \vdash}
			\infer1[($\exists$R)]{\defn, \Gamma, \exists x: \mathit{Reachable}(x) \land \mathit{Edge}(x,t) \vdash}
			
			\infer2[($\lor$L)]{\defn, \Gamma, t=s \lor \exists x: \mathit{Reachable}(x) \land \mathit{Edge}(x,t) \vdash}
			\infer1[\clarkrul]{\defn, \Gamma, \mathit{Reachable}(t) \vdash}
		\end{prooftree}
	\end{equation*}
	Here we assumed that $x \notin \free(\Gamma)$.
	Hence, to prove the sequent $\defn, \Gamma, \mathit{Reachable}(t) \vdash$, we would generally need to prove the sequent $\defn, \Gamma, \mathit{Reachable}(x), \mathit{Edge}(x,t) \vdash$, which again has a defined atom on its left side.
\end{example}

The rule \clarkrul becomes useful when allowing infinite branches in proofs.
In fact, it is (essentially) the left introduction rule for defined atoms in the infinitary sequent calculus \lkidinf and the cyclic sequent calculus \clkid for inductive definitions by Brotherston and Simpson \cite{jlc/BrotherstonS11}.

Another possible left introduction rule for defined atoms would be the following:
\begin{equation*} %
	\begin{prooftree}
		\hypo{\text{minor premises}}
		\infer1[\deflalt]{\defn, \Gamma, P(\bar{t}) \vdash \ihf{P}{\bar{t}}, \Delta}
	\end{prooftree}
\end{equation*}
Here the minor premises are as in (def L).
Following \cite{jlc/BrotherstonS11}, we opt for (def L) over \deflalt, as (def L) enables cut-elimination (for a fragment of \foid in our case).
In fact, (def L) can be seen as a combination of \clarkrul and a cut on the induction hypothesis, as illustrated in the proof of Proposition \ref{prop:admissibility-def-L-from-L''}.
In terms of provability (with cut), the rules (def L) and \deflalt are interchangeable in \scfoid, as made precise by the following two propositions:

\begin{proposition} \label{prop:admissibility-def-L''}
	The inference rule \emph{\deflalt} is derivable in \scfoid.
	\hfill\markfull
\end{proposition}

\begin{proof}
	Consider an arbitrary instance of \deflalt
	\begin{equation} \label{eq:given-instance-def-L''}
		\begin{prooftree}
			\hypo{\text{minor premises}}
			\infer1[\deflalt]{\defn, \Gamma, P(\bar{t}) \vdash \ihf{P}{\bar{t}}, \Delta}
		\end{prooftree}
	\end{equation}
	with predicate set $\Pi$ and induction hypotheses $\ihs{\Pi}$.
	For every definitional rule $\forall \bar{y}: Q(\bar{s}) \rul \formulatwo$ in $\defn$ with $Q \in \Pi$, (\ref{eq:given-instance-def-L''}) has a minor premise $\defn, \Gamma, \formulatwo\substposihpi \vdash \ihf{Q}{\bar{s}}, \Delta$.
	We can derive the conclusion of (\ref{eq:given-instance-def-L''}) %
	with (def L), picking the same predicate set $\Pi$ and the same induction hypotheses $\ihs{\Pi}$:
	\begin{equation} \label{eq:applic-def-L}
		\begin{prooftree}
			\hypo{\text{minor premises}}
			\hypo{\defn, \Gamma, \ihf{P}{\bar{t}} \vdash \ihf{P}{\bar{t}}, \Delta}
			\infer2[(def L)]{\defn, \Gamma, P(\bar{t}) \vdash \ihf{P}{\bar{t}}, \Delta}
		\end{prooftree}
	\end{equation}
	For every definitional rule $\forall \bar{y}: Q(\bar{s}) \rul \formulatwo$ in $\defn$ with $Q \in \Pi$, (\ref{eq:applic-def-L}) has a minor premise $\defn, \Gamma, \formulatwo\substposihpi \vdash \ihf{Q}{\bar{s}}, \ihf{P}{\bar{t}}, \Delta$, which can be derived from the corresponding minor premise of (\ref{eq:given-instance-def-L''}) with (wk).
	The major premise of (\ref{eq:applic-def-L}) is trivially proven with (ax).
	This shows that \deflalt is derivable in \scfoid.
\end{proof}

\begin{proposition} \label{prop:admissibility-def-L-from-L''}
	Let \scfoiddefl be the sequent calculus obtained from \scfoid by replacing \emph{(def L)} with \emph{\deflalt}.
	The inference rule \emph{(def L)} is derivable in \scfoiddefl.
	\hfill\markfull
\end{proposition}

\begin{proof}
	Consider an arbitrary instance of \deflalt
	\begin{equation} \label{eq:given-instance-def-L}
		\begin{prooftree}
			\hypo{\text{minor premises}}
			\infer1[\deflalt]{\defn, \Gamma, P(\bar{t}) \vdash \ihf{P}{\bar{t}}, \Delta}
		\end{prooftree}
	\end{equation}
	with predicate set $\Pi$ and induction hypotheses $\ihs{\Pi}$.
	For every definitional rule $\forall \bar{y}: Q(\bar{s}) \rul \formulatwo$ in $\defn$ with $Q \in \Pi$, (\ref{eq:given-instance-def-L}) has a minor premise $\defn, \Gamma, \formulatwo\substposihpi \vdash \ihf{Q}{\bar{s}}, \Delta$.
	Consider the following \scfoid-derivation of the conclusion of (\ref{eq:given-instance-def-L}) from the premises of (\ref{eq:given-instance-def-L}):
	\begin{equation*}
		\begin{prooftree}
			\hypo{\text{minor premises}}
			\infer1[\deflalt]{\defn, \Gamma, P(\bar{t}) \vdash \ihf{P}{\bar{t}}, \Delta}
			
			\hypo{\defn, \Gamma, \ihf{P}{\bar{t}} \vdash \Delta}
			\infer1[(wk)]{\defn, \Gamma, P(\bar{t}), \ihf{P}{\bar{t}} \vdash \Delta}
			
			\infer2[(cut)]{\defn, \Gamma, P(\bar{t}) \vdash \Delta}
		\end{prooftree}
	\end{equation*}
	It shows that (def L) is derivable in \scfoiddefl.
\end{proof}

\subsection{Examples} \label{sec:examples}

\begin{example} \label{ex:proof-one-not-even}
	Let $\defn$ be the non-monotone inductive definition of $\mathit{Even}$ from Example \ref{ex:def_even}, let $\peax$ denote the Peano axioms $\forall x, y: \mathit{succ}(x) = \mathit{succ}(y) \Rightarrow x=y$ and $\forall x: \mathit{zero} \neq \mathit{succ}(x)$, and let $\mathit{one} \doteq \mathit{succ}(\mathit{zero})$. 
	We can formally prove that one is not even, given $\peax$ and $\defn$, by means of the induction rule (def L):
	\[ 
	\begin{prooftree}
		\hypo{\text{minor premises}}
		\hypo{\defn, \peax, \ihf{\mathit{Even}}{\mathit{one}} \vdash}
		\infer2[(def L)]{\defn, \peax, \mathit{Even}(\mathit{one}) \vdash}
		\infer1[($\lnot$R)]{\defn, \peax \vdash \lnot \mathit{Even}(\mathit{one})}
	\end{prooftree}
	\]
	The minor premises take the form $\defn, \peax, \top \vdash \ihf{\mathit{Even}}{\mathit{zero}}$ and $\defn, \peax, \mathit{Nat}(n) \land \lnot \mathit{Even}(n) \vdash \ihf{\mathit{Even}}{\mathit{succ}(n)}$. 
	(Note that the only option for $\Pi$ is $\{\mathit{Even}\}$.)
	To complete this proof, we need to choose an induction hypothesis $\ih_{\mathit{Even}}$ that contains every even number but that does not contain one. 
	A straightforward candidate is $\ih_{\mathit{Even}} \doteq \setexp{z}{z \neq \mathit{one}}$.
	This induction hypothesis indeed allows us to prove all the premises:
	\[
	\begin{prooftree}
		\infer0[(ax)]{\mathit{zero} \neq \mathit{succ}(\mathit{zero}) \vdash \mathit{zero} \neq \mathit{one}}
		\infer1[($\forall$L)]{\forall x: \mathit{zero} \neq \mathit{succ}(x) \vdash \mathit{zero} \neq \mathit{one}}
		\infer1[(wk)]{\defn, \peax, \top \vdash \ihf{\mathit{Even}}{\mathit{zero}}}
	\end{prooftree}
	\]
	\[
	\begin{prooftree}
		\infer0[(ax)]{\defn, \mathit{succ}(n) = \mathit{one} \vdash \mathit{succ}(n) = \mathit{succ}(\mathit{zero}), \mathit{Even}(n)}
		\infer0[(ax)]{\defn \vdash \top}
		\infer1[(def R)]{\defn \vdash \mathit{Even}(\mathit{zero})}
		\infer1[($=$L)]{\defn, n=\mathit{zero} \vdash \mathit{Even}(n)}
		\infer1[(wk)]{\defn, n=\mathit{zero}, \mathit{succ}(n) = \mathit{one} \vdash \mathit{Even}(n)}
		\infer2[($\Rightarrow$L)]{\defn, \mathit{succ}(n) = \mathit{succ}(\mathit{zero}) \Rightarrow n=\mathit{zero}, \mathit{succ}(n) = \mathit{one} \vdash \mathit{Even}(n)}
		\infer1[($\forall$L)]{\defn, \forall x, y: \mathit{succ}(x) = \mathit{succ}(y) \Rightarrow x=y, \mathit{succ}(n) = \mathit{one} \vdash \mathit{Even}(n)}
		\infer1[($\lnot$L)]{\defn, \forall x, y: \mathit{succ}(x) = \mathit{succ}(y) \Rightarrow x=y, \lnot \mathit{Even}(n), \mathit{succ}(n) = \mathit{one} \vdash}
		\infer1[($\lnot$R)]{\defn, \forall x, y: \mathit{succ}(x) = \mathit{succ}(y) \Rightarrow x=y, \lnot \mathit{Even}(n) \vdash \mathit{succ}(n) \neq \mathit{one}}
		\infer1[(wk)]{\defn, \peax, \mathit{Nat}(n), \lnot \mathit{Even}(n) \vdash \mathit{succ}(n) \neq \mathit{one}}
		\infer1[($\land$L)]{\defn, \peax, \mathit{Nat}(n) \land \lnot \mathit{Even}(n) \vdash \ihf{\mathit{Even}}{\mathit{succ}(n)}}
	\end{prooftree}
	\]
	
	\vsp
	
	\[ 
	\begin{prooftree}
		\infer0[($=$R)]{\defn, \peax \vdash \mathit{one} = \mathit{one}}
		\infer1[($\lnot$L)]{\defn, \peax, \ihf{\mathit{Even}}{\mathit{one}} \vdash}
	\end{prooftree}
	\]	
\end{example}

\begin{example} \label{ex:proof_sat}
	Let $\defn$ be the definition of $\mathit{Sat}$ from Example \ref{ex:sat}, and let $\coax$ be the following set of axioms, enforcing the functions $\mathit{and}$ and $\mathit{not}$ to behave as constructors:
	\begin{align*}
		& \forall f_1, f_2, g_1, g_2: \mathit{and}(f_1, g_1) = \mathit{and}(f_2, g_2) \Rightarrow f_1 = f_2 \land g_1 = g_2 \\
		& \forall f, g: \mathit{not}(f) = \mathit{not}(g) \Rightarrow f=g \\
		& \forall f, g, h: \mathit{and}(f, g) \neq \mathit{not}(h) \\
		& \forall p, i: \mathit{In}(p,i) \Rightarrow \lnot \exists f, g: p = \mathit{and}(f,g) \\
		& \forall p, i: \mathit{In}(p,i) \Rightarrow \lnot \exists f: p = \mathit{not}(f)
	\end{align*}
	More specifically, these axioms enforce that $\mathit{and}$ and $\mathit{not}$ are injective, that they have disjoint images, and that no propositional atom is a conjunction or a negation, respectfully.\footnote{
		These axioms are necessary but not sufficient for an appropriate $\defn$-context.
		For instance, they allow $\defn$-contexts with infinite sequences of formulas $f_1, f_2, \dots$ such that $f_i = \mathit{not}(f_{i+1})$
		Nevertheless, these axioms are already sufficient to prove various theorems about $\defn$.
	}
	
	We can formally prove that, under $\defn$ and $\coax$, a structure $\structuretwo$ containing a propositional symbol $Q$ does not satisfy $\lnot Q$, by proving the sequent $\defn, \coax, \mathit{In}(q,j) \vdash \lnot \mathit{Sat}(j,\mathit{not}(q))$:
	\[
	\begin{prooftree}
		\hypo{\text{minor premises}}
		\hypo{\defn, \coax, \mathit{In}(q,j), \ihf{\mathit{Sat}}{j,\mathit{not}(q)} \vdash}
		\infer2[(def L)]{\defn, \coax, \mathit{In}(q,j), \mathit{Sat}(j,\mathit{not}(q)) \vdash}
		\infer1[($\lnot$R)]{\defn, \coax, \mathit{In}(q,j) \vdash \lnot \mathit{Sat}(j,\mathit{not}(q))}
	\end{prooftree}
	\]	
	The minor premises take the form 
	\[
		\defn, \coax, \mathit{In}(q,j), \mathit{In}(p,i) \vdash \ihf{\mathit{Sat}}{i,p} ,
	\]
	\[
		\defn, \coax, \mathit{In}(q,j), \ihf{\mathit{Sat}}{i,f} \land \ihf{\mathit{Sat}}{i,g} \vdash \ihf{\mathit{Sat}}{i,\mathit{and}(f,g)} ,
	\]
	\[
		\defn, \coax, \mathit{In}(q,j), \lnot \mathit{Sat}(i,f) \vdash \ihf{\mathit{Sat}}{i,\mathit{not}(f)} .
	\]
	The induction hypothesis $\ih_{\mathit{Sat}}$ needs to include every pair $(i,f)$ such that $i$ satisfies $f$ but exclude $(j,\mathit{not}(q))$. 
	Similarly to Example \ref{ex:proof-one-not-even}, we can let $\ihs{\mathit{Sat}}$ be $\ihs{1} \doteq \{ (i,f) \mid i \neq j \lor f \neq \mathit{not}(q) \}$.
	Another option would be $\ihs{2} \doteq \{ (i,f) \mid \forall p: \mathit{In}(p,i) \Rightarrow f \neq \mathit{not}(p) \}$ and yet another option would be $\ihs{3} \doteq \{ (i,f) \mid \forall p: \mathit{Sat}(p,i) \Rightarrow f \neq \mathit{not}(p) \}$.
	All three induction hypotheses can be used to prove the above premises.
	Conceptually, the three induction hypotheses correspond to three different upper bounds on the defined predicate $\mathit{Sat}$.
	Note that under $\defn$, $\coax$ and $\mathit{In}(q,j)$, $\ihf{3}{i,f}$ implies $\ihf{2}{i,f}$ and $\ihf{2}{i,f}$ implies $\ihf{1}{i,f}$ for all pairs $(i,f)$, but not vice versa.
	Hence, $\ihs{3}$ describes a tighter upper bound on $\mathit{Sat}$ than $\ihs{2}$, and $\ihs{2}$ a tighter upper bound than $\ihs{1}$.
	Tighter upper bounds generally require more work to prove, but they can be used to prove more sequents.
	For instance, $\ihs{2}$ can be used to prove the sequent $\defn, \coax, \mathit{In}(p,i) \vdash \lnot \mathit{Sat}(i,\mathit{not}(p))$ for any pair $(p,i)$, whereas $\ihs{1}$ only allows proving this sequent for the specific pair $(q,j)$.
	Example \ref{ex:app-proof-sat} %
	contains a full proof of $\defn, \coax, \mathit{In}(q,j) \vdash \lnot \mathit{Sat}(j,\mathit{not}(q))$ with induction hypothesis $\ihs{1}$.
\end{example}

\begin{example} \label{ex:temporal-graph-proof}
	Let $\defn$ be the definition of $\mathit{Active}$ and $\mathit{ActivePath}$ from Example \ref{ex:temporal_graph}. 
	We can formally prove that if no edge of the form $(a,p)$ is ever activated, then all active edges remain active. 
	We do so by proving the sequent $\defn, \Gamma \vdash \forall n,m,t: \mathit{Active}(n,m,t) \Rightarrow \mathit{Active}(n,m,\mathit{succ}(t))$, where $\Gamma$ is the formula $ \lnot \exists p,s: \mathit{MakeActive}(a,p,s)$.
	For space reasons, we abbreviate $\mathit{MakeActive}$ to $\mathit{MA}$, $\mathit{Active}$ to $\mathit{A}$, and $\mathit{ActivePath}$ to $\mathit{AP}$.
	\begin{equation} \label{eq:proof-active}
		\begin{prooftree}
			\infer0[(ax)]{\defn, \Gamma, \mathit{A}(n,m,t) \vdash \mathit{A}(n,m,t)}
			\hypo{\text{minor premises}}
			\hypo{\defn, \Gamma, \ihf{\mathit{AP}}{a,b,t} \vdash}
			\infer2[(def L)]{\defn, \Gamma, \mathit{AP}(a,b,t) \vdash}
			\infer1[($\lnot$R)]{\defn, \Gamma \vdash \lnot \mathit{AP}(a,b,t)}
			\infer1[(wk)]{\defn, \Gamma, \mathit{A}(n,m,t) \vdash \lnot \mathit{AP}(a,b,t)}
			\infer2[($\land$R)]{\defn, \Gamma, \mathit{A}(n,m,t) \vdash \mathit{A}(n,m,t) \land \lnot \mathit{AP}(a,b,t)}
			\infer1[(def R)]{\defn, \Gamma, \mathit{A}(n,m,t) \vdash \mathit{A}(n,m,\mathit{succ}(t))}
			\infer1[($\Rightarrow$R)]{\defn, \Gamma \vdash \mathit{A}(n,m,t) \Rightarrow \mathit{A}(n,m,\mathit{succ}(t))}
			\infer1[($\forall$R)]{\defn, \Gamma \vdash \forall n,m,t: \mathit{A}(n,m,t) \Rightarrow \mathit{A}(n,m,\mathit{succ}(t))}
		\end{prooftree}
	\end{equation}
	To complete this proof, we can instantiate the above application of (def L) with the predicate set $\Pi = \{\mathit{A}, \mathit{AP}\}$ and the induction hypotheses $\ihs{\mathit{A}} \doteq \ihs{\mathit{AP}} \doteq \{ (x,y,s) \mid x \neq a \}$.
	Intuitively, this amounts to showing that no active edge or path starts with $a$.
	Our instance has four minor premises; two corresponding to the predicate $\mathit{A}$:
	\[
		\defn, \Gamma, \mathit{MA}(n,m,t) \land \lnot \mathit{AP}(a,b,t) \vdash n \neq a ,
	\]
	\[
	\defn, \Gamma, n \neq a \land \lnot \mathit{AP}(a,b,t) \vdash n \neq a ,
	\]
	and two corresponding to the predicate $\mathit{AP}$:
	\[
	\defn, \Gamma, n \neq a \vdash n \neq a ,
	\]
	\[
	\defn, \Gamma, n \neq a \land m \neq a \vdash n \neq a .
	\]
	The major premise instantiates to $\defn, \Gamma, a \neq a \vdash$.
	The first minor premise can be proven as follows:
	\begin{equation*}
		\begin{prooftree}
			\infer0[(ax)]{\mathit{MA}(a,m,t) \vdash \mathit{MA}(a,m,t)}
			\infer1[($\exists$R)]{\mathit{MA}(a,m,t) \vdash \exists p,s: \mathit{MA}(a,p,s)}
			\infer1[($\lnot$L)]{\lnot \exists p,s: \mathit{MA}(a,p,s), \mathit{MA}(a,m,t) \vdash}
			\infer1[($=$L)]{\lnot \exists p,s: \mathit{MA}(a,p,s), \mathit{MA}(n,m,t), n=a \vdash}
			\infer1[($\lnot$R)]{\lnot \exists p,s: \mathit{MA}(a,p,s), \mathit{MA}(n,m,t) \vdash n \neq a}
			\infer1[(wk)]{\defn, \Gamma, \mathit{MA}(n,m,t), \lnot \mathit{AP}(a,b,t) \vdash n \neq a}
			\infer1[($\land$L)]{\defn, \Gamma, \mathit{MA}(n,m,t) \land \lnot \mathit{AP}(a,b,t) \vdash n \neq a}
		\end{prooftree}
	\end{equation*}
	All the other premises are straightforward to prove.
	
	Alternatively, (\ref{eq:proof-active}) can be completed by instantiating the application of (def L) with the predicate set $\Pi = \{ \mathit{AP} \}$ and the induction hypothesis $\ihs{\mathit{AP}} \doteq \{ (x,y,s) \mid \exists z : \mathit{A}(x,z,s) \}$.
	Intuitively, this amounts to showing that the starting node of any active path has an outgoing active edge.
	Our instance has two minor premises:
	\[\defn, \Gamma, \mathit{A}(n,m,t) \vdash \exists z: \mathit{A}(n,z,t) ,\]
	\[\defn, \Gamma, \exists z: \mathit{A}(n,z,t) \land \exists z: \mathit{A}(m,z,t) \vdash \exists z: \mathit{A}(n,z,t) ,\]
	both of which are straightforward to prove.
	The major premise $\defn, \Gamma, \exists z:\mathit{A}(a,z,t) \vdash$ can be proven as follows:
	\begin{equation*}
		\begin{prooftree}
			\hypo{\text{minor premises}}
			\hypo{\defn, \Gamma, \ihf{\mathit{A}}{a,z,t} \vdash}
			\infer2[(def L)]{\defn, \Gamma, \mathit{A}(a,z,t) \vdash}
			\infer1[($\exists$L)]{\defn, \Gamma, \exists z:\mathit{A}(a,z,t) \vdash}
		\end{prooftree}
	\end{equation*}
	We can complete this proof by instantiating the above application of (def L) with the predicate set $\Pi = \{ \mathit{A} \}$ and the induction hypothesis $\ihs{\mathit{A}} \doteq \{ (x,y,s) \mid \exists r: \mathit{MA}(x,y,r) \}$.
	Intuitively, this amounts to showing that every active edge must have been activated at some time point.
	The instance has two minor premises 
	\begin{equation*}
		\defn, \Gamma, \mathit{MA}(n,m,t) \land \lnot \mathit{AP}(a,b,t) \vdash \exists r: \mathit{MA}(n,m,r) ,
	\end{equation*}
	\begin{equation*}
		\defn, \Gamma, (\exists r: \mathit{MA}(n,m,r)) \land \lnot \mathit{AP}(a,b,t) \vdash \exists r: \mathit{MA}(n,m,r) 
	\end{equation*}
	and major premise 
	\begin{equation*}
		\defn, \Gamma, \exists r: \mathit{MA}(a,z,r) \vdash 
	\end{equation*}
	all of which are straightforward to prove.
	
	Note that our proofs did not require Peano axioms for the time terms (i.e., the terms that occur as the third argument in $\mathit{MA}$, $\mathit{A}$ or $\mathit{PA}$).
	Hence, our proven property holds for any interpretation of $\mathit{succ}$.
\end{example}

\begin{example} \label{ex:proof_distance}
	Let $\defn$ be the definition of $\mathit{Distance}$ from Example \ref{ex:distance}. 
	We can formally prove that if a graph has an edge from a node $a$ to a different node $b$, then the distance from $a$ to $b$ in this graph is one.
	We do so by proving the sequent $\defn, \Gamma \vdash \mathit{Distance}(a,c,\mathit{one})$, where $\mathit{one} \doteq \mathit{succ}(\mathit{zero})$ and $\Gamma$ consists of the assumptions $\mathit{Edge}(a,b)$ and $a \neq b$, as well as the axioms $\forall x: \mathit{succ}(x) \neq \mathit{zero}$ and $\forall x: x \leq \mathit{zero} \Rightarrow x =\mathit{zero}$.\footnote{
		Since these are the only axioms we need, the property also holds for non-standard interpretations of $\mathit{succ}$ and $\leq$.
	}
	For space reasons, we abbreviate $\mathit{Edge}$ to $\mathit{E}$ and $\mathit{Distance}$ to $\mathit{D}$.
	\begin{equation*}
		\begin{prooftree}
			\infer0[(ax)]{\defn, \Gamma \vdash \top}
			\infer1[(def R)]{\defn, \Gamma \vdash \mathit{D}(a,a,\mathit{zero})}
			\infer0[(ax)]{\mathit{E}(a,b) \vdash \mathit{E}(a,b) }
			\infer1[(wk)]{\defn, \Gamma \vdash \mathit{E}(a,b) }
			\infer2[($\land$R)]{\defn, \Gamma \vdash \mathit{D}(a,a,\mathit{zero}) \land \mathit{E}(a,b) }
			\infer1[($\exists$R)]{\defn, \Gamma \vdash \exists z : \mathit{D}(a,z,\mathit{zero}) \land \mathit{E}(z,b) }
			
			\hypo{\defn, \Gamma, m \leq \mathit{zero}, \mathit{D}(a,b,m) \vdash }
			\infer1[($\land$R)]{\defn, \Gamma, m \leq \mathit{zero} \land \mathit{D}(a,b,m) \vdash }
			\infer1[($\exists$L)]{\defn, \Gamma, \exists m: m \leq \mathit{zero} \land \mathit{D}(a,b,m) \vdash }
			\infer1[($\lnot$R)]{\defn, \Gamma \vdash \lnot (\exists m: m \leq \mathit{zero} \land \mathit{D}(a,b,m)) }
			
			\infer2[($\land$R)]{\defn, \Gamma \vdash (\exists z : \mathit{D}(a,z,\mathit{zero}) \land  \mathit{E}(z,b)) \land \lnot (\exists m: m \leq \mathit{zero} \land \mathit{D}(a,b,m)) }
			\infer1[(def R)]{\defn, \Gamma \vdash \mathit{D}(a,c,\mathit{one})}
		\end{prooftree}
	\end{equation*}
	Using the axiom $\forall x: x \leq \mathit{zero} \Rightarrow x =\mathit{zero}$, we can simplify $\defn, \Gamma, m \leq \mathit{zero}, \mathit{D}(a,b,m) \vdash$ to $\defn, \Gamma, \mathit{D}(a,b,\mathit{zero}) \vdash$:
	\begin{equation*}
		\begin{prooftree}
			\infer0[(ax)]{\defn, \Gamma, m \leq \mathit{zero}, \mathit{D}(a,b,m) \vdash m \leq \mathit{zero}}
			\hypo{\text{minor premises}}
			\hypo{\defn, \Gamma, \ihf{\mathit{D}}{a,b,\mathit{zero}} \vdash}
			\infer2[(def L)]{\defn, \Gamma, \mathit{D}(a,b,\mathit{zero}) \vdash }
			\infer1[($=$L)]{\defn, \Gamma, m = \mathit{zero}, \mathit{D}(a,b,m) \vdash }
			\infer1[(wk)]{\defn, \Gamma, m = \mathit{zero}, m \leq \mathit{zero}, \mathit{D}(a,b,m) \vdash }
			\infer2[($\Rightarrow$L)]{\defn, \Gamma, m \leq \mathit{zero} \Rightarrow m = \mathit{zero}, m \leq \mathit{zero}, \mathit{D}(a,b,m) \vdash }
			\infer1[($\forall$L)]{\defn, \Gamma, m \leq \mathit{zero}, \mathit{D}(a,b,m) \vdash }
		\end{prooftree}
	\end{equation*} 
	In the application of (def L), the only option for $\Pi$ is $\{ \mathit{D} \}$.
	Hence, the minor premises are of the form
	\begin{equation*}
		\defn, \Gamma, \top \vdash \ihf{\mathit{D}}{x, x, \mathit{zero}} ,
	\end{equation*}
	\begin{equation*}
		\defn, \Gamma, (\exists z: \ihf{\mathit{D}}{x, y, n} \land \mathit{E}(z,y)) \land \lnot (\exists m: m \leq n \land \mathit{D}(x,y,m)) \vdash \ihf{\mathit{D}}{x, y, \mathit{succ}(n)} .
	\end{equation*}
	We pick the induction hypothesis $\ihs{\mathit{D}} \doteq \setexp{(x,y,m)}{x \neq a \lor y \neq b \lor m \neq \mathit{zero}}$ and prove the premises as follows:
	\begin{equation*}
		\begin{prooftree}
			\infer0[(ax)]{x=b \vdash x=b}
			\infer1[($=$L)]{x=a, x=b \vdash a=b}
			\infer1[($\lnot$L)]{a \neq b, x=a, x=b \vdash}
			\infer1[($\lnot$R)]{a \neq b, x=a \vdash x \neq b}
			\infer1[($\lnot$R)]{a \neq b \vdash x \neq a, x \neq b}
			\infer1[(wk)]{\defn, \Gamma, \top \vdash x \neq a, x \neq b, \mathit{zero} \neq \mathit{zero}}
			\infer1[($\lor$R)]{\defn, \Gamma, \top \vdash \ihf{\mathit{D}}{x, x, \mathit{zero}}}
		\end{prooftree}
	\end{equation*}
	\begin{equation*}
	\begin{prooftree}
		\infer0[(ax)]{\mathit{succ}(n) \neq \mathit{zero} \vdash \mathit{succ}(n) \neq \mathit{zero}}
		\infer1[($\forall$L)]{\forall x: \mathit{succ}(x) \neq \mathit{zero} \vdash \mathit{succ}(n) \neq \mathit{zero}}
		\infer1[(wk)]{\defn, \Gamma, (\exists z: \ihf{\mathit{D}}{x, y, n} \land \mathit{E}(z,y)) \land \lnot (\exists m: m \leq n \land \mathit{D}(x,y,m)) \vdash x \neq a, y \neq b, \mathit{succ}(n) \neq \mathit{zero}}
		\infer1[($\lor$R)]{\defn, \Gamma, (\exists z: \ihf{\mathit{D}}{x, y, n} \land \mathit{E}(z,y)) \land \lnot (\exists m: m \leq n \land \mathit{D}(x,y,m)) \vdash \ihf{\mathit{D}}{x, y, \mathit{succ}(n)}}
	\end{prooftree}
	\end{equation*}
	
	\begin{equation*}
	\begin{prooftree}
		\infer0[($=$R)]{\vdash a=a}
		\infer1[($\lnot$L)]{a \neq a \vdash}
		\infer0[($=$R)]{\vdash b=b}
		\infer1[($\lnot$L)]{b \neq b \vdash}
		\infer0[($=$R)]{\vdash \mathit{zero} = \mathit{zero}}
		\infer1[($\lnot$L)]{\mathit{zero} \neq \mathit{zero} \vdash}
		\infer3[($\lor$L)]{a \neq a \lor b \neq b \lor \mathit{zero} \neq \mathit{zero} \vdash}
		\infer1[(wk)]{\defn, \Gamma, \ihf{\mathit{D}}{a,b,\mathit{zero}} \vdash}
	\end{prooftree}
	\end{equation*}
\end{example}

\begin{example} \label{ex:proof-reachable}
	In this example, we prove a sequent with a formula that contains definitions as atomic components. 
	Let $\defn$ be the definition of $\mathit{Reachable}$ from Example \ref{ex:def-reachable}.
	By $\defn[s,R]$, we denote the definition obtained from $\defn$ by replacing $\mathit{Reachable}$ with $R$, and by $\defn[t,Q]$, we denote the definition obtained from $\defn$ by replacing $\mathit{Reachable}$ with $Q$ and $s$ with $t$.
	The defined predicate $Q$ of $\defn[t,Q]$ represents the set of all nodes reachable from the node $t$.
	We can prove that for all nodes $s$, $t$ and $r$, if $t$ is reachable from $s$, and $r$ is reachable from $t$, then $r$ is reachable from $s$ as well, by proving the sequent
	$
	\vdash \forall s,t,r : \defn[s,R] \land R(t) \land \defn[t,Q] \land Q(r) \Rightarrow R(r) 
	$:\footnote{
		This sequent can be read as: ``for all nodes $s$, $t$ and $r$, if $R$ is the set of reachable nodes from $s$, $t$ belongs to $R$, $Q$ is the set of reachable nodes from $t$, and $r$ belongs to $G$, then $r$ belongs to $R$''.
	}
	\[ %
	\begin{prooftree}%
		\hypo
		{\defn[s,R], R(t), \defn[t,Q], Q(r) \vdash R(r)}
		\infer1[($\land$L)]{\defn[s,R] \land R(t) \land \defn[t,Q] \land Q(r) \vdash R(r)}
		\infer1[($\Rightarrow$R)]{\vdash \defn[s,R] \land R(t) \land \defn[t,Q] \land Q(r) \Rightarrow R(r)}
		\infer1[($\forall$R)]{\vdash \forall s,t,r : \defn[s,R] \land R(t) \land \defn[t,Q] \land Q(r) \Rightarrow R(r)}
	\end{prooftree}
	\]
	To prove the sequent $\defn[s,R], R(t), \defn[t,Q], Q(r) \vdash R(r)$, we apply (def L) w.r.t.\ $\defn[t,Q]$, picking $\Pi = \{Q\}$ and $\ihs{Q} \doteq \{ x \mid R(x) \}$.
	Informally, this amounts to showing that all nodes reachable from $t$ are also reachable from $s$ (which holds since $t$ is reachable from $s$).
	This instance of (def L) has two minor premises:
	\begin{equation*}
		\defn[s,R], R(t), \defn[t,Q], \top \vdash R(t) ,
	\end{equation*}
	\begin{equation*}
		\defn[s,R], R(t), \defn[t,Q], R(x) \land E(x,y) \vdash R(y) ,
	\end{equation*}
	and one major premise:
	\begin{equation*}
		\defn[s,R], R(t), \defn[t,Q], R(r) \vdash R(r) ,
	\end{equation*}
	all of which are straightforward to prove.
\end{example}

As illustrated by the following example, \scfoid enables proving non-totality of definitions.

\begin{example} \label{ex:proof_non-total_liar}
	Let $\defn$ be the definition %
	$\{P \rul \lnot P\}$. 
	We can formally prove that $\defn$ is non-total by deriving its negation.
	Indeed, if $\lnot \defn$ is valid, then $\defn$ has no two-valued models under the well-founded semantics.
	Hence, its well-founded model (in any trivial $\defn$-context) must be strictly three-valued.
	Informally, we prove $\lnot \defn$ by showing that $\defn$ entails both the truth and the falsity of $P$.
	\[ 
	\begin{prooftree}[separation=5pt]
		\hypo{}
		\infer1[(ax)]{\defn, \lnot P \vdash \lnot P}
		\infer1[(def R)]{\defn, \lnot P \vdash P}
		\infer1[($\lnot$R)]{\defn \vdash P, \lnot \lnot P}
		\infer1[($*$)]{\defn \vdash P}
		
		\hypo{}
		\infer1[(ax)]{\defn, \lnot P \vdash \lnot P}
		
		\hypo{}
		\infer1[(ax)]{\defn, \lnot P \vdash \lnot P}
		\infer1[(def R)]{\defn, \lnot P \vdash P}
		\infer1[($\lnot$R)]{\defn \vdash P, \lnot \lnot P}
		\infer1[($*$)]{\defn \vdash P}
		\infer1[($\lnot$L)]{\defn, \lnot P \vdash}
		
		\infer2[(def L)]{\defn, P \vdash}
		
		\infer2[(cut)]{\defn \vdash}
		\infer1[($\lnot$R)]{\vdash \lnot \defn}
	\end{prooftree}
	\]
	We picked $\ih_P \doteq \lnot P$ in the application of (def L).\footnote{
		By abuse of notation, we let induction hypotheses $\ihs{P}$ of propositional symbols $P$ be formulas $\formula$ instead of set expressions $\{ () \mid \formula\}$.
	} 
	The transition labeled with ($*$) can be completed as follows:	
	\[
	\begin{prooftree}[separation=5pt]
		\hypo{\defn \vdash P, \lnot \lnot P}
		\infer1[($\lor$R)]{\defn \vdash P \lor \lnot \lnot P}

		\hypo{}
		\infer1[(ax)]{\defn, P \vdash P}
		
		\hypo{}
		\infer1[(ax)]{\defn, P \vdash P}
		\infer1[($\lnot$R)]{\defn \vdash P, \lnot P}
		\infer1[($\lnot$L)]{\defn, \lnot \lnot P \vdash P}
		
		\infer2[($\lor$L)]{\defn, P \lor \lnot \lnot P \vdash P}
		
		\infer2[(cut)]{\defn \vdash P}
	\end{prooftree}
	\]
	\hfill\markfull
\end{example}

\begin{example} \label{ex:proof_access}
	Let $\defn$ be the definition of $\mathit{Access}$ from Example \ref{ex:access_control}. 
	We can formally prove that if the owner grants access to another user $p$ who blocks themself, and if $p$ is not blocked by another user, then $\defn$ is non-total. 
	We do so by proving the sequent $\Gamma \vdash \lnot \defn$, where $\Gamma$ consists of the assumptions $\mathit{Grants}(\mathit{owner},p)$, $\mathit{Blocks}(p,p)$, $\mathit{owner} \neq p$ and $\forall x : \mathit{Blocks}(x,p) \Rightarrow x=p$.
	Validity of $\Gamma \vdash \lnot \defn$ means that no two-valued model of $\Gamma$ is a model of $\defn$; hence
	the well-founded model of $\defn$ in any $\defn$-context that satisfies $\Gamma$ is strictly three-valued, and therefore, $\defn$ is non-total in any such $\defn$-context.
	We abbreviate $\mathit{Grant}$ to $\mathit{G}$, $\mathit{Blocks}$ to $\mathit{B}$, and $\mathit{Access}$ to $\mathit{A}$.
	\begin{equation} \label{eq:proof-access}
		\begin{prooftree}
			\hypo{\defn, \Gamma, \lnot \mathit{A}(p) \vdash \exists v: A(v) \land G(v,p)}
			
			\hypo{\defn, \Gamma, \lnot \mathit{A}(p) \vdash \lnot (\exists v: A(v) \land B(v,p))}
			
			\infer2[($\land$R)]{\defn, \Gamma, \lnot \mathit{A}(p) \vdash (\exists v: A(v) \land G(v,p)) \land \lnot (\exists v: A(v) \land B(v,p))}
			\infer1[(def R)]{\defn, \Gamma, \lnot \mathit{A}(p) \vdash \mathit{A}(p)}
			\infer1[($\lnot$R)]{\defn, \Gamma \vdash \mathit{A}(p), \lnot \lnot \mathit{A}(p)}
			\infer1[($*$)]{\defn, \Gamma \vdash \mathit{A}(p)}
			
			\hypo{\defn, \Gamma, \mathit{A}(p) \vdash}
			
			\infer2[(cut)]{\defn, \Gamma \vdash}
			\infer1[($\lnot$R)]{\Gamma \vdash \lnot \defn}
		\end{prooftree}
	\end{equation}
	The transition labeled with ($*$) can be completed analogously to the transition labeled with ($*$) in Example \ref{ex:proof_non-total_liar}.	
	We prove the bud sequents in the above derivation as follows:
	\begin{equation*}
		\begin{prooftree}
			\infer0[(ax)]{\defn, \Gamma, \lnot \mathit{A}(p) \vdash \top}
			\infer1[(def R)]{\defn, \Gamma, \lnot \mathit{A}(p) \vdash A(\mathit{owner})}
			\infer0[(ax)]{\defn, \Gamma, \lnot \mathit{A}(p) \vdash G(\mathit{owner},p)}
			\infer2[($\land$R)]{\defn, \Gamma, \lnot \mathit{A}(p) \vdash A(\mathit{owner}) \land G(\mathit{owner},p)}
			\infer1[($\exists$R)]{\defn, \Gamma, \lnot \mathit{A}(p) \vdash \exists v: A(v) \land G(v,p)}
		\end{prooftree}
	\end{equation*}
	
	\begin{equation*}
		\begin{prooftree}
			\infer0[(ax)]{A(v), B(v,p) \vdash \mathit{B}(v,p),  \mathit{A}(p)}
			\infer0[(ax)]{A(p), B(v,p) \vdash \mathit{A}(p)}
			\infer1[($=$L)]{v=p, A(v), B(v,p) \vdash \mathit{A}(p)}
			\infer2[($\Rightarrow$L)]{\mathit{B}(v,p) \Rightarrow v=p, A(v), B(v,p) \vdash \mathit{A}(p)}
			\infer1[($\forall$L)]{\forall x : \mathit{B}(x,p) \Rightarrow x=p, A(v), B(v,p) \vdash \mathit{A}(p)}
			\infer1[(wk)]{\defn, \Gamma, A(v), B(v,p) \vdash \mathit{A}(p)}
			\infer1[($\land$L)]{\defn, \Gamma, A(v) \land B(v,p) \vdash \mathit{A}(p)}
			\infer1[($\exists$L)]{\defn, \Gamma, \exists v: A(v) \land B(v,p) \vdash \mathit{A}(p)}
			\infer1[($\lnot$L)]{\defn, \Gamma, \lnot \mathit{A}(p), \exists v: A(v) \land B(v,p) \vdash}
			\infer1[($\lnot$R)]{\defn, \Gamma, \lnot \mathit{A}(p) \vdash \lnot (\exists v: A(v) \land B(v,p))}
		\end{prooftree}
	\end{equation*}
	
	\vsp
	
	\begin{equation*}
		\begin{prooftree}
			\hypo{\text{minor premises}}
			\hypo{\defn, \Gamma, \ihf{A}{p} \vdash}
			\infer2[(def L)]{\defn, \Gamma, \mathit{A}(p) \vdash}
		\end{prooftree}
	\end{equation*}
	We instantiate the application of (def L) with $\Pi = \{A\}$ and $\ihs{A} \doteq \setexp{u}{u \neq p}$.
	This yields two minor premises:
	\begin{equation*}
		\defn, \Gamma, \top \vdash \mathit{owner} \neq p ,
	\end{equation*}
	\begin{equation*}
		\defn, \Gamma, (\exists v: v \neq p \land G(v,u)) \land \lnot (\exists v: A(v) \land B(v,u)) \vdash u \neq p .
	\end{equation*}
	The major premise instantiates to $\defn, \Gamma, p \neq p \vdash$.
	The first minor premise and the major premise are straightforward to prove.
	We prove the second minor premise as follows:
	\begin{equation*}
		\begin{prooftree}
			\hypo{\defn, \Gamma \vdash A(p)}
			\infer0[(ax)]{\defn, \Gamma \vdash B(p,p)}
			\infer2[($\land$R)]{\defn, \Gamma \vdash A(p) \land B(p,p)}
			\infer1[($\exists$R)]{\defn, \Gamma \vdash \exists v: A(v) \land B(v,p)}
			\infer1[($=$L)]{\defn, \Gamma, u=p \vdash \exists v: A(v) \land B(v,u)}
			\infer1[($\lnot$L)]{\defn, \Gamma, u=p, \lnot (\exists v: A(v) \land B(v,u)) \vdash}
			\infer1[($\lnot$R)]{\defn, \Gamma, \lnot (\exists v: A(v) \land B(v,u)) \vdash u \neq p}
			\infer1[(wk)]{\defn, \Gamma, \exists v: v \neq p \land G(v,u), \lnot (\exists v: A(v) \land B(v,u)) \vdash u \neq p}
			\infer1[($\land$L)]{\defn, \Gamma, (\exists v: v \neq p \land G(v,u)) \land \lnot (\exists v: A(v) \land B(v,u)) \vdash u \neq p}
		\end{prooftree}
	\end{equation*}
	To complete this proof, it suffices to plug in the proof of $\defn, \Gamma \vdash A(p)$ from (\ref{eq:proof-access}).
\end{example}

\subsection{The Predicate Set} \label{sec:pi}

Proofs by mathematical induction generally involve multiple defined predicates.
In particular, this is the case for mutually defined predicates, such as $\mathit{Even}$ and $\mathit{Odd}$ in the definition from Example \ref{ex:def-even-odd}.
Proving that every even number is a natural number, for instance, is done by simultaneously proving that every odd number is a natural number.
Formally, this is reflected by the predicate set $\Pi$ in (def L), which would be $\{\mathit{Even}, \mathit{Odd}\}$ in this example.
In general, to derive a sequent $\defn, \Gamma, P(\bar{t}) \vdash \Delta$ with (def L) in \scfoid, $\Pi$ can be any subset of $\defi{\defn}$ containing $P$.\footnote{
	Here we mean that (def L) is sound (w.r.t.~the well-founded, stable and Henkin semantics) for any choice of $\Pi$, in the sense that validity of the premises implies validity of the conclusion.
	However, one generally needs to choose $\Pi$ wisely, as not all choices permit provable premises.
}
In \lkid, $\Pi$ is fixed to be the set consisting of all defined predicates that are \emph{mutually dependent with} $P$ \cite{jlc/BrotherstonS11}. 

\begin{definition}
	Let $\defn$ be a definition and $P, Q \in \defi{\defn}$.
	We say that $P$ \emph{depends directly on} $Q$ \emph{w.r.t.~$\defn$} if $Q$ occurs in the body $\formula$ of a rule $\forall \bar{x}: P(\bar{t}) \rul \formula$ of $\defn$.
	The \emph{direct dependency relation} $\ddep{\defn}$ \emph{w.r.t.~$\defn$} contains all pairs of defined predicates $(R,S)$ for which $R$ depends directly on $S$.
	The \emph{dependency relation} $\dep{\defn}$ \emph{w.r.t.~$\defn$} is the reflexive-transitive closure of the direct dependency relation $\ddep{\defn}$ w.r.t.~$\defn$.
	We say that $P$ \emph{depends on} $Q$ \emph{w.r.t.~$\defn$} if $(P, Q) \in \dep{\defn}$.
	We say that $P$ and $Q$ are \emph{mutually dependent w.r.t.~$\defn$} if $(P, Q) \in \dep{\defn}$ and $(Q, P) \in \dep{\defn}$.
	We denote by $\md{\defn}{P}$ the set of all defined predicates of $\defn$ that are mutually dependent with $P$ w.r.t.~$\defn$.
\end{definition}

As demonstrated by Example \ref{ex:temporal-graph-proof}, it is not always necessary to include all mutually dependent predicates in the predicate set $\Pi$.
In this example, we proved a theorem about the definition $\defn$ of $\mathit{Active}$ and $\mathit{ActivePath}$ from Example \ref{ex:temporal_graph}.
Even though $\mathit{Active}$ and $\mathit{ActivePath}$ are mutually dependent w.r.t.~$\defn$, the second proof from Example \ref{ex:temporal-graph-proof} applied (def L) once with $\Pi = \{\mathit{Active}\}$ and once with $\Pi = \{\mathit{ActivePath}\}$.
Thus, the induction rule (def L) in \scfoid is more flexible than the induction rule in \lkid as it permits general predicate sets $\Pi$.
This extension of \scfoid w.r.t.~\lkid is of interest in its own right, independently of its extension to non-monotone definitions.

As a default option, one can always pick $\Pi$ to be $\md{\defn}{P}$ in (def L), like in \lkid.
Although this property may seem evident, its proof is non-trivial.
To show that one can indeed restrict $\Pi$ to $\md{\defn}{P}$ (Proposition \ref{prop:predicate-set-restr-md}), we first establish that one can restrict $\Pi$ to a subset of $\md{\defn}{P}$ (Lemma \ref{lem:predicate-set-restr-subsets-md}), and we then establish that any sequent that can be proven with a predicate set $\Pi$ can also be proven with a larger predicate set $\Pi'$ (Lemma \ref{lem:predicate-set-larger}).

\begin{restatable}{lemma}{predicatesetrestrsubsetsmd} \label{lem:predicate-set-restr-subsets-md}
	Let \scfoidsmd be the sequent calculus obtained from \scfoid by restricting $\Pi$ in \emph{(def L)} to subsets of sets of mutually defined predicates $\md{\defn}{P}$.
	A sequent is \scfoid-provable iff it is \scfoidsmd-provable.
\end{restatable}

\begin{proof}
	See Appendix \ref{app:proofs-section-scfoid}.
\end{proof}

\begin{restatable}{lemma}{predicatesetlarger} \label{lem:predicate-set-larger}
	Let 
	\begin{equation} \label{eq:def-L-lem-predlarger}
		\begin{prooftree}
			\hypo{\text{minor premises}}
			\hypo{\defn, \Gamma, \ihf{P}{\bar{t}} \vdash \Delta}
			\infer2[\textup{(def L)}]{\defn, \Gamma, P(\bar{t}) \vdash \Delta}
		\end{prooftree}
	\end{equation}
	be an instance of \textup{(def L)} with predicate set $\Pi$, and let $\Pi' \subseteq \defi{\defn}$ be a superset of\/ $\Pi$.
	Then there exists an \scfoid-derivation of $\defn, \Gamma, P(\bar{t}) \vdash \Delta$ from the premises of (\ref{eq:def-L-lem-predlarger}) with a single instance of \textup{(def L)}, with predicate set $\Pi'$.
\end{restatable}

\begin{proof}
	See Appendix \ref{app:proofs-section-scfoid}.
\end{proof}

\begin{proposition} \label{prop:predicate-set-restr-md}
	Let \scfoidmd be the sequent calculus obtained from \scfoid by restricting $\Pi$ in \emph{(def L)} to sets of mutually defined predicates $\md{\defn}{P}$.
	A sequent is \scfoid-provable iff it is \scfoidmd-provable.
\end{proposition}

\begin{proof}
	The `if' part is trivial, as every \scfoidmd-proof is also an \scfoid-proof.
	For the `only if' part, let $\seqshort$ be an \scfoid-provable sequent.
	By Lemma \ref{lem:predicate-set-restr-subsets-md}, $\seqshort$ is also \scfoidsmd-provable.
	Let $\prf$ be an \scfoidsmd-proof of $\seqshort$.
	By Lemma \ref{lem:predicate-set-larger}, we can replace every instance (\ref{eq:def-L-lem-predlarger})
	of (def L) in $\prf$ by an \scfoid-derivation of $\defn, \Gamma, P(\bar{t}) \vdash \Delta$ from the premises of (\ref{eq:def-L-lem-predlarger}), in which the only instance of (def L) has predicate set $\md{\defn}{P}$, since $\md{\defn}{P}$ is a superset of the predicate set $\Pi$ of (\ref{eq:def-L-lem-predlarger}).
	This results in an \scfoidmd-proof of $\seqshort$, showing that $\seqshort$ is \scfoidmd-provable.
\end{proof}

%% file: correspondence.tex
\section{Correspondence to \texorpdfstring{\scfo}{SCFO}} \label{sec:correspondence}

The introduction rules for defined atoms in \scfoid are based on two first-order 
properties of \foid-definitions: the right introduction rule (def R) relies on the fact that bodies imply heads, and the left introduction rule (def L) relies on a generalization of the induction scheme.
In this section, we make these underlying first-order principles explicit by establishing a correspondence between \scfoid-theorems and \scfo-theorems.
In the \scfo-theorems, \foid-definitions are replaced by their \emph{first-order approximations}:

Given a definition $\defn$, we let $\matimps{\defn}$ denote the set of material implications $\forall \bar{x} : \formula \Rightarrow P(\bar{t})$ corresponding to the definitional rules $\forall \bar{x} : P(\bar{t}) \rul \formula$ of $\defn$.
With $\indscheme{\defn}$, we denote the \emph{induction scheme} corresponding to $\defn$.
This is the (countably infinite) set of \fo-formulas of the form
\begin{equation*}
	\bigwedge_{\Pi} (\forall \bar{y} : \formulatwo\substpos{\ihs{\Pi}}{\Pi} \Rightarrow \ihf{Q}{\bar{s}})
	\Rightarrow (\forall \bar{x} :  P(\bar{x}) \Rightarrow \ihf{P}{\bar{x}} )  
\end{equation*}
where $P$ is a defined predicate of $\defn$, $\Pi$ is a subset of $\defi{\defn}$ containing $P$, and for every $Q \in \Pi$,
$\ihs{Q}$ is an induction hypothesis associated with $Q$.
The conjunction $\bigwedge_{\Pi} (\forall \bar{y} : \formulatwo\substpos{\ihs{\Pi}}{\Pi} \Rightarrow \ihf{Q}{\bar{s}})$ ranges over all definitional rules $\forall \bar{y} : Q(\bar{s}) \rul \formulatwo$ in $\defn$ with $Q \in \Pi$.
The \emph{first-order approximation} $\foapp{\defn}$ of $\defn$ is the set $\matimps{\defn} \cup \indscheme{\defn}$.

As a consequence of Propositions \ref{prop:completion-wf} and \ref{prop:induction-scheme-wf},
any definition $\defn$ logically entails its first-order approximation $\foapp{\defn}$ under the well-founded %
semantics.
Conversely, $\foapp{\defn}$ does not always logically entail $\defn$.
A simple counterexample is the definition of the natural numbers $\defn_{\mathit{Nat}}$ from Section \ref{sec:introduction}, since non-standard models of first-order Peano arithmetic satisfy $\foapp{\defn_{\mathit{Nat}}}$ but not $\defn_{\mathit{Nat}}$.

\begin{lemma} \label{lem:proof-FO-principles-from-def}
	Let $\defn$ be a definition and $\formulathree \in \foapp{\defn}$.
	Then $\defn \vdash \formulathree$ is \scfoid-provable.
	\hfill\markfull
\end{lemma}

\begin{proof}
	\textbf{Case}:
	$\formulathree \in \matimps{\defn}$. 
	Then $\formulathree$ is of the form $\forall \bar{x} : \formula \Rightarrow P(\bar{t})$ for a definitional rule $\forall \bar{x} : P(\bar{t}) \rul \formula$ of $\defn$.
	We can prove $\defn \vdash \formulathree$ as follows:
	\begin{equation*}
		\begin{prooftree}
			\hypo{}
			\infer1[(ax)]{\defn, \formula \vdash \formula}
			\infer1[(def R)]{\defn, \formula \vdash P(\bar{t})}
			\infer1[($\Rightarrow$R)]{\defn \vdash \formula \Rightarrow P(\bar{t})}
			\infer1[($\forall$R)]{\defn \vdash \forall \bar{x} : \formula \Rightarrow P(\bar{t})}
		\end{prooftree}
	\end{equation*}
	
	\textbf{Case:}
	$\formulathree \in \indscheme{\defn}$. 
	Then $\formulathree$ is of the form $\bigwedge_{\Pi} (\forall \bar{y} : \formulatwo\substpos{\ihs{\Pi}}{\Pi} \Rightarrow \ihf{Q}{\bar{s}})
	\Rightarrow (\forall \bar{x} :  P(\bar{x}) \Rightarrow \ihf{P}{\bar{x}} )$
	for a defined predicate $P$ of $\defn$, a subset $\Pi$ of $\defi{\defn}$ containing $P$, and for every $Q \in \Pi$, an induction hypothesis $\ihs{Q}$ associated with $Q$.
	We can prove $\defn \vdash \formulathree$ as follows:
	\begin{equation*}
		\begin{prooftree}
			\hypo{\text{minor premises}}
			\hypo{}
			\infer1[(ax)]{\defn, \bigwedge_{\Pi} (\forall \bar{y} : \formulatwo\substpos{\ihs{\Pi}}{\Pi} \Rightarrow \ihf{Q}{\bar{s}}),  \ihf{P}{\bar{v}} \vdash \ihf{P}{\bar{v}}}
			\infer2[(def L)]{\defn, \bigwedge_{\Pi} (\forall \bar{y} : \formulatwo\substpos{\ihs{\Pi}}{\Pi} \Rightarrow \ihf{Q}{\bar{s}}), P(\bar{v}) \vdash \ihf{P}{\bar{v}}}
			\infer1[($\Rightarrow$R)]{\defn, \bigwedge_{\Pi} (\forall \bar{y} : \formulatwo\substpos{\ihs{\Pi}}{\Pi} \Rightarrow \ihf{Q}{\bar{s}}) \vdash P(\bar{v}) \Rightarrow \ihf{P}{\bar{v}}}
			\infer1[($\forall$R)]{\defn, \bigwedge_{\Pi} (\forall \bar{y} : \formulatwo\substpos{\ihs{\Pi}}{\Pi} \Rightarrow \ihf{Q}{\bar{s}}) \vdash \forall \bar{x} :  P(\bar{x}) \Rightarrow \ihf{P}{\bar{x}}}
			\infer1[($\Rightarrow$R)]{\defn \vdash \bigwedge_{\Pi} (\forall \bar{y} : \formulatwo\substpos{\ihs{\Pi}}{\Pi} \Rightarrow \ihf{Q}{\bar{s}}) \Rightarrow (\forall \bar{x} :  P(\bar{x}) \Rightarrow \ihf{P}{\bar{x}} )}
		\end{prooftree}
	\end{equation*}
	Here $\bar{v}$ is a tuple of variables that do not occur freely in $\defn$ or $\bigwedge_{\Pi} (\forall \bar{y} : \formulatwo\substpos{\ihs{\Pi}}{\Pi} \Rightarrow \ihf{Q}{\bar{s}})$, with the same length as $\bar{x}$.
	We apply (def L) with the predicate set $\Pi$ and with the induction hypotheses $\ihs{Q} \doteq \setexp{\bar{z}}{\mbd{Q}{\defn}(\bar{z})}$ for every $Q \in \Pi$.
	For each definitional rule $\forall \bar{z} : R(\bar{u}) \rul \formulathree$ in $\defn$ with $R \in \Pi$, this rule instance has a minor premise $\defn, \bigwedge_{\Pi} (\forall \bar{y} : \formulatwo\substpos{\ihs{\Pi}}{\Pi} \Rightarrow \ihf{Q}{\bar{s}}), \formulathree\substpos{\ihs{\Pi}}{\Pi} \vdash \ihf{R}{\bar{u}}, \ihf{P}{\bar{v}}$.
	By ($\land$L) and (wk), it suffices to prove $\defn, \forall \bar{z} : \formulathree\substpos{\ihs{\Pi}}{\Pi} \Rightarrow \ihf{R}{\bar{u}},  \formulathree\substpos{\ihs{\Pi}}{\Pi} \vdash \ihf{R}{\bar{u}}$ for any such definitional rule, which we do as follows:
	\begin{equation*}
		\begin{prooftree}
			\infer0[(ax)]{\defn, \formulatwo\substpos{\ihs{\Pi}}{\Pi} \vdash \formulatwo\substpos{\ihs{\Pi}}{\Pi}, \ihf{R}{\bar{u}}}
			\infer0[(ax)]{\defn, \ihf{R}{\bar{u}}, \formulathree\substpos{\ihs{\Pi}}{\Pi} \vdash \ihf{R}{\bar{u}}}
			\infer2[($\Rightarrow$L)]{\defn, \formulathree\substpos{\ihs{\Pi}}{\Pi} \Rightarrow \ihf{R}{\bar{u}}, \formulathree\substpos{\ihs{\Pi}}{\Pi} \vdash \ihf{R}{\bar{u}}}
			\infer1[($\forall$L)]{\defn, \forall \bar{z} : \formulathree\substpos{\ihs{\Pi}}{\Pi} \Rightarrow \ihf{R}{\bar{u}}, \formulathree\substpos{\ihs{\Pi}}{\Pi} \vdash \ihf{R}{\bar{u}}}
		\end{prooftree}
	\end{equation*}
\end{proof}

As with most of the subsequent proof-theoretic results, the correspondence theorem restricts to \emph{regular sequents}.
These are sequents of the form $\regularsequent$, where $\defns$ are \foid-definitions and $\Gamma$ and $\Delta$ contain only \fo-formulas.
Even though non-regular sequents are generally sensible and useful (as illustrated in Examples \ref{ex:proof-reachable}, \ref{ex:proof_non-total_liar} and \ref{ex:proof_access}), many statements about inductive definitions are regular, as they are often of the form ``given these definitions, the following property holds''.
The correspondence theorem says that the \scfoid-provability of a regular sequent $\regularsequent$ comes down to the \scfo-provability of the sequent $\foapp{\defn_1}, \dots, \foapp{\defn_n}, \Gamma \vdash \Delta$, obtained from $\regularsequent$ by replacing every definition $\defn_i$ by its first-order approximation $\foapp{\defn_i}$.
Note that $\foapp{\defn_1}, \dots, \foapp{\defn_n}, \Gamma \vdash \Delta$ is a infinite sequent.
As for finite sequents, an infinite sequent $\sequent$ is said to be \seqcal-provable for a sequent calculus \seqcal if there exists a (finite) \seqcal-proof of $\sequent$. %

\begin{theorem}[Correspondence] \label{thm:correspondence}
	Let $\regularsequent$ be a regular sequent.
	Then $\regularsequent$ is \scfoid-provable iff\/ $\foapp{\defn_1}, \dots, \foapp{\defn_n}, \Gamma \vdash \Delta$ is \scfo-provable.
	\hfill\markfull
\end{theorem}

\begin{proof}
	Suppose that $\regularsequent$ is \scfoid-provable.
	Then there is an \scfoid-proof $\prf$ of $\regularsequent$.
	We show that $\foapp{\defn_1}, \dots, \foapp{\defn_n}, \Gamma \vdash \Delta$ is \scfo-provable by induction on the height of $\prf$.
	
	In the base case, %
	$\prf$ derives $\defn_1, \dots, \defn_n, \Gamma \vdash \Delta$ by (ax) or (=R). 
	This means that $\Gamma \cap \Delta \neq \emptyset$ or $\top \in \Delta$ or $\bot \in \Gamma$, or if $\Delta$ contains an atom of the form $t=t$. 
	In either case, $\foapp{\defn_1}, \dots, \foapp{\defn_n}, \Gamma \vdash \Delta$ can be proven in \scfo by the same inference rule.
	For the inductive step, most cases are trivial.
	The only non-trivial cases are where $\defn_1, \dots, \defn_n, \Gamma \vdash \Delta$ is derived by (def R) or by (def L): 
	
	\textbf{Case}: $\defn_1, \dots, \defn_n, \Gamma \vdash \Delta$ is derived by (def R).
	Then the final step in $\prf$ is of the form
	\begin{equation*}
		\begin{prooftree}
			\hypo{\defns, \Gamma \vdash \formula[\bar{s}/\bar{x}], \Delta'}
			\infer1[(def R)]{\defns, \Gamma \vdash P(\bar{t}[\bar{s}/\bar{x}]), \Delta'}
		\end{prooftree}
	\end{equation*}
	where $\Delta = \Delta' \cup \{P(\bar{t}[\bar{s}/\bar{x}])\}$.
	By induction hypothesis, $\foapp{\defn_1}, \dots, \foapp{\defn_n}, \Gamma \vdash \formula[\bar{s}/\bar{x}], \Delta'$ is \scfo-provable.
	The following \scfo-derivation shows that $\foapp{\defn_1}, \dots, \foapp{\defn_n}, \Gamma \vdash \Delta$ is \scfo-provable as well:
	\begin{equation*}
		\begin{prooftree}
			\hypo{\foapp{\defn_1}, \dots, \foapp{\defn_n}, \Gamma \vdash \formula[\bar{s}/\bar{x}], \Delta'}
			\infer1[(wk)]{\foapp{\defn_1}, \dots, \foapp{\defn_n}, \Gamma \vdash \formula[\bar{s}/\bar{x}], P(\bar{t}[\bar{s}/\bar{x}]), \Delta'}
			\hypo{}
			\infer1[(ax)]{\foapp{\defn_1}, \dots, \foapp{\defn_n}, \Gamma, P(\bar{t}[\bar{s}/\bar{x}]) \vdash P(\bar{t}[\bar{s}/\bar{x}]), \Delta'}
			\infer2[($\Rightarrow$L)]{\foapp{\defn_1}, \dots, \foapp{\defn_n}, \Gamma, \formula[\bar{s}/\bar{x}] \Rightarrow P(\bar{t}[\bar{s}/\bar{x}]) \vdash P(\bar{t}[\bar{s}/\bar{x}]), \Delta'}
			\infer1[($\forall$L)]{\foapp{\defn_1}, \dots, \foapp{\defn_n}, \Gamma \vdash P(\bar{t}[\bar{s}/\bar{x}]), \Delta'}
		\end{prooftree}
	\end{equation*}
	In the application of ($\forall$L), we used that $\forall \bar{x}: \formula \Rightarrow P(\bar{t}) \in \foapp{\defn_i}$ for the definition $\defn_i$ such that $\defrul \in \defn_i$.
	
	\textbf{Case}: $\defn_1, \dots, \defn_n, \Gamma \vdash \Delta$ is derived by (def L).
	Then the final step in $\prf$ is of the form
	\begin{equation*}
		\begin{prooftree}
			\hypo{\text{minor premises}}
			\hypo{\defns, \Gamma', \ihf{P}{\bar{t}} \vdash \Delta}
			\infer2[(def L)]{\defns, \Gamma', P(\bar{t}) \vdash \Delta}
		\end{prooftree}
	\end{equation*}
	where $\Gamma = \Gamma' \cup \{P(\bar{t})\}$.
	For every definitional rule $\forall \bar{y} : Q(\bar{s}) \rul \formulatwo$ in $\defn$ with $Q \in \Pi$, this rule has a minor premise $\defns, \Gamma', \formulatwo\substposihpi \vdash \ihf{Q}{\bar{s}}, \Delta$.
	By the condition on (def L), the variables in $\bar{y}$ do not occur freely in $\defns$, $\Gamma'$ or $\Delta$.
	Consequently, they do not occur freely in $\foapp{\defn_1}, \dots, \foapp{\defn_n}$.
	By induction hypothesis, 
	\begin{equation} \label{eq:fo-app-minor-premise}
		\foapp{\defn_1}, \dots, \foapp{\defn_n}, \Gamma', \formulatwo\substposihpi \vdash \ihf{Q}{\bar{s}}, \Delta
	\end{equation}
	is \scfo-provable for any definitional rule $\forall \bar{y} : Q(\bar{s}) \rul \formulatwo$ in $\defn$ with $Q \in \Pi$, and so is
	\begin{equation} \label{eq:fo-app-major-premise}
		\foapp{\defn_1}, \dots, \foapp{\defn_n}, \Gamma', \ihf{P}{\bar{t}} \vdash \Delta .
	\end{equation}
	From (\ref{eq:fo-app-minor-premise}), we can derive the sequent $\foapp{\defn_1}, \dots, \foapp{\defn_n}, \Gamma' \vdash \forall \bar{y} : \formulatwo\substposihpi \Rightarrow \ihf{Q}{\bar{s}}, \Delta$ as follows:
	\begin{equation*}
		\begin{prooftree}
			\hypo{\foapp{\defn_1}, \dots, \foapp{\defn_n}, \Gamma', \formulatwo\substposihpi \vdash \ihf{Q}{\bar{s}}, \Delta}
			\infer1[($\Rightarrow$R)]{\foapp{\defn_1}, \dots, \foapp{\defn_n}, \Gamma' \vdash \formulatwo\substposihpi \Rightarrow \ihf{Q}{\bar{s}}, \Delta}
			\infer1[($\forall$R)]{\foapp{\defn_1}, \dots, \foapp{\defn_n}, \Gamma' \vdash \forall \bar{y} : \formulatwo\substposihpi \Rightarrow \ihf{Q}{\bar{s}}, \Delta}
		\end{prooftree}
	\end{equation*}
	In the application of ($\forall$R), we used that the variables in $\bar{y}$ do not occur freely in $\foapp{\defn_1}, \dots, \foapp{\defn_n}$, $\Gamma'$ or $\Delta$.
	By repeatedly application of ($\land$L) and (wk), we can derive 
	\begin{equation} \label{eq:antecedent-indscheme}
		\foapp{\defn_1}, \dots, \foapp{\defn_n}, \Gamma', P(\bar{t}) \vdash \bigwedge_\Pi (\forall \bar{y} : \formulatwo\substposihpi \Rightarrow \ihf{Q}{\bar{s}}), \Delta .
	\end{equation}
	From (\ref{eq:fo-app-major-premise}), we can derive the sequent 
	\begin{equation} \label{eq:consequent-indscheme}
		\foapp{\defn_1}, \dots, \foapp{\defn_n}, \Gamma', P(\bar{t}), \forall \bar{x} :  P(\bar{x}) \Rightarrow \ihf{P}{\bar{x}} \vdash \Delta 
	\end{equation}
	as follows:
	\begin{equation*}
		\begin{prooftree}
			\hypo{}
			\infer1[(ax)]{\foapp{\defn_1}, \dots, \foapp{\defn_n}, \Gamma', P(\bar{t}) \vdash P(\bar{t}), \Delta}
			\hypo{\foapp{\defn_1}, \dots, \foapp{\defn_n}, \Gamma', \ihf{P}{\bar{t}} \vdash \Delta}
			\infer1[(wk)]{\foapp{\defn_1}, \dots, \foapp{\defn_n}, \Gamma', P(\bar{t}), \ihf{P}{\bar{t}} \vdash \Delta}
			\infer2[($\Rightarrow$L)]{\foapp{\defn_1}, \dots, \foapp{\defn_n}, \Gamma', P(\bar{t}), P(\bar{t}) \Rightarrow \ihf{P}{\bar{t}} \vdash \Delta}
			\infer1[($\forall$L)]{\foapp{\defn_1}, \dots, \foapp{\defn_n}, \Gamma', P(\bar{t}), \forall \bar{x} :  P(\bar{x}) \Rightarrow \ihf{P}{\bar{x}} \vdash \Delta}
		\end{prooftree}
	\end{equation*}
	Finally, we can derive $\foapp{\defn_1}, \dots, \foapp{\defn_n}, \Gamma', P(\bar{t}) \vdash \Delta$ by applying ($\Rightarrow$L) to (\ref{eq:antecedent-indscheme}) and (\ref{eq:consequent-indscheme}).
	Indeed, this is because $\bigwedge_\Pi (\forall \bar{y} : \formulatwo\substposihpi \Rightarrow \ihf{Q}{\bar{s}}) \Rightarrow (\forall \bar{x} :  P(\bar{x}) \Rightarrow \ihf{P}{\bar{x}})$ is in $\foapp{\defn_i}$ for the definition $\defn_i$ defining $P$.
	This shows that $\foapp{\defn_1}, \dots, \foapp{\defn_n}, \Gamma \vdash \Delta$ is \scfo-provable.
	
	Suppose conversely that $\foapp{\defn_1}, \dots, \foapp{\defn_n}, \Gamma \vdash \Delta$ is provable in \scfo.
	Since proofs are finite, there must be some finite subset $X$ of $\foapp{\defn_1} \cup \dots \cup \foapp{\defn_n}$ such that $\Gamma, X \vdash \Delta$ is \scfo-provable.\footnote{
		This can be shown by induction on the height of a proof of $\foapp{\defn_1}, \dots, \foapp{\defn_n}, \Gamma \vdash \Delta$.
	}
	Since every inference rule of \scfo is also an inference rule of \scfoid, it follows that $\Gamma, X \vdash \Delta$ is also \scfoid-provable.
	By (wk), 
	\begin{equation} \label{eq:Xcut1}
		\defns, \Gamma, X \vdash \Delta
	\end{equation}
	is \scfoid-provable.
	Pick some $\formulathree \in X$.
	By Lemma \ref{lem:proof-FO-principles-from-def}, $\defn_i \vdash \formulathree$ is \scfoid-provable for the definition $\defn_i$ corresponding to $\formulathree$, and by application of (wk), so is 
	\begin{equation} \label{eq:Xcut2}
		\defns, \Gamma, X \setminus \{\formulathree\} \vdash \formulathree, \Delta .
	\end{equation}
	Applying (cut) to (\ref{eq:Xcut2}) and (\ref{eq:Xcut1}), we find that $\defns, \Gamma, X \setminus \{\formulathree\} \vdash \Delta$ is \scfoid-provable as well:
	\begin{equation*}
		\begin{prooftree}
			\hypo{\defns, \Gamma, X \setminus \{\formulathree\} \vdash \formulathree, \Delta}
			\hypo{\defns, \Gamma, X \vdash \Delta}
			\infer2[(cut)]{\defns, \Gamma, X \setminus \{\formulathree\} \vdash \Delta}
		\end{prooftree}
	\end{equation*}
	Repeating this argument a finite number of times, we obtain that $\defns, \Gamma \vdash \Delta$ is \scfoid-provable.
\end{proof}

The correspondence theorem characterizes the deductive strength of a definition $\defn$ in \scfoid by its first-order approximation $\foapp{\defn} = \matimps{\defn} \cup \indscheme{\defn}$.
Note that the induction scheme $\indscheme{\defn}$ corresponding to a definition $\defn$ is generally weaker than the induction axioms $\indax{\defn, \Pi}$ for $\Pi \subseteq \defi{\defn}$ from Definition \ref{def:induction-axiom}.
This is because the second-order variables in the induction axioms range over \emph{all} relations over the domain, while the induction scheme restricts to the relations defined by \fo-formulas.
For an infinite domain $D$, there are uncountably many relations over $D$, but only countably many \fo-definable ones.

%% file: alternative-semantics.tex
\section{Alternative Semantics} \label{sec:alternative-semantics}

\subsection{Stable Semantics} \label{sec:stable-semantics}

The induction rule (def L) of \scfoid handles the non-monotonicity of \foid-definitions in a surprisingly simple way: 
in the minor premises, only positive occurrences of defined predicates are replaced by an induction hypothesis, whereas negative occurrences remain untouched.
While this may appear as an arbitrary and perhaps naive approach, it emerges naturally from the \emph{stable (model) semantics} for logic programming \cite{iclp/GelfondL88}.
This alternative to the well-founded semantics underlies the paradigm of \emph{answer set programming (ASP)} \cite{marek99stable}.
Based on the presentation in \cite{tocl/DeneckerT08}, we will define the stable semantics for \foid,\footnote{
	Even though Denecker and Ternovska do not explicitly define the notion of stable model, they define the notion of \emph{stable operator}, the fixpoints of which are stable models.
} 
allowing us to capture (fragments of) logic programs under the stable semantics in \foid.
By its non-constructive nature, the stable semantics is not a semantics of inductive definitions.
Furthermore, an \foid-definition (or program) $\defn$ %
can have multiple
stable models in a $\defn$-context, whereas %
a definition ought to provide a single interpretation of its defined predicates given an interpretation of its parameters.
Nevertheless, the stable semantics coincides with the well-founded semantics on all `sensible' definitions $\defn$ and $\defn$-contexts $\context$.
More concretely, if the well-founded model $\structure$ of a definition $\defn$ in a $\defn$-context $\context$ is two-valued, then $\defn$ has a unique stable model in $\context$, which is precisely $\structure$.

Henceforth, we fix an \foid-definition $\defn$ and a $\defn$-context $\context$.
The \emph{structure space} $\structurespace$ is the set of all $\sym{\defn}$-structures expanding $\context$.
For $\structure, \structuretwo \in \structurespace$, we write $\structure \sqsubseteq \structuretwo$ if $P^{\structure} \subseteq P^{\structuretwo}$ for all $P \in \defi{\defn}$.
As can be easily checked, the relation $\sqsubseteq$ defines a \emph{complete lattice order} on $\structurespace$.
This means that any subset $A$ of $\structurespace$ has a (unique) least upper bound and greatest lower bound, which we denote by $\sqcup A$ and $\sqcap A$, respectively.
Let $\defi{\defn}'$ be a set that contains, for each $P \in \defi{\defn}$, a predicate symbol $P'$ with the same arity as $P$ that does not occur in $\defn$.
Given $\structure, \structuretwo \in \structurespace$, we define $\structure_{\structuretwo}$ as the expansion of $\structure$ with vocabulary $\sym{\defn} \cup \defi{\defn}'$ that interprets every $P'$ as $P^{\structuretwo}$ for every $P \in \defi{\defn}$. 
In short, $\structure_{\structuretwo} = \structure[\defi{\defn}' : \defi{\defn}^{\structuretwo}]$.
Given an \fo-formula $\formula$, we denote by $\formula'$ the formula obtained from $\formula$ by replacing all negative occurrences of defined atoms $P(\bar{t})$ of $\defn$ with $P'(\bar{t})$.
We define the operator $\cons : \structurespace \times \structurespace \to \structurespace$ by $P^{\cons(\structure, \structuretwo)} \coloneq \{ \bar{a} \mid \structure_{\structuretwo}[\bar{y} : \bar{a}] \models \mbd{P}{\defn}(\bar{y})' \}$ for all $P \in \defi{\defn}$, where $\bar{y}$ is an $\arity{P}$-tuple of variables not occurring %
in $\defn$.\footnote{
	We only need to define the interpretation of defined predicates, as the interpretation of parameters is fixed by $\context$.
}
Given $\structuretwo \in \structurespace$, we denote by $\cons(\cdot, \structuretwo)$ the unary operator $\structurespace \to \structurespace: \structure \mapsto \cons(\structure, \structuretwo)$.
It is straightforward to verify that $\cons(\cdot, \structuretwo)$ is \emph{monotone}, i.e., that $\cons(\structure_1, \structuretwo) \sqsubseteq \cons(\structure_2, \structuretwo)$ if $\structure_1 \sqsubseteq \structure_2$ for all $\structure_1, \structure_2 \in \structurespace$.
Since $\cons(\cdot, \structuretwo)$ is monotone, and since $(\structurespace, \sqsubseteq)$ is a complete lattice, the Knaster-Tarski theorem \cite{Tarski} implies that $\cons(\cdot, \structuretwo)$ has a (unique) \emph{least fixpoint} $\lfp(\cons(\cdot, \structuretwo))$, i.e., an $\structure \in \structurespace$ such that $\structure = \cons(\structure, \structuretwo)$, and $\structure \sqsubseteq \structure'$ for every $\structure' \in \structurespace$ such that $\structure '= \cons(\structure', \structuretwo)$.
Furthermore, this least fixpoint is also the least pre-fixpoint of $\cons(\cdot, \structuretwo)$, i.e., the least $\structure \in \structurespace$ such that $\cons(\structure, \structuretwo) \sqsubseteq \structure$.
We define the \emph{stable operator} $\stabop : \structurespace \to \structurespace$ %
by
\begin{equation*}
	\stabop(\structure) \coloneq \lfp(\cons(\cdot, \structure)) .
\end{equation*}
The fixpoints of $\stabop$ are called the \emph{stable models} of $\defn$ in $\context$.
The \emph{stable satisfaction relation} $\modelsst$ is defined analogously to the well-founded satisfaction relation, the only difference being the rule for \foid-definitions:
\begin{itemize}
	\item $\structure \modelsst \defn$ if $\structure|_{\sym{\defn}}$ is a stable model of $\defn$ in $\structure|_{\pars{\defn}}$.
\end{itemize}

\begin{remark}
	Since stable models are not defined in terms of individual definitional rules, but only in terms of merged bodies $\mbd{P}{\defn}(\bar{y})$, any definition is equivalent to its normalization under the stable semantics.
	\hfill\markfull
\end{remark}

To make the link between \scfoid and the stable semantics more explicit, let us unfold what it means for a structure $\structure \in \structurespace$ to be a stable model of $\defn$ in $\context$.
First, $\structure$ needs to be a fixpoint of $\cons(\cdot, \structure)$.
This means that $P^{\structure} = P^{\cons(\structure, \structure)} = \{ \bar{a} \mid \structure_\structure[\bar{y} : \bar{a}] \models \mbd{P}{\defn}(\bar{y})' \} = \{ \bar{a} \mid \structure[\bar{y} : \bar{a}] \models \mbd{P}{\defn}(\bar{y}) \}$ for every $P \in \defi{\defn}$.
In other words, $\structure$ needs to satisfy $\forall \bar{y} : P(\bar{y}) \Leftrightarrow \mbd{P}{\defn}(\bar{y})$ for every $P \in \defi{\defn}$.
Second, $\structure$ needs to be the \emph{least} fixpoint, or equivalently, pre-fixpoint, of $\cons(\cdot, \structure)$.
This means that for every structure $\structuretwo \in \structurespace$, if $\cons(\structuretwo, \structure) \sqsubseteq \structuretwo$, then $\structure \sqsubseteq \structuretwo$.
Unraveling the definition of $\sqsubseteq$, the statement $\structure \sqsubseteq \structuretwo$ says that $P^{\structure} \subseteq P^{\structuretwo}$ for every $P \in \defi{\defn}$, and $\cons(\structuretwo, \structure) \sqsubseteq \structuretwo$ says that $\{ \bar{a} \mid \structuretwo_\structure[\bar{y} : \bar{a}] \models \mbd{P}{\defn}(\bar{y})' \} = P^{\cons(\structuretwo, \structure)} \subseteq P^{\structuretwo}$ for every $P \in \defi{\defn}$.
In the evaluation $\structuretwo_\structure[\bar{y} : \bar{a}] \models \mbd{P}{\defn}(\bar{y})'$, the positive occurrences of defined predicates in $\mbd{P}{\defn}(\bar{y})'$ are interpreted by $\structuretwo$, and the negative occurrences by $\structure$.
Since $\structuretwo$ can interpret the defined predicates $P$ of $\defn$ as any $\arity{P}$-relation over $\dom{\structure}$, we can restate the fact that $\structure$ is the least pre-fixpoint of $\cons(\cdot, \structure)$ as %
\begin{equation} \label{eq:ind-axiom-stable}
	\structure \models \forall G_{\defi{\defn}}: \Big(  \bigwedge_{Q \in \defi{\defn}} \forall \bar{y} : \mbd{Q}{\defn}(\bar{y})\substpos{G_{\defi{\defn}}}{\defi{\defn}} \Rightarrow G_Q(\bar{y}) \Big) \Rightarrow \Big( \bigwedge_{P\in\defi{\defn}} \forall \bar{x} : P(\bar{x}) \Rightarrow G_P(\bar{x}) \Big) 
\end{equation}
From Definition \ref{def:merged-body}, it is not difficult to see that the second-order formula in (\ref{eq:ind-axiom-stable}) is equivalent to $\indax{\defn, \defi{\defn}}$ from Definition \ref{def:induction-axiom}.
Thus, the stable models of $\defn$ in $\context$ are characterized precisely by $\indax{\defn, \defi{\defn}}$ and the equivalences $\forall \bar{y} : P(\bar{y}) \Leftrightarrow \mbd{P}{\defn}(\bar{y})$.

By the reasoning in the previous paragraph, analogous versions of Propositions \ref{prop:completion-wf} and \ref{prop:induction-scheme-wf} hold for the stable semantics.
However, we have only argued that an analogous version of \ref{prop:induction-scheme-wf} holds for $\Pi = \defi{\defn}$.
The result for general $\Pi \subseteq \defi{\defn}$ follows from the fact that $\structure \modelsst \defn$ implies $\structure \modelsst \defn'$ if $\defn'$ is a \emph{subdefinition} of $\defn$.
We say that $\defn'$ is a subdefinition of $\defn$ if $\defn'$ is a subset of $\defn$ such that, if $\defn'$ contains a definitional rule of $\defn$ with $P$ in its head, then $\defn'$ contains \emph{all} definitional rules of $\defn$ with $P$ in its head.
We will prove this property more generally for the \emph{Henkin semantics} in Section \ref{sec:henkin-semantics}.

The well-founded model of $\defn$ in $\context$, as defined in Section \ref{sec:well-founded-semantics}, can also be defined in terms of the stable operator $\stabop$ \cite{lpnmr/DeneckerV07,tocl/DeneckerT08}.
This is done by viewing three-valued structures $\structure$ as pairs $(\structure_l, \structure_u)$ of two-valued structures such that $\structure_l \leqt \structure_u$, as described in Appendix \ref{app:induction-scheme-wf}.
These pairs have the property that for any formula $\formula$ interpreted by $\structure$, $\formula^{\structure} = \true$ iff $\formula^{\structure_l} = \formula^{\structure_u} = \true$, $\formula^{\structure} = \false$ iff $\formula^{\structure_l} = \formula^{\structure_u} = \false$, and $\formula^{\structure} = \unknown$ iff $\formula^{\structure_l} = \false$ and $\formula^{\structure_u} = \true$.
The well-founded model of $\defn$ in $\context$ corresponds to the \emph{maximal oscillating pair} of $\stabop$.
This is the pair $(\structure, \structuretwo)$ consisting of the $\leqt$-minimal $\structure \in \structurespace$ and the $\leqt$-maximal $\structuretwo \in \structurespace$ such that $\stabop(\structure) = \structuretwo$ and $\stabop(\structuretwo) = \structure$.
Equivalently, $\structure$ is the $\leqt$-minimal and $\structuretwo$ the $\leqt$-maximal fixpoint of $\stabop \circ \stabop$. %
Translating a result by Van Gelder et al.~\cite{GelderRS91} to the logic \foid, Denecker and Ternovska showed that if the well-founded model $\structure$ of $\defn$ in $\context$ is two-valued, then there is precisely one stable model of $\defn$ in $\context$, which is $\structure$ \cite{tocl/DeneckerT08}.
Hence, the well-founded and the stable semantics coincide whenever $\defn$ is total in $\context$.

\subsection{Henkin Semantics} \label{sec:henkin-semantics}

We generalize the stable semantics to the \emph{Henkin semantics} for \foid.
This semantics is based on the homonymous semantics for the inductive definition logic of Brotherston and Simpson \cite{jlc/BrotherstonS11}, which is in turn inspired by Henkin's semantics for second- and higher-order logic \cite{jsl/Henkin50}.
In the Henkin semantics for \foid, a (two-valued) structure $\structure$ is supplemented with a \emph{Henkin class} $\henkinclass$ consisting of relations over $\dom{\structure}$.
The \emph{Henkin models} of a definition $\defn$ are then defined by restricting the structure space $\structurespace$ to those $\sym{\defn}$-structures that interpret all defined predicates of $\defn$ as relations in $\henkinclass$.
The relations in a Henkin class must be sufficiently rich.
Given a set $A$, we denote with $\Pow{A}$ the power set of $A$. 

\begin{definition} \label{def:henkin-class}
	Let $\structure$ be a structure with domain $D$.
	A \emph{Henkin class} $\henkinclass$ for $\structure$ is a set $\{ H_k \mid k \in \N \}$ such that $H_k \subseteq \Pow{D^k}$ for all $k \in \N$, and
	\begin{enumerate}[label=(\roman*)]
		\item $\{(a,a) \mid a \in D\} \in H_2$; \label{it:henkin-class-equality}
		\item $Q^{\structure} \in H_{\arity{Q}}$ for every predicate symbol $Q$ in $\voc{\structure}$; \label{it:henkin-class-predicate}
		\item if $R \in H_{k+1}$ and $a \in D$, then $\{(a_1, \dots, a_k) \in D^k \mid (a_1, \dots, a_k, a) \in R\} \in H_k$; \label{it:henkin-class-projection}
		\item if $R \in H_k$, $x_1, \dots, x_n$ are distinct object symbols not in $\voc{\structure}$ and $\bar{t} =  (t_1, \dots, t_k)$ a tuple of terms over $\voc{\structure} \cup \{x_1, \dots, x_n\}$, then
		$\{ (a_1, \dots, a_n) \in D^n \mid \bar{t}^{\structure[x_1 : a_1, \dots, x_n : a_n]} \in R\} \in H_n$; 
		\label{it:henkin-class-terms}
		\item if $R \in H_k$, then $D^k \setminus R \in H_k$; \label{it:henkin-class-complement}
		\item if $R \in H_k$ and $S \in H_k$, then $R \cap S \in H_k$; \label{it:henkin-class-intersection}
		\item if $R \in H_{k+1}$, then $\{(a_1, \dots, a_k) \in D^k \mid (a_1, \dots, a_k, a) \in R \text{ for all } a \in D \} \in H_k$. \label{it:henkin-class-forall}
	\end{enumerate}
	A \emph{Henkin structure} is a pair $(\structure, \henkinclass)$ consisting of a (regular) structure $\structure$ and a Henkin class $\henkinclass$ for $\structure$.
	\hfill\markfull
\end{definition}

Henkin classes are sufficiently rich to interpret \fo-formulas.
This property will be used in establishing the soundness of \scfoid w.r.t.~the Henkin semantics.

\begin{restatable}{proposition}{henkinclassinterpretsallfoformulas} \label{prop:henkin-class-interprets-all-fo-formulas}
	Let $\structure$ be a structure %
	and $\{ H_k \mid k \in \N \}$ a Henkin class for $\structure$.
	Let $\bar{x} = (x_1, \dots, x_k)$ be a tuple of distinct object symbols not in $\voc{\structure}$ and $\formula$ an \fo-formula over $\voc{\structure} \cup \{x_1, \dots, x_k\}$.
	Then
	\begin{equation*}
		\{ \bar{a} \in \dom{\structure}^k \mid \structure[\bar{x}: \bar{a}] \models \formula \} \in H_k .
	\end{equation*}
\end{restatable}

\begin{proof}
	See Appendix \ref{app:proofs-section-alternative-semantics}.
\end{proof}

Conversely, the collection of \fo-\emph{definable relations} forms a Henkin class for \emph{domain-closed} structures, i.e., structures $\structure$ with the property that for every $a \in \dom{\structure}$, there exists a term $t$ over $\voc{\structure}$ such that $t^{\structure} = a$.

\begin{restatable}{proposition}{fodefinablesetsformhenkinclass} \label{prop:fo-definable-sets-form-henkin-class}
	Let $\structure$ be a domain-closed structure. %
	For every $k \in \N$, let $H_k$ be the set of all relations of the form %
	$\{ \bar{a} \in \dom{\structure}^k \mid \structure[\bar{x}:\bar{a}] \models \formula \}$
	such that $\bar{x} = (x_1, \dots, x_k)$ is a tuple of distinct object symbols not in $\voc{\structure}$ and $\formula$ an \fo-formula over $\voc{\structure} \cup \{x_1, \dots, x_k\}$.
	Then $\{ H_k \mid k \in \N \}$ is a Henkin class for $\structure$, which we call the \emph{Henkin class of \fo-definable relations for $\structure$}.
\end{restatable}

\begin{proof}
	See Appendix \ref{app:proofs-section-alternative-semantics}.
\end{proof}

Henceforth, we fix an \foid-definition $\defn$ and a $\defn$-context $\context$.
We define the structure space $\structurespace$, the relation $\sqsubseteq$ and the operator $\cons$ as in Section \ref{sec:stable-semantics}.
Let $\structure \in \structurespace$, and let $\henkinclass = \{ H_k \mid k \in \N \}$ be a Henkin class for $\structure$.
We say that a structure $\structuretwo \in \structurespace$ is an $\henkinclass$-\emph{point} if $P^{\structuretwo} \in H_{\arity{P}}$ for every $P \in \defi{\defn}$.
We say that $\henkinstructure$ is a \emph{Henkin model} of $\defn$ in $\context$ if $\structure$ is the \emph{least pre-fixed $\henkinclass$-point} of $\cons(\cdot, \structure)$ (given that it exists).
This means that $\structure = \cons(\structure, \structure)$, and that $\structuretwo = \cons(\structuretwo, \structure)$ implies $\structure \sqsubseteq \structuretwo$ for every $\henkinclass$-point $\structuretwo$.
The \emph{Henkin satisfaction relation} $\modelsh$ in \foid is a relation between Henkin structures and \foid-formulas, defined by the following rules:
\begin{itemize}
	\item $(\structure, \henkinclass) \modelsh t=s$ if $t^{\structure} = s^{\structure}$;
	\item $(\structure, \henkinclass) \modelsh P(\bar{t})$ if $\bar{t}^{\structure} \in P^{\structure}$;
	\item $(\structure, \henkinclass) \modelsh \formula \land \formulatwo$ if $(\structure, \henkinclass) \modelsh \formula$ and $(\structure, \henkinclass) \modelsh \formulatwo$;
	\item $(\structure, \henkinclass) \modelsh \lnot \formula$ if not $(\structure, \henkinclass) \modelsh \formula$;
	\item $(\structure, \henkinclass) \modelsh \forall x: \formula$ if $(\structure[x:a], \henkinclass) \modelsh \formula$ for all $a \in \dom{\structure}$;
	\item $(\structure, \henkinclass) \modelsh \defn$ if $(\structure|_{\sym{\defn}}, \henkinclass)$ is a Henkin model of $\defn$ in $\structure|_{\pars{\defn}}$.\footnote{
		A Henkin class $\henkinclass$ for a structure $\structure$ is also a Henkin class for any reduct $\structure|_{\vocab}$ of $\structure$, as is straightforward to check.
	}
\end{itemize}
As before, the rules for $\lor$, $\Rightarrow$, $\Leftrightarrow$ and $\exists$ can de derived from the above rules through reformulation in terms of $\lnot$, $\land$ and $\forall$.
Note that if $\formula$ is an \fo-formula, then $(\structure, \henkinclass) \modelsh \formula$ comes down to $\structure \models \formula$.

\begin{remark}
	Since the definition of Henkin model does not refer to individual definition rules, any \foid-definition is equivalent to its normalization under the Henkin semantics.
\end{remark}

Any stable model $\structure$ of $\defn$ in $\context$ corresponds to a Henkin model $\henkinstructure$ of $\defn$ in $\context$, where $\henkinclass$ is the \emph{trivial Henkin class} for $\structure$, i.e., the Henkin class $\{ H_k \mid k \in \N \}$ where $H_k = \Pow{D^k}$ for all $k \in \N$.
Indeed, if $\henkinclass$ is the trivial Henkin class for $\structure$, then the least pre-fixed $\henkinclass$-point of $\cons(\cdot,\structure)$ comes down to the least (pre-)fixpoint of $\cons(\cdot,\structure)$.
Conversely, there generally exist Henkin models $\henkinstructure$ of $\defn$ in $\context$ such that $\structure$ is not a stable model of $\defn$ in $\context$.
Brotherston derived the existence of such Henkin models for the natural number definition, using the existence of non-standard models of arithmetic \cite{Brotherston2006PhD}.

We now show that under the Henkin semantics, every definition entails its Clark completion (Proposition \ref{prop:completion-henkin}) as well as a refined version of the induction axiom in which the second-order variables range over the relations in a Henkin class (Proposition \ref{prop:induction-scheme-henkin}).
We obtain the corresponding results for the stable semantics as a corollary by instantiating the Henkin classes in these results with trivial Henkin classes. 

\begin{proposition} \label{prop:completion-henkin}
	Let $\defn$ be a definition, $P$ a defined predicate of $\defn$, and $\bar{y}$ an $\arity{P}$-tuple of variables not occurring %
	in $\defn$.
	Let $\henkinstructure$ be a Henkin structure such that $\henkinstructure \modelsh \defn$. 
	Then $\structure \models \forall \bar{y} : P(\bar{y}) \Leftrightarrow \mbd{P}{\defn}(\bar{y})$.
	\hfill\markfull
\end{proposition}

\begin{proof}
	We can assume without loss of generality that $\voc{\structure} = \sym{\defn}$.
	Let $\context \coloneq \structure|_{\pars{\defn}}$.
	Since $(\structure, \henkinclass)$ is a Henkin model of $\defn$, $\structure = \cons(\structure, \structure)$.
	This implies that $P^\structure = P^{\cons(\structure, \structure)}$.
	By definition of $\cons$, this means that $\structure[\bar{x} : \bar{a}] \models P(\bar{x})$ iff $\structure[\bar{y} : \bar{a}] \models \mbd{P}{\defn}(\bar{y})$ for all $\bar{a} \in \dom{\defn}^{\arity{P}}$, or equivalently, $(\structure, \henkinclass) \modelsh \forall \bar{y} : P(\bar{y}) \Leftrightarrow \mbd{P}{\defn}(\bar{y})$.
\end{proof}

\begin{corollary} \label{cor:completion-st}
	Let $\defn$ be a definition, $P$ a defined predicate of $\defn$, and $\bar{y}$ an $\arity{P}$-tuple of variables not occurring %
	in $\defn$.
	Let $\structure$ be a structure such that $\structure \modelsst \defn$. 
	Then $\structure \models \forall \bar{y} : P(\bar{y}) \Leftrightarrow \mbd{P}{\defn}(\bar{y})$.
\end{corollary}

\begin{proof}
	Let $\henkinclass$ be the trivial Henkin class for $\structure$.
	Then $\henkinstructure \modelsh \defn$. 
	By Proposition \ref{prop:completion-henkin}, it follows that $\structure \models \forall \bar{y} : P(\bar{y}) \Leftrightarrow \mbd{P}{\defn}(\bar{y})$.
\end{proof}

\begin{restatable}{proposition}{subdefinitionhenkin}\label{prop:subdefinition-Henkin-model}
	Let $\henkinstructure$ be an Henkin structure, $\defn$ a definition and $\defn'$ a subdefinition of $\defn$.
	If $\henkinstructure \modelsh \defn$, then $\henkinstructure \modelsh \defn'$.
	\hfill\markfull
\end{restatable}

\begin{proof}
	See Appendix \ref{app:proofs-section-alternative-semantics}.
\end{proof}

\begin{proposition} \label{prop:induction-scheme-henkin}
	Let $\defn$ be a definition and $\Pi$ a subset of $\defi{\defn}$.
	Let $\structure$ be a structure and $\henkinclass = \{ H_k \mid k \in \N \}$ a Henkin class for $\structure$ such that $\henkinstructure \modelsh \defn$. 
	Let $A_Q \in H_{\arity{Q}}$ for every $Q \in \Pi$.
	Then 
	\begin{equation} \label{eq:h-point-satisfies-indscheme}
		\structure[G_\Pi : A_\Pi] \models \Big(  \bigwedge_\Pi \forall \bar{y} : \formulatwo\substpos{G_{\Pi}}{\Pi} \Rightarrow G_Q(\bar{s}) \Big) \Rightarrow \Big( \bigwedge_{P\in\Pi} \forall \bar{x} : P(\bar{x}) \Rightarrow G_P(\bar{x}) \Big) 
	\end{equation}
	where the first conjunction ranges over all definitional rules $\forall \bar{y} : Q(\bar{s}) \rul \formulatwo$ in $\defn$ with $Q\in \Pi$. %
	As before, $\formulatwo\substpos{G_\Pi}{\Pi}$ denotes the formula obtained from $\formulatwo$ by replacing all positive occurrences of atoms $R(\bar{u})$ such that $R \in \Pi$ with $G_R(\bar{u})$.
	\hfill\markfull
\end{proposition}

\begin{proof}
	We first show that (\ref{eq:h-point-satisfies-indscheme}) holds for $\Pi = \defi{\defn}$.
	We can assume without loss of generality that $\voc{\structure} = \sym{\defn}$, so that $\structure$ is the least pre-fixed $\henkinclass$-point of $\cons(\cdot, \structure)$, where $\context \coloneq \structure|_{\pars{\defn}}$.
	Suppose that $\structure[G_\Pi : A_\Pi] \models \forall \bar{y} : \formulatwo\substpos{G_\Pi}{\Pi} \Rightarrow G_Q(\bar{s})$ for all definitional rules $\forall \bar{y} : Q(\bar{s}) \rul \formulatwo$ in $\defn$.
	By Definition \ref{def:merged-body}, this implies that 
	\begin{equation} \label{eq:ass-proof-prop-ind-scheme-henkin}
		\structure[G_\Pi : A_\Pi] \models \forall \bar{y} : \mbd{Q}{\defn}(\bar{y})\substpos{G_\Pi}{\Pi} \Rightarrow G_Q(\bar{y}) \qquad \text{%
			for all $Q \in \defi{\defn}$.
		}
	\end{equation}
	Let $\structuretwo \coloneq \structure[\Pi : A_\Pi]$.
	Since $\structure$ is an $\henkinclass$-point, and since $A_Q \in H_{\arity{Q}}$ for every $Q \in \Pi$, $\structuretwo$ is an $\henkinclass$-point as well.
	We show that $\structuretwo$ is a pre-fixpoint of $\cons(\cdot, \structure)$.
	To do so, take $P \in \defi{\defn}$, and let $\bar{a} \in P^{\cons(\structuretwo, \structure)}$.
	Then $\structuretwo_{\structure}[\bar{x}:\bar{a}] \models \mbd{P}{\defn}(\bar{x})'$.
	Since $\structuretwo_{\structure}[\bar{x}:\bar{a}]$ interprets the positive occurrences of defined predicates of $\defn$ in $\mbd{P}{\defn}(\bar{x})'$ like $\structuretwo[\bar{x}:\bar{a}]$ and other occurrences of non-logical symbols like $\structure[\bar{x}:\bar{a}]$, this comes down to 
	\begin{equation*}
		\structure[G_\Pi : A_\Pi, \, \bar{x} : \bar{a}] \models \mbd{P}{\defn}(\bar{x})\substpos{G_\Pi}{\Pi} .
	\end{equation*}
	By (\ref{eq:ass-proof-prop-ind-scheme-henkin}), this implies that $\structure[G_\Pi : A_\Pi, \, \bar{x} : \bar{a}] \models G_P(\bar{x})$.
	Thus, $\bar{a} \in G_P^{\structure[G_\Pi : A_\Pi]} = A_P = P^{\structuretwo}$.
	This shows that $P^{\cons(\structuretwo, \structure)} \subseteq P^{\structuretwo}$, and hence, that $\cons(\structuretwo, \structure) \sqsubseteq \structuretwo$.
	In sum, $\structuretwo$ is a pre-fixed $\henkinclass$-point of $\cons(\cdot, \structure)$.
	Since $\structure$ is the least pre-fixed $\henkinclass$-point of $\cons(\cdot, \structure)$, we obtain that $\structure \sqsubseteq \structuretwo$.
	This means that $P^{\structure} \subseteq P^{\structuretwo}$ or, equivalently, that $\structure[G_\Pi : A_\Pi] \models \forall \bar{x} : P(\bar{x}) \Rightarrow G_P(\bar{x})$ for all defined predicates $P$ of $\defn$.
	This shows that (\ref{eq:h-point-satisfies-indscheme}) holds for $\Pi = \defi{\defn}$.
	
	The statement for general $\Pi$ now follows by Proposition \ref{prop:subdefinition-Henkin-model}.
	Indeed, let $\Pi$ be an arbitrary subset of $\defi{\defn}$, and consider the subdefinition $\defn'$ of $\defn$ consisting of all definitional rules $\forall \bar{y} : Q(\bar{s}) \rul \formulatwo$ of $\defn$ with $Q\in \Pi$.
	By Proposition \ref{prop:subdefinition-Henkin-model}, $\henkinstructure \modelsh \defn'$.
	By what we proved in the previous paragraph, (\ref{eq:h-point-satisfies-indscheme}) holds for the set of defined predicates $\defi{\defn'}$ of $\defn'$, which is precisely $\Pi$.
\end{proof}

\begin{corollary} \label{cor:induction-scheme-st}
	Let $\defn$ be a definition and $\Pi$ a subset of $\defi{\defn}$.
	Let $\structure$ be a structure such that $\structure \modelsst \defn$. 
	Then $\structure \models \indax{\defn, \Pi}$. 
\end{corollary}

\begin{proof}
	Let $\henkinclass$ be the trivial Henkin class for $\structure$.
	Then $\henkinstructure \modelsh \defn$. 
	By Proposition \ref{prop:induction-scheme-henkin}, (\ref{eq:h-point-satisfies-indscheme}) holds for every subset $\Pi$ of $\defi{\defn}$ and for all relations $A_Q \in H_{\arity{Q}}$ corresponding to the predicate symbols $Q \in \Pi$.
	Since $H_{k} = \dom{\structure}^{k}$ for every $k \in \N$, this comes down to $\structure \models \indax{\defn, \Pi}$ for every subset $\Pi$ of $\defi{\defn}$.
\end{proof}

%% file: soundness.tex
\section{Soundness} \label{sec:soundness}

In this section, we show that %
\scfoid %
is \emph{sound} w.r.t.~the well-founded, stable and Henkin semantics, meaning that every \scfoid-theorem is valid under all three semantics.
As a consequence, \scfoid can be used to prove theorems about (fragments of) logic programs under the well founded and the stable semantics.

\begin{theorem}[Soundness] \label{thm:soundness}
	\scfoid is sound w.r.t.~the well-founded, the stable and the Henkin semantics. %
	\hfill\markfull
\end{theorem}

\begin{proof}
	Let $\seqshort$ be an \scfoid-theorem, and let $\prf$ be an \scfoid-proof of $\seqshort$.
	We prove by induction on the height of $\prf$ that $\seqshort$ is valid under the Henkin semantics.
	The proof for the well-founded and the stable semantics is similar.
	We only discuss the cases where $\prf$ derives $\seqshort$ by (def R) or (def L), as the cases corresponding to the other inference rules are covered by the soundness of \scfo:\footnote{
		We refer to \cite{Brotherston2006PhD} for an explicit proof of these cases.
	}
	
	\textbf{Case}: %
	the final inference in $\prf$ is of the form
	\[
	\begin{prooftree}
		\hypo{\defn, \Gamma \vdash \formula[\bar{s}/\bar{x}], \Delta}
		\infer1[(def R)]{\defn, \Gamma \vdash P(\bar{t}[\bar{s}/\bar{x}]), \Delta}
	\end{prooftree}
	\]
	By induction hypothesis, $\defn, \Gamma \vdash \formula[\bar{s}/\bar{x}], \Delta$ is valid under the Henkin semantics.
	Let $\henkinstructure$ be a Henkin structure, and assume that $\henkinstructure \modelsh \defn \land \bigwedge \Gamma$.
	Since $\henkinstructure \modelsh \defn, \Gamma \vdash \formula[\bar{s}/\bar{x}], \Delta$, this implies that $\henkinstructure \modelsh \formula[\bar{s}/\bar{x}] \lor \bigvee \Delta$.
	If $\henkinstructure \modelsh\formula[\bar{s}/\bar{x}]$, then $\structure \modelsh P(\bar{t}[\bar{s}/\bar{x}])$ by Proposition \ref{prop:completion-henkin} (and Definition \ref{def:merged-body}).
	If $\henkinstructure \not\modelsh\formula[\bar{s}/\bar{x}]$, then $\henkinstructure \modelsh \bigvee \Delta$.
	Thus, in either case, $\henkinstructure \modelsh P(\bar{t}[\bar{s}/\bar{x}]) \lor \bigvee \Delta$.
	This shows that $\bigwedge \defn, \Gamma \vdash P(\bar{t}[\bar{s}/\bar{x}]), \Delta$ is valid under the Henkin semantics.
	
	\textbf{Case}: %
	the final inference in $\prf$ is of the form
	\[
	\begin{prooftree}
		\hypo{\text{minor premises}}
		\hypo{\defn, \Gamma, \ihf{P}{\bar{t}} \vdash \Delta}
		\infer2[(def L)]{\defn, \Gamma, P(\bar{t}) \vdash \Delta}
	\end{prooftree}
	\]
	for a predicate set $\Pi$ and an induction hypothesis $\ihs{Q}$ for every $Q \in \Pi$.
	For every definitional rule $\forall \bar{y}: Q(\bar{s}) \rul \formulatwo$ in $\defn$ with $Q \in \Pi$, this rule has a minor premise $\defn, \Gamma, \formulatwo\substposihpi \vdash \ihf{Q}{\bar{s}}, \Delta$.
	By induction hypothesis, every premise is valid under the Henkin semantics.
	Let $\henkinstructure$ be a Henkin structure, and assume that $\henkinstructure \modelsh \defn \land \bigwedge \Gamma \land P(\bar{t})$.
	Assume for contradiction that $\henkinstructure \not\modelsh \bigvee \Delta$.
	Since the minor premises are valid under the Henkin semantics, and since, for every definitional rule $\forall \bar{y}: Q(\bar{s}) \rul \formulatwo$ in $\defn$ with $Q \in \Pi$, none of the variables in $\bar{y}$ occurs freely in $\defn$, $\Gamma$ or $\Delta$, we have that
	\begin{equation} \label{eq:henkin-struct-minor-premises}
		\structure \models \forall \bar{y}: \formulatwo\substposihpi \Rightarrow \ihf{Q}{\bar{s}}
	\end{equation}
	for all definitional rules $\forall \bar{y}: Q(\bar{s}) \rul \formulatwo$ in $\defn$ with $Q \in \Pi$.
	For every $Q \in \Pi$, we define %
	\begin{equation} \label{eq:A_Q}
		A_Q \coloneq \{ \bar{b} \in \dom{\structure}^{\arity{Q}} \mid \structure[\bar{y}:\bar{b}] \models \ihf{Q}{\bar{y}} \} .
	\end{equation}
	By Proposition \ref{prop:henkin-class-interprets-all-fo-formulas}, $A_Q \in H_{\arity{Q}}$ for every $Q \in\Pi$, writing $\henkinclass = \{H_k \mid k \in \N \}$.\footnote{
		To satisfy the conditions of Proposition \ref{prop:henkin-class-interprets-all-fo-formulas}, we can rename the variables $\bar{y}$ if necessary.
	}
	By (\ref{eq:A_Q}), we can restate (\ref{eq:henkin-struct-minor-premises}) as $\structure[G_\Pi : A_\Pi] \models \forall \bar{y}: \formulatwo\substpos{G_\Pi}{\Pi} \Rightarrow G_Q(\bar{s})$.

	By Proposition \ref{prop:induction-scheme-henkin}, $\structure[G_\Pi : A_\Pi] \models \forall \bar{x}: P(\bar{x}) \Rightarrow G_P(\bar{x})$.
	Stated differently, $\henkinstructure \modelsh \forall \bar{x}: P(\bar{x}) \Rightarrow \ihf{P}{\bar{x}}$.
	Since $\henkinstructure \modelsh P(\bar{t})$, we find that $\henkinstructure \modelsh \ihf{P}{\bar{t}}$.
	Since $\henkinstructure \modelsh \bigwedge \Gamma \land \defn$, and since $\defn, \Gamma, \ihf{P}{\bar{t}} \vdash \Delta$ is valid under the Henkin semantics, this entails that $\henkinstructure \modelsh \bigvee \Delta$, which contradicts our assumption. %
	Thus, $\defn, \Gamma, P(\bar{t}) \vdash \Delta$ is valid under the Henkin semantics.
\end{proof}

The fact that \scfoid is sound with respect to all three semantics provides insights into the limitations of \scfoid, as it cannot prove statements on which the semantics disagree. 

\begin{example} \label{ex:def-choice-rule-stable}
	The definition
	\begin{equation*}
		\defn = \defin{
			P \rul \lnot Q\\
			Q \rul \lnot P
		}
	\end{equation*}
	has no two-valued well-founded models, but it has two (two-valued) stable models: one in which $P$ is true and $Q$ is false, and one in which $P$ is false and $Q$ is true.
	Hence, even though $\vdash \lnot \defn$ is valid under the well-founded semantics, it is not \scfoid-provable, as it is not valid under the stable semantics.
\end{example}

%% file: completeness.tex
\section{Completeness} \label{sec:completeness}

\scfoid is not \emph{complete} w.r.t.~the well-founded semantics (nor w.r.t.~the stable semantics),
meaning that there exist valid sequents that are not \scfoid-provable.
This is the case even when restricting to positive definitions (on which the well-founded and the stable semantics coincide).
Like \lkid, the incompleteness of \scfoid is a consequence of G\"odel's first incompleteness theorem, as there exists a canonical \foid-theory of the natural numbers. 
Indeed, let $\natvoc$ be the vocabulary consisting of the unary predicate symbol $\mathit{Nat}$, the constant symbol $\mathit{zero}$, the unary function symbol $\mathit{succ}$ and the binary function symbols $\mathit{plus}$ and $\mathit{times}$.
Let $\natstr$ be the standard $\natvoc$-structure, i.e., the $\natvoc$-structure such that $\dom{\natstr} = \mathit{Nat}^{\natstr} = \N$, $\mathit{zero}^{\natstr} = 0$, $\mathit{succ}^{\natstr}(n) = n+1$ for all $n \in \N$, and $\mathit{plus}^{\natstr}(n,m) = n+m$ and $\mathit{times}^{\natstr}(n,m) = n \cdot m$ for all $n,m \in \N$.
Let $\natth$ be the \foid-theory consisting of the two first-order Peano axioms imposing unique names for  $\mathit{zero}/0$ and $\mathit{succ}/1$:
\begin{align*}
	&\forall n: \mathit{succ}(n) \neq \mathit{zero}, 
	&
	&\forall n, m: \mathit{succ}(n) = \mathit{succ}(m) \Rightarrow n=m ;
\end{align*}
	a definition and an \fo-formula enforcing the second-order Peano axiom;
\begin{align*}
	& \defin{
		\mathit{Nat}(\mathit{zero})\\
		\forall n: \mathit{Nat}(\mathit{succ}(n)) \rul \mathit{Nat}(n)
	},
	&
	&\forall n: \mathit{Nat}(n) ; 
\end{align*}
and four \fo-formulas defining the arithmetic operators $\mathit{plus}$ and $\mathit{times}$:
\begin{align*}
	&\forall n: \mathit{plus}(n, \mathit{zero}) = n, 
	&
	&\forall n,m: \mathit{plus}(n, \mathit{succ}(m)) = \mathit{succ}(\mathit{plus}(n,m)),\\
	&\forall n: \mathit{times}(n, \mathit{zero}) = \mathit{zero},
	&
	&\forall n,m: \mathit{times}(n, \mathit{succ}(m)) = \mathit{plus}(\mathit{times}(n,m) ,m).
\end{align*}
Using Proposition \ref{prop:induction-scheme-wf}, the definition of $\mathit{Nat}$ and $\forall n: \mathit{Nat}(n)$ indeed enforce the second-order Peano axiom
\begin{equation*}
	\forall G: G(\mathit{zero}) \land (\forall n: G(n) \Rightarrow G(\mathit{succ}(n))) \Rightarrow \forall n: G(n) .
\end{equation*}
Informally, this axiom says that every object can be obtained from $\mathit{zero}$ through a finite number of applications of $\mathit{succ}$.
It follows that $\natth$ is \emph{categorical} for $\natstr$, i.e., that every model of $\natth$ (under the well-founded semantics) is isomorphic to $\natth$.
Therefore, the set $\{\formula \text{ \fo-formula} \mid \natth \modelswf \formula \}$ of first-order consequences of $\natth$ is equal to the set $\{\formula \text{ \fo-formula} \mid \natstr \modelswf \formula \}$ of true first-order formulas in $\natstr$, also known as \emph{true arithmetic}.
Now suppose for contradiction that \scfoid is complete w.r.t.~the well-founded semantics.
Then, since \scfoid is also sound w.r.t.~the well-founded semantics, we have that
\begin{equation*}
	\{\formula \text{ \fo-formula} \mid \natth \modelswf \formula \} = \{\formula \text{ \fo-formula} \mid \natth \vdash \formula \text{ is \scfoid-provable} \} .
\end{equation*}
Since \scfoid-proofs are finite, the set $\{\formula \text{ \fo-formula} \mid \natth \vdash \formula \text{ is \scfoid-provable} \}$ is recursively enumerable.
It follows that true arithmetic $\{\formula \text{ \fo-formula} \mid \natstr \modelswf \formula \}$ is recursively enumerable.
However, it is a well-known corollary of G\"odel's first incompleteness theorem that this is not the case.

Even though no finitary proof system for \foid can be complete w.r.t.~well-founded or the stable semantics, we can still prove weaker completeness results for \scfoid.
In Sections \ref{sec:prop-completeness} and \ref{sec:henkin-completeness}, we will prove that \scfoid is complete w.r.t.~the stable semantics on %
a propositional fragment of \foid,
and that \scfoid is complete w.r.t.~the Henkin semantics on
a first-order order fragment of \foid.
These results provide insight into the deductive strength of \scfoid within the proof-theoretic limits imposed by inductive definitions.

\subsection{Propositional Completeness} \label{sec:prop-completeness}

In Section \ref{sec:correspondence}, we saw that any definition $\defn$ logically entails its first-order approximation $\foapp{\defn}$ under the well-founded semantics, %
but that $\foapp{\defn}$ does not always logically entail $\defn$.
By a similar reasoning, these properties also hold under the stable semantics.
We prove that for 
\emph{propositional definitions}, i.e., definitions in which all predicates are $0$-ary, 
$\foapp{\defn}$ logically entails $\defn$ under the stable semantics.

\begin{lemma} \label{lem:foapp-prop-def-entails-def}
	Let $\defn$ be a propositional definition.
	Then $\foapp{\defn} \modelsst \defn$.
\end{lemma}

\begin{proof}
	Let $\structure$ be a structure such that $\structure \models \foapp{\defn}$.
	Without loss of generality, we can assume that $\voc{\structure} = \sym{\defn}$, and that $\defn$ is in normal form.
	To prove that $\structure \modelsst \defn$, we need to show that $\structure$ is the least pre-fixpoint of $\cons(\cdot, \structure)$, where $\context \coloneq \structure|_{\pars{\defn}}$.
	Since $\structure \models \matimps{\defn}$, $\structure$ is a pre-fixpoint of $\cons(\cdot, \structure)$.
	Indeed, for any defined predicate $P$ of $\defn$, we have that $P^{\cons(\structure, \structure)} = (\mbd{P}{\defn})^{\structure} \leqt P^{\structure}$, as $\structure \modelsst \mbd{P}{\defn} \Rightarrow P$.
	
	Since $\structure \models \indscheme{\defn}$, $\structure$ is moreover the \emph{least} pre-fixpoint of $\cons(\cdot, \structure)$.
	Indeed, suppose that $\structuretwo$ is a pre-fixpoint of $\cons(\cdot, \structure)$, i.e., $\cons(\structuretwo, \structure) \sqsubseteq \structuretwo$.
	Let $P$ be a defined predicate of $\defn$.
	By definition of $\indscheme{\defn}$, 
	\begin{equation} \label{eq:ind-scheme-propositional}
		\structure \models \bigwedge_{\Pi} (\mbd{Q}{\defn}\substposihpi \Rightarrow \ihs{Q})
		\Rightarrow (P \Rightarrow \ihs{P} )  
	\end{equation}
	for any subset $\Pi$ of $\defi{\defn}$ and any choice of induction hypotheses $\ihs{Q}$ for the predicate symbols $Q$ in $\Pi$.
	Let $\Pi \coloneq \defi{\defn}$ and for every $Q \in \Pi$, let $\ihs{Q} \coloneq \top$ if $\structuretwo \models Q$ and $\ihs{Q} \coloneq \bot$ if $\structuretwo \not\models Q$.
	Since we chose $\ihs{Q}$ in accordance with $Q^{\structuretwo}$, 
	we can restate $Q^{\cons(\structuretwo, \structure)} \subseteq Q^{\structuretwo}$ as $\structure \models \mbd{Q}{\defn}\substposihpi \Rightarrow \ihs{Q}$ for every $Q \in \defi{\defn}$.
	By (\ref{eq:ind-scheme-propositional}), it follows that $\structure \models P \Rightarrow \ihs{P}$ or, stated differently, that $P^{\structure} \leqt P^{\structuretwo}$.
	This shows that $\structure \sqsubseteq \structuretwo$, thus proving that $\structure$ is indeed the least pre-fixpoint of $\cons(\cdot, \structure)$.
\end{proof}

\begin{theorem}[Propositional completeness] \label{thm:propositional-completeness}
	Let $\regularsequent$ be a regular sequent with propositional definitions $\defns$.
	If $\regularsequent$ is valid under the stable semantics, then it is \scfoid-provable. 
\end{theorem}

\begin{proof}
	Suppose that $\regularsequent$ is valid under the stable semantics.
	By Lemma \ref{lem:foapp-prop-def-entails-def}, it follows that $\foapp{\defn_1}, \dots, \foapp{\defn_n}, \Gamma \vdash \Delta$ is valid in \fo.
	By completeness of \scfo, we obtain that $\foapp{\defn_1}, \dots, \foapp{\defn_n}, \Gamma \vdash \Delta$ is \scfo-provable.
	By Theorem \ref{thm:correspondence}, $\regularsequent$ is \scfoid-provable.
\end{proof}

As shown by Example \ref{ex:def-choice-rule-stable}, propositional completeness does not hold under the well-founded semantics.
Indeed, for the definition $\defn$ in this example, the sequent $\defn \vdash$ is valid under the well-founded semantics, but it is not \scfoid-provable (as it is not valid under the stable semantics).

\subsection{Henkin Completeness} \label{sec:henkin-completeness}

The Henkin completeness result for \scfoid relies on the following lemma, which provides a relation between the first-order approximations of definitions and Henkin classes of \fo-definable sets.

\begin{lemma} \label{lem:satisfaction-fo-app-def-fo-def-class}
	Let $\defn$ be a definition, $\structure$ a domain-closed structure, and $\henkinclass$ the Henkin class of \fo-definable sets for $\structure$ from Proposition \ref{prop:fo-definable-sets-form-henkin-class}.
	If $\structure \models \foapp{\defn}$, then $\henkinstructure \modelsh \defn$.
\end{lemma}

\begin{proof}
	Suppose that $\structure \models \foapp{\defn}$.
	Without loss of generality, we can assume that $\voc{\structure} = \sym{\defn}$, and that $\defn$ is in normal form.
	We show that $\structure$ is the least pre-fixed $\henkinclass$-point of $\cons(\cdot, \structure)$.
	Trivially, $\structure$ is an $\henkinclass$-point.
	 The fact that $\structure$ is a pre-fixpoint of $\cons(\cdot, \structure)$ follows from the fact that $\structure \models \matimps{\defn}$.
	 Indeed, for any defined predicate $P$ of $\defn$, we have that $\structure \models \forall \bar{y}: \mbd{P}{\defn}(\bar{y}) \Rightarrow P(\bar{y})$, and therefore, $P^{\cons(\structure, \structure)} = \{ \bar{a} \mid \structure[\bar{y}:\bar{a}] \models \mbd{P}{\defn}(\bar{y}) \} \subseteq \{ \bar{a} \mid \structure[\bar{y}:\bar{a}] \models P(\bar{y}) \} = P^{\structure}$.
	 Here $\bar{y}$ is an $\arity{P}$-tuple of variables not occurring %
	 in $\defn$.
	The fact that $\structure$ is the \emph{least} pre-fixed $\henkinclass$-point of $\cons(\cdot, \structure)$ follows from the fact that $\structure \models \indscheme{\defn}$.
	Indeed, let $\structuretwo$ be a pre-fixed $\henkinclass$-point of $\cons(\cdot, \structure)$. 
	Since $\structuretwo$ is an $\henkinclass$-point, there exists, for every $Q \in \defi{\defn}$, an \fo-formula $\formula_Q$ and an $\arity{Q}$-tuple of variables $\bar{x}_Q$ such that $Q^{\structuretwo} = \{ \bar{b} \mid \structure[\bar{x}_Q:\bar{b}] \models \formula_{Q} \}$. 
	Defining $\ihs{Q} \coloneq \setexp{\bar{x}_Q}{\formula_{Q}}$, we can restate the fact that $Q^{\cons(\structuretwo, \structure)} \subseteq Q^{\structuretwo}$ as $\structure \models \forall \bar{z}_Q : \mbd{Q}{\defn}(\bar{z}_Q)\substpos{\ihs{\defi{\defn}}}{\defi{\defn}} \Rightarrow \ihf{Q}{\bar{z}_Q}$ for every $Q \in \defi{\defn}$, where $\bar{z}_Q$ is an $\arity{Q}$-tuple of variables not occurring in $\defn$.
	Since 
	\begin{equation*}
		\structure \models \bigwedge_{Q \in \defi{\defn}} \left( 
		\forall \bar{z}_Q : \mbd{Q}{\defn}(\bar{z}_Q)\substpos{\ihs{\defi{\defn}}}{\defi{\defn}} \Rightarrow \ihf{Q}{\bar{z}_Q}
		\right) 
		\Rightarrow 
		\left( \forall \bar{y} : P(\bar{y}) \Rightarrow \ihf{P}{\bar{y}} \right) ,
	\end{equation*}
	this implies that $\structure \models \forall \bar{y} : P(\bar{y}) \Rightarrow \ihf{P}{\bar{y}}$ for all $P\in \defi{\defn}$.
	Consequently, 
	\begin{equation*}
		P^{\structure} = \{ \bar{a} \mid \structure[\bar{y}:\bar{a}] \models P(\bar{y}) \} \subseteq \{ \bar{a} \mid \structure[\bar{y}:\bar{a}] \models \ihf{P}{\bar{y}} \} = \{ \bar{a} \mid \structure[\bar{x}_P:\bar{a}] \models \formula_P \} = P^{\structuretwo}
	\end{equation*}
	for all $P\in \defi{\defn}$,
	which means that $\structure \sqsubseteq \structuretwo$.
	This shows that 
	$\structure$ is the least pre-fixed $\henkinclass$-point of $\cons(\cdot, \structure)$.
\end{proof}

\begin{theorem}[Henkin completeness] \label{thm:henkin-completeness}
	Let $\regularsequent$ be a regular sequent.
	If $\regularsequent$ is valid under the Henkin semantics, then it is \scfoid-provable.
\end{theorem}

\begin{proof}
	Suppose that $\regularsequent$ is valid under the Henkin semantics.
	Then $\foapp{\defn_1}, \dots, \foapp{\defn_n}, \Gamma \vdash \Delta$ is valid in \fo.
	Indeed, let $\structure$ be a structure such that $\structure \models \bigwedge \Gamma \land \bigwedge \foapp{\defn_1} \land \dots \land \bigwedge \foapp{\defn_n}$. 
	By expanding $\structure$ if necessary, we can assume $\structure$ to be domain-closed.
	Let $\henkinclass$ be the Henkin class of \fo-definable sets for $\structure$.
	By Lemma \ref{lem:satisfaction-fo-app-def-fo-def-class}, $\henkinstructure \modelsh \bigwedge \Gamma \land \defn_1 \land \dots \land \defn_n$. 
	Since $\regularsequent$ is valid under the Henkin semantics, $\henkinstructure \modelsh \bigvee \Delta$. 
	Since $\Delta$ contains only \fo-formulas, $\structure \models \bigvee \Delta$.
	This shows that $\foapp{\defn_1}, \dots, \foapp{\defn_n}, \Gamma \vdash \Delta$ is indeed valid in \fo.
	By completeness of \scfo, we obtain that $\foapp{\defn_1}, \dots, \foapp{\defn_n}, \Gamma \vdash \Delta$ is \scfo-provable.
	By Theorem \ref{thm:correspondence}, $\regularsequent$ is \scfoid-provable.
\end{proof}

%% file: cut-elimination.tex
\section{Cut-Elimination} \label{sec:cut-elimination}

Brotherston and Simpson showed that cut-elimination holds for \lkid \cite{jlc/BrotherstonS11}, meaning that every \lkid-provable sequent can be proven without the cut rule.
From the perspective of proof search, cut-elimination has limited relevance in induction-based sequent calculi when compared to first-order sequent calculi.
This is because, contrary to first-order sequent calculi, cut-elimination in induction-based sequent calculi does not entail the \emph{subformula property}.
This property guarantees that one may restrict to \emph{analytic} inference rules, in which all the formulas in the premises are subformulas of the formulas in the conclusion.
In induction-based calculi like \lkid and \scfoid, the subformula property does not hold because of the induction rule, as the induction hypotheses in the premises of the induction rule do not appear %
in the conclusion.
Hence, the induction hypotheses are unconstrained by the conclusion, meaning that they can be instantiated in infinitely many ways.
Cut-elimination is nevertheless meaningful in induction-based calculi, as it says that these unconstrained formula need only emerge from the need to find suitable induction hypotheses, and not from the need to find suitable lemmas \cite{Brotherston2006PhD}.

As demonstrated by Example \ref{ex:proof_non-total_liar}, cut-elimination does not hold in \scfoid.
Indeed, from a simple inspection of the inference rules of \scfoid, one can see that no rule other than (cut) applies to derive the sequent %
$\defn \vdash$ (where $\defn = \defintight{P \rul \neg P}$).
Three more counterexamples are included below:

\begin{example} \label{ex:counterexample-cut-elimination-P-if-O-or-not-P}
	Consider the definition $\defn = \defintight{ P \rul \neg P \land A }$.
	Given a $\defn$-context $\context$, the well-founded model of $\defn$ in $\context$ interprets $P$ as $\false$ if $A^\context = \false$ and interprets $P$ as $\unknown$ if $A^\context = \true$, as is straightforward to check.
	Consequently, any two-valued model of $\defn$ under the well-founded semantics interprets $A$ as $\false$.
	This means that $\defn \modelswf \lnot A$.
	It is also straightforward to verify that $\defn \modelsst \lnot A$.
	We can prove that $\defn$ entails $\lnot A$ (under both semantics) as follows:
	\begin{equation*}
		\begin{prooftree}
			\hypo{\defn, A \vdash P}
			
			\infer0[(ax)]{\defn, A, \lnot P \land A \vdash \lnot P \land A}
			
			\hypo{\defn, A \vdash P}
			\infer1[($\lnot$L)]{\defn, A, \lnot P \vdash}
			\infer1[($\land$L)]{\defn, A, \lnot P \land A \vdash}
			
			\infer2[(def L)]{\defn, A, P \vdash}
			
			\infer2[(cut)]{\defn, A \vdash}
			\infer1[($\lnot$R)]{\defn \vdash \lnot A}
		\end{prooftree}
	\end{equation*}
	Here we applied (def L) with $\Pi = \{P\}$ and $\ihs{P} \doteq \lnot P \land A$.
	We can prove the bud sequent $\defn, A \vdash $ as follows:	
	\begin{equation*}
	\begin{prooftree}
		\infer0[(ax)]{\defn, A, \lnot P \vdash \lnot P}
		\infer0[(ax)]{\defn, A, \lnot P \vdash A}
		\infer2[($\land$R)]{\defn, A, \lnot P \vdash \lnot P \land A}
		\infer1[(def R)]{\defn, A, \lnot P \vdash P}
		\infer1[($\lnot$R)]{\defn, A \vdash P, \lnot \lnot P}
		\infer1[($*$)]{\defn, A \vdash P}
	\end{prooftree}
	\end{equation*}
	The transition labeled with ($*$) is completed analogously to the homonymous transition in Example \ref{ex:proof_non-total_liar}.
	The sequent $\defn, A \vdash $ can only be derived via (cut), as can be checked by inspection of the inference rules of \scfoid.
\end{example}

\begin{example} \label{ex:counterexample-cut-elimination-P-if_O-or-not-O}
	Consider the definition $\defn = \defin{
		P \rul A\\
		P \rul \lnot A
	}$. 
	We can prove that $\defn$ entails $P$ as follows:
	\begin{equation*}
		\begin{prooftree}
			\hypo{}
			\infer1[(ax)]{\defn, A \vdash A}
			\infer1[($\neg$R)]{\defn \vdash A, \neg A}
			\infer1[(def R)]{\defn \vdash A, P}
			\hypo{}
			\infer1[(ax)]{\defn, A \vdash A}
			\infer1[(def R)]{\defn, A \vdash P}
			\infer2[(cut)]{\defn \vdash P}
		\end{prooftree}
	\end{equation*}
	The sequent 
	$\defn \vdash P$ is not provable without (cut), as can be checked by inspection of the rules of \scfoid.
	Indeed, the only other rule that can derive $\defn \vdash P$ is (def R), either from the premise $\defn \vdash A$ or from the premise $\defn \vdash \lnot A$.
	However, neither of these sequents is \scfoid-provable, as neither is valid (under any of the three semantics).
	
	As can be easily checked, the sequent $\normalization{\defn} \vdash P$ (where $\normalization{\defn}$ is the normalization of $\defn$) is provable without (cut).
	This shows that Proposition \ref{prop:normalization-equiderivable} does not extend to cut-free provability.
\end{example}

\begin{example} \label{ex:counterexample-cut-elimination-two-defs}
	Consider the definitions $\defn_1 = \defintight{P \rul A}$ and $\defn_2 = \defintight{P \rul B}$.
	The Clark completions of these definitions are $P \Leftrightarrow A$ and $P \Leftrightarrow B$, respectively.
	Consequently, $\defn_1$ and $\defn_2$ logically entail $A \Rightarrow B$, under the well-founded, the stable and the Henkin semantics.
	We can formally prove this as follows:\footnote{
		This example comes from \cite{lpnmr/HouWD07}.
	}
	\begin{equation*}
		\begin{prooftree}
			\hypo{}
			\infer1[(ax)]{\defn_1, \defn_2, A \vdash A, B}
			\infer1[(def R)]{\defn_1, \defn_2, A \vdash P, B}
			\hypo{}
			\infer1[(ax)]{\defn_1, \defn_2, A, B \vdash B}
			\hypo{}
			\infer1[(ax)]{\defn_1, \defn_2, A, B \vdash B}
			\infer2[(def L)]{\defn_1, \defn_2, A, P \vdash B}
			\infer2[(cut)]{\defn_1, \defn_2, A \vdash B}
			\infer1[($\Rightarrow$R)]{\defn_1, \defn_2 \vdash A \Rightarrow B}
		\end{prooftree}
	\end{equation*}
	Here we applied (def L) w.r.t.~$\defn_2$, with $\Pi = \{P\}$ and $\ihs{P} \doteq B$.
	As before, a simple inspection of the inference rules of \scfoid reveals that $\defn_1, \defn_2, A \vdash B$ can only be proven via (cut).
\end{example}

Even though cut-elimination does not hold for general sequents in \scfoid, we can identify a fragment of sequents for which it holds.
One option is the fragment of sequents $\defn, \Gamma \vdash \Delta$ such that $\Gamma$ and $\Delta$ are finite sets of \fo-formulas and $\defn$ is a definition in which the body of every definitional rule is a conjunction of atoms.
Cut-elimination holds for this fragment because it corresponds to the logic of inductive definitions considered by Brotherston and Simpson \cite{jlc/BrotherstonS11}.
In Theorem \ref{thm:brotherston-cut-elimination}, we establish cut-elimination for a slightly broader fragment of \foid-sequents $\regularsequent$, containing multiple definitions $\defns$, in which the bodies of the rules may also contain disjunctions.
The definitions $\defns$ must be such that no two (distinct) definitions $\defn_i$, $\defn_j$ define the same predicate, and the defined predicates of a definition $\defn_i$ do not occur in a preceding definition $\defn_j$.

\begin{definition} \label{def:canonical-sequent}
	We call a regular sequent $\regularsequent$ \emph{canonical} if for all $i,j \in \{1, \dots, n\}$:
	\begin{enumerate}[label=(\roman*)]
		\item the body of each definitional rule of $\defn_i$ is a quantifier-free formula in negation normal form;\footnote{
			In other words, the body of each definitional rule of $\defn_i$ is composed of literals, connected only by conjunctions and disjunctions.
		}
		\item $\defi{\defn_i} \cap \defi{\defn_j} = \emptyset$ if $i \neq j$;
		\item $\defi{\defn_i} \cap \pars{\defn_j} = \emptyset$ if $j < i$. \label{it:def-predicates-do-not-occur-in-previous-defs}
	\end{enumerate}
\end{definition}

To prove Theorem \ref{thm:brotherston-cut-elimination}, we follows the semantic approach by Brotherston and Simpson, who proved cut-elimination for \lkid by proving \emph{cut-free Henkin completeness} for \lkid.
This means that all sequents that are valid under the Henkin semantics are provable without the cut rule.
Combined with Henkin soundness, cut-free Henkin completeness implies cut-elimination.

We say that an \foid-sequent is \emph{cut-free provable} (in \scfoid) if it admits an \scfoid-proof without (cut).

\begin{restatable}[Cut-free Henkin completeness]{theorem}{brotherstonlemma} \label{thm:brotherston}
	Let $\regularsequent$ be a canonical sequent with positive definitions $\defns$.
	If\/ $\regularsequent$ is valid under the Henkin semantics, then it is cut-free provable.
\end{restatable}

The proof of Theorem \ref{thm:brotherston} is analogous to the proof of cut-free Henkin completeness for \lkid \cite{jlc/BrotherstonS11}, and is included in Appendix \ref{app:brotherston-proof}.
The rough outline of the proof is as follows:

\begin{enumerate}[label=(\roman*)]
	\item Assume that $\regularsequent$ is not cut-free provable.
	\item We construct infinite sets of \fo-formulas $\Gamma_\omega$ and $\Delta_\omega$, extending $\Gamma$ resp.~$\Delta$, such that 
	$\limitsequent$ is not cut-free provable.
	\item We construct a \emph{Henkin counter-model} for $\limitsequent$, i.e., a Henkin structure $\henkincountermodel$ such that $\henkincountermodel \modelsh \defn_i$ for all $i \in \{1, \dots, n\}$, $\henkincountermodel \modelsh \formula$ for all $\formula \in \Gamma_\omega$ and $\henkincountermodel \not\modelsh \formula$ for all $\formula \in \Delta_\omega$.
	In particular, we obtain that $\henkincountermodel \modelsh \formula$ for all $\formula \in \Gamma$ and $\henkincountermodel \not\modelsh \formula$ for all $\formula \in \Delta$.
	\item Hence, $\regularsequent$ is not valid under the Henkin semantics.
	By contraposition, Theorem \ref{thm:brotherston} holds.
\end{enumerate}

\begin{theorem}[Cut-elimination] \label{thm:brotherston-cut-elimination}
	Let $\, \regularsequent$ be a canonical sequent with positive definitions $\, \defns$.
	If $\regularsequent$ is provable, then it is cut-free provable.
\end{theorem}

\begin{proof}
	If $\regularsequent$ is provable, then it is valid under the Henkin semantics by Theorem \ref{thm:soundness}, whence it is cut-free provable by Theorem \ref{thm:brotherston}.
\end{proof}

As illustrated by Example \ref{ex:counterexample-cut-elimination-P-if_O-or-not-O}, we cannot extend Theorems \ref{thm:brotherston} and \ref{thm:brotherston-cut-elimination} to allow for negation in the bodies of definitional rules, not even at the level of parameters.
Nevertheless, we will establish a weakened version of cut-elimination for canonical sequents $\regularsequent$ in which %
the definitions $\defns$ may contain negation, as long as %
they
are stratified (Theorem \ref{thm:cut-restriction}).
In this weakened version, cuts are allowed, but only on formulas of the form $\forall \bar{y}: Q(\bar{y}) \lor \lnot Q(\bar{y})$ such that $Q$ appears negatively in the body of a rule of a definition $\defn_i$. %
We call a cut of this form an \emph{elementary cut}.
We call a proof with only elementary cuts an \emph{EC-proof}, and a sequent that admits an EC-proof \emph{EC-provable}.
This weakened version of cut-elimination, which we call \emph{cut-restriction}, is valuable for proof search, as it narrows the options for a cut formula down to a finite list (if we fix the tuple of variables $\bar{y}$ for every predicate symbol $Q$ appearing negatively in the body of a definitional rule).
Before we prove the result for stratified definitions, we prove it for definitions in which negation occurs only at the level of parameters (Lemma \ref{lem:elementary-cuts-negation-parameters}).
Our strategy is to transform a canonical sequent $\regularsequent$ with definitions of this form to a semantically equivalent sequent $\positiverewriting$ with positive definitions $\compl{\defn}_1, \dots, \compl{\defn}_n$ by means of auxiliary predicates, and to apply Theorem \ref{thm:brotherston} to $\positiverewriting$.

\begin{definition} \label{def:positive-rewriting}
	Let $\regularsequent$ be a canonical sequent such that no defined atom occurs negatively in the body of a rule of a definition $\defn_i$.
	For every predicate symbol $Q$ that occurs negatively in the body of a rule of a definition $\defn_i$, consider a distinct predicate symbol $\compl{Q}$ that does not occur in $\regularsequent$.\footnote{
		Such a predicate symbol exists, because of the assumption that $\Gamma$ and $\Delta$ are finite.
	}
	Let $\equivs$ denote the set of formulas $\forall \bar{y} : \compl{Q}(\bar{y}) \Leftrightarrow \lnot Q(\bar{y})$ such that $Q$ occurs negatively in the body of a rule of a definition $\defn_i$.
	Here $\bar{y}$ is a fixed $\arity{Q}$-tuple of variables.
	For each $i \in \{1, \dots, n\}$, let $\compl{\defn_i}$ denote the definition obtained from $\defn_i$ by replacing all negative literals $\lnot Q(\bar{s})$ in the bodies of its definitional rules with $\compl{Q}(\bar{s})$.
	We call the sequent $\positiverewriting$ the \emph{positive rewriting} of $\regularsequent$.
\end{definition}

\begin{restatable}{lemma}{cfppositiverewriting} \label{lem:cfp-positive-rewriting}
	Let $\regularsequent$ be a canonical sequent.
	If the positive rewriting $\positiverewriting$ of $\regularsequent$ is cut-free provable, then $\regularsequent$ is EC-provable.
\end{restatable}

\begin{proof}
	See Appendix \ref{app:cfp-positive-rewriting}. 
\end{proof}

\begin{lemma} \label{lem:elementary-cuts-negation-parameters}
	Let $\regularsequent$ be a canonical sequent such that no defined atom occurs negatively in the body of a rule of a definition $\defn_i$.
	If $\regularsequent$ is valid under the Henkin semantics, then it is EC-provable.
\end{lemma}

\begin{proof}
	Suppose that $\regularsequent$ is valid under the Henkin semantics.
	Let $\positiverewriting$ be the positive rewriting of $\regularsequent$.
	Then $\positiverewriting$ is valid under the Henkin semantics.
	Indeed, let $\henkinstructure$ be a Henkin structure, and suppose that $\henkinstructure \modelsh \compl{\defn_i}$ for every $i \in \{1, \dots, n\}$ and $\henkinstructure \modelsh \formula$ for all $\formula \in \Gamma \cup \Omega$.
	Choose $i \in \{1, \dots, n\}$ arbitrarily.
	Since $\henkinstructure \modelsh \forall \bar{y} : \compl{Q}(\bar{y}) \Leftrightarrow \lnot Q(\bar{y})$, the structure $\structure$ interprets $\compl{Q}$ as the complement of $Q$ for every predicate symbol $Q$ occurring negatively in the body of a rule of $\defn_i$.
	Consequently, $\consg{\compl{\defn_i}}{\context}(\structuretwo, \structure|_{\sym{\compl{\defn}_i}})|_{\sym{\defn_i}} = \consg{\defn_i}{\context}(\structuretwo|_{\sym{\defn_i}}, \structure|_{\sym{\defn_i}})$ for every $\sym{\compl{\defn}_i}$-structure $\structuretwo$ expanding $\context \coloneq \structure|_{\pars{\compl{\defn}_i}} = \structure|_{\pars{\defn_i}}$.
	Since $\structure|_{\sym{\compl{\defn}_i}}$ is the least pre-fixed $\henkinclass$-point of $\consg{\compl{\defn_i}}{\context}(\cdot, \structure|_{\sym{\compl{\defn}_i}})$, it follows that $\structure|_{\sym{\defn_i}}$ is the least pre-fixed $\henkinclass$-point of $\consg{\defn_i}{\context}(\cdot, \structure|_{\sym{\defn_i}})$.
	Therefore, $\henkinstructure \modelsh \defn_i$.
	Since $i \in \{1, \dots, n\}$ was chosen arbitrarily, we deduce that $\henkinstructure \modelsh \defn_i$ for all $i \in \{1, \dots, n\}$.
	Since $\regularsequent$ is valid under the Henkin semantics, $\henkinstructure \modelsh \formula$ for some $\formula \in \Delta$.
	This shows that $\positiverewriting$ is indeed valid under the Henkin semantics.
	By Theorem \ref{thm:brotherston}, $\positiverewriting$ is cut-free provable, and by Lemma \ref{lem:cfp-positive-rewriting}, $\regularsequent$ is EC-provable.
\end{proof}

We now generalize the previous lemma to canonical sequents with stratified definitions.
We do so by decomposing the stratified definitions into monotone subdefinitions.

\begin{definition} \label{def:decomposition-stratified-def}
	Let $\defn$ be a stratified definition.
	A \emph{decomposition} of $\defn$ is a tuple $(\defn_1, \dots, \defn_k)$ of subdefinitions of $\defn$ such that the definitions $\defn_1, \dots, \defn_k$ form a partitioning of $\defn$, and there exists a stratification $\ell$ of $\defn$ such that for all $i,j \in \{1, \dots, k\}$, $P \in \defi{\defn_i}$ and $Q \in \defi{\defn_j}$: $\ell(P) < \ell(Q)$ iff $i<j$.
	We say that a stratification $\ell$ of $\defn$ \emph{corresponds to} the decomposition $(\defn_1, \dots, \defn_k)$ if it satisfies this condition.
\end{definition}

\begin{remark} \label{rem:decomposition-strat-def}
	It is straightforward to see that every stratified definition $\defn$ admits a decomposition $(\defn_1, \dots, \defn_k)$, and that $\ell(P) = \ell(Q)$ for every two predicates $P$ and $Q$ defined by the same definition $\defn_i$ in a decomposition $(\defn_1, \dots, \defn_k)$ and every stratification $\ell$ corresponding to $(\defn_1, \dots, \defn_k)$.
	Furthermore, %
	no defined atom occurs negatively in the body of a rule of a definition $\defn_i$ in a decomposition $(\defn_1, \dots, \defn_k)$.
	Indeed, if a defined predicate $Q$ of %
	$\defn_{i}$ were to occur negatively in the body $\formula$ of a definitional rule $\defrul$ in $\defn_{i}$, then, since $P$ and $Q$ are both defined by $\defn_{i}$, we could infer the contradiction $\ell(Q) < \ell(P) = \ell(Q)$.
\end{remark}

\begin{lemma} \label{lem:stratified-def-henkin-entails-decomposition}
	Let $\defn$ be a stratified definition, $(\defn_1, \dots, \defn_k)$ a decomposition of $\defn$, and $\henkinstructure$ a Henkin structure.
	If $\henkinstructure \modelsh\defn_i$ for all $i \in \{1, \dots, k\}$, then $\henkinstructure \modelsh \defn$.
\end{lemma}

\begin{proof}
	Suppose that $\henkinstructure \modelsh \defn_i$ for all $i \in \{1, \dots, k\}$.
	Assume without loss of generality that $\voc{\structure} = \sym{\defn}$.
	We need to show that $\structure$ is the least pre-fixed $\henkinclass$-point of $\cons(\cdot, \structure)$, where $\context \coloneq \structure|_{\pars{\defn}}$.
	Since $\henkinclass$ is a Henkin class for $\structure$, $\structure$ is trivially an $\henkinclass$-point.
	To show that $\structure$ is a pre-fixpoint of $\cons(\cdot, \structure)$, pick a defined predicate $P$ of $\defn$.
	Let $i$ be such that $P \in \defi{\defn_i}$.
	Since $\mbd{P}{\defn}(\bar{y}) \doteq \mbd{P}{\defn_i}(\bar{y})$ for every $\arity{P}$-tuple of variables $\bar{y}$ not occurring in $\defn$, we have that $P^{\cons (\structure, \structure)} = P^{\consg{\defn_i}{\context} (\structure, \structure)}$.
	Since $\henkinstructure \modelsh \defn_i$, it follows that $P^{\cons (\structure, \structure)} = P^{\consg{\defn_i}{\context} (\structure, \structure)} \subseteq P^{\structure}$.
	This shows that $\cons(\structure, \structure) \sqsubseteq \structure$, meaning that $\structure$ is indeed a pre-fixpoint of $\cons(\cdot, \structure)$.
	
	We now show that $\structure$ is the \emph{least} pre-fixed $\henkinclass$-point of $\cons(\cdot, \structure)$.
	Let $\structuretwo$ be a pre-fixed $\henkinclass$-point of $\cons (\cdot, \structure)$.
	Let $P$ be a defined predicate of $\defn$, and let $i$ be such that $P \in \defi{\defn_i}$. %
	It is straightforward to check that $\structuretwo |_{\sym{\defn_i}}$ is a pre-fixed $\henkinclass$-point of $\consg{\defn_i}{\context} (\cdot, \structure |_{\sym{\defn_i}})$.
	Since $\henkinstructure \modelsh\defn_i$, the structure $\structure |_{\sym{\defn_i}}$ is the least pre-fixed $\henkinclass$-point of $\consg{\defn_i}{\context} (\cdot, \structure |_{\sym{\defn_i}})$, which implies that $\structure |_{\sym{\defn_i}} \sqsubseteq \structuretwo |_{\sym{\defn_i}}$.
	Hence, $P^{\structure} = P^{\structure |_{\sym{\defn_i}}} \subseteq P^{\structuretwo |_{\sym{\defn_i}}} = P^{\structuretwo}$.
	This shows that $\structuretwo \sqsubseteq \structure$, and thus, that $\structure$ is the least pre-fixed $\henkinclass$-point of $\cons(\cdot, \structure)$.
\end{proof}

\begin{theorem}[Elementary cut Henkin completeness]
	\label{thm:henkin-valid-elementary-cuts}
	Let $\regularsequent$ be a canonical sequent %
	with stratified definitions $\defns$.
	If $\regularsequent$ is valid under the Henkin semantics, then it is EC-provable.
\end{theorem}

\begin{proof}
	Suppose that $\regularsequent$ is valid under the Henkin semantics.
	For each $i \in \{1, \dots, n\}$, let $(\defn_{i,1}, \dots, \defn_{i,k_i})$ be a decomposition of $\defn_i$.
	Consider $\defn_{1,1}, \dots, \defn_{1,k_1}, \dots, \defn_{n,k_n}, \Gamma \vdash \Delta$.
	By Lemma \ref{lem:stratified-def-henkin-entails-decomposition}, this sequent is valid under the Henkin semantics.
	We observe that $\defn_{1,1}, \dots, \defn_{1,k_1}, \dots, \defn_{n,k_n}, \Gamma \vdash \Delta$ is a canonical sequent, as for all $i, i' \in \{ 1, \dots, n\}$, $j \in \{ 1, \dots, k_i\}$ and $j' \in \{ 1, \dots, k_{i'}\}$:
	\begin{enumerate}[label=(\roman*)]
		\item $\defi{\defn_{i,j}} \cap \defi{\defn_{i',j'}} = \emptyset$ if $i \neq i'$ or if $j \neq j'$. 
		If, $i \neq i'$, this follows from the fact that $\defi{\defn_{i}} \cap \defi{\defn_{i'}} = \emptyset$, and if $i = i'$ and $j \neq j'$, from the fact that different subdefinitions of a definition define different predicates.
		\item $\defi{\defn_{i,j}} \cap \pars{\defn_{i',j'}} = \emptyset$ if $i' < i$ or if $i'=i$ and $j' < j$.
		If $i'<i$, this follows from the fact that $\defi{\defn_{i}} \cap \pars{\defn_{i'}} = \emptyset$, and if $i'=i$ and $j' < j$, from Definition \ref{def:decomposition-stratified-def}.
		Indeed, suppose that a defined predicate $P$ of $\defn_{i,j}$ occurred as a parameter of $\defn_{i,j'}$ for some $j' < j$.
		Then $P$ would occur in the body $\psi$ of a definitional rule $\forall \bar{y}: Q(\bar{s}) \rul \formulatwo$ of $\defn_{i,j'}$.
		However, letting $\ell$ be a stratification of $\defn_i$ corresponding to the decomposition $(\defn_{i,1}, \dots, \defn_{i,k_i})$, this would imply the contradiction $\ell(P) \leq \ell(Q) < \ell(P)$.
	\end{enumerate}
	Furthermore, by Remark \ref{rem:decomposition-strat-def}, %
	no defined atom occurs negatively in the body of a rule of a definition $\defn_{i,j}$.
	Thus, $\defn_{1,1}, \dots, \defn_{1,k_1}, \dots, \defn_{n,k_n}, \Gamma \vdash \Delta$ satisfies the conditions of Lemma \ref{lem:elementary-cuts-negation-parameters}, which yields
	that $\defn_{1,1}, \dots, \defn_{1,k_1}, \dots, \defn_{n,k_n}, \Gamma \vdash \Delta$ is EC-provable.
	By induction on the height of an EC-proof of $\defn_{1,1}, \dots, \defn_{1,k_1}, \dots, \defn_{n,k_n}, \Gamma \vdash \Delta$, it follows that $\regularsequent$ is also EC-provable.
	Here we used the fact that any application of (def L) w.r.t.~a definition $\defn_{i,j}$ is also an application of (def L) w.r.t.~$\defn_i$.
\end{proof}

\begin{theorem}[Cut-restriction] 
	\label{thm:cut-restriction}
	Let $\regularsequent$ be a canonical sequent %
	with stratified definitions $\defns$.
	If $\regularsequent$ is provable, then it is EC-provable.
\end{theorem}

\begin{proof}
	If $\regularsequent$ is provable, then it is valid under the Henkin semantics by Theorem \ref{thm:soundness}, whence it is EC-provable by Theorem \ref{thm:henkin-valid-elementary-cuts}.
\end{proof}

Note that Theorems \ref{thm:henkin-valid-elementary-cuts} and \ref{thm:cut-restriction} generalize Theorems \ref{thm:brotherston} and \ref{thm:brotherston-cut-elimination}, respectively, since if the bodies of the rules of the definitions $\defns$ contain no negation, then there are no elementary cuts.

The counterexamples from the start of this section leave open the possibility that cut-restriction holds for a broader fragment of sequents than the one described in Theorem \ref{thm:cut-restriction}.
Indeed, it is straightforward to check that the sequents in Examples \ref{ex:proof_non-total_liar}, \ref{ex:counterexample-cut-elimination-P-if-O-or-not-P} and \ref{ex:counterexample-cut-elimination-P-if_O-or-not-O} are EC-provable.\footnote{
	The sequent in Example \ref{ex:counterexample-cut-elimination-two-defs} is not EC-provable (as the involved definitions are positive), which shows that there are still limitations on cut-restriction in \scfoid.
}
Hence, the stratification condition on the definitions $\defns$ in Theorem \ref{thm:henkin-valid-elementary-cuts} could potentially be relaxed.
Moreover, it is straightforward to check that the sequent $\defn \vdash P$ from Example \ref{ex:counterexample-cut-elimination-P-if_O-or-not-O} is cut-free provable if we replace (def R) with the following alternative right introduction rule for defined atoms:\footnote{
	It is not difficult to show that (def R) and (def R') are interchangeable in \scfoid, i.e., that replacing (def R) with (def R') does not affect the set of \scfoid-theorems.
}
\begin{equation*}
	\begin{prooftree}
		\hypo{\defn, \Gamma \vdash \mbd{P}{\defn}(\bar{v}), \Delta}
		\infer1[(def R')]{\defn, \Gamma \vdash P(\bar{v}), \Delta}
	\end{prooftree}
\end{equation*}
or alternatively, that the sequent $\normalization{\defn} \vdash P$, with the normalization $\normalization{\defn} = \defintight{P \rul O \lor \lnot O}$ of $\defn$ is cut-free provable in \scfoid.
This raises the possibility that cut-restriction can be strengthened to cut-elimination when replacing (def R) with (def R') in \scfoid, or when restricting to normal definitions $\defns$.
Thus, there remain multiple open questions regarding the fragments of \foid on which \scfoid allows cut-elimination and cut-restriction.

%% file: conclusion.tex
\section{Conclusion} \label{sec:conclusion}

We have introduced a sequent calculus \scfoid for the logic \foid, thereby enabling formal proofs about general non-monotone inductive definitions.
\scfoid extends the sequent calculus \lkid by Brotherston and Simpson, formalizing the principle of mathematical induction \cite{jlc/BrotherstonS11}.
While covering a substantially larger class of inductive definitions, our sequent calculus remains syntactically close to \lkid, the most fundamental difference being the asymmetry between positive and negative occurrences of defined atoms in the induction rule.
We motivated our %
induction rule by linking it to the stable semantics, which is closely related to the well-founded semantics.
We established a correspondence between regular \scfoid-theorems and \scfo-theorems, in which the definitions are replaced with (infinite) first-order approximations, consisting of material implications and induction schemes (Section \ref{sec:correspondence}).
We proved that \scfoid is sound w.r.t.~the well-founded, the stable and the Henkin semantics (Section \ref{sec:soundness}), that 
\scfoid is complete w.r.t.~the Henkin semantics on regular definitions, and that
\scfoid is complete w.r.t.~the stable semantics on regular sequents with propositional definitions (Section \ref{sec:completeness}).
Finally, we showed that \scfoid admits cut-elimination on canonical sequents with positive definitions, and that \scfoid admits cut-restriction (which says that every \scfoid-theorem admits an \scfoid-proof in which all cuts are of a specific form) on canonical sequents with stratified definitions (Section \ref{sec:cut-elimination}).
Notably, \scfoid supports proofs of non-totality, thereby offering formal tools to demonstrate that certain non-monotone definitions are nonsensical.

This work paves the way for several directions for future research: 
\begin{itemize}
	\item As remarked at the end of Section \ref{sec:cut-elimination}, the examples from the start of the section leave open the possibility that cut-restriction holds for a broader fragment than the one described in Theorem \ref{thm:cut-restriction}.
	Furthermore, the cut-restriction result for the sequents $\regularsequent$ in this fragment may be strengthened to cut-elimination by replacing (def R) in \scfoid with the alternative right introduction rule for defined atoms (def R'), which derives heads of definitional rules from merged bodies $\mbd{P}{\defn}$, or by imposing normality on the definitions $\defns$.
	\item Following the intuition from Section \ref{sec:scfoid} of viewing induction hypotheses as upper bounds on defined predicates, we could refine our induction rule by substituting the negative occurrences of defined predicates with formulas describing a lower bound on these predicates.
	Such a rule could be stronger than our current induction rule.
	On the other hand, this refined rule would be more complex, requiring the choice of an additional sort of ``induction hypothesis'',
	it would presumably not be sound w.r.t.~the stable semantics, and it would deviate further from the principle of mathematical induction as used in practice.
	\item Hou Ping and Denecker enhanced the deductive strength of their sequent calculus for \foid by including least fixpoint expressions, thereby enlarging the class of expressible sets (while also introducing an infinitary inference rule) \cite{lpnmr/HouD09}.
	One could investigate whether the same principle applies to \scfoid, as the addition of least fixpoint expressions would enlarge the class of expressible induction hypotheses.
	\item As mentioned in Section \ref{sec:related-work}, 	Brotherston and Simpson also devised an infinitary sequent calculus \lkidinf and a cyclic sequent calculus \clkid for their logic of inductive definitions \cite{jlc/BrotherstonS11}.
	Compared to \lkid, a main advantage of these calculi is that they do not require induction hypotheses.
	Furthermore, \lkidinf and \clkid are strictly stronger than \lkid, in the sense that they can prove strictly more sequents \cite{Berardi2019}.
	A straightforward question is whether %
	\lkidinf and \clkid can be extended to \foid as well, and whether the corresponding proof-theoretic properties generalize to the non-monotone definitions of \foid.
	\item Even though \scfoid admits proofs of non-totality, it does not admit proofs of totality of definitions.
	Such proofs would nevertheless be valuable, as they show why certain definitions construct well-defined sets and relations.
	To enable proofs of totality, one would first need to extend the language of \foid, 
	as it is currently not possible to express that a definition is total in \foid.
	\item Besides sequent calculi, one could also develop natural deduction systems for \foid.
	While these systems are typically less susceptible to theoretical analysis than sequent calculi, their proofs come ``as close as possible to actual reasoning'', according to Gentzen \cite{Gentzen35}.
	\item One could extend \scfoid to extensions of \foid with additional language constructs, such as types (or sorts), aggregates, partial functions and modal operators.
	\item \scfoid could be implemented in proof assistants and proof logging algorithms, to enable interactive theorem proving and proof logging in \foid, respectively.
\end{itemize}

%% file: proof-induction-scheme-wf.tex
\section{Proof of Proposition \ref{prop:induction-scheme-wf}} \label{app:induction-scheme-wf}

\inductionschemewf*

To prove Proposition \ref{prop:induction-scheme-wf}, we associate with each three-valued structure $\tvs$ a pair of two-valued structures $(\tvs^l, \tvs^u)$. 
Specifically, we define $\tvs^l$ and $\tvs^u$ to have the same vocabulary and domain $D$ as $\tvs$, the same interpretation of non-predicate symbols, and for each predicate symbol $P$, $P^{\tvs^l} \coloneq \{ \bar{a} \in D^{\arity{P}} \mid P^{\tvs}(\bar{a}) = \true \}$ and $P^{\tvs^u} \coloneq \{ \bar{a} \in D^{\arity{P}} \mid P^{\tvs}(\bar{a}) = \true \text{ or } P^{\tvs}(\bar{a}) = \unknown \}$. Intuitively, $\tvs^l$ and $\tvs^u$ can be seen as two-valued approximations of $\tvs$ by interpreting all unknown facts as false and as true respectively. 
It is straightforward to see that $\tvs^l \leqt \tvs \leqt \tvs^u$, when viewing $\tvs^l$ and $\tvs^u$ as three-valued structures, and furthermore that $\tvs \leqp \tvstwo$ if and only if $\tvs^l \leqt \tvstwo^l$ and $\tvstwo^u \leqt \tvs^u$ for all three-valued structures $\tvs$ and $\tvstwo$. 
We define a satisfaction relation $\models$ between pairs $(\structure, \structuretwo)$ of two-valued structures with the same domain and the same vocabulary, and \fo-formulas $\formula$.
We do so by structural induction on $\formula$, via the following rules:
\begin{itemize}
	\item $(\structure, \structuretwo) \models t=s$ if $t^{\structure}=s^{\structure}$;
	\item $(\structure, \structuretwo) \models P(\bar{t})$ if $\bar{t}^{\structure} \in P^{\structure}$;
	\item $(\structure, \structuretwo) \models \lnot \formulatwo$ if $(\structuretwo, \structure) \not\models \formulatwo$;
	\item $(\structure, \structuretwo) \models \formulatwo \land \formulathree$ if $(\structure, \structuretwo) \models \formulatwo$ and $(\structure, \structuretwo) \models \formulathree$;
	\item $(\structure, \structuretwo) \models \forall x: \formulatwo$ if $(\structure[x:a], \structuretwo[x:a]) \models \formulatwo$ for all $a \in \dom{\structure}$.
\end{itemize}
The rules for $\lor$, $\Rightarrow$, $\Leftrightarrow$ and $\exists$ can de derived from the above rules through reformulation in terms of $\lnot$, $\land$ and $\forall$.
Under this definition, all positive occurrences of predicate symbols in $\formula$ are evaluated by $\structure$ and all negative occurrences by $\structuretwo$.
Intuitively, one can read $(\structure, \structuretwo) \models \formula$ as saying that $\formula$ is \emph{certainly} true when viewing $\structure$ as a lower bound and $\structuretwo$ as an upper bound on what is known (i.e., what is true in $\structuretwo$ is possibly true and what is false in $\structuretwo$ is certainly false) \cite{iclp06/WittocxVMDB06}. 
We will write $\formula^{(\structure, \structuretwo)} = \true$ if $(\structure, \structuretwo) \models \formula$, and $\formula^{(\structure, \structuretwo)} = \false$ if $(\structure, \structuretwo) \not\models \formula$.
For a merged body $\mb{P}(\bar{y})$, %
we write $\mb{P}^{(\structure, \structuretwo)}(\bar{a})$ for $\mbd{P}{\defn}(\bar{y})^{(\structure[\bar{y} : \bar{a}], \structuretwo[\bar{y} : \bar{a}])}$ (given that the definition $\defn$ is clear from the context). 
In the proof of Proposition \ref{prop:induction-scheme-wf}, we will use the following two properties of this satisfaction relation:

\begin{proposition} \label{prop:decomposition property}
	Let $\tvs$ be a three-valued structure and $(\tvs^l, \tvs^u)$ the corresponding pair of two-valued structures. Let $\formula$ be an \fo-formula. If $\formula^{\tvs} \in \{\false, \true\}$, then $\formula^{\tvs} = \formula^{(\tvs^l, \tvs^u)} = \formula^{(\tvs^u, \tvs^l)}$.\footnote{
		This property extends to \emph{four-valued structures}, which have a fourth truth value $\inconsistent$, standing for `inconsistent'.
		The truth order $\leqt$ and the precision order $\leqp$ on $\{\false, \true, \unknown\}$ extend to $\{\false, \true, \unknown, \inconsistent\}$ as follows: $\false \leqt \inconsistent$, $\inconsistent \leqt \true$ and $\unknown$ is incomparable with $\inconsistent$ w.r.t.~$\leqt$; furthermore, 
		$\inconsistent$ is the greatest element w.r.t. $\leqp$, i.e., $\false \leqp \inconsistent$, $\unknown \leqp \inconsistent$ and $\true \leqp \inconsistent$.
		The operator $\lnot$ on $\{\false, \true, \unknown\}$ moreover extends to $\{\false, \true, \unknown, \inconsistent\}$ by letting $\lnot \inconsistent \coloneq \inconsistent$.
		Under this extended truth order, Kleene's truth assignment from Section \ref{sec:well-founded-semantics} extends to Belnap's truth assignment \cite{Belnap77}.
		The pair of two-valued structures $(\tvs^l, \tvs^u)$ associated with a four-valued structure $\tvs$ is defined as for a three-valued structure, with the difference that $P^{\tvs^l} \coloneq \{ \bar{a} \in D^{\arity{P}} \mid P^{\tvs}(\bar{a}) = \true \text{ or } P^{\tvs}(\bar{a}) = \inconsistent \}$ for each predicate symbol $P$.
		In this way, the truth values $\true$, $\false$, $\unknown$ and $\inconsistent$ in four-valued structures correspond to the pairs of truth values $(\true, \true)$, $(\false, \false)$, $(\false, \true)$ resp.~$(\true, \false)$ in pairs of two-valued structures.
		The proof of Proposition \ref{prop:decomposition property} extends directly to the four-valued setting.
	} 
\end{proposition}

\begin{proof}
	We show by structural induction on $\formula$ that for all three-valued structures $\tvs$, if $\formula^{\tvs} \in \{\false, \true\}$, then $\formula^{\tvs} = \formula^{(\tvs^l, \tvs^u)} = \formula^{(\tvs^u, \tvs^l)}$:
	
	\textbf{Case}: $\formula$ is of the form $t=s$.
	Let $\tvs$ be a three-valued structure such that $(t=s)^{\tvs} \in \{\false, \true\}$.
	Since $\tvs$, $\tvs^l$ and $\tvs^u$ have the same interpretation of terms, $(t=s)^{\tvs} = (t=s)^{\tvs^l} = (t=s)^{(\tvs^l, \tvs^u)} = (t=s)^{\tvs^u} = (t=s)^{(\tvs^u, \tvs^l)}$.
	
	\textbf{Case}: $\formula$ is of the form $P(\bar{t})$.
	Let $\tvs$ be a three-valued structure such that $P(\bar{t})^{\tvs} \in \{\false, \true\}$.
	If $P(\bar{t})^{\tvs} \in \{\false, \true\}$, then $P(\bar{t})^{(\tvs^l, \tvs^u)} = P(\bar{t})^{\tvs^l} = P(\bar{t})^{\tvs}$ and $P(\bar{t})^{(\tvs^u, \tvs^l)} = P(\bar{t})^{\tvs^u} = P(\bar{t})^{\tvs}$.
	
	\textbf{Case}: $\formula$ is of the form $\lnot \formulatwo$.
	Let $\tvs$ be a three-valued structure such that $(\lnot \formulatwo)^{\tvs} \in \{\false, \true\}$.
	Then also $\formulatwo^{\tvs} \in \{\false, \true\}$.
	By induction hypothesis, $\formulatwo^{\tvs} = \formulatwo^{(\tvs^l, \tvs^u)} = \formulatwo^{(\tvs^u, \tvs^l)}$.
	It follows that $(\lnot \formulatwo)^{\tvs} = \lnot \formulatwo^{\tvs} = \lnot \formulatwo^{(\tvs^l, \tvs^u)} = (\lnot \formulatwo)^{(\tvs^u, \tvs^l)}$ and $(\lnot \formulatwo)^{\tvs} = \lnot \formulatwo^{\tvs} = \lnot \formulatwo^{(\tvs^u, \tvs^l)} = (\lnot \formulatwo)^{(\tvs^l, \tvs^u)}$.
	
	\textbf{Case}: $\formula$ is of the form $\formulatwo \land \formulathree$.
	Let $\tvs$ be a three-valued structure such that $(\formulatwo \land \formulathree)^{\tvs} \in \{\false, \true\}$.
	Suppose that $(\formulatwo \land \formulathree)^{\tvs} = \false$. 
	Then $\formulatwo^{\tvs} = \false$ or $\formulathree^{\tvs} = \false$.
	Assume without loss of generality that $\formulatwo^{\tvs} = \false$.
	By induction hypothesis, $\formulatwo^{(\tvs^l, \tvs^u)} = \formulatwo^{(\tvs^u, \tvs^l)} = \formulatwo^{\tvs} = \false$.
	It follows that $(\formulatwo \land \formulathree)^{(\tvs^l, \tvs^u)} = \formulatwo^{(\tvs^l, \tvs^u)} \land \formulathree^{(\tvs^l, \tvs^u)} = \false$ and $(\formulatwo \land \formulathree)^{(\tvs^u, \tvs^l)} = \formulatwo^{(\tvs^u, \tvs^l)} \land \formulathree^{(\tvs^u, \tvs^l)} = \false$.
	
	Now suppose that $(\formulatwo \land \formulathree)^{\tvs} = \true$.
	Then $\formulatwo^{\tvs} = \true$ and $\formulathree^{\tvs} = \true$.
	By induction hypothesis, $\formulatwo^{(\tvs^l, \tvs^u)} = \formulatwo^{(\tvs^u, \tvs^l)} = \formulatwo^{\tvs} = \true$ and $\formulathree^{(\tvs^l, \tvs^u)} = \formulathree^{(\tvs^u, \tvs^l)} = \formulathree^{\tvs} = \true$.
	It follows that $(\formulatwo \land \formulathree)^{(\tvs^l, \tvs^u)} = \formulatwo^{(\tvs^l, \tvs^u)} \land \formulathree^{(\tvs^l, \tvs^u)} = \true$ and $(\formulatwo \land \formulathree)^{(\tvs^u, \tvs^l)} = \formulatwo^{(\tvs^u, \tvs^l)} \land \formulathree^{(\tvs^u, \tvs^l)} = \true$.
	
	\textbf{Cases}: $\formula$ is of the form $\formulatwo \lor \formulathree$, $\formulatwo \Rightarrow \formulathree$, or $\formulatwo \Leftrightarrow \formulathree$.
	These cases are similar to the previous case.
	
	\textbf{Case}: $\formula$ is of the form $\forall x: \formulatwo$.
	Let $\tvs$ be a three-valued structure such that $(\forall x: \formulatwo)^{\tvs} \in \{\false, \true\}$.
	Let $D$ be the domain of $\tvs$.
	Suppose that $(\forall x: \formulatwo)^{\tvs} = \false$.
	Then $\formulatwo^{\tvs[x:a]} = \false$ for some $a \in D$.
	By induction hypothesis, $\formulatwo^{(\tvs[x:a]^l, \tvs[x:a]^u)} = \formulatwo^{(\tvs[x:a]^u, \tvs[x:a]^l)} = \formulatwo^{\tvs[x:a]} = \false$.
	Note that $\tvs[x:a]^l = \tvs^l[x:a]$ and $\tvs[x:a]^u = \tvs^u[x:a]$, i.e., $(\tvs^l[x:a], \tvs^u[x:a])$ is the pair of two-valued structures corresponding to $\tvs[x:a]$.
	Therefore, $\formulatwo^{(\tvs^l[x:a], \tvs^u[x:a])} = \formulatwo^{(\tvs^u[x:a], \tvs^l[x:a])} = \false$.
	It follows that $(\forall x: \formulatwo)^{(\tvs^l, \tvs^u)} = (\forall x: \formulatwo)^{(\tvs^u, \tvs^l)} = \false$.
	
	Now suppose that $(\forall x: \formulatwo)^{\tvs} = \true$. 
	Then $\formulatwo^{\tvs[x:a]} = \true$ for all $a \in D$.
	By induction hypothesis, $\formulatwo^{(\tvs[x:a]^l, \tvs[x:a]^u)} = \formulatwo^{(\tvs[x:a]^u, \tvs[x:a]^l)} = \formulatwo^{\tvs[x:a]} = \true$ for all $a \in D$.
	Since $\tvs[x:a]^l = \tvs^l[x:a]$ and $\tvs[x:a]^u = \tvs^u[x:a]$, we have that $\formulatwo^{(\tvs^l[x:a], \tvs^u[x:a])} = \formulatwo^{(\tvs^u[x:a], \tvs^l[x:a])} = \true$ for all $a \in D$.
	It follows that $(\forall x: \formulatwo)^{(\tvs^l, \tvs^u)} = (\forall x: \formulatwo)^{(\tvs^u, \tvs^l)} = \true$.
	
	\textbf{Case}: $\formula$ is of the form $\exists x: \formulatwo$.
	This case is similar to the previous case.
\end{proof}

\begin{proposition} \label{prop:(anti)monotonicity truth assignment}
	Let $\structure, \structure', \structuretwo$ and $\structuretwo'$ be two-valued structures with the same domain, vocabulary, and interpretation of non-predicate symbols. 
	Let $\formula$ be an \fo-formula. 
	If $\structure \leqt \structure'$ and $\structuretwo' \leqt \structuretwo$, then $\formula^{(\structure, \structuretwo)} \leqt \formula^{(\structure', \structuretwo')}$ and $\formula^{(\structuretwo', \structure')} \leqt \formula^{(\structuretwo, \structure)}$.
\end{proposition}

\begin{proof}
	We show by structural induction on $\formula$ that $\formula^{(\structure, \structuretwo)} \leqt \formula^{(\structure', \structuretwo')}$ and $\formula^{(\structuretwo', \structure')} \leqt \formula^{(\structuretwo, \structure)}$ for all two-valued structures $\structure, \structure', \structuretwo$ and $\structuretwo'$ with the same domain, vocabulary, and interpretation of non-predicate symbols, such that $\structure \leqt \structure'$ and $\structuretwo' \leqt \structuretwo$.
	
	\textbf{Case}: $\formula$ is of the form $t=s$.
	Let $\structure, \structure', \structuretwo$ and $\structuretwo'$ be as described above.
	Since $\structure, \structure', \structuretwo$ and $\structuretwo'$ have the same interpretation of terms, $(t=s)^{\structure} = (t=s)^{\structure'} = (t=s)^{\structuretwo} = (t=s)^{\structuretwo'}$, and hence, $(t=s)^{(\structure, \structuretwo)} = (t=s)^{(\structure', \structuretwo')} = (t=s)^{(\structuretwo, \structure)} = (t=s)^{(\structuretwo', \structure')}$.
	In particular, $(t=s)^{(\structure, \structuretwo)} \leqt (t=s)^{(\structure', \structuretwo')}$ and $(t=s)^{(\structuretwo', \structure')} \leqt (t=s)^{(\structuretwo, \structure)}$.
	
	\textbf{Case} $\formula$ is of the form $P(\bar{t})$.
	Let $\structure, \structure', \structuretwo$ and $\structuretwo'$ be as described above.
	Then $P(\bar{t})^{(\structure, \structuretwo)} = P(\bar{t})^{\structure} \leqt P(\bar{t})^{\structure'} = P(\bar{t})^{(\structure', \structuretwo')}$ and $P(\bar{t})^{(\structuretwo', \structure')} = P(\bar{t})^{\structuretwo'} \leqt P(\bar{t})^{\structuretwo} = P(\bar{t})^{(\structuretwo, \structure)}$.
	
	\textbf{Case}: $\formula$ is of the form $\lnot \formulatwo$.
	Let $\structure, \structure', \structuretwo$ and $\structuretwo'$ be as described above.
	By induction hypothesis, $\formulatwo^{(\structure, \structuretwo)} \leqt \formulatwo^{(\structure', \structuretwo')}$ and $\formulatwo^{(\structuretwo', \structure')} \leqt \formulatwo^{(\structuretwo, \structure)}$.
	It follows that $\lnot \formulatwo^{(\structure', \structuretwo')} \leqt \lnot \formulatwo^{(\structure, \structuretwo)}$ and $\lnot \formulatwo^{(\structuretwo, \structure)} \leqt \lnot \formulatwo^{(\structuretwo', \structure')}$.
	Using the fact that $(\lnot \formulathree)^{(\mathcal{K}, \mathcal{L})} = \lnot \formulathree^{(\mathcal{L}, \mathcal{K})}$ for all \fo-formulas $\formulathree$ and all pairs of two-valued structures $(\mathcal{K}, \mathcal{L})$, we deduce that $(\lnot \formulatwo)^{(\structuretwo', \structure')} = \lnot \formulatwo^{(\structure', \structuretwo')} \leqt \lnot \formulatwo^{(\structure, \structuretwo)} = (\lnot \formulatwo)^{(\structuretwo, \structure)}$ and $(\lnot \formulatwo)^{(\structure, \structuretwo)} = \lnot \formulatwo^{(\structuretwo, \structure)} \leqt \lnot \formulatwo^{(\structuretwo', \structure')} = (\lnot \formulatwo)^{(\structure', \structuretwo')}$.
	
	\textbf{Case}: $\formula$ is of the form $\formulatwo \land \formulathree$.
	Let $\structure, \structure', \structuretwo$ and $\structuretwo'$ be as described above.
	By induction hypothesis, we have that $\formulatwo^{(\structure, \structuretwo)} \leqt \formulatwo^{(\structure', \structuretwo')}$, $\formulathree^{(\structure, \structuretwo)} \leqt \formulathree^{(\structure', \structuretwo')}$, $\formulatwo^{(\structuretwo', \structure')} \leqt \formulatwo^{(\structuretwo, \structure)}$ and $\formulathree^{(\structuretwo', \structure')} \leqt \formulathree^{(\structuretwo, \structure)}$.
	It follows that $(\formulatwo \land \formulathree)^{(\structure, \structuretwo)} = \min_{\leqt} \{\formulatwo^{(\structure, \structuretwo)}, \formulathree^{(\structure, \structuretwo)}\} \leqt \min_{\leqt} \{\formulatwo^{(\structure', \structuretwo')}, \formulathree^{(\structure', \structuretwo')}\} = (\formulatwo \land \formulathree)^{(\structure', \structuretwo')}$ and $(\formulatwo \land \formulathree)^{(\structuretwo', \structure')} = \min_{\leqt} \{\formulatwo^{(\structuretwo', \structure')}, \formulathree^{(\structuretwo', \structure')}\}$\\ 
	$\leqt \min_{\leqt} \{\formulatwo^{(\structuretwo, \structure)}, \formulathree^{(\structuretwo, \structure)}\} = (\formulatwo \land \formulathree)^{(\structuretwo, \structure)}$.
	
	\textbf{Cases}: $\formula$ is of the form $\formulatwo \lor \formulathree$, $\formulatwo \Rightarrow \formulathree$, or $\formulatwo \Leftrightarrow \formulathree$.
	These cases are similar to the previous case.
	
	\textbf{Case}: $\formula$ is of the form $\forall x: \formulatwo$.
	Let $\structure, \structure', \structuretwo$ and $\structuretwo'$ be as described above.
	Let $D$ be the domain of $\structure, \structure', \structuretwo$ and $\structuretwo'$.
	Then $\structure[x:a] \leqt \structure'[x:a]$ and $\structuretwo'[x:a] \leqt \structuretwo[x:a]$ for all $a \in D$.
	By induction hypothesis, $\formulatwo^{(\structure[x:a], \structuretwo[x:a])} \leqt \formulatwo^{(\structure'[x:a], \structuretwo'[x:a])}$ and $\formulatwo^{(\structuretwo'[x:a], \structure'[x:a])} \leqt \formulatwo^{(\structuretwo[x:a], \structure[x:a])}$ for all $a \in D$.
	It follows that $(\forall x: \formulatwo)^{(\structure, \structuretwo)} = \min_{\leqt} \{\formulatwo^{(\structure[x:a], \structuretwo[x:a])} \mid a \in D \} \leqt \min_{\leqt} \{\formulatwo^{(\structure'[x:a], \structuretwo'[x:a])} \mid a \in D \} = (\forall x: \formulatwo)^{(\structure', \structuretwo')}$ and $(\forall x: \formulatwo)^{(\structuretwo', \structure')} = \min_{\leqt} \{\formulatwo^{(\structuretwo'[x:a], \structure'[x:a])} \mid a \in D \} \leqt \min_{\leqt} \{\formulatwo^{(\structuretwo[x:a], \structure[x:a])} \mid a \in D \} = (\forall x: \formulatwo)^{(\structuretwo, \structure)}$.
	
	\textbf{Case}: $\formula$ is of the form $\exists x: \formulatwo$.
	This case is similar to the previous case.
\end{proof}

\begin{proof}[Proof of Proposition \ref{prop:induction-scheme-wf}]
	Suppose that $\structure \modelswf \defn$.
	We assume without loss of generality that $\defn$ is normal and that $\voc{\structure} = \sym{\defn}$, so that $\structure$ is the well-founded model of $\defn$ in $\structure|_{\pars{\formula}}$.
	Let $A_\Pi$ be a tuple of sets $A_Q \subseteq D^{\arity{Q}}$ for every $Q \in \Pi$, and assume that
	\begin{equation} \label{eq:assumption minor premise}
		\structure[G_\Pi : A_\Pi] \models \forall \bar{y}: \mbd{Q}{\defn}(\bar{y})\substpos{G_\Pi}{\Pi} \Rightarrow G_Q(\bar{y})
	\end{equation}
	for every $Q \in \Pi$, where $\bar{y}$ is an $\arity{Q}$-tuple of variables not occurring in $\defn$. %
	Take $P \in \Pi$.
	We need to show that $\structure[G_\Pi : A_\Pi] \models \forall \bar{x}: P(\bar{x}) \Rightarrow G_P(\bar{x})$, which comes down to $P^{\structure} \subseteq A_P$. 
	Suppose by contradiction that $P^{\structure} \not\subseteq A_P$. 
	For every $Q \in \Pi$, let $V_Q$ denote the set $Q^{\structure} \setminus A_Q$. 
	Then $V_P \neq \emptyset$. 
	Let $\wellfoundedinduction$ be a well-founded induction for $\defn$ in $\structure|_{\pars{\defn}}$.
	Take $i$ such that $Q^{\tvs_i}(\bar{b})=\unknown$ for all $Q \in \Pi$ and all $\bar{b} \in V_Q$, but $Q^{\tvs_{i+1}}(\bar{b})=\true$ for some $Q \in \Pi$ and some $\bar{b} \in V_Q$. 
	Note that such $i$ must exist, since all domain atoms $Q(\bar{b})$ with $Q \in \Pi$ and $\bar{b} \in V_Q$ are unknown in $\tvs_0$, but true in $\tvs_\beta = \structure$, by definition of $V_Q$. 
	Hence, their truth must be derived at a certain step $i+1$. 
	Furthermore, the fact that $V_P$ is non-empty guarantees that there is at least one domain atom of this form. 
	Fix a predicate symbol $Q \in \Pi$ and a tuple $\bar{b} \in V_Q$ for which $Q^{\tvs_{i+1}}(\bar{b})=\true$. 
	In words, the truth of $Q(\bar{b})$ gets derived at step $i+1$ of the well-founded induction. 
	This is only possible if $\mb{Q(\bar{b})}^{\tvs_{i}}=\true$. 
	By Proposition \ref{prop:decomposition property}, it follows that $\mb{Q}^{(\tvs_{i}^l, \tvs_{i}^u)}(\bar{b})=\true$. 
	Furthermore, we have that $\mb{Q}^{(\structure[\Pi : A_\Pi], \structure)}(\bar{b}) = \false$. 
	Indeed, $\mb{Q}^{(\structure[\Pi : A_\Pi], \structure)}(\bar{b}) = \true$ would mean that $\structure[G_\Pi : A_\Pi, \bar{y}:\bar{b}] \models \mbd{Q}{\defn}(\bar{y})\substpos{G_\Pi}{\Pi}$, %
	and by (\ref{eq:assumption minor premise}), this entails $\structure[G_\Pi : A_\Pi, \bar{y}: \bar{b}] \models G_Q(\bar{y})$. 
	However, the latter comes down to $\bar{b} \in A_Q$, which contradicts $\bar{b} \in V_Q$. 
	Thus,
	\begin{equation} \label{eq:t>f}
		\mb{Q}^{(\structure[\Pi : A_\Pi], \structure)}(\bar{b}) = \false \lesst \true = \mb{Q}^{(\tvs_i^l, \tvs_i^u)}(\bar{b}).
	\end{equation}
	
	We proceed by observing the inequalities
	\begin{equation} \label{eq:inequalities}
		\tvs_i^l \leqt \structure[\Pi : A_\Pi] \qquad \text{ and } \qquad \structure \leqt \tvs_i^u.
	\end{equation}
	The right inequality follows from the fact that $\tvs_i \leqp \tvs_\beta$, and hence $\structure = \tvs_\beta \leqt \tvs_\beta^u \leqt \tvs_i^u$. 
	Similarly, we find that $\tvs_i^l \leq \structure$. 
	Thus, to prove the left inequality, it suffices to show that $R^{\tvs_i^l} \subseteq R^{\structure[\Pi : A_\Pi]} = A_R$ for all $R \in \Pi$. 
	Take $\bar{c} \in R^{\tvs_i^l}$. 
	Then $R^{\tvs_i}(\bar{c}) = \true$. 
	This means that the truth of $R(\bar{c})$ has been derived by step $i$ of the well-founded induction. 
	Hence, $\bar{c} \in R^{\structure}$. 
	This entails that $\bar{c} \in A_R$. 
	Indeed, $\bar{c} \notin A_R$ would imply that $\bar{c} \in R^{\structure} \setminus A_R= V_R$, but then $R^{\tvs_i}(\bar{c}) = \unknown$ by the choice of $i$ contradicting the fact that $R^{\tvs_i}(\bar{c}) = \true$.
	This proves the left inequality of (\ref{eq:inequalities}).
	By Proposition \ref{prop:(anti)monotonicity truth assignment}, the inequalities of (\ref{eq:inequalities}) imply that $\mb{Q}^{(\tvs_i^l, \tvs_i^u)}(\bar{b}) \leqt \mb{Q}^{(\structure[\Pi : A_\Pi], \structure)}(\bar{b})$. 
	However, this contradicts (\ref{eq:t>f}). 
	We can thus conclude that $P^{\structure} \subseteq A_P$, which finishes the proof.
\end{proof}

%% file: proofs-section-scfoid.tex
\section{Proofs of Section \ref{sec:scfoid}} \label{app:proofs-section-scfoid}

\begin{lemma} \label{lem:safe-replacement-bound-vars}
	Let $\Gamma$ and $\Delta$ be sets of \foid-formulas, $\formula$ an $\foid$-formula, $x$ a variable and $y$ a variable that does not occur freely in $\formula$.
	Then 
	\begin{itemize}
		\item $\Gamma, \forall x: \formula \vdash \Delta$ is \scfoid-provable iff\/ $\Gamma, \forall y: \formula[y/x] \vdash \Delta$ is \scfoid-provable;
		\item $\Gamma \vdash \forall x: \formula, \Delta$ is \scfoid-provable iff\/ $\Gamma \vdash \forall y: \formula[y/x], \Delta$ is \scfoid-provable;
		\item $\Gamma, \exists x: \formula \vdash \Delta$ is \scfoid-provable iff\/ $\Gamma, \exists y: \formula[y/x] \vdash \Delta$ is \scfoid-provable;
		\item $\Gamma \vdash \exists x: \formula, \Delta$ is \scfoid-provable iff\/ $\Gamma \vdash \exists y: \formula[y/x], \Delta$ is \scfoid-provable.
	\end{itemize}
\end{lemma}

\begin{proof}
If $\Gamma, \forall x: \formula \vdash \Delta$ is \scfoid-provable, then so is $\Gamma, \forall y: \formula[y/x] \vdash \Delta$, as can be seen from the following \scfoid-derivation:
\[
\begin{prooftree}
	\infer0[(ax)]{\formula \vdash \formula}
	\infer1[($\forall$L)]{\forall y: \formula[y/x] \vdash \formula}
	\infer1[($\forall$R)]{\forall y: \formula[y/x] \vdash \forall x: \formula}
	\infer1[(wk)]{\Gamma, \forall y: \formula[y/x] \vdash \forall x: \formula, \Delta}
	\hypo{\Gamma, \forall x: \formula \vdash \Delta}
	\infer1[(wk)]{\Gamma, \forall x: \formula, \forall y: \formula[y/x] \vdash \Delta}
	\infer2[(cut)]{\Gamma, \forall y: \formula[y/x] \vdash \Delta}
\end{prooftree}
\]
In the application of ($\forall$L), we used that $\formula[y/x][x/y] \doteq \formula$, and in the application of ($\forall$R), we used that $x$ does not occur freely in $\forall y: \formula[y/x]$.
The proof for the other direction and the other cases is similar.
\end{proof}

\normalizationequiderivable*

\begin{proof}
	Suppose that $\defn, \Gamma \vdash \Delta$ is \scfoid-provable.
	We show that $\normalization{\defn}, \Gamma\vdash \Delta$ is \scfoid-provable by induction on the height of an \scfoid-proof $\prf$ of $\defn, \Gamma \vdash \Delta$.
	The base cases are trivial, and so are most cases in the inductive step.
	We only discuss the interesting cases, which are those where $\defn, \Gamma \vdash \Delta$ is derived by (def R) or (def L) in $\prf$.
	
	\textbf{Case}: $\defn, \Gamma \vdash \Delta$ is derived by (def R):
	\begin{equation*}
		\begin{prooftree}
			\hypo{\defn, \Gamma \vdash \formula[\bar{s}/\bar{x}], \Delta}
			\infer1[(def R)]{\defn, \Gamma \vdash P(\bar{t}[\bar{s}/\bar{x}]), \Delta}
		\end{prooftree}
	\end{equation*}
	Then $\defn, \Gamma \vdash \formula[\bar{s}/\bar{x}], \Delta$ admits an \scfoid-proof of smaller height than $\prf$.
	By induction hypothesis, $\normalization{\defn}, \Gamma \vdash \formula[\bar{s}/\bar{x}], \Delta$ is \scfoid-provable, and by the following \scfoid-derivation, so is $\normalization{\defn}, \Gamma \vdash P(\bar{t}[\bar{s}/\bar{x}]), \Delta$.
	\begin{equation*}
		\begin{prooftree}
			\infer0[($=$R)]{\normalization{\defn}, \Gamma \vdash \bar{t}[\bar{s}/\bar{x}] = \bar{t}[\bar{s}/\bar{x}], \Delta}
			\hypo{\normalization{\defn}, \Gamma \vdash \formula[\bar{s}/\bar{x}], \Delta}
			\infer2[($\land$R)]{\normalization{\defn}, \Gamma \vdash \bar{t}[\bar{s}/\bar{x}] = \bar{t}[\bar{s}/\bar{x}] \land \formula[\bar{s}/\bar{x}], \Delta}
			\infer1[($\exists$R)]{\normalization{\defn}, \Gamma \vdash \exists \bar{x}: \bar{t}[\bar{s}/\bar{x}] = \bar{t} \land \formula, \Delta}
			\infer1[(wk), ($\lor$R)]{\normalization{\defn}, \Gamma \vdash \mbd{Q}{\defn}(\bar{t}[\bar{s}/\bar{x}]), \Delta}
			\infer1[(def R)]{\normalization{\defn}, \Gamma \vdash P(\bar{t}[\bar{s}/\bar{x}]), \Delta}
		\end{prooftree}
	\end{equation*}
	In the application of ($\exists$R), we assumed that none of the variables in $\bar{x}$ occur in $\bar{s}$.
	If this condition is not met, we can rename the variables $\bar{x}$, as is justified by Lemma \ref{lem:safe-replacement-bound-vars}.
	In the (repeated) application of (wk) and ($\lor$R), we used the fact that $\exists \bar{x}: \bar{t}[\bar{s}/\bar{x}] = \bar{t} \land \formula$ is one of the disjuncts in the merged body $\mbd{Q}{\defn}(\bar{t}[\bar{s}/\bar{x}])$ of $P(\bar{t}[\bar{s}/\bar{x}])$, and in the application of (def R), we used that $\normalization{\defn}$ contains a definitional rule $\forall \bar{y} : P(\bar{y}) \rul \mbd{Q}{\defn}(\bar{y})$.
	
	\textbf{Case}: $\defn, \Gamma \vdash \Delta$ is derived by (def L):
	\begin{equation*}
		\begin{prooftree}
			\hypo{\text{minor premises}}
			\hypo{\defn, \Gamma, \ihf{P}{\bar{t}} \vdash \Delta}
			\infer2[(def L)]{\defn, \Gamma, P(\bar{t}) \vdash \Delta}
		\end{prooftree}
	\end{equation*}
	For every definitional rule $\forall \bar{y}: Q(\bar{s}) \rul \formulatwo$ in $\defn$ with $Q \in \Pi$, this rule has a minor premise $\defn, \Gamma, \formulatwo\substposihpi \vdash \ihf{Q}{\bar{s}}, \Delta$.
	The premises admit \scfoid-proofs of smaller height than $\prf$.
	By induction hypothesis, the sequents 
	\begin{equation} \label{eq:minor-premise-normalization}
		\normalization{\defn}, \Gamma, \formulatwo\substposihpi \vdash \ihf{Q}{\bar{s}}, \Delta
	\end{equation}
	corresponding to the definitional rules $\forall \bar{y}: Q(\bar{s}) \rul \formulatwo$ in $\defn$ with $Q \in \Pi$ are \scfoid-provable, and so is 
	\begin{equation} \label{eq:major-premise-normalization}
		\normalization{\defn}, \Gamma, \ihf{P}{\bar{t}} \vdash \Delta .
	\end{equation}
	We derive the sequent $\normalization{\defn}, \Gamma, P(\bar{t}) \vdash \Delta$ from these sequents by application of (def L) with the same set predicate set $\Pi$ and the same induction hypotheses $\ihs{\Pi}$ as the application above.
	The major premise is precisely (\ref{eq:major-premise-normalization}).
	The minor premises are of the form $\normalization{\defn}, \Gamma, \mbd{Q}{\defn}(\bar{z})\substposihpi \vdash \ihf{Q}{\bar{z}}, \Delta$ for $Q \in \Pi$.
	Since $\mbd{Q}{\defn}(\bar{z})\substposihpi$ is the disjunction of formulas of the form $\exists \bar{y}: \bar{z} = \bar{s} \land \formulatwo\substposihpi$ corresponding to the rules $\forall \bar{y}: Q(\bar{s}) \rul \formulatwo$ in $\defn$ with $Q \in \Pi$, it can be derived from sequents of the form 
	\begin{equation} \label{eq:minor-prem-norm-disjunct}
		\normalization{\defn}, \Gamma, \exists \bar{y}: \bar{z} = \bar{s} \land \formulatwo\substposihpi \vdash \ihf{Q}{\bar{z}}, \Delta
	\end{equation}
	corresponding to the same rules.
	Fixing a rule $\forall \bar{y}: Q(\bar{s}) \rul \formulatwo$ in $\defn$ with $Q \in \Pi$, we derive (\ref{eq:minor-prem-norm-disjunct}) from (\ref{eq:major-premise-normalization}) as follows:
	\begin{equation*}
		\begin{prooftree}
			\hypo{\normalization{\defn}, \Gamma, \formulatwo\substposihpi \vdash \ihf{Q}{\bar{s}}, \Delta}
			\infer1[($=$L)]{\normalization{\defn}, \Gamma, \bar{z} = \bar{s}, \formulatwo\substposihpi \vdash \ihf{Q}{\bar{z}}, \Delta}
			\infer1[($\land$L)]{\normalization{\defn}, \Gamma, \bar{z} = \bar{s} \land \formulatwo\substposihpi \vdash \ihf{Q}{\bar{z}}, \Delta}
			\infer1[($\exists$L)]{\normalization{\defn}, \Gamma, \exists \bar{y}: \bar{z} = \bar{s} \land \formulatwo\substposihpi \vdash \ihf{Q}{\bar{z}}, \Delta}
		\end{prooftree}
	\end{equation*}
	In the application of ($\exists$L), we assumed that none of the variables in $\bar{y}$ occur freely in $\normalization{\defn}$, $\Gamma$, $\ihf{Q}{\bar{z}}$ or $\Delta$.
	If this condition is not met, we can rename the variables $\bar{y}$, as is justified by Lemma \ref{lem:safe-replacement-bound-vars}.
	
	Suppose conversely that $\normalization{\defn}, \Gamma \vdash \Delta$ is \scfoid-provable.
	We show that $\defn, \Gamma \vdash \Delta$ is \scfoid-provable by induction on the height of an \scfoid-proof $\prf$ of $\normalization{\defn}, \Gamma \vdash \Delta$.
	Again, we restrict to the non-trivial cases, which are those where $\normalization{\defn}, \Gamma\vdash \Delta$ is derived by (def R) or (def L) in $\prf$.
	
	\textbf{Case}: $\normalization{\defn}, \Gamma \vdash \Delta$ is derived by (def R). 
	Since the definitional rules of $\normalization{\defn}$ are of the form $\forall \bar{y}: P(\bar{y}) \rul \mbd{Q}{\defn}(\bar{y})$, the final inference in $\prf$ is of the form:
	\begin{equation*}
		\begin{prooftree}
			\hypo{\normalization{\defn}, \Gamma\vdash \mbd{Q}{\defn}(\bar{s}), \Delta}
			\infer1[(def R)]{\normalization{\defn}, \Gamma\vdash P(\bar{s}), \Delta}
		\end{prooftree}
	\end{equation*}
	Then $\defn, \Gamma \vdash \mbd{Q}{\defn}(\bar{s}), \Delta$ admits an \scfoid-proof of smaller height than $\prf$.
	By induction hypothesis, $\defn, \Gamma \vdash \mbd{Q}{\defn}(\bar{s}), \Delta$ 
	is \scfoid-provable.
	To show that $\defn, \Gamma \vdash P(\bar{s}), \Delta$ is \scfoid-provable it suffices to show that $\defn, \Gamma, \mbd{Q}{\defn}(\bar{s}) \vdash P(\bar{s}), \Delta$ is \scfoid-provable, as illustrated by the following \scfoid-derivation: %
	\begin{equation*}
		\begin{prooftree}
			\hypo{\defn, \Gamma \vdash \mbd{Q}{\defn}(\bar{s}), \Delta}
			\infer1[(wk)]{\defn, \Gamma \vdash \mbd{Q}{\defn}(\bar{s}), P(\bar{s}), \Delta}
			\hypo{\defn, \Gamma, \mbd{Q}{\defn}(\bar{s}) \vdash P(\bar{s}), \Delta}
			\infer2[(cut)]{\defn, \Gamma \vdash P(\bar{s}), \Delta}
		\end{prooftree}
	\end{equation*}
	By (repeated) application of ($\lor$L), it suffices to show that $\defn, \Gamma, \exists \bar{x} : \bar{s} = \bar{t} \land \formula \vdash P(\bar{s}), \Delta$ is \scfoid-provable for a given rule $\defrul$ in $\defn$ defining $P$, which is done as follows:
	\begin{equation*}
		\begin{prooftree}
			\infer0[(ax)]{\defn, \Gamma, \formula[\bar{z}/\bar{x}] \vdash \formula[\bar{z}/\bar{x}], \Delta}
			\infer1[(def R)]{\defn, \Gamma, \formula[\bar{z}/\bar{x}] \vdash P(\bar{t}[\bar{z}/\bar{x}]), \Delta}
			\infer1[($=$L)]{\defn, \Gamma, \bar{s} = \bar{t}[\bar{z}/\bar{x}], \formula[\bar{z}/\bar{x}] \vdash P(\bar{s}), \Delta}
			\infer1[($\land$L)]{\defn, \Gamma, \bar{s} = \bar{t}[\bar{z}/\bar{x}] \land \formula[\bar{z}/\bar{x}] \vdash P(\bar{s}), \Delta}
			\infer1[($\exists$L)]{\defn, \Gamma, \exists \bar{x} : \bar{s} = \bar{t} \land \formula \vdash P(\bar{s}), \Delta}
		\end{prooftree}
	\end{equation*}
	Here $\bar{z}$ is a tuple of variables that do not occur freely in $\defn$, $\Gamma$ or $\Delta$, with the same length as $\bar{x}$.
	In the application of ($\exists$L), we assumed that none of the variables in $\bar{x}$ occur in $\bar{s}$. 
	Again, we can rename the variables $\bar{x}$ if this condition is not met, as is justified by Lemma \ref{lem:safe-replacement-bound-vars}.
	
	\textbf{Case}: $\normalization{\defn}, \Gamma \vdash \Delta$ is derived by (def L):
	\begin{equation*}
		\begin{prooftree}
			\hypo{\text{minor premises}}
			\hypo{\normalization{\defn}, \Gamma, \ihf{P}{\bar{t}} \vdash \Delta}
			\infer2[(def L)]{\normalization{\defn}, \Gamma, P(\bar{t}) \vdash \Delta}
		\end{prooftree}
	\end{equation*}
	For every defined predicate $Q \in \Pi$, this rule has a minor premise $\normalization{\defn}, \Gamma, \mbd{Q}{\defn}(\bar{y})\substposihpi \vdash \ihf{Q}{\bar{y}}, \Delta$.
	By the condition on (def L) none of the variables in $\bar{y}$ occurs freely in $\defn$, $\Gamma$ or $\Delta$.
	The premises admit \scfoid-proofs of smaller height than $\prf$.
	By induction hypothesis, the sequents 
	\begin{equation} \label{eq:proof-normalization-0}
		\defn, \Gamma, \mbd{Q}{\defn}(\bar{y})\substposihpi \vdash \ihf{Q}{\bar{y}}, \Delta
	\end{equation}
	corresponding to the predicates $Q \in \Pi$ are \scfoid-provable, and so is
	\begin{equation} \label{eq:proof-normalization-05}
		\defn, \Gamma, \ihf{P}{\bar{t}} \vdash \Delta.
	\end{equation}
	We derive the sequent $\normalization{\defn}, \Gamma, P(\bar{t}) \vdash \Delta$ from these sequents by application of (def L) with the same predicate set $\Pi$ and the same induction hypotheses $\ihs{\Pi}$ as the application above.
	The major premise is precisely (\ref{eq:proof-normalization-05}).
	The minor premises are of the form $\defn, \Gamma, \formulatwo\substposihpi \vdash \ihf{Q}{\bar{s}}, \Delta$ for definitional rules $\forall \bar{y}: Q(\bar{s}) \rul \formulatwo$ in $\defn$ with $Q \in \Pi$.
	Fix such rule $\forall \bar{y}: Q(\bar{s}) \rul \formulatwo$ and consider the following \scfoid-derivation:
	\begin{equation*}
		\begin{prooftree}
			\hypo{\formulatwo\substposihpi \vdash \mbd{Q}{\defn}(\bar{s})\substposihpi}
			\infer1[(wk)]{\defn, \Gamma, \formulatwo\substposihpi \vdash \mbd{Q}{\defn}(\bar{s})\substposihpi, \ihf{Q}{\bar{s}}, \Delta}
			\hypo{\defn, \Gamma, \mbd{Q}{\defn}(\bar{y})\substposihpi  \vdash \ihf{Q}{\bar{y}}, \Delta}
			\infer1[(subst)]{\defn, \Gamma, \mbd{Q}{\defn}(\bar{s})\substposihpi  \vdash \ihf{Q}{\bar{s}}, \Delta}
			\infer1[(wk)]{\defn, \Gamma, \formulatwo\substposihpi, \mbd{Q}{\defn}(\bar{s})\substposihpi  \vdash \ihf{Q}{\bar{s}}, \Delta}
			\infer2[(cut)]{\defn, \Gamma, \formulatwo\substposihpi \vdash \ihf{Q}{\bar{s}}, \Delta}
		\end{prooftree}
	\end{equation*}
	In the application of (subst), we used that no variable in $\bar{y}$ occurs freely in $\defn$, $\Gamma$ or $\Delta$.
	Since the top right sequent is (\ref{eq:proof-normalization-0}), which we know to be \scfoid-provable, it suffices to show that the top left sequent $\formulatwo\substposihpi \vdash \mbd{Q}{\defn}(\bar{s})\substposihpi$ is \scfoid-provable.
	We do so as follows:
	\begin{equation*}
		\begin{prooftree}
			\hypo{}
			\infer1[($=$R)]{\formulatwo\substposihpi \vdash \bar{s} = \bar{s}}
			\hypo{}
			\infer1[(ax)]{\formulatwo\substposihpi \vdash \formulatwo\substposihpi}
			\infer2[($\land$R)]{\formulatwo\substposihpi \vdash \bar{s} = \bar{s} \land \formulatwo\substposihpi}
			\infer1[($\exists$R)]{\formulatwo\substposihpi \vdash \exists \bar{y} : \bar{s} = \bar{s} \land \formulatwo\substposihpi}
			\infer1[(wk), ($\lor$L)]{\formulatwo\substposihpi \vdash \mbd{Q}{\defn}(\bar{s})\substposihpi}
		\end{prooftree}
	\end{equation*}
	In the final step, we used that $\exists \bar{y} : \bar{s} = \bar{s} \land \formulatwo\substposihpi$ is a disjunct of $\mbd{Q}{\defn}(\bar{s})\substposihpi$.
\end{proof}

\safereplacementofpositiveandnegativeoccurrences*

\begin{proof}
	We show the statement by structural induction on $\chi$.
	
	\textbf{Case}:
	$\formulathree$ is an atom.
	Then $\formulathree\ssubstnegg \doteq \formulathree$, and hence $\Gamma, \formulathree\ssubstnegg \vdash \formulathree, \Delta$ can be derived by (ax).
	If $\formulathree \doteq \formula$, then $\Gamma, \formulathree \vdash \formulathree\ssubstposg, \Delta$ comes down to $\Gamma, \formula \vdash \formulatwo, \Delta$, which is provable by assumption.
	Otherwise, $\formulathree\ssubstposg \doteq \formulathree$, and $\Gamma, \formulathree \vdash \formulathree\ssubstposg, \Delta$ can be derived by (ax).
	
	\textbf{Case}:
	$\formulathree$ is of the form $\lnot \formulathree'$.
	Then $\formulathree\ssubstposg$ is of the form $\lnot \formulathree'\ssubstnegg$ for an \fo-formula $\formulathree'\ssubstnegg$ obtained from $\formulathree'$ by replacing negative occurrences of $\formula$ by $\formulatwo$.
	By induction hypothesis, $\Gamma, \formulathree'\ssubstnegg \vdash \formulathree', \Delta$ is provable.
	We can then derive $\Gamma, \formulathree \vdash \formulathree\ssubstposg, \Delta$ as follows:
	\begin{equation*}
		\begin{prooftree}
			\hypo{\Gamma, \formulathree'\ssubstnegg \vdash \formulathree', \Delta}
			\infer1[($\lnot$R)]{\Gamma \vdash \lnot \formulathree'\ssubstnegg, \formulathree', \Delta}
			\infer1[($\lnot$L)]{\Gamma, \lnot \formulathree' \vdash \lnot \formulathree'\ssubstnegg, \Delta}
		\end{prooftree}
	\end{equation*}
	
	The formula $\formulathree\ssubstnegg$ is of the form $\lnot \formulathree'\ssubstposg$ for an \fo-formula $\formulathree'\ssubstposg$ obtained from $\formulathree'$ by replacing positive occurrences of $\formula$ by $\formulatwo$.
	By induction hypothesis, $\Gamma, \formulathree' \vdash \formulathree'\ssubstposg, \Delta$ is provable.
	We can then derive $\Gamma, \formulathree\ssubstnegg \vdash \formulathree, \Delta$ as follows:
	\begin{equation*}
		\begin{prooftree}
			\hypo{\Gamma, \formulathree' \vdash \formulathree'\ssubstposg, \Delta}
			\infer1[($\lnot$R)]{\Gamma \vdash \lnot \formulathree', \formulathree'\ssubstposg, \Delta}
			\infer1[($\lnot$L)]{\Gamma, \lnot \formulathree'\ssubstposg \vdash \lnot \formulathree', \Delta}
		\end{prooftree}
	\end{equation*}
	
	\textbf{Case}:
	$\formulathree$ is of the form $\formulathree_1 \land \formulathree_2$.
	Then $\formulathree\ssubstposg$ is of the form $\formulathree_1\ssubstposg \land \formulathree_2\ssubstposg$ for \fo-formulas $\formulathree_1\ssubstposg$ and $\formulathree_2\ssubstposg$ obtained from $\formulathree_1$ and $\formulathree_2$, respectively, by replacing positive occurrences of $\formula$ by $\formulatwo$.
	By induction hypothesis, $\Gamma, \formulathree_1 \vdash \formulathree_1\ssubstposg, \Delta$ and $\Gamma, \formulathree_2 \vdash \formulathree_2\ssubstposg, \Delta$ are provable.
	We can then derive $\Gamma, \formulathree \vdash \formulathree\ssubstposg, \Delta$ as follows:
	\begin{equation*}
		\begin{prooftree}
			\hypo{\Gamma, \formulathree_1 \vdash \formulathree_1\ssubstposg, \Delta}
			\infer1[(wk)]{\Gamma, \formulathree_1, \formulathree_2 \vdash \formulathree_1\ssubstposg, \Delta}
			\hypo{\Gamma, \formulathree_2 \vdash \formulathree_2\ssubstposg, \Delta}
			\infer1[(wk)]{\Gamma, \formulathree_1, \formulathree_2 \vdash \formulathree_2\ssubstposg, \Delta}
			\infer2[($\land$R)]{\Gamma, \formulathree_1, \formulathree_2 \vdash \formulathree_1\ssubstposg \land \formulathree_2\ssubstposg, \Delta}
			\infer1[($\land$L)]{\Gamma, \formulathree_1 \land \formulathree_2 \vdash \formulathree_1\ssubstposg \land \formulathree_2\ssubstposg, \Delta}
		\end{prooftree}
	\end{equation*}
	
	The formula $\formulathree\ssubstnegg$ is of the form $\formulathree_1\ssubstnegg \land \formulathree_2\ssubstnegg$ for \fo-formulas $\formulathree_1\ssubstnegg$ and $\formulathree_2\ssubstnegg$ obtained from $\formulathree_1$ and $\formulathree_2$, respectively, by replacing negative occurrences of $\formula$ by $\formulatwo$.
	By induction hypothesis, $\Gamma, \formulathree_1\ssubstnegg \vdash \formulathree_1, \Delta$ and $\Gamma, \formulathree_2\ssubstnegg \vdash \formulathree_2, \Delta$ are provable.
	We can then derive $\Gamma, \formulathree\ssubstnegg \vdash \formulathree, \Delta$ as follows:
	\begin{equation*}
		\begin{prooftree}
			\hypo{\Gamma, \formulathree_1\ssubstnegg \vdash \formulathree_1, \Delta}
			\infer1[(wk)]{\Gamma, \formulathree_1\ssubstnegg, \formulathree_2\ssubstnegg \vdash \formulathree_1, \Delta}
			\hypo{\Gamma, \formulathree_2\ssubstnegg \vdash \formulathree_2, \Delta}
			\infer1[(wk)]{\Gamma, \formulathree_1\ssubstnegg, \formulathree_2\ssubstnegg \vdash \formulathree_2, \Delta}
			\infer2[($\land$R)]{\Gamma, \formulathree_1\ssubstnegg, \formulathree_2\ssubstnegg \vdash \formulathree_1 \land \formulathree_2, \Delta}
			\infer1[($\land$L)]{\Gamma, \formulathree_1\ssubstnegg \land \formulathree_2\ssubstnegg \vdash \formulathree_1 \land \formulathree_2, \Delta}
		\end{prooftree}
	\end{equation*}
	
	\textbf{Cases}: 
	$\formulathree$ is of the form $\formulathree_1 \lor \formulathree_2$, $\formulathree_1 \Rightarrow \formulathree_2$ or $\formulathree_1 \Leftrightarrow \formulathree_2$.
	These cases are similar to the previous case.
	
	\textbf{Case}:
	$\formulathree$ is of the form $\forall x: \formulathree'$.
	Then $\formulathree\ssubstposg$ is of the form $\forall x: \formulathree'\ssubstposg$ for an \fo-formula $\formulathree'\ssubstposg$ obtained from $\formulathree'$ by replacing positive occurrences of $\formula$ by $\formulatwo$.
	By induction hypothesis, $\Gamma, \formulathree' \vdash \formulathree'\ssubstposg, \Delta$ is provable.
	We can then derive $\Gamma, \formulathree \vdash \formulathree\ssubstposg, \Delta$ as follows:
	\begin{equation*}
		\begin{prooftree}
			\hypo
			{\Gamma, \formulathree' \vdash \formulathree'\ssubstposg, \Delta}
			\infer1[($\forall$L)]{\Gamma, \forall x: \formulathree' \vdash \formulathree'\ssubstposg, \Delta}
			\infer1[($\forall$R)]{\Gamma, \forall x: \formulathree' \vdash \forall x: \formulathree'\ssubstposg, \Delta}
		\end{prooftree}
	\end{equation*}
	In the application of ($\forall$R), we assumed that $x$ does not occur freely in $\Gamma$ or $\Delta$, which is justified by Lemma \ref{lem:safe-replacement-bound-vars}.
	Indeed, if necessary, we can replace $x$ in $\forall x: \formulathree'\ssubstposg$ by a variable $y$ that does not occur freely in $\formula$, $\Gamma$ or $\Delta$
	 (and such variable exists, as $\Gamma$ and $\Delta$ are finite).
	
	The formula $\formulathree\ssubstnegg$ is of the form $\forall x: \formulathree'\ssubstnegg$ for an \fo-formula $\formulathree'\ssubstnegg$ obtained from $\formulathree'$ by replacing negative occurrences of $\formula$ by $\formulatwo$.
	By induction hypothesis, $\Gamma, \formulathree'\ssubstnegg \vdash \formulathree', \Delta$ is provable.
	We can then derive $\Gamma, \formulathree\ssubstnegg \vdash \formulathree, \Delta$ as follows:
	\begin{equation*}
		\begin{prooftree}
			\hypo{\Gamma, \formulathree'\ssubstnegg \vdash \formulathree', \Delta}
			\infer1[($\forall$L)]{\Gamma, \forall x: \formulathree'\ssubstnegg \vdash \formulathree', \Delta}
			\infer1[($\forall$R)]{\Gamma, \forall x: \formulathree'\ssubstnegg \vdash \forall x: \formulathree', \Delta}
		\end{prooftree}
	\end{equation*}
	Again, we assumed that $x$ does not occur freely in $\formula$, $\Gamma$ or $\Delta$, which is justified by Lemma \ref{lem:safe-replacement-bound-vars}.
	
	\textbf{Case}: 
	$\formulathree$ is of the form $\exists x: \formulathree'$.
	This case is similar to the previous case. \qedhere
\end{proof}

\begin{example} \label{ex:app-proof-sat}
	In this example, we complete the proof of the sequent $\defn, \coax, \mathit{In}(q,j) \vdash \lnot \mathit{Sat}(j,\mathit{not}(q))$ from Example \ref{ex:proof_sat}.
	Recall that $\defn$ is the definition of $\mathit{Sat}$ and $\coax$ the set of axioms on the constructors $\mathit{and}$ and $\mathit{not}$ from Example \ref{ex:proof_sat}. 
	The bottom of the proof takes the form
	\[
	\begin{prooftree}
		\hypo{\text{minor premises}}
		\hypo{\defn, \coax, \mathit{In}(q,j), \ihf{\mathit{Sat}}{j,\mathit{not}(q)} \vdash}
		\infer2[(def L)]{\defn, \coax, \mathit{In}(q,j), \mathit{Sat}(j,\mathit{not}(q)) \vdash}
		\infer1[($\lnot$R)]{\defn, \coax, \mathit{In}(q,j) \vdash \lnot \mathit{Sat}(j,\mathit{not}(q))}
	\end{prooftree}
	\]	
	with minor premises
	\[
	\defn, \coax, \mathit{In}(q,j), \mathit{In}(p,i) \vdash \ihf{\mathit{Sat}}{i,p} ,
	\]
	\[
	\defn, \coax, \mathit{In}(q,j), \ihf{\mathit{Sat}}{i,f} \land \ihf{\mathit{Sat}}{i,g} \vdash \ihf{\mathit{Sat}}{i,\mathit{and}(f,g)} ,
	\]
	\[
	\defn, \coax, \mathit{In}(q,j), \lnot \mathit{Sat}(i,f) \vdash \ihf{\mathit{Sat}}{i,\mathit{not}(f)} .
	\]
	We prove these minor premises and the major premise %
	with the induction hypothesis $\ihs{\mathit{Sat}} \doteq \{ (i,f) \mid i \neq j \lor f \neq \mathit{not}(q) \}$.
	For space reasons, we abbreviate %
	$\mathit{Sat}$ to $S$.
	\begin{equation*}
		\begin{prooftree}
			\infer0[(ax)]{\mathit{In}(\mathit{not}(q),j) \vdash \mathit{In}(\mathit{not}(q),j)}
			\infer0[($=$R)]{\vdash \mathit{not}(q) = \mathit{not}(q)}
			\infer1[($\exists$R)]{\vdash \exists f: \mathit{not}(q) = \mathit{not}(f)}
			\infer1[($\lnot$L)]{\lnot \exists f: \mathit{not}(q) = \mathit{not}(f) \vdash}
			\infer1[(wk)]{\lnot \exists f: \mathit{not}(q) = \mathit{not}(f), \mathit{In}(\mathit{not}(q),j) \vdash}
			\infer2[($\Rightarrow$L)]{\mathit{In}(\mathit{not}(q),j) \Rightarrow \lnot \exists f: \mathit{not}(q) = \mathit{not}(f), \mathit{In}(\mathit{not}(q),j) \vdash}
			\infer1[($\forall$L)]{\forall p, i: \mathit{In}(p,i) \Rightarrow \lnot \exists f: p = \mathit{not}(f),  \mathit{In}(\mathit{not}(q),j) \vdash}
			\infer1[(wk)]{\defn, \coax,  \mathit{In}(\mathit{not}(q),j) \vdash}
			\infer1[($=$L)]{\defn, \coax, \mathit{In}(p,j), p=\mathit{not}(q) \vdash}
			\infer1[($=$L)]{\defn, \coax, \mathit{In}(p,i), i=j, p=\mathit{not}(q) \vdash}
			\infer1[($\lnot$R)]{\defn, \coax, \mathit{In}(p,i), i=j \vdash p \neq \mathit{not}(q)}
			\infer1[($\lnot$R)]{\defn, \coax, \mathit{In}(p,i) \vdash i \neq j, p \neq \mathit{not}(q)}
			\infer1[($\lor$R)]{\coax, \mathit{In}(p,i) \vdash i \neq j \lor p \neq \mathit{not}(q)}
			\infer1[(wk)]{\defn, \coax, \mathit{In}(q,j), \mathit{In}(p,i) \vdash \ihf{\mathit{S}}{i,p}}
		\end{prooftree}
	\end{equation*}
	
	\begin{equation*}
	\begin{prooftree}
		\infer0[(ax)]{\mathit{and}(f, g) \neq \mathit{not}(q) \vdash \mathit{and}(f,g) \neq \mathit{not}(q)}
		\infer1[($\forall$L)]{\forall f, g, h: \mathit{and}(f, g) \neq \mathit{not}(h) \vdash \mathit{and}(f,g) \neq \mathit{not}(q)}
		\infer1[(wk)]{\defn, \coax, \mathit{In}(q,j), \ihf{\mathit{S}}{i,f} \land \ihf{\mathit{S}}{i,g} \vdash i \neq j, \mathit{and}(f,g) \neq \mathit{not}(q)}
		\infer1[($\lor$R)]{\defn, \coax, \mathit{In}(q,j), \ihf{\mathit{S}}{i,f} \land \ihf{\mathit{S}}{i,g} \vdash \ihf{\mathit{S}}{i,\mathit{and}(f,g)}}
	\end{prooftree}
	\end{equation*}

	\begin{equation*}
		\begin{prooftree}[separation=3mm]
			\infer0[(ax)]{\defn, \mathit{In}(q,j), \mathit{not}(f) = \mathit{not}(q) \vdash \mathit{not}(f) = \mathit{not}(q),  \mathit{S}(j,f)}
			\infer0[(ax)]{\defn, \mathit{In}(q,j) \vdash \mathit{In}(q,j)}
			\infer1[(def R)]{\defn, \mathit{In}(q,j) \vdash \mathit{S}(j,q)}
			\infer1[($=$L)]{\defn, f=q, \mathit{In}(q,j) \vdash \mathit{S}(j,f)}
			\infer1[(wk)]{\defn, f=q, \mathit{In}(q,j), \mathit{not}(f) = \mathit{not}(q) \vdash \mathit{S}(j,f)}
			\infer2[($\Rightarrow$L)]{\defn, \mathit{not}(f) = \mathit{not}(q) \Rightarrow f=q, \mathit{In}(q,j), \mathit{not}(f) = \mathit{not}(q) \vdash \mathit{S}(j,f)}
			\infer1[($\forall$L)]{\defn, \forall f, g: \mathit{not}(f) = \mathit{not}(g) \Rightarrow f=g, \mathit{In}(q,j), \mathit{not}(f) = \mathit{not}(q) \vdash \mathit{S}(j,f)}
			\infer1[(wk)]{\defn, \coax, \mathit{In}(q,j), \mathit{not}(f) = \mathit{not}(q) \vdash \mathit{S}(j,f)}
			\infer1[($=$L)]{\defn, \coax, \mathit{In}(q,j), i=j, \mathit{not}(f) = \mathit{not}(q) \vdash \mathit{S}(i,f)}
			\infer1[($\lnot$L)]{\defn, \coax, \mathit{In}(q,j), \lnot \mathit{S}(i,f), i=j, \mathit{not}(f) = \mathit{not}(q) \vdash}
			\infer1[($\lnot$R)]{\defn, \coax, \mathit{In}(q,j), \lnot \mathit{S}(i,f), i=j \vdash \mathit{not}(f) \neq \mathit{not}(q)}
			\infer1[($\lnot$R)]{\defn, \coax, \mathit{In}(q,j), \lnot \mathit{S}(i,f) \vdash i \neq j, \mathit{not}(f) \neq \mathit{not}(q)}
			\infer1[($\lor$R)]{\defn, \coax, \mathit{In}(q,j), \lnot \mathit{S}(i,f) \vdash \ihf{\mathit{S}}{i,\mathit{not}(f)}}
		\end{prooftree}
	\end{equation*}
	
	\begin{equation*}
		\begin{prooftree}
			\infer0[($=$R)]{\vdash j=j}
			\infer1[($\lnot$L)]{j \neq j \vdash}
			\infer0[($=$R)]{\vdash \mathit{not}(q) = \mathit{not}(q)}
			\infer1[($\lnot$L)]{\mathit{not}(q) \neq \mathit{not}(q) \vdash}
			\infer2[($\lor$L)]{j \neq j \lor \mathit{not}(q) \neq \mathit{not}(q) \vdash}
			\infer1[(wk)]{\defn, \coax, \mathit{In}(q,j), \ihf{\mathit{S}}{j,\mathit{not}(q)} \vdash}
		\end{prooftree}
	\end{equation*}
	
\end{example}

\predicatesetrestrsubsetsmd*

\begin{proof}
	The `if' part is trivial, as every \scfoidsmd-proof is also an \scfoid-proof.
	For the `only if' part, we take an \scfoid-provable sequent $\seqshort$ and show by induction on the height of an \scfoid-proof $\prf$ of $\seqshort$ that $\seqshort$ is \scfoidsmd-provable as well.
	By Propositions \ref{prop:admissibility-def-L''} and \ref{prop:admissibility-def-L-from-L''}, we can replace the rule (def L) with \deflalt in \scfoid and \scfoidsmd.\footnote{
		The proofs of Propositions \ref{prop:admissibility-def-L''} and \ref{prop:admissibility-def-L-from-L''} extend directly to \scfoidsmd.
	}
	The base cases are trivial, and so are most cases in the inductive step.
	The only non-trivial case is where $\seqshort$ is derived by \deflalt:
	\begin{equation} \label{eq:application-(def-L)-large_Pi}
		\begin{prooftree}
			\hypo{\text{minor premises}}
			\infer1[\deflalt]{\defn, \Gamma, P(\bar{t}) \vdash \ihf{P}{\bar{t}}, \Delta}
		\end{prooftree}
	\end{equation}
	Let $\Pi$ be the corresponding predicate set and $\ihs{\Pi}$ the corresponding induction hypotheses.
	For every definitional rule $\forall \bar{y}: Q(\bar{s}) \rul \psi$ of $\defn$ with $Q \in \Pi$, this rule has a minor premise
	\begin{equation*} %
		\defn, \Gamma, \psi\substposihpi \vdash \ihf{Q}{\bar{s}}, \Delta .
	\end{equation*}
	Since these sequents admit \scfoid-proofs of lower height than $\prf$, they are \scfoidsmd-provable by the induction hypothesis.
	We show that $\defn, \Gamma, P(\bar{t}) \vdash \ihf{P}{\bar{t}}, \Delta$ is \scfoidsmd-provable by induction over the \emph{dependency order} $\prec$ on $\defi{\defn}$, which is defined by $Q \prec R$ iff $(R,Q) \in \dep{\defn}$ and $(Q,R) \not\in \dep{\defn}$ for all $Q, R \in \defi{\defn}$.\footnote{
		The dependency order $\prec$ is a \emph{strict partial order}, as it is irreflexive, asymmetric and transitive.
	}
	This means that we prove the statement for $P$ under the assumption that it holds for all $Q \in \defi{\defn}$ such that $Q \prec P$.
	
	We show that $\defn, \Gamma, P(\bar{t}) \vdash \ihf{P}{\bar{t}}, \Delta$ is \scfoidsmd-provable via an instance of (def L) with predicate set $\Pi_{P} \doteq \Pi \cap \md{\defn}{P} \subseteq \md{\defn}{P}$ and induction hypotheses $\ihs{\Pi_P}$ (i.e., $\ihs{\Pi}$ restricted to $\Pi_P$).
	For every definitional rule $\forall \bar{y}: Q(\bar{s}) \rul \psi$ of $\defn$ with $Q \in \Pi_P$, this rule has a minor premise
	\begin{equation*} %
		\defn, \Gamma, \psi\substposih{\Pi_P} \vdash \ihf{Q}{\bar{s}}, \Delta .
	\end{equation*}
	Fix such a definitional rule $\forall \bar{y}: Q(\bar{s}) \rul \psi$ and consider the following \scfoidsmd-derivation:
	\begin{equation*}
		\begin{prooftree}
			\hypo{\defn, \Gamma, \psi\substposih{\Pi_P} \vdash \psi\substposihpi, \Delta}
			\infer1[(wk)]{\defn, \Gamma, \psi\substposih{\Pi_P} \vdash \psi\substposihpi, \ihf{Q}{\bar{s}}, \Delta}
			\hypo{\defn, \Gamma, \psi\substposihpi \vdash \ihf{Q}{\bar{s}}, \Delta}
			\infer1[(wk)]{\defn, \Gamma, \psi\substposih{\Pi_P}, \psi\substposihpi \vdash \ihf{Q}{\bar{s}}, \Delta}
			\infer2[(cut)]{\defn, \Gamma, \psi\substposih{\Pi_P} \vdash \ihf{Q}{\bar{s}}, \Delta}
		\end{prooftree}
	\end{equation*}
	Since $\defn, \Gamma, \psi\substposihpi \vdash \ihf{Q}{\bar{s}}, \Delta$ is a minor premise %
	from (\ref{eq:application-(def-L)-large_Pi}), we know that it is \scfoidsmd-provable.
	It thus suffices to show that 
	\begin{equation} \label{eq:sequent-phi_Q[Pi_P]-implies-phi_Q[Pi]}
		\defn, \Gamma, \psi\substposih{\Pi_P} \vdash \psi\substposihpi, \Delta
	\end{equation}
	is \scfoidsmd-provable.
	To do so, we start by observing that $\psi\substposihpi \doteq \psi\substposih{\Pi_P \cup \Pi_{\prec P}}$, where $\Pi_{\prec P}$ denotes the set of predicate symbols $Q$ in $\Pi$ such that $Q \prec P$.
	Indeed, by definition of $\dep{\defn}$, $\formulatwo$ can only contain predicate symbols $R$ such that $(Q, R) \in \dep{\defn}$.
	Since $P$ and $Q$ are mutually dependent, these are precisely the predicate symbols $R$ such that $(P, R) \in \dep{\defn}$, and by definition of $\Pi_P$ and $\Pi_{\prec P}$, these are precisely the elements of $\Pi_P \cup \Pi_{\prec P}$.
	Now note that $\psi\substposih{\Pi_P \cup \Pi_{\prec P}}$ can be obtained from $\psi\substposih{\Pi_P}$ by replacing positive occurrences of atoms $R(\bar{u})$ such that $R \in \Pi_{\prec P}$ with $\ihf{R}{\bar{u}}$ (more precisely, those occurrences that were not introduced as a result of the substitution $\substposih{\Pi_P}$). %
	By the induction hypothesis (corresponding to the induction over $\prec$), the sequent $\defn, \Gamma, R(\bar{u}) \vdash \ihf{R}{\bar{u}}, \Delta$ is \scfoidsmd-provable for every atom $R(\bar{u})$ with $R \in \Pi_{\prec P}$.
	Repeatedly applying Lemma \ref{lem:safe-replacement-of-positive-and-negative-occurrences}, we find that (\ref{eq:sequent-phi_Q[Pi_P]-implies-phi_Q[Pi]}) is \scfoidsmd-provable. 
	This finishes the proof.
\end{proof}

\predicatesetlarger*

\begin{proof}
	Consider an instance (\ref{eq:def-L-lem-predlarger}) of (def L) with predicate set $\Pi$ and induction hypotheses $\ihs{\Pi}$.
	For every definitional rule $\forall \bar{y}: Q(\bar{s}) \rul \formulatwo$ in $\defn$ with $Q \in \Pi$, this instance has a minor premise 
	\begin{equation*}
		\defn, \Gamma, \formulatwo\substposih{\Pi} \vdash \ihf{Q}{\bar{s}}, \Delta .
	\end{equation*}
	Let $\Pi' \subseteq \defi{\defn}$ be a superset of\/ $\Pi$.
	We will construct an \scfoid-derivation of the conclusion $\defn, \Gamma, P(\bar{t}) \vdash \Delta$ of (\ref{eq:def-L-lem-predlarger}) from the premises of (\ref{eq:def-L-lem-predlarger}) containing a single instance of (def L), with predicate set $\Pi'$.
	
	For each $Q \in \Pi$, we define $\ihstwo{Q} \coloneq \setexp{\bar{z}}{Q(\bar{z}) \land \ihf{Q}{\bar{z}}}$, and for each $Q \in \Pi' \setminus \Pi$, we define $\ihstwo{Q} \coloneq \setexp{\bar{z}}{Q(\bar{z})}$.
	Consider the following instance of (def L) with predicate set $\Pi'$ and induction hypotheses $\ihstwo{\Pi}$:
	\begin{equation} \label{eq:def-L-Pi'}
		\begin{prooftree}
			\hypo{\text{minor premises}}
			\hypo{\defn, \Gamma, \ihftwo{P}{\bar{t}} \vdash \Delta}
			\infer2[(def L)]{\defn, \Gamma, P(\bar{t}) \vdash \Delta}
		\end{prooftree}
	\end{equation}
	For every definitional rule $\forall \bar{y}: Q(\bar{s}) \rul \formulatwo$ in $\defn$ with $Q \in \Pi'$, this instance has a minor premise 
	\begin{equation*} %
		\defn, \Gamma, \formulatwo\substpos{\ihstwo{\Pi'}}{\Pi'} \vdash \ihftwo{Q}{\bar{s}}, \Delta .
	\end{equation*}
	Since $\ihftwo{R}{\bar{u}} \doteq R(\bar{u})$ for any defined atom $R(\bar{u})$ with $R \in \Pi' \setminus \Pi$, we find that $\formulatwo\substpos{\ihstwo{\Pi'}}{\Pi'} \doteq \formulatwo\substpos{\ihstwo{\Pi}}{\Pi}$.
	
	We will obtain the desired \scfoid-derivation of $\defn, \Gamma, P(\bar{t}) \vdash \Delta$ by expanding the buds of (\ref{eq:def-L-Pi'}).
	First, we replace the bud labeled with the major premise $\defn, \Gamma, \ihftwo{P}{\bar{t}} \vdash \Delta$ with the following \scfoid-derivation of $\defn, \Gamma, \ihftwo{P}{\bar{t}} \vdash \Delta$ from the major premise $\defn, \Gamma, \ihf{P}{\bar{t}} \vdash \Delta$ of (\ref{eq:def-L-lem-predlarger}):
	\[
	\begin{prooftree}
		\hypo{\defn, \Gamma, \ihf{P}{\bar{t}} \vdash \Delta}
		\infer1[(wk)]{\defn, \Gamma, P(\bar{t}), \ihf{P}{\bar{t}} \vdash \Delta}
		\infer1[($\land$L)]{\defn, \Gamma, P(\bar{t}) \land \ihf{P}{\bar{t}} \vdash \Delta}
	\end{prooftree}
	\]
	
	Next, we consider the buds in (\ref{eq:def-L-Pi'}) labeled with minor premises $\defn, \Gamma, \formulatwo\substpos{\ihstwo{\Pi}}{\Pi} \vdash \ihftwo{Q}{\bar{s}}, \Delta$ corresponding to definitional rules $\forall \bar{y}: Q(\bar{s}) \rul \formulatwo$ in $\defn$ with $Q \in \Pi' \setminus \Pi$.
	For every such definitional rule, we replace the bud labeled with the corresponding minor premise by the following \scfoid-derivation:
	\begin{equation} \label{eq:proof-psi-IH'-Qs}
		\begin{prooftree}
			\hypo{\formulatwo\substpos{\ihstwo{\Pi}}{\Pi} \vdash \formulatwo}
			\infer1[(wk)]{\defn, \Gamma, \formulatwo\substpos{\ihstwo{\Pi}}{\Pi} \vdash \formulatwo, Q(\bar{s}), \Delta}
			\infer0[(ax)]{\defn, \formulatwo \vdash \formulatwo}
			\infer1[(def R)]{\defn, \formulatwo \vdash Q(\bar{s})}
			\infer1[(wk)]{\defn, \Gamma, \formulatwo\substpos{\ihstwo{\Pi}}{\Pi}, \formulatwo \vdash Q(\bar{s}), \Delta}
			\infer2[(cut)]{\defn, \Gamma, \formulatwo\substpos{\ihstwo{\Pi}}{\Pi} \vdash Q(\bar{s}), \Delta}
		\end{prooftree}
	\end{equation}
	Furthermore, we replace the bud labeled with $\formulatwo\substpos{\ihstwo{\Pi}}{\Pi} \vdash \formulatwo$ in this derivation by an \scfoid-proof of $\formulatwo\substpos{\ihstwo{\Pi}}{\Pi} \vdash \formulatwo$.
	Such a proof exists by repeated application of Lemma \ref{lem:safe-replacement-of-positive-and-negative-occurrences}, since $\formulatwo$ can be obtained from $\formulatwo\substpos{\ihstwo{\Pi}}{\Pi}$ by replacing positive occurrences of atoms $R(\bar{u})$ such that $R \in \Pi$ with formulas $\ihftwo{R}{\bar{u}} \doteq R(\bar{u}) \land \ihf{R}{\bar{u}}$, and since for any such atom, the sequent $R(\bar{u}) \land \ihf{R}{\bar{u}} \vdash R(\bar{u})$ is clearly provable.
	Note that this proof has no instance of (def L), as the sequent $\formulatwo\substpos{\ihstwo{\Pi}}{\Pi} \vdash \formulatwo$ contains no definitions.\footnote{
		This follows from the proof of Lemma \ref{lem:safe-replacement-of-positive-and-negative-occurrences}.
		In general, an \scfoid-proof of an \fo-sequent may have instances of (def L), as it can introduce definitions by means of the cut rule.
		However, the \scfoid-proof obtained from the proof of Lemma \ref{lem:safe-replacement-bound-vars} has no such cuts.
	}
	In sum, we have replaced all the buds labeled with minor premises $\defn, \Gamma, \formulatwo\substpos{\ihstwo{\Pi}}{\Pi} \vdash \ihftwo{Q}{\bar{s}}, \Delta$ in (\ref{eq:def-L-Pi'}) corresponding to definitional rules $\forall \bar{y}: Q(\bar{s}) \rul \formulatwo$ in $\defn$ such that $Q \in \Pi' \setminus \Pi$ with \scfoid-proofs of $\defn, \Gamma, \formulatwo\substpos{\ihstwo{\Pi}}{\Pi} \vdash \ihftwo{Q}{\bar{s}}, \Delta$.
	
	Finally, we consider the buds in (\ref{eq:def-L-Pi'}) labeled with minor premises $\defn, \Gamma, \formulatwo\substpos{\ihstwo{\Pi}}{\Pi} \vdash \ihftwo{Q}{\bar{s}}, \Delta$ corresponding to definitional rules $\forall \bar{y}: Q(\bar{s}) \rul \formulatwo$ in $\defn$ with $Q \in \Pi$.
	For every such definitional rule, we replace the bud labeled with the corresponding minor premise by the following \scfoid-derivation:
	\[
	\begin{prooftree}
		\hypo{\defn, \Gamma, \formulatwo\substpos{\ihstwo{\Pi}}{\Pi} \vdash Q(\bar{s}), \Delta}
		\hypo{\defn, \Gamma, \formulatwo\substpos{\ihstwo{\Pi}}{\Pi} \vdash \ihf{Q}{\bar{s}}, \Delta}
		\infer2[($\land$R)]{\defn, \Gamma, \formulatwo\substpos{\ihstwo{\Pi}}{\Pi} \vdash Q(\bar{s}) \land \ihf{Q}{\bar{s}}, \Delta}
	\end{prooftree}
	\]
	We replace the bud labeled with $\defn, \Gamma, \formulatwo\substpos{\ihstwo{\Pi}}{\Pi} \vdash Q(\bar{s}), \Delta$ by an \scfoid-proof of $\defn, \Gamma, \formulatwo\substpos{\ihstwo{\Pi}}{\Pi} \vdash Q(\bar{s}), \Delta$, constructed as in (\ref{eq:proof-psi-IH'-Qs}).
	Furthermore, we replace the bud labeled with $\defn, \Gamma, \formulatwo\substpos{\ihstwo{\Pi}}{\Pi} \vdash \ihf{Q}{\bar{s}}, \Delta$ by the following \scfoid-derivation:
	\[
	\begin{prooftree}
		\hypo{\formulatwo\substpos{\ihstwo{\Pi}}{\Pi} \vdash \formulatwo\substpos{\ihs{\Pi}}{\Pi}}
		\infer1[(wk)]{\defn, \Gamma, \formulatwo\substpos{\ihstwo{\Pi}}{\Pi} \vdash \formulatwo\substpos{\ihs{\Pi}}{\Pi}, \ihf{Q}{\bar{s}}, \Delta}
		\hypo{\defn, \Gamma, \formulatwo\substpos{\ihs{\Pi}}{\Pi} \vdash \ihf{Q}{\bar{s}}, \Delta}
		\infer1[(wk)]{\defn, \Gamma, \formulatwo\substpos{\ihstwo{\Pi}}{\Pi}, \formulatwo\substpos{\ihs{\Pi}}{\Pi} \vdash \ihf{Q}{\bar{s}}, \Delta}
		\infer2[(cut)]{\defn, \Gamma, \formulatwo\substpos{\ihstwo{\Pi}}{\Pi} \vdash \ihf{Q}{\bar{s}}, \Delta}
	\end{prooftree}
	\]
	Note that $\defn, \Gamma, \formulatwo\substpos{\ihs{\Pi}}{\Pi} \vdash \ihf{Q}{\bar{s}}, \Delta$ is the minor premise in (\ref{eq:def-L-lem-predlarger}) corresponding to the definitional rule $\forall \bar{y}: Q(\bar{s}) \rul \formulatwo$ in $\defn$.
	We replace the bud labeled with $\formulatwo\substpos{\ihstwo{\Pi}}{\Pi} \vdash \formulatwo\substpos{\ihs{\Pi}}{\Pi}$ by an \scfoid-proof of $\formulatwo\substpos{\ihstwo{\Pi}}{\Pi} \vdash \formulatwo\substpos{\ihs{\Pi}}{\Pi}$.
	Similarly to before, such a proof exists by repeated application of Lemma \ref{lem:safe-replacement-of-positive-and-negative-occurrences}, since $\formulatwo\substpos{\ihs{\Pi}}{\Pi}$ can be obtained from $\formulatwo\substpos{\ihstwo{\Pi}}{\Pi}$ by replacing positive occurrences of formulas of the form $\ihftwo{R}{\bar{u}} \doteq R(\bar{u}) \land \ihf{R}{\bar{u}}$ such that $R \in \Pi$ with formulas $\ihf{R}{\bar{u}}$, and since for any such formula, the sequent $R(\bar{u}) \land \ihf{R}{\bar{u}} \vdash \ihf{R}{\bar{u}}$ is clearly provable.
	Again, this proof has no instance of (def L), as the sequent $\formulatwo\substpos{\ihstwo{\Pi}}{\Pi} \vdash \formulatwo\substpos{\ihs{\Pi}}{\Pi}$ contains no definitions.
	In sum, we have replaced all the buds labeled with minor premises $\defn, \Gamma, \formulatwo\substpos{\ihstwo{\Pi}}{\Pi} \vdash \ihftwo{Q}{\bar{s}}, \Delta$ in (\ref{eq:def-L-Pi'}) corresponding to definitional rules $\forall \bar{y}: Q(\bar{s}) \rul \formulatwo$ in $\defn$ such that $Q \in \Pi$ with \scfoid-derivations of $\defn, \Gamma, \formulatwo\substpos{\ihstwo{\Pi}}{\Pi} \vdash \ihftwo{Q}{\bar{s}}, \Delta$ from the corresponding minor premise $\defn, \Gamma, \formulatwo\substpos{\ihs{\Pi}}{\Pi} \vdash \ihf{Q}{\bar{s}}, \Delta$ in (\ref{eq:def-L-lem-predlarger}).
	Altogether, we have constructed an \scfoid-derivation of the conclusion $\defn, \Gamma, P(\bar{t}) \vdash \Delta$ of (\ref{eq:def-L-lem-predlarger}) from the premises of (\ref{eq:def-L-lem-predlarger}) with a single instance of (def L), with predicate set $\Pi'$.
\end{proof}

%% file: proofs-section-alternative-semantics.tex
\section{Proofs of Section \ref{sec:alternative-semantics}} \label{app:proofs-section-alternative-semantics}

\henkinclassinterpretsallfoformulas*

\begin{proof}
	We prove by structural induction on $\formula$ that $\{ \bar{a} \in \dom{\structure}^k \mid \structure[\bar{x}: \bar{a}] \models \formula \} \in H_k$ for all structures $\structure$, Henkin classes $\{ H_k \mid k \in \N \}$ for $\structure$, tuples of distinct object symbols $\bar{x} = (x_1, \dots, x_k)$ and \fo-formulas $\formula$ over $\voc{\structure} \cup \{x_1, \dots, x_k\}$.
	We write $D$ for $\dom{\structure}$.
	
	\textbf{Case}: $\formula$ is of the form $t=s$.
	By condition \ref{it:henkin-class-equality} of Definition \ref{def:henkin-class}, $R \coloneq \{(a,a) \mid a \in D\} \in H_2$.
	By condition \ref{it:henkin-class-terms}, $\{ \bar{a} \in D^k \mid \structure[\bar{x}:\bar{a}] \models t=s \} = \{ \bar{a} \in D^k \mid (t^{\structure[\bar{x}:\bar{a}]},s^{\structure[\bar{x}:\bar{a}]}) \in R \} \in H_k$.
	
	\textbf{Case}: $\formula$ is of the form $Q(\bar{t})$.
	Let $k$ denote the arity of $Q$.
	By condition \ref{it:henkin-class-predicate}, $Q^{\structure} \in H_k$.
	By condition \ref{it:henkin-class-terms}, $\{ \bar{a} \in D^k \mid \structure[\bar{x}:\bar{a}] \models Q(\bar{t}) \} = \{ \bar{a} \in D^k \mid \bar{t}^{\structure[\bar{x}:\bar{a}]} \in Q^{\structure} \} \in H_k$.
	
	\textbf{Case}: $\formula$ is of the form $\neg \formulatwo$.
	By induction hypothesis, $\{ \bar{a} \in D^k \mid \structure[\bar{x}: \bar{a}] \models \formulatwo \} \in H_k$.
	By condition \ref{it:henkin-class-complement}, $\{ \bar{a} \in D^k \mid \structure[\bar{x}: \bar{a}] \models \neg \formulatwo \} = D^k \setminus \{ \bar{a} \in D^k \mid \structure[\bar{x}: \bar{a}] \models \formulatwo \} \in H_k$.
	
	\textbf{Case}: $\formula$ is of the form $\formulatwo \land \formulathree$.
	By induction hypothesis, $\{ \bar{a} \in D^k \mid \structure[\bar{x}: \bar{a}] \models \formulatwo \} \in H_k$ and $\{ \bar{a} \in D^k \mid \structure[\bar{x}: \bar{a}] \models \formulathree \} \in H_k$.
	By condition \ref{it:henkin-class-intersection}, $\{ \bar{a} \in D^k \mid \structure[\bar{x}: \bar{a}] \models \formulatwo \land \formulathree \} = \{ \bar{a} \in D^k \mid \structure[\bar{x}: \bar{a}] \models \formulatwo \} \cap \{ \bar{a} \in D^k \mid \structure[\bar{x}: \bar{a}] \models \formulatwo \} \in H_k$.
	
	\textbf{Cases}: $\formula$ is of the form $\formulatwo \lor \formulathree$, $\formulatwo \Rightarrow \formulathree$ or $\formulatwo \Leftrightarrow \formulathree$.
	Using the fact that $\formulatwo \lor \formulathree$ is \fo-equivalent to $\neg (\neg \formulatwo \land \neg \formulathree)$, $\formulatwo \Rightarrow \formulathree$ to $\neg \formulatwo \lor \formulathree$, and $\formulatwo \Leftrightarrow \formulathree$ to $(\formulatwo \Rightarrow \formulathree) \land (\formulathree \Rightarrow \formulatwo)$, these cases follow from the previous cases.
	
	\textbf{Case}: $\formula$ is of the form $\forall x: \formulatwo$.
	Let $y$ be an object symbol that does not occur in $\bar{x}$ nor in $\formulatwo$.
	By induction hypothesis, $R\coloneq \{ (a_1, \dots, a_k, a) \in D^{k+1} \mid \structure[x_1: a_1, \dots, x_k: a_k, \, y: a] \models \formulatwo[y/x] \} \in H_{k+1}$.
	By condition \ref{it:henkin-class-forall}, $\{ \bar{a} \in D^{k} \mid \structure[\bar{x}: \bar{a}] \models \forall x: \formulatwo \} = \{ \bar{a} \in D^{k} \mid \structure[\bar{x}: \bar{a}] \models \forall y: \formulatwo[y/x] \} = \{ (a_1, \dots, a_k) \in D^{k} \mid (a_1, \dots, a_k, a) \in R \text{ for all } a \in D \} \in H_k$.
	
	\textbf{Case}: $\formula$ is of the form $\exists x: \formulatwo$.
	Using the fact that $\exists x: \formulatwo$ is \fo-equivalent to $\neg \forall x: \neg \formulatwo$, this follows from the previous case.
\end{proof}

\fodefinablesetsformhenkinclass*

\begin{proof}
	Let $D \coloneq \dom{\structure}$ and $\vocab \coloneq \voc{\structure}$.
	We need to check that $\{ H_k \mid k \in \N \}$ fulfills all the conditions of Definition \ref{def:henkin-class}:
	\begin{enumerate}[label=(\roman*)]
		\item $\{(a,a) \mid a \in D\} = \{(a_1,a_2) \in D^2 \mid \structure[x_1:a_1, x_2:a_2] \models x_1 = x_2 \} \in H_2$.
		\item Let $Q$ be a predicate symbol in $\vocab$ with arity $k$.
		Then $Q^{\structure} = \{ \bar{a} \in D^k \mid \structure[\bar{x}:\bar{a}] \models Q(\bar{x}) \} \in H_{k}$.
		\item Let $R \in H_{k+1}$ and $a \in D$.
		By definition of $H_{k+1}$, the set $R$ is of the form $\{(a_1, \dots, a_{k+1}) \in D^{k+1} \mid \structure[x_1:a_1, \dots, x_{k+1}:a_{k+1}] \models \formula \}$ for distinct object symbols $x_1, \dots, x_{k+1}$ and an \fo-formula $\formula$ over $\vocab \cup \{x_1, \dots, x_{k+1}\}$.
		Since $\structure$ is domain-closed, there exists a term $t$
		such that $t^\structure = a$.
		After renaming object symbols if necessary, %
		we can assume that none of the object symbols $x_1, \dots, x_{k}$ occurs in $t$.
		Since $\formula[t/x_{k+1}]$ is an \fo-formula, $\{(a_1, \dots, a_k) \in D^k \mid (a_1, \dots, a_k, a) \in R\} = \{(a_1, \dots, a_{k}) \in D^{k} \mid \structure[x_1:a_1, \dots, x_{k}:a_{k}, x_{k+1} : a] \models \formula\} = \{(a_1, \dots, a_{k}) \in D^{k} \mid \structure[x_1:a_1, \dots, x_{k}:a_{k}] \models \formula[t/x_{k+1}] \} \in H_k$.
		\item Let $R \in H_k$, $x_1, \dots, x_n$ distinct object symbols not in $\vocab$ and $\bar{t} = (t_1, \dots, t_k)$ a tuple of terms over $\vocab \cup \{x_1, \dots, x_n\}$.
		Then $R$ is of the form $\{\bar{b} \in D^{k} \mid \structure[\bar{y}:\bar{b}] \models \formula \}$ for a tuple of distinct object symbols $\bar{y} = (y_1, \dots, y_k)$ not in $\vocab$ and an \fo-formula $\formula$ over $\vocab \cup \{y_1, \dots, y_k\}$.
		Since $\formula[\bar{t}/\bar{y}]$ is an \fo-formula over $\vocab \cup \{x_1, \dots, x_n\}$, we have that $\{ \bar{a} \in D^n \mid \bar{t}^{\structure[\bar{x}:\bar{a}]} \in R\} = \{ \bar{a} \in D^n \mid \structure[\bar{y} : \bar{t}^{\structure[\bar{x}:\bar{a}]}] \models \formula \} = \{ \bar{a} \in D^n \mid \structure[\bar{x} : \bar{a}] \models \formula[\bar{t}/\bar{y}] \} \in H_n$.
		\item Let $R \in H_k$.
		Then $R$ is of the form $\{\bar{a} \in D^{k} \mid \structure[\bar{x}:\bar{a}] \models \formula \}$ for a tuple of distinct object symbols $\bar{x} = (x_1, \dots, x_k)$ not in $\vocab$ and $\formula$ an \fo-formula over $\vocab \cup \{x_1, \dots, x_k\}$.
		Then $D^k \setminus R = \{\bar{a} \in D^{k} \mid \structure[\bar{x}:\bar{a}] \models \neg \formula \} \in H_k$.
		\item Let $R \in H_k$ and $S \in H_k$.
		Then $R$ is of the form $\{\bar{a} \in D^{k} \mid \structure[\bar{x}: \bar{a}] \models \formula \}$ for a tuple of distinct object symbols $\bar{x} = (x_1, \dots, x_k)$ not in $\vocab$ and $\formula$ an \fo-formula over $\vocab \cup \{x_1, \dots, x_k\}$, and $S$ is of the form $\{\bar{a} \in D^{k} \mid \structure[\bar{y}:\bar{a}] \models \formulatwo \}$ for a tuple of distinct object symbols $\bar{y} = (y_1, \dots, y_k)$ not in $\vocab$ and $\formulatwo$ an \fo-formula over $\vocab \cup \{y_1, \dots, y_k\}$.
		After renaming object symbols, %
		we can assume that $\bar{x} \doteq \bar{y}$.
		Then $R \cap S = \{ \bar{a} \in D^k \mid \structure[\bar{x}: \bar{a}] \models \formula \land \formulatwo \} \in H_k$.
		\item Let $R \in H_{k+1}$.
		Then $R$ is of the form $\{(a_1, \dots, a_{k+1}) \in D^{k+1} \mid \structure[x_1:a_1, \dots, x_{k+1}:a_{k+1}] \models \formula \}$ for a tuple of distinct object symbols $\bar{x} = (x_1, \dots, x_{k+1})$ not in $\vocab$ and $\formula$ an \fo-formula over $\vocab \cup \{x_1, \dots, x_{k+1}\}$.
		Then $\{(a_1, \dots, a_k) \in D^k \mid (a_1, \dots, a_k, a) \in R \text{ for all } a \in D \} =  \{(a_1, \dots, a_k) \in D^k \mid \structure[x_1:a_1, \dots, x_{k}:a_{k}] \models \forall x_{k+1}: \formula \}  \in H_k$. \qedhere
	\end{enumerate}
\end{proof}

\subdefinitionhenkin*

\begin{proof}
	We can assume without loss of generality that $\voc{\structure} = \sym{\defn}$.
	Suppose that $\henkinstructure \modelsh \defn$.
	Then $\structure$ is the least pre-fixed $\henkinclass$-point of $\cons(\cdot, \structure)$, where $\context \coloneq \structure|_{\pars{\defn}}$.
	Let $\structure' \coloneq \structure|_{\sym{\defn'}}$.
	We need to show that $\structure'$ is the least pre-fixed $\henkinclass$-point of $\consg{\defn'}{\context'}(\cdot, \structure')$, where $\context' \coloneq \structure'|_{\pars{\defn'}}$.
	Since $\sym{\defn'} \subseteq \sym{\defn}$, $\structure'$ is an $\henkinclass$-point.
	
	To show that $\structure'$ is a pre-fixpoint of $\consg{\defn'}{\context'}(\cdot, \structure')$, take $P \in \defi{\defn'}$.
	Since $\structure$ is a pre-fixpoint of $\cons(\cdot, \structure)$, we have that $\{ \bar{a} \mid \structure[\bar{x}:\bar{a}] \models \mbd{P}{\defn}(\bar{x}) \} \subseteq P^{\structure}$.
	Note that $P^{\structure'} = P^{\structure}$ and that $\structure'[\bar{x}:\bar{a}] \models \mbd{P}{\defn'}(\bar{x})$ iff $\structure[\bar{x}:\bar{a}] \models \mbd{P}{\defn}(\bar{x})$, since $\mbd{P}{\defn'}(\bar{x}) \doteq \mbd{P}{\defn}(\bar{x})$. 
	It follows that $\{ \bar{a} \mid \structure'[\bar{x}:\bar{a}] \models \mbd{P}{\defn'}(\bar{x}) \} \subseteq P^{\structure'}$.
	This shows that $\structure'$ is a pre-fixpoint of $\consg{\defn'}{\context'}(\cdot, \structure')$.
	
	To show that $\structure'$ is the \emph{least} pre-fixed $\henkinclass$-point of $\consg{\defn'}{\context'}(\cdot, \structure')$, let $\structuretwo'$ be a pre-fixed $\henkinclass$-point of $\consg{\defn'}{\context'}(\cdot, \structure')$.
	Consider the $\sym{\defn}$-structure $\structuretwo$ that interprets any $P \in \defi{\defn'}$ as $P^{\structuretwo'} \cap P^{\structure}$, and any $\sigma \in \sym{\defn} \setminus \defi{\defn'}$ as $\sigma^{\structure}$.
	Note that $\structuretwo \sqsubseteq \structure$ and $\structuretwo|_{\sym{\defn'}} \sqsubseteq \structuretwo'$.
	We show that $\structuretwo$ is a pre-fixed $\henkinclass$-point of $\cons(\cdot, \structure)$.
	The fact that $\structuretwo$ is an $\henkinclass$-point follows from the fact that $\structuretwo'$ and $\structure$ are $\henkinclass$-points, using condition \ref{it:henkin-class-intersection} of Definition \ref{def:henkin-class}.
	To show that $\structuretwo$ is a pre-fixpoint of $\cons(\cdot, \structure)$, take $P \in \defi{\defn}$.
	We distinguish two cases:
	
	\textbf{Case}: $P \notin \defi{\defn'}$.
	Then $P^{\cons(\structuretwo, \structure)} \subseteq P^{\cons(\structure, \structure)} \subseteq P^{\structure} = P^{\structuretwo}$, using the monotonicity of $\cons(\cdot, \structure)$, the fact that $\structuretwo \sqsubseteq \structure$, and that $\structure$ is a pre-fixpoint of $\cons(\cdot, \structure)$.
	
	\textbf{Case}: $P \in \defi{\defn'}$.
	Similarly to the previous case, we find that $P^{\cons(\structuretwo, \structure)} \subseteq P^{\cons(\structure, \structure)} \subseteq P^{\structure}$.
	Furthermore, $P^{\cons(\structuretwo, \structure)} = P^{\consg{\defn'}{\context'}(\structuretwo|_{\sym{\defn'}}, \structure')} \subseteq P^{\consg{\defn'}{\context'}(\structuretwo', \structure')} \subseteq P^{\structuretwo'}$, using the monotonicity of $\consg{\defn'}{\context'}(\cdot, \structure')$, the fact that $\structuretwo|_{\sym{\defn'}} \sqsubseteq \structuretwo'$, and that $\structuretwo'$ is a pre-fixpoint of $\consg{\defn'}{\context'}(\cdot, \structure')$.
	Combining these inclusions, we obtain $P^{\cons(\structuretwo, \structure)} \subseteq P^{\structure} \cap P^{\structuretwo'} = P^{\structuretwo}$.
	
	This shows that $\structuretwo$ is indeed a pre-fixed $\henkinclass$-point of $\cons(\cdot, \structure)$.
	Since $\structure$ is the least pre-fixed $\henkinclass$-point of $\cons(\cdot, \structure)$, we find that $\structure \sqsubseteq \structuretwo$, and hence, that $\structure' = \structure|_{\sym{\defn'}} \sqsubseteq \structuretwo|_{\sym{\defn'}} \sqsubseteq \structuretwo'$, 
	This shows that $\structure'$ is the least pre-fixed $\henkinclass$-point of $\consg{\defn'}{\context'}(\cdot, \structure')$.
	In other words, $\structure' = \structure|_{\sym{\defn'}}$ is a Henkin model of $\defn'$ in $\context' = \structure|_{\pars{\defn'}}$, which means that $\henkinstructure \modelsh \defn'$.
\end{proof}

%% file: proofs-section-cut-elimination.tex
\section{Proofs of Section \ref{sec:cut-elimination}} \label{app:proofs-section-cut-elimination}

\subsection{Proof of Theorem \ref{thm:brotherston}} \label{app:brotherston-proof}

We prove Theorem \ref{thm:brotherston}, following the approach by Brotherston and Simpson \cite{jlc/BrotherstonS11}:

\brotherstonlemma*

Let $\fixedcanseq$ be a canonical sequent with positive definitions $\defns$. 
Assume that $\fixedcanseq$ is not cut-free provable.
We will show that $\fixedcanseq$ is not valid under the Henkin semantics, which will prove Theorem \ref{thm:brotherston} by contraposition.
To show that $\fixedcanseq$ is not valid under the Henkin semantics, we will construct a \emph{limit sequent} $\limitsequent$, where $\Gamma_\omega$ and $\Delta_\omega$ are infinite sets of \fo-formulas, containing $\Gamma$ and $\Delta$, respectively, such that $\limitsequent$ is not cut-free provable.
Next, we will construct a \emph{Henkin counter-model} for $\limitsequent$, i.e., a Henkin structure $\henkincountermodel$ such that 
$\henkincountermodel \modelsh \defn_i$ for all $i \in \{ 1, \dots, n \}$, $\henkincountermodel \modelsh \formula$ for all $\formula \in \Gamma_\omega$ and $\henkincountermodel \not\modelsh \formula$ for all $\formula \in \Delta_\omega$.

Let $\vocab$ be the vocabulary consisting of all non-logical symbols in $\fixedcanseq$, and let $\terms$ be the set of all terms built from function symbols in $\vocab$.
We construct the limit sequent $\limitsequent$ via a \emph{schedule}.

\begin{definition}[Schedule]
	A \emph{schedule element} is defined by the following rules:
	\begin{itemize}
		\item If $\formula$ and $\formulatwo$ are \fo-formulas over $\vocab$, then $\lnot \formula$, $\formula \land \formulatwo$, $\formula \lor \formulatwo$, $\formula \Rightarrow \formulatwo$ and $\formula \Leftrightarrow \formulatwo$ are schedule elements.
		\item If $\formula$ is an \fo-formula over $\vocab$, $x$ an object symbol in $\vocab$, and $t \in \terms$, then $\langle \forall x: \formula,\, t \rangle$ and $\langle \exists x: \formula,\, t \rangle$ are schedule elements.
		\item If $P(\bar{t})$ is an atom over $\vocab$ such that $P \in \defi{\defn_i}$ for some $i \in \{1, \dots, n\}$, and $\ihs{\defi{\defn_i}}$ a tuple of induction hypotheses over $\vocab$ for the predicates in $\defi{\defn_i}$, then $\langle P(\bar{t}), \ihs{\defi{\defn_i}} \rangle$ is a schedule element.
	\end{itemize}
	A \emph{schedule} $\sched$ is a sequence of schedule elements in which every schedule element appears infinitely often.
\end{definition}

Schedule elements correspond to applications of introduction rules in \scfoid, which we will use to construct the limit sequent $\limitsequent$.
Henceforth, we fix a schedule $\sched$, whose existence follows from the fact that there are countably many schedule elements.

\begin{definition}[Limit sequent] 
	\label{def:limit-sequent}
	For every $i \in \N$, we construct finite sets $\Gamma_i$ and $\Delta_i$ of \fo-formulas over $\vocab$, subject to the conditions:
	\begin{enumerate}[label=(\roman*)]
		\item $\Gamma_{i-1} \subseteq \Gamma_{i}$ and $\Delta_{i-1} \subseteq \Delta_{i}$ if $i > 0$; \label{it:inv-seq-inclusion}
		\item $\defns, \Gamma_i \vdash \Delta_i$ is not cut-free provable. \label{it:inv-seq-cut-free-provable}
	\end{enumerate}
	We do so by induction on $i$:
	
	For $i=0$, the sets $\Gamma_0$ and $\Delta_0$ are given by the fixed canonical sequent $\fixedcanseq$, which we assumed not to be cut-free provable.
	
	Suppose we have constructed sets $\Gamma_i$ and $\Delta_i$ subject to the conditions \ref{it:inv-seq-inclusion} and \ref{it:inv-seq-cut-free-provable}.
	Note that $\Gamma_i$ and $\Delta_i$ cannot have any formula in common, for otherwise $\defns, \Gamma_i \vdash \Delta_i$ would be cut-free provable via (ax), contradicting condition \ref{it:inv-seq-cut-free-provable}.
	We construct $\Gamma_{i+1}$ and $\Delta_{i+1}$ based on the form of 
	the $i$-th schedule element $\schedelem$ of $\sched$ and based on whether $\schedelem \in \Gamma_i$, $\schedelem \in \Delta_i$ or $\schedelem \notin \Gamma_i \cup \Delta_i$. 
	In case $\schedelem \notin \Gamma_i \cup \Delta_i$, we always pick $\Gamma_{i+1} \coloneq \Gamma_i$ and $\Delta_{i+1} \coloneq \Delta_i$.
	Hence, we can assume that $\schedelem \in \Gamma_i$ or $\schedelem \in \Delta_i$.
		
		\textbf{Case}: $\schedelem$ is of the form $\lnot \formula$. 
		
		\textbf{Subcase}: $\lnot \formula \in \Gamma_i$.
		Consider the \scfoid-derivation
		\begin{equation*}
			\begin{prooftree}
				\hypo{\defns, \Gamma_i \vdash \formula, \Delta_i}
				\infer1[($\lnot$L)]{\defns, \Gamma_i \vdash \Delta_i}
			\end{prooftree}
		\end{equation*}
		It shows that $\defns, \Gamma_i \vdash \formula, \Delta_i$ is not cut-free provable, for otherwise, $\defns, \Gamma_i \vdash \Delta_i$ would be cut-free provable.
		We pick $\Gamma_{i+1} \coloneq \Gamma_i$ and $\Delta_{i+1} \coloneq \Delta_i \cup \{\formula\}$.
		
		\textbf{Subcase}: $\lnot \formula \in \Delta_i$.
		Consider the \scfoid-derivation
		\begin{equation*}
			\begin{prooftree}
				\hypo{\defns, \Gamma_i, \formula \vdash \Delta_i}
				\infer1[($\lnot$R)]
				{\defns, \Gamma_i \vdash \Delta_i}
			\end{prooftree}
		\end{equation*}
		It shows that $\defns, \Gamma_i, \formula \vdash \Delta_i$ is not cut-free provable, for otherwise, $\defns, \Gamma_i \vdash \Delta_i$ would be cut-free provable.
		We pick $\Gamma_{i+1} \coloneq \Gamma_i \cup \{\formula\}$ and $\Delta_{i+1} \coloneq \Delta_i$.
		
		\textbf{Case}: $\schedelem$ is of the form $\formula \land \formulatwo$. 
		
		\textbf{Subcase}: $\formula \land \formulatwo \in \Gamma_i$.
		Consider the \scfoid-derivation
		\begin{equation*}
			\begin{prooftree}
				\hypo{\defns, \Gamma_i, \formula, \formulatwo \vdash \Delta_i}
				\infer1[($\land$L)]
				{\defns, \Gamma_i \vdash \Delta_i}
			\end{prooftree}
		\end{equation*}
		It shows that $\defns, \Gamma_i, \formula, \formulatwo \vdash \formula, \Delta_i$ is not cut-free provable, for otherwise, $\defns, \Gamma_i \vdash \Delta_i$ would be cut-free provable.
		We pick $\Gamma_{i+1} \coloneq \Gamma_i \cup \{\formula, \formulatwo\}$ and $\Delta_{i+1} \coloneq \Delta_i$.
		
		\textbf{Subcase}: $\formula \land \formulatwo \in \Delta_i$.
		Consider the \scfoid-derivation
		\begin{equation*}
			\begin{prooftree}
				\hypo{\defns, \Gamma_i \vdash \formula, \Delta_i}
				\hypo{\defns, \Gamma_i \vdash \formulatwo, \Delta_i}
				\infer2[($\land$R)]
				{\defns, \Gamma_i \vdash \Delta_i}
			\end{prooftree}
		\end{equation*}
		It shows that $\defns, \Gamma_i \vdash \formula, \Delta_i$ or $\defns, \Gamma_i \vdash \formulatwo, \Delta_i$ is not cut-free provable, for otherwise, $\defns, \Gamma_i \vdash \Delta_i$ would be cut-free provable.
		If $\defns, \Gamma_i \vdash \formula, \Delta_i$ is not cut-free provable, we pick $\Gamma_{i+1} \coloneq \Gamma_i$ and $\Delta_{i+1} \coloneq \Delta_i \cup \{\formula\}$.
		Otherwise, we pick $\Gamma_{i+1} \coloneq \Gamma_i$ and $\Delta_{i+1} \coloneq \Delta_i \cup \{\formulatwo\}$.

		\textbf{Case}: $\schedelem$ is of the form $\formula \lor \formulatwo$.
		
		\textbf{Subcase}: $\formula \lor \formulatwo \in \Gamma_i$. 
		Consider the derivation
		\begin{equation*}
			\begin{prooftree}
				\hypo{\defns, \Gamma_i, \formula \vdash \Delta_i}
				\hypo{\defns, \Gamma_i, \formulatwo \vdash \Delta_i}
				\infer2[($\lor$L)]
				{\defns, \Gamma_i \vdash \Delta_i}
			\end{prooftree}
		\end{equation*}
		It shows that $\defns, \Gamma_i, \formula \vdash \Delta_i$ or $\defns, \Gamma_i, \formulatwo \vdash \Delta_i$ is not cut-free provable, for otherwise, $\defns, \Gamma_i \vdash \Delta_i$ would be cut-free provable.
		If $\defns, \Gamma_i, \formula \vdash \Delta_i$ is not cut-free provable, we pick $\Gamma_{i+1} \coloneq \Gamma_i \cup \{\formula\}$ and $\Delta_{i+1} \coloneq \Delta_i$.
		Otherwise, we pick $\Gamma_{i+1} \coloneq \Gamma_i \cup \{\formula\}$ and $\Delta_{i+1} \coloneq \Delta_i$.

		\textbf{Subcase}: $\formula \lor \formulatwo \in \Delta_i$.
		Consider the derivation
		\begin{equation*}
			\begin{prooftree}
				\hypo{\defns, \Gamma_i \vdash \formula, \formulatwo, \Delta_i}
				\infer1[($\lor$R)]
				{\defns, \Gamma_i \vdash \Delta_i}
			\end{prooftree}
		\end{equation*}
		It shows that $\defns, \Gamma_i \vdash \formula, \formulatwo, \Delta_i$ is not cut-free provable, for otherwise, $\defns, \Gamma_i \vdash \Delta_i$ would be cut-free provable.
		We pick $\Gamma_{i+1} \coloneq \Gamma_i$ and $\Delta_{i+1} \coloneq \Delta_i \cup \{\formula, \formulatwo\}$.

		\textbf{Case}: $\schedelem$ is of the form $\formula \Rightarrow \formulatwo$.
		
		\textbf{Subcase}: $\formula \Rightarrow \formulatwo \in \Gamma_i$. 
		Consider the derivation
		\begin{equation*}
			\begin{prooftree}
				\hypo{\defns, \Gamma_i \vdash \formula, \Delta_i}
				\hypo{\defns, \Gamma_i, \formulatwo \vdash \Delta_i}
				\infer2[($\Rightarrow$L)]
				{\defns, \Gamma_i \vdash \Delta_i}
			\end{prooftree}
		\end{equation*}
		It shows that $\defns, \Gamma_i \vdash \formula, \Delta_i$ or $\defns, \Gamma_i, \formulatwo \vdash \Delta_i$ is not cut-free provable, for otherwise, $\defns, \Gamma_i \vdash \Delta_i$ would be cut-free provable.
		If $\defns, \Gamma_i \vdash \formula, \Delta_i$ is not cut-free provable, we pick $\Gamma_{i+1} \coloneq \Gamma_i$ and $\Delta_{i+1} \coloneq \Delta_i \cup \{\formula\}$.
		Otherwise, we pick $\Gamma_{i+1} \coloneq \Gamma_i \cup \{\formulatwo\}$ and $\Delta_{i+1} \coloneq \Delta_i$.
		
		\textbf{Subcase}: $\formula \Rightarrow \formulatwo \in \Delta_i$. 
		Consider the derivation
		\begin{equation*}
			\begin{prooftree}
				\hypo{\defns, \Gamma_i, \formula \vdash \formulatwo, \Delta_i}
				\infer1[($\Rightarrow$R)]
				{\defns, \Gamma_i \vdash \Delta_i}
			\end{prooftree}
		\end{equation*}
		It shows that $\defns, \Gamma_i, \formula \vdash \formulatwo, \Delta_i$ is not cut-free provable, for otherwise, $\defns, \Gamma_i \vdash \Delta_i$ would be cut-free provable.
		We pick $\Gamma_{i+1} \coloneq \Gamma_i \cup \{\formula\}$ and $\Delta_{i+1} \coloneq \Delta_i \cup \{\formulatwo\}$.
		
		\textbf{Case}: $\schedelem$ is of the form $\formula \Leftrightarrow \formulatwo$.

		\textbf{Subcase}: $\formula \Leftrightarrow \formulatwo \in \Gamma_i$. 
		Consider the derivation
		\begin{equation*}
			\begin{prooftree}
				\hypo{\defns, \Gamma_i \vdash \formula, \formulatwo, \Delta_i}
				\hypo{\defns, \Gamma_i, \formula, \formulatwo \vdash \Delta_i}
				\infer2[($\Leftrightarrow$L)]{\defns, \Gamma_i \vdash \Delta_i}
			\end{prooftree}
		\end{equation*}
		It shows that $\defns, \Gamma_i \vdash \formula, \formulatwo, \Delta_i$ or $\defns, \Gamma_i, \formula, \formulatwo \vdash \Delta_i$ is not cut-free provable, for otherwise, the conclusion $\defns, \Gamma_i \vdash \Delta_i$ would be cut-free provable.
		If $\defns, \Gamma_i \vdash \formula, \formulatwo, \Delta_i$ is not cut-free provable, we pick $\Gamma_{i+1} \coloneq \Gamma_i$ and $\Delta_{i+1} \coloneq \Delta_i \cup \{\formula, \formulatwo\}$.
		Otherwise, we pick $\Gamma_{i+1} \coloneq \Gamma_i \cup \{\formula, \formulatwo\}$ and $\Delta_{i+1} \coloneq \Delta_i$.
		
		\textbf{Subcase}: $\formula \Leftrightarrow \formulatwo \in \Delta_i$. 
		Consider the derivation
		\begin{equation*}
			\begin{prooftree}
				\hypo{\defns, \Gamma_i, \formula \vdash \formulatwo, \Delta_i}
				\hypo{\defns, \Gamma_i, \formulatwo \vdash \formula, \Delta_i}
				\infer2[($\Leftrightarrow$R)]{\defns, \Gamma_i \vdash \Delta_i}
			\end{prooftree}
		\end{equation*}
		It shows that $\defns, \Gamma_i, \formula \vdash \formulatwo, \Delta_i$ or $\defns, \Gamma_i, \formulatwo \vdash \formula, \Delta_i$ is not cut-free provable, for otherwise, the conclusion $\defns, \Gamma_i \vdash \Delta_i$ would be cut-free provable.
		If $\defns, \Gamma_i, \formula \vdash \formulatwo, \Delta_i$ is not cut-free provable, we pick $\Gamma_{i+1} \coloneq \Gamma_i \cup \{\formula\}$ and $\Delta_{i+1} \coloneq \Delta_i \cup \{\formulatwo\}$.
		Otherwise, we pick $\Gamma_{i+1} \coloneq \Gamma_i \cup \{\formulatwo\}$ and $\Delta_{i+1} \coloneq \Delta_i \cup \{\formula\}$.
		
		\textbf{Case}: $\schedelem$ is of the form $\langle \forall x : \formula, t \rangle$.
		
		\textbf{Subcase}: $\forall x : \formula \in \Gamma_i$. 
		Consider the derivation
		\begin{equation*}
			\begin{prooftree}
				\hypo{\defns, \Gamma_i, \formula[t/x] \vdash \Delta_i}
				\infer1[($\forall$L)]
				{\defns, \Gamma_i \vdash \Delta_i}
			\end{prooftree}
		\end{equation*}
		It shows that $\defns, \Gamma_i, \formula[t/x] \vdash \Delta_i$ is not cut-free provable, for otherwise, $\defns, \Gamma_i \vdash \Delta_i$ would be cut-free provable.
		We pick $\Gamma_{i+1} \coloneq \Gamma_i \cup \{\formula[t/x]\}$ and $\Delta_{i+1} \coloneq \Delta_i$.
		
		\textbf{Subcase}: $\forall x : \formula \in \Delta_i$. 
		Consider the derivation
		\begin{equation*}
			\begin{prooftree}
				\hypo{\defns, \Gamma_i \vdash \formula[z/x], \Delta_i}
				\infer1[($\forall$R)]
				{\defns, \Gamma_i \vdash \Delta_i}
			\end{prooftree}
		\end{equation*}
		for some variable $z$ not occurring freely in $\defns, \Gamma_i$, or $\Delta_i$.
		It shows that $\defns, \Gamma_i \vdash \formula[z/x], \Delta_i$ is not cut-free provable, for otherwise, $\defns, \Gamma_i \vdash \Delta_i$ would be cut-free provable.
		We pick $\Gamma_{i+1} \coloneq \Gamma_i$ and $\Delta_{i+1} \coloneq \Delta_i \cup \{\formula[z/x]\}$.

		\textbf{Case}: $\schedelem$ is of the form $\langle \exists x : \formula, t \rangle$.
		
		\textbf{Subcase}: $\exists x : \formula \in \Gamma_i$. 
		Consider the derivation
		\begin{equation*}
			\begin{prooftree}
				\hypo{\defns, \Gamma_i, \formula[z/x] \vdash \Delta_i}
				\infer1[($\exists$L)]
				{\defns, \Gamma_i \vdash \Delta_i}
			\end{prooftree}
		\end{equation*}
		for some variable $z$ not occurring freely in $\defns, \Gamma_i$, or $\Delta_i$.
		It shows that $\defns, \Gamma_i, \formula[z/x] \vdash \Delta_i$ is not cut-free provable, for otherwise, $\defns, \Gamma_i \vdash \Delta_i$ would be cut-free provable.
		We pick $\Gamma_{i+1} \coloneq \Gamma_i \cup \{\formula[z/x]\}$ and $\Delta_{i+1} \coloneq \Delta_i$.
		
		\textbf{Subcase}: $\exists x : \formula \in \Delta_i$. 
		Consider the derivation
		\begin{equation*}
			\begin{prooftree}
				\hypo{\defns, \Gamma_i \vdash \formula[t/x], \Delta_i}
				\infer1[($\exists$R)]
				{\defns, \Gamma_i \vdash \Delta_i}
			\end{prooftree}
		\end{equation*}
		It shows that $\defns, \Gamma_i \vdash \formula[t/x], \Delta_i$ is not cut-free provable, for otherwise, $\defns, \Gamma_i \vdash \Delta_i$ would be cut-free provable.
		We pick $\Gamma_{i+1} \coloneq \Gamma_i$ and $\Delta_{i+1} \coloneq \Delta_i \cup \{\formula[t/x]\}$.

		\textbf{Case}: $\schedelem$ is of the form $\langle P(\bar{t}), \ihs{\defi{\defn_i}} \rangle$.
		
		\textbf{Subcase}: $P(\bar{t}) \in \Gamma_i$.
		Consider the derivation
		\begin{equation*}
			\begin{prooftree}
				\hypo{\text{minor premises}}
				\hypo{\defns, \Gamma_i, \ihf{P}{\bar{t}} \vdash \Delta_i}
				\infer2[(def L)]
				{\defns, \Gamma_i \vdash \Delta_i}
			\end{prooftree}
		\end{equation*}
		that has a minor premise $\defns, \Gamma_i, \formulatwo\substihpi \vdash \ihf{Q}{\bar{s}}, \Delta_i$ for every rule. $\forall \bar{y}: Q(\bar{s}) \rul \formulatwo$ in $\defn_i$.\footnote{
		Since there are no negative occurrences of defined atoms in $\psi$, we can simplify $\formulatwo\substposihpi$ to $\formulatwo\substihpi$.
		} 
		It shows that one of the premises is not cut-free provable, for otherwise, the conclusion $\defns, \Gamma_i \vdash \Delta_i$ would be cut-free provable.
		If the major premise $\defns, \Gamma_i, \ihf{P}{\bar{t}} \vdash \Delta_i$ is not cut-free provable, we pick $\Gamma_{i+1} \coloneq \Gamma_i \cup \{\ihf{P}{\bar{t}}\}$ and $\Delta_{i+1} \coloneq \Delta_i$.
		Otherwise, we choose a minor premise $\defns, \Gamma_i, \formulatwo\substihpi \vdash \ihf{Q}{\bar{s}}, \Delta_i$ that is not cut-free provable (any choice works), and pick $\Gamma_{i+1} \coloneq \Gamma_i \cup \{\formulatwo\substihpi\}$ and $\Delta_{i+1} \coloneq \Delta_i \cup \{ \ihf{Q}{\bar{s}} \}$.
		
		\textbf{Subcase}: $P(\bar{t}) \in \Delta_i$.
		In this case, we pick $\Gamma_{i+1} \coloneq \Gamma_i$ and $\Delta_{i+1} \coloneq \Delta_i$.

	This covers all the possible cases.
	Note that in every possible case,
	we ensured that $\Gamma_i$ and $\Delta_i$ are finite, and that they satisfy conditions \ref{it:inv-seq-inclusion} and \ref{it:inv-seq-cut-free-provable}.

	We let $\Gamma_\omega \coloneq \cup_{i=0}^\infty \Gamma_i$ and $\Delta_\omega \coloneq \cup_{i=0}^\infty \Delta_i$, and we call $\limitsequent$ the \emph{limit sequent}.
	This sequent is not cut-free provable.
	Indeed, if it were, then there would exist finite sets $\Gamma' \subseteq \Gamma$ and $\Delta' \subseteq \Delta$ such that $\defns, \Gamma' \vdash \Delta'$ is cut-free provable.
	This follows by a straightforward argument by induction on the height of a cut-free proof of $\limitsequent$.\footnote{
		Recall that \scfoid-proofs are finite.
	}
	However, since $\Gamma' \subseteq \Gamma_i$ and $\Delta' \subseteq \Delta_i$ for some $i \in \N$, this would entail that $\defns, \Gamma_i \vdash \Delta_i$ is cut-free provable, which is a contradiction.
\end{definition}

By carefully constructing the sets $\Gamma_i$ and $\Delta_i$ along the schedule $\sched$, we encoded much structural information in the limit sequent $\limitsequent$.
Note that we did not employ the rules ($=$L) and (def R) in our construction.
Instead, the structural information reflected by these rules are covered by lemmas below.

We proceed by constructing the \emph{Henkin counter-model} $\henkincountermodel$ for $\limitsequent$. 
The domain of $\structure_\omega$ consists of all equivalence classes $[t]$ of terms $t \in \terms$ under an equivalence relation $\eqrel$ over $\terms$:

\begin{definition} \label{def:eqrel}
	We define a binary relation $\eqrel$ on $\terms$ via the rules:
	\begin{enumerate}[label=(\roman*)]
		\item if $t=s \in \Gamma_\omega$, then $t \eqrel s$; \label{it:base} %
		\item $t \eqrel t$; \label{it:refl} %
		\item if $t \eqrel s$, then $s \eqrel t$; \label{it:symm} %
		\item if $t \eqrel s$ and $s \eqrel u$, then $t \eqrel u$; \label{it:trans} %
		\item if $t_1 \eqrel s_1, \dots, t_k \eqrel s_k$ and $f$ is a $k$-ary function symbol, then $f(t_1, \dots, t_k) \eqrel f(s_1, \dots, s_k)$. \label{it:cong} %
	\end{enumerate}
	By \ref{it:refl}-\ref{it:trans}, $\eqrel$ is an equivalence relation.
	We denote the equivalence class of a term $t \in \terms$ under $\eqrel$ by $[t]$. %
	We will sometimes write $(t_1, \dots, t_k) \eqrel (s_1, \dots, s_k)$ for $t_1 \eqrel s_1, \dots, t_k \eqrel s_k$, 
	and $[(t_1, \dots, t_k)]$ for $([t_1], \dots, [t_k])$.
\end{definition}

We construct $\structure_\omega$ as the last structure $\structure_k$ of a finite sequence of structures $\structure_0, \structure_1, \dots, \structure_k$.
In defining this sequence, we use operators $\consg{\defn_i}{\context_i}$ for $i \in \{1, \dots, k\}$, as defined in Section \ref{sec:stable-semantics}.
Since the definitions $\defn_i$ are positive, the operators $\consg{\defn_i}{\context_i}$ depend only on the first argument.
Hence, by abuse of notation, we simply write $\consg{\defn_i}{\context_i}(\structure)$ instead of $\consg{\defn_i}{\context_i}(\structure, \structuretwo)$. 

\begin{definition} \label{def:I_omega}
	For every $i \in \{1, \dots, k\}$, we define a structure $\structure_i$ with vocabulary $\vocab \setminus \bigcup_{j=i+1}^k \defi{\defn_j}$:
	
	First, we define $\structure_0$ as the structure with vocabulary $\vocab \setminus \bigcup_{j=1}^k \defi{\defn_j}$ and with domain the set of equivalence classes $[t]$ of terms $t \in \terms$ under $\eqrel$ such that:

	\begin{itemize}
		\item for every $l$-ary function symbol $f$ in $\vocab$: $f^{\structure_0}([t_1], \dots, [t_l]) = [f^{\structure_0}(t_1, \dots, t_l)]$ for all $t_1, \dots, t_l \in \terms$;
		\item for every $l$-ary predicate symbol $P$ in $\vocab$: $P^{\structure_0} = \{ ([t_1], \dots, [t_l]) \mid P(t_1, \dots, t_l) \in \Gamma_\omega \}$.
	\end{itemize}
	Note that $f^{\structure_0}$ is well-defined by rule \ref{it:cong} of Definition \ref{def:eqrel}.
	
	Suppose that we have defined the structure $\structure_{i-1}$ for some $i \in \{1, \dots, k\}$.
	Since $\fixedcanseq$ is a canonical sequent, no parameter of $\defn_i$ is a defined predicate of $\defn_j$ for any $j>i$.
	This means that every parameter of $\defn_i$ is in $\voc{\structure_{i-1}} = \vocab \setminus \bigcup_{j=i}^k \defi{\defn_j}$, and therefore, $\context_{i} \coloneq \structure_{i-1}|_{\pars{\defn_{i}}}$ is a $\defn_{i}$-context.
	Let $\structuretwo_{i}$ be the least $\sym{\defn_{i}}$-structure expanding $\context_{i}$ such that: 
	\begin{enumerate}[label=(\roman*)]
		\item for all $P \in \defi{\defn_i}$ and $\bar{t} \in \terms^{\arity{P}}$: if $P(\bar{t}) \in \Gamma_\omega$, then $[\bar{t}] \in P^{\structuretwo_i}$; \label{it:J_j-base}
		\item $\consg{\defn_i}{\context_i}(\structuretwo_i) \sqsubseteq \structuretwo_i$.
	\end{enumerate}
	The existence of $\structuretwo_{i}$ follows from the Knaster-Tarski theorem \cite{Tarski}, as the $\sym{\defn_{i}}$-structures expanding $\context_{i}$ and satisfying condition \ref{it:J_j-base} form a complete lattice w.r.t.~$\sqsubseteq$, and $\consg{\defn_i}{\context_i}$ is a monotone operator on this lattice.
	We define $\structure_i$ as the expansion of $\structure_{i-1}$ with %
	$\voc{\structure_i} = \voc{\structure_{i-1}} \cup \defi{\defn_i}$ such that $P^{\structure_i} = P^{\structuretwo_i}$ for all $P \in \defi{\defn_i}$.
	
	In sum, we have defined a finite sequence of structures $\structure_0, \structure_1, \dots, \structure_k$.
	We define $\structure_\omega$ as $\structure_k$.
\end{definition}

\begin{proposition} \label{prop:t^I=[t]}
	$t^{\structure_\omega} = [t]$ for all $t \in \terms$.
\end{proposition}

\begin{proof}
	We show that $t^{\structure_\omega} = [t]$ for all $t \in \terms$, by structural induction on $t$:\footnote{
		Note that we do not need a separate case for object symbols, as we defined them to be $0$-ary function symbols.
	}

	\textbf{Case}: $t$ is of the form $f(t_1, \dots, t_k)$ for a $k$-ary function symbol $f$ and terms $t_1, \dots, t_k \in \terms$.
	By induction hypothesis, $t_i^{\structure_\omega} = [t_i]$ for each $i \in \{ 1, \dots, k \}$.
	By Definition \ref{def:I_omega},
	$f(t_1, \dots, t_k)^{\structure_\omega} = f^{\structure_\omega}(t_1^{\structure_\omega}, \dots, t_k^{\structure_\omega}) = f^{\structure_\omega}([t_1], \dots, [t_k]) = [f(t_1, \dots, t_k)]$.
\end{proof}

\begin{lemma} \label{lem:CFP-preserved-under-equivalence-of-terms}
	Let $\formula$ be an \fo-formula, $x$ an object symbol and $t,s \in \terms$.
	If $t \eqrel s$, then $\defns, \Gamma_\omega \vdash \formula[t / x]$ is cut-free provable iff $\defns, \Gamma_\omega \vdash \formula[s / x]$ is cut-free provable.
\end{lemma}

\begin{proof}
	Let $t$ and $s$ be terms such that $t \eqrel s$.
	Then there exists a tree-shaped derivation $T$ for the fact that $t \eqrel s$, according to the rules \ref{it:base}-\ref{it:cong} of Definition \ref{def:eqrel}.
	We show by induction on the height of $T$ that for every \fo-formula $\formula$ and object symbol $x$, the sequent $\defns, \Gamma_\omega \vdash \formula[t / x]$ is cut-free provable iff the sequent $\defns, \Gamma_{\omega} \vdash \formula[s / x]$ is cut-free provable:
	
	\textbf{Case}: $t \eqrel s$ is derived by \ref{it:base}. 
	Then $t=s \in \Gamma_\omega$.
	By application of (=L), a cut-free proof of $\defns, \Gamma_{\omega} \vdash \formula[t / x]$ can be extended to a cut-free proof of $\defns, \Gamma_{\omega} \vdash \formula[s / x]$ and vice versa.
	
	\textbf{Case}: $t \eqrel s$ is derived by \ref{it:refl}. 
	Then $s \doteq t$.
	Trivially, $\defns, \Gamma_{\omega} \vdash \formula[t / x]$ is cut-free provable iff $\defns, \Gamma_{\omega} \vdash \formula[t / x]$ is cut-free provable.
	
	\textbf{Case}: $t \eqrel s$ is derived by \ref{it:symm} from $s \eqrel t$. 
	By induction hypothesis, $\defns, \Gamma_{\omega} \vdash \formula[s / x]$ is cut-free provable iff $\defns, \Gamma_{\omega} \vdash \formula[t / x]$ is cut-free provable.
	Then, obviously, $\defns, \Gamma_{\omega} \vdash \formula[t / x]$ is cut-free provable iff $\defns, \Gamma_{\omega} \vdash \formula[s / x]$ is cut-free provable.
	
	\textbf{Case}: $t \eqrel s$ is derived by \ref{it:trans} from $t \eqrel r$ and $r \eqrel s$ for a term $r$. 
	By induction hypothesis, $\defns, \Gamma_{\omega} \vdash \formula[t / x]$ is cut-free provable %
	iff $\defns, \Gamma_{\omega} \vdash \formula[r / x]$ is cut-free provable, %
	and $\defns, \Gamma_{\omega} \vdash \formula[r / x]$ is cut-free provable %
	iff $\defns, \Gamma_{\omega} \vdash \formula[s / x]$ is cut-free provable. %
	Then $\defns, \Gamma_{\omega} \vdash \formula[t / x]$ is cut-free provable 
	iff $\defns, \Gamma_{\omega} \vdash \formula[s / x]$ is cut-free provable. 
	
	\textbf{Case}: $t \eqrel s$ is derived by \ref{it:cong} from $t_1 \eqrel s_1$, \dots, $t_k, \eqrel s_k$ for terms $t_1, \dots, t_k, s_1, \dots, s_k \in \terms$.
	Then $t \doteq f(t_1, \dots, t_k)$ and  $s \doteq f(s_1, \dots, s_k)$.
	By induction hypothesis, we know that for every $i \in \{1, \dots, k\}$ and every \fo-formula $\formulatwo$, the sequent $\defns, \Gamma_{\omega} \vdash \formulatwo[t_i / x]$ is cut-free provable iff $\defns, \Gamma_{\omega} \vdash \formulatwo[s_i / x]$ is cut-free provable.
	We infer that
	\begin{align*}
		& \defns, \Gamma_\omega \vdash \formula[f(t_1, \dots, t_k) / x] \text{ is cut-free provable} \\
		\text{iff } & \defns, \Gamma_\omega \vdash \formula[f(s_1, t_2 \dots, t_k) / x] \text{ is cut-free provable} \\
		\vdots & \\
		\text{iff } & \defns, \Gamma_\omega \vdash \formula[f(s_1, \dots, s_k) / x] \text{ is cut-free provable.}
	\end{align*} 
\end{proof}

\begin{lemma} \label{lem:satisfied-by-countermodel-implies-provable-from-gamma_omega}
	Let $P \in \defi{\defn_i}$ for some $i \in \{1, \dots, n\}$.
	If\/ $[\bar{t}] \in P^{\structure_\omega}$, then $\defns, \Gamma_\omega \vdash P(\bar{t})$ is cut-free provable.
\end{lemma}

\begin{proof}
	We prove by induction on $i$ that, for every $i \in \{1, \dots, n\}$ and every $P \in \defi{\defn_i}$, if $[\bar{t}] \in P^{\structure_\omega}$, then $\defns, \Gamma_\omega \vdash P(\bar{t})$ is cut-free provable, i.e., we prove this for a given $i \in \{1, \dots, n\}$, assuming that 
	\begin{equation} \label{eq:ih1-cfp-lemma}
		\text{for every $j < i$ and $Q \in \defi{\defn_j}$, if $[\bar{r}] \in Q^{\structure_\omega}$, then $\defns, \Gamma_\omega \vdash Q(\bar{r})$ is cut-free provable.}
	\end{equation}
	
	For each $P \in \defi{\defn_i}$, we let $X_P$ denote the set $\{ [\bar{t}] \mid \defns, \Gamma_\omega \vdash P(\bar{t}) \text{ is cut-free provable} \}$.
	Note that if $[\bar{t}] \in X_P$ for some $P \in \defi{\defn}$, then there exists a tuple of terms $\bar{s}$ such that $\bar{s} \eqrel \bar{t}$ and $\Gamma_\omega \vdash P(\bar{s})$ is cut-free provable.
	By repeated application of Lemma \ref{lem:CFP-preserved-under-equivalence-of-terms}, it follows that $\Gamma_\omega \vdash P(\bar{t})$ is cut-free provable as well.
	We will apply this fact multiple times 
	in the remainder of the proof.
	
	Let $\structuretwo \coloneq \structure_\omega |_{\pars{\defn_i}} [\defi{\defn_i} : X_{\defi{\defn_i}}]$.
	In words, $\structuretwo$ is the $\sym{\defn_i}$-structure that coincides with $\structure_\omega$ on $\pars{\defn_i}$ and that interprets every defined predicate $P$ of $\defn_i$ as $X_P$.
	We will show that $\structuretwo$ satisfies:
	\begin{enumerate}[label=(\roman*)]
		\item for all $P \in \defi{\defn_i}$ and $\bar{t} \in \terms^{\arity{P}}$: if $P(\bar{t}) \in \Gamma_\omega$, then $[\bar{t}] \in P^{\structuretwo}$; \label{it:J-base-1}
		\item $\consg{\defn_i}{\context_i}(\structuretwo) \sqsubseteq \structuretwo$, where $\context_i \coloneq \structure_\omega|_{\pars{\defn_i}}$. \label{it:J-prefixpoint-1}
	\end{enumerate}
	Since, by Definition \ref{def:I_omega}, $\structure_\omega|_{\sym{\defn_i}}$ is the least such structure, this will entail that $\structure_\omega|_{\sym{\defn_i}} \sqsubseteq \structuretwo$, and hence, that
	\begin{equation*}
		P^{\structure_\omega} = P^{\structure_\omega|_{\sym{\defn_i}}} \subseteq P^\structuretwo = X_P 
	\end{equation*}
	for all $P \in \defi{\defn_i}$.
	By definition of $X_P$, this means that if $[\bar{t}] \in P^{\structure_\omega}$, then $\defns, \Gamma_\omega \vdash P(\bar{t})$ is cut-free provable, thereby finishing the proof.
	
	To show that \ref{it:J-base-1} holds, take $P \in \defi{\defn_i}$, and assume that $P(\bar{t}) \in \Gamma_\omega$.
	Then $\defns, \Gamma_\omega \vdash P(\bar{t})$ is cut-free provable via an application of (ax). 
	It thus remains to show \ref{it:J-prefixpoint-1}, i.e., that $\consg{\defn_i}{\context_i}(\structuretwo) \sqsubseteq \structuretwo$.
	Take $P\in \defi{\defn_i}$.
	We need to show that $P^{\consg{\defn_i}{\context_i}(\structuretwo)} \subseteq P^\structuretwo$.
	Take $[\bar{t}] \in P^{\consg{\defn_i}{\context_i}(\structuretwo)}$.
	By definition of $\consg{\defn_i}{\context_i}$, this means that $\structuretwo[\bar{y}:[\bar{t}]] \models \mbd{P}{\defn_i}(\bar{y})$, where $\bar{y}$ is an $\arity{P}$-tuple of object symbols not occurring %
	in $\defn_i$.
	Since %
	$[\bar{t}] = \bar{t}^{\structure_\omega} = \bar{t}^{\structuretwo}$ 
	by Proposition \ref{prop:t^I=[t]}, this comes down to $\structuretwo \models \mbd{P}{\defn_i}(\bar{t})$.
	By Definition \ref{def:merged-body}, there must be a definitional rule $\forall \bar{x}: P(\bar{s}) \rul \formulatwo$ in $\defn_i$ such that %
	$\structuretwo \models \bar{t} = \bar{s}$ and $\structuretwo \models \formulatwo$.
	Since $\structuretwo$ coincides with $\structure_\omega$ on terms, we infer that $\bar{t} \eqrel \bar{s}$.
	We claim that $\defns, \Gamma_\omega \vdash \formulatwo$ is cut-free provable.
	Via (def R), this claim entails that $\defns, \Gamma_\omega \vdash P(\bar{s})$ is cut-free provable, and since $\bar{t} \eqrel \bar{s}$, furthermore that $\defns, \Gamma_\omega \vdash P(\bar{t})$ is cut-free provable.
	Since this implies that $[\bar{t}] \in X_P = P^\structuretwo$, proving the claim finished the proof of Lemma \ref{lem:satisfied-by-countermodel-implies-provable-from-gamma_omega}.
	We prove that $\defns, \Gamma_\omega \vdash \formulatwo$ is indeed cut-free provable, by structural induction on $\formulatwo$, using the fact that $\structuretwo \models \formulatwo$.
	Since $\formulatwo$ is the body of a definitional rule in $\defn_i$, and since $\fixedcanseq$ is a canonical sequent,
	$\formulatwo$ must be composed of atoms, connected only by conjunctions and disjunctions.
	
	\textbf{Case}: $\formulatwo$ is of the form $t=s$. 
	Then $t \eqrel s$.
	Via ($=$R),  $\defns, \Gamma_\omega \vdash t=t$ is cut-free provable.
	By Lemma \ref{lem:CFP-preserved-under-equivalence-of-terms}, $\defns, \Gamma_\omega \vdash t=s$ is cut-free provable as well.
	
	\textbf{Case}: $\formulatwo$ is of the form $Q(\bar{r})$.
	The predicate symbol $Q$ is either defined by $\defn_i$, by $\defn_j$ for some $j<i$, or not defined by any $\defn_j$ for $j \in \{1, \dots, n\}$.
	(Note that, since $\fixedcanseq$ is a canonical sequent, $Q$ cannot be defined by $\defn_j$ for any $j>i$.)
	
	\textbf{Subcase}: $Q$ is defined by $\defn_i$.
	Since $\structuretwo \models Q(\bar{r})$, we have that $[\bar{r}] \in X_Q$, which implies that $\defns, \Gamma_\omega \vdash Q(\bar{r})$. %
	
	\textbf{Subcase}: $Q$ is defined by $\defn_j$ for some $j<i$.
	Since $\structuretwo$ and $\structure_\omega$ coincide on $Q$, we have that $[\bar{r}] \in Q^{\structure_\omega}$. 
	By (\ref{eq:ih1-cfp-lemma}), %
	$\defns, \Gamma_\omega \vdash Q(\bar{r})$ is cut-free provable.
	
	\textbf{Subcase}: $Q$ is not defined by $\defn_j$ for any $j \in \{1, \dots, n\}$.
	Since $\structuretwo$ and $\structure_\omega$ coincide on $Q$, we have that $[\bar{r}] \in Q^{\structure_\omega}$. 
	By definition of $\structure_\omega$, there exists a tuple of terms $\bar{v}$ such that $\bar{v} \eqrel \bar{r}$ and $Q(\bar{v}) \in \Gamma_\omega$.
	Via (ax), $\defns, \Gamma_\omega \vdash Q(\bar{v})$ is cut-free provable, and by repeated application of Lemma \ref{lem:CFP-preserved-under-equivalence-of-terms}, so is $\defns, \Gamma_\omega \vdash Q(\bar{r})$.
	
	\textbf{Case}: $\formulatwo$ is of the form $\formula_1 \land \formula_2$.
	Then $\structuretwo \models \formula_1$ and $\structuretwo \models \formula_2$.
	By %
	induction hypothesis, $\defns, \Gamma_\omega \vdash \formula_1$ and $\defns, \Gamma_\omega \vdash \formula_2$ are cut-free provable, and via ($\land$R), so is $\defns, \Gamma_\omega \vdash \formula_1 \land \formula_2$.
	
	\textbf{Case}: $\formulatwo$ is of the form $\formula_1 \lor \formula_2$.
	Then $\structuretwo \models \formula_1$ or $\structuretwo \models \formula_2$.
	By %
	induction hypothesis, $\defns, \Gamma_\omega \vdash \formula_1$ or $\defns, \Gamma_\omega \vdash \formula_2$ is cut-free provable, and via (wk) and ($\lor$R), so is $\defns, \Gamma_\omega \vdash \formula_1 \lor \formula_2$.
\end{proof}

\begin{lemma} \label{lem:counter-model-satisfies-Gamma_omega-but-not-Delta_Omega}
	$\structure_\omega \models \formula$ for all $\formula \in \Gamma_\omega$ and $\structure_\omega \not\models \formula$ for all $\formula \in \Delta_\omega$.
\end{lemma}

\begin{proof}
	We prove this by structural induction on $\formula$:

	\textbf{Case}: $\formula$ is of the form $t=s$.
	Suppose that $t=s \in \Gamma_\omega$.
	Then $t \eqrel s$ by rule \ref{it:base} of Definition \ref{def:eqrel}.
	Hence, $t^{\structure_\omega} = [t] = [s] = s^{\structure_\omega}$, and therefore $\structure_\omega \models t=s$.
	
	Now suppose that $t=s \in \Delta_\omega$, and
	suppose for contradiction that $\structure_\omega \models t=s$.
	Then $[t] = t^{\structure_\omega} = s^{\structure_\omega} = [s]$, whence $t \eqrel s$.
	Since $\defns, \Gamma_\omega \vdash t=t$ is obviously cut-free provable, so is $\defns, \Gamma_\omega \vdash t=s$, by Lemma \ref{lem:CFP-preserved-under-equivalence-of-terms}.
	However, this implies that $\limitsequent$ is cut-free provable, which is a contradiction.
	Thus, $\structure_\omega \not\models t=s$.
	
	\textbf{Case}: $\formula$ is of the form $P(\bar{t})$.
	Suppose that $P(\bar{t}) \in \Gamma_\omega$.
	Then $[\bar{t}] \in P^{\structure_\omega}$ by Definition \ref{def:I_omega}.
	Hence, $\structure_\omega \models P(\bar{t})$.
	
	Now suppose that $P(\bar{t}) \in \Delta_\omega$,
	and suppose for contradiction that $\structure_\omega \models P(\bar{t})$.
	Then $[\bar{t}] \in P^{\structure_\omega}$.
	Consequently, $\defns, \Gamma_\omega \vdash P(\bar{t})$ is cut-free provable.
	Indeed, if $P$ is defined by one of the definitions $\defns$, this follows from Lemma \ref{lem:satisfied-by-countermodel-implies-provable-from-gamma_omega}.
	Otherwise, there exists a tuple of terms $\bar{s}$ such that $\bar{s} \eqrel \bar{t}$ such that $P(\bar{s}) \in \Gamma_\omega$, by Definition \ref{def:I_omega}.
	Through (ax), $\defns, \Gamma_\omega \vdash P(\bar{s})$ is cut-free provable, and by repeated application of Lemma \ref{lem:CFP-preserved-under-equivalence-of-terms}, so is $\defns, \Gamma_\omega \vdash P(\bar{t})$.
	Thus, in either case, $\defns, \Gamma_\omega \vdash P(\bar{t})$ is cut-free provable.
	However, since $P(\bar{t}) \in \Delta_\omega$, we obtain that $\limitsequent$ is cut-free provable, which is a contradiction.

	\textbf{Case}: $\formula$ is of the form $\lnot \formulatwo$.
	Suppose that $\lnot \formulatwo \in \Gamma_\omega$.
	Then
	$\formulatwo \in \Delta_\omega$ by construction of $\limitsequent$.
	Indeed, as $\lnot \formulatwo \in \Gamma_\omega$, there exists a $j \in \N$ such that $\lnot \formulatwo \in \Gamma_j$.
	Since $\lnot \formulatwo$ occurs infinitely often on the schedule $\sched$, it occurs as the $i$-th schedule element $\schedelem$ for some $i > j$.
	By %
	Definition \ref{def:limit-sequent}, 
	the set $\Delta_{i+1}$ contains $\formulatwo$, and hence, so does $\Delta_\omega$.
	By induction hypothesis, $\structure_\omega \not\models \formulatwo$.
	Therefore, $\structure_\omega \models \lnot \formulatwo$.
	
	Now suppose that $\lnot \formulatwo \in \Delta_\omega$.
	Then
	$\formulatwo \in \Gamma_\omega$ %
	by a similar reasoning as before.
	By induction hypothesis, $\structure_\omega \models \formulatwo$.
	Therefore, $\structure_\omega \not\models \lnot \formulatwo$.
	
	\textbf{Case}: $\formula$ is of the form $\formulatwo \land \formulathree$.
	Suppose that $\formulatwo \land \formulathree \in \Gamma_\omega$.
	Then
	$\formulatwo,\formulathree \in \Gamma_\omega$ by Definition \ref{def:limit-sequent}.
	By induction hypothesis, $\structure_\omega \models \formulatwo$ and $\structure_\omega \models \formulathree$.
	Hence, $\structure_\omega \models \formulatwo \land \formulathree$.
	
	Now suppose that $\formulatwo \land \formulathree \in \Delta_\omega$.
	Then $\formulatwo \in \Delta_\omega$ or $\formulathree \in \Delta_\omega$ by %
	Definition \ref{def:limit-sequent}.
	By induction hypothesis, $\structure_\omega \not\models \formulatwo$ or $\structure_\omega \not\models \formulathree$.
	Hence, $\structure_\omega \not\models \formulatwo \land \formulathree$.
	
	\textbf{Cases}: $\formula$ is of the form $\formulatwo \lor \formulathree$, $\formulatwo \Rightarrow \formulathree$ or $\formulatwo \Leftrightarrow \formulathree$. 
	These cases are similar to the previous case.
	
	\textbf{Case}: $\formula$ is of the form $\forall x: \formulatwo$.
	Suppose that $\forall x: \formulatwo \in \Gamma_\omega$.
	Then $\formulatwo[t/x] \in \Gamma_\omega$ for any term $t$, by %
	Definition \ref{def:limit-sequent}.
	Indeed, suppose that $\forall x: \formulatwo \in \Gamma_\omega$.
	Then $\forall x: \formulatwo \in \Gamma_j$ for some $j \in \N$.
	Let $t$ be a term.
	Since $\langle \forall x: \formulatwo, t \rangle$ occurs infinitely often on the schedule $\sched$, it occurs as the $i$-th schedule element of $\sched$ for some $i > j$.
	By %
	Definition \ref{def:limit-sequent}, 
	the set $\Gamma_{i+1}$ contains $\formulatwo[t/x]$.
	Hence, $\formulatwo[t/x] \in \Gamma_\omega$.
	By induction hypothesis, $\structure_\omega \models \formulatwo[t/x]$. %
	Since $t^{\structure_\omega} = [t]$ by Proposition \ref{prop:t^I=[t]}, this can be restated as $\structure_\omega[x:[t]] \models \formulatwo$.
	Since any domain element of $\structure_\omega$ is of the form $[t]$ for some term $t$, this shows that $\structure_\omega \models \forall x: \formulatwo$.
	
	Now suppose that $\forall x: \formulatwo \in \Delta_\omega$.
	Then, by Definition \ref{def:limit-sequent}, %
	the set $\Delta_\omega$ contains $\formulatwo[z/x]$ for some object symbol $z$. 
	By induction hypothesis, $\structure_\omega \not\models \formulatwo[z/x]$.
	Since $z^{\structure_\omega} = [z]$ by Definition \ref{def:I_omega}, this can be restated as $\structure_\omega[x:[z]] \not\models \formulatwo$.
	Consequently, $\structure_\omega \not\models \forall x: \formulatwo$.
	
	\textbf{Case}: $\formula$ is of the form $\exists x: \formulatwo$. 
	This case is similar to the previous case.
\end{proof}

\begin{definition}[Henkin counter-class]
	\label{def:henkin-counterclass}
	For every $k \in \N$, let $H_k$ be the set of relations of the form $\{ [\bar{t}] \in \dom{\structure_\omega}^k \mid \structure_\omega \models \formula[\bar{t} / \bar{x}] \}$ such that $\bar{x} = (x_1, \dots, x_k)$ is a tuple of distinct object symbols not in $\voc{\structure_\omega}$, $\bar{t}$ is a $k$-tuple of terms over $\voc{\structure_\omega} \cup \{x_1, \dots, x_k\}$ and $\formula$ an \fo-formula over $\voc{\structure_\omega} \cup \{x_1, \dots, x_k\}$.
	We call $\henkinclass_\omega = \{H_k \mid k \in \N \}$ the \emph{Henkin counter-class} for $\structure_\omega$.
\end{definition}

\begin{lemma} \label{lem:h-omega-henkin-class}
	The Henkin counter-class $\henkinclass_\omega$ for $\structure_\omega$ is a Henkin class for $\structure_\omega$.
\end{lemma}

\begin{proof}
	Since $\structure_\omega \models \formula[\bar{t}/\bar{x}]$ is equivalent to $\structure_\omega[\bar{x} : [\bar{t}]] \models \formula$ for all $\bar{x}$, $\bar{t}$ and $\formula$ %
	as specified in
	Definition \ref{def:henkin-counterclass},
	and since all elements of $\dom{\structure_\omega}$ are of the form $[t]$ for a term $t$,
	the Henkin counter-class $\henkinclass_\omega$ is precisely the Henkin class of \fo-definable sets for $\structure_\omega$ from Proposition \ref{prop:fo-definable-sets-form-henkin-class}, which is a Henkin class for $\structure_\omega$ by Proposition \ref{prop:fo-definable-sets-form-henkin-class}.
	(Note that $\structure_\omega$ %
	is domain closed, since every element of $\dom{\structure_\omega}$ is of the form $[t]$ for a term $t$, and $t^{\structure_\omega} = [t]$ by Proposition \ref{prop:t^I=[t]}.)
\end{proof}

\begin{lemma} \label{lem:henkin-structure-satisfies-defs}
	$(\structure_\omega, \henkinclass_\omega) \modelsh \defn_i$ for all $i \in \{1, \dots, n\}$.
\end{lemma}

\begin{proof}
	Take $i \in \{1, \dots, n\}$.
	We need to show that $\structure_\omega |_{\sym{\defn_i}}$ is the least pre-fixed $\henkinclass_\omega$-point of $\consg{\defn_i}{\context_i}$, where $\context_{i} \coloneq \structure_\omega |_{\pars{\defn_i}}$.
	We write $\henkinclass_\omega = \{ H_k \mid k \in \N \}$.
	First, $\structure_\omega |_{\sym{\defn_i}}$ is an $\henkinclass_\omega$-point.
	This is because for every $P \in \defi{\defn_i}$, we can write 
	$P^{\structure_\omega |_{\sym{\defn_i}}}$ as $\{ [\bar{t}] \in \dom{\structure_\omega} \mid \structure_\omega \models P(\bar{t}) \}$, 
	and since $P(\bar{t})$ is an \fo-formula, $P^{\structure_\omega |_{\sym{\defn_i}}} \in H_{\arity{P}}$.
	Furthermore, $\structure_\omega |_{\sym{\defn_i}}$ is a pre-fixpoint of $\consg{\defn_i}{\context_i}$ by Definition \ref{def:I_omega}. %
	It thus remains to show that $\structure_\omega |_{\sym{\defn_i}}$ is the \emph{least} pre-fixed $\henkinclass_\omega$-point of $\consg{\defn_i}{\context_i}$. %
	Let $\structuretwo$ be a pre-fixed $\henkinclass_\omega$-point of $\consg{\defn_i}{\context_i}$. %
	We will show that:
	\begin{enumerate}[label=(\roman*)]
		\item for all $P \in \defi{\defn_i}$ and $\bar{t} \in \terms^{\arity{P}}$: if $P(\bar{t}) \in \Gamma_\omega$, then $[\bar{t}] \in P^{\structuretwo}$; \label{it:J-base}
		\item $\consg{\defn_i}{\context_i}(\structuretwo) \sqsubseteq \structuretwo$. \label{it:J-prefixpoint}
	\end{enumerate}
	Since, by Definition \ref{def:I_omega}, $\structure_\omega |_{\sym{\defn_i}}$ is the least $\sym{\defn_i}$-structure expanding %
	$\context_{i}$
	that satisfies these properties, this implies that $\structure_\omega |_{\sym{\defn_i}} \sqsubseteq \structuretwo$, showing that $\structure_\omega |_{\sym{\defn_i}}$ is indeed the least pre-fixed $\henkinclass_\omega$-point of $\consg{\defn_i}{\context_i}$. %
	Note that \ref{it:J-prefixpoint} is immediately satisfied, as we took $\structuretwo$ to be a pre-fixed $\henkinclass_\omega$-point of $\consg{\defn_i}{\context_i}$.
	It thus remains to show \ref{it:J-base}.
	
	Since $\structuretwo$ is an $\henkinclass_\omega$-point, there exists, for every $P \in \defi{\defn_i}$, an \fo-formula $\formula_P$ and an $\arity{P}$-tuple $\bar{x}_P$ of distinct object symbols not in $\voc{\structure_\omega}$ such that 
	\begin{equation} \label{eq:FO-description-of-P^J}
		P^{\structuretwo} = \{ [\bar{t}] \in \dom{\structuretwo}^{\arity{P}} \mid \structure_\omega \models \formula_P[\bar{t} / \bar{x}_P]\} .
	\end{equation}
	For every $P \in \defi{\defn_i}$, we let $\ihs{P}$ denote the induction hypothesis $\setexp{\bar{z}}{\formula_P[\bar{z}/\bar{x}_P]}$. 
	
	Let $P(\bar{t}) \in \Gamma_\omega$. 
	Then there exists a $j \in \N$ such that $P(\bar{t}) \in \Gamma_j$. 
	Consider the schedule element $\langle P(\bar{t}), \ihs{\Pi} \rangle$, where $\Pi \coloneq \defi{\defn_i}$ and $\ihs{\Pi}$ consists of all induction hypotheses $\ihs{Q}$ for $Q \in \Pi$. %
	Since this schedule element occurs infinitely often in the schedule $\sched$, there exists a $k \geq j$ such that $\schedelemg{k} = \langle P(\bar{t}), \ihs{\Pi} \rangle$.
	By %
	Definition \ref{def:limit-sequent},
	the sequent $\defns, \Gamma_{k+1} \vdash \Delta_{k+1}$ is one of the premises of the rule instance
	\begin{equation*}
		\begin{prooftree}
			\hypo{\text{minor premises}}
			\hypo{\defns, \Gamma_k, \ihf{P}{\bar{t}} \vdash \Delta_k}
			\infer2[(def L)]
			{\defns, \Gamma_k \vdash \Delta_k}
		\end{prooftree}
	\end{equation*}
	with predicate set $\Pi$ and induction hypotheses $\ihs{\Pi}$.
	For every definitional rule $\forall \bar{y}: Q(\bar{s}) \rul \formulatwo$ in $\defn_i$, this instance has a minor premise $\defns, \Gamma_k, \formulatwo\substihpi \vdash \ihf{Q}{\bar{s}}, \Delta_k$.
	We claim that $\defns, \Gamma_{k+1} \vdash \Delta_{k+1}$  is the major premise $\defns, \Gamma_{k}, \ihf{P}{\bar{t}} \vdash \Delta_{k}$.
	This would mean that $\formula_P[\bar{t} / \bar{x}_P] \doteq \ihf{P}{\bar{t}} \in \Gamma_\omega$, which would imply that $\structure_\omega \models \formula_P[\bar{t} / \bar{x}_P]$ by Lemma \ref{lem:counter-model-satisfies-Gamma_omega-but-not-Delta_Omega}, and that $[\bar{t}] \in P^\structuretwo$ by (\ref{eq:FO-description-of-P^J}), thereby proving \ref{it:J-base}.
	
	We prove the claim by contradiction.
	Suppose that $\defns, \Gamma_{k+1} \vdash \Delta_{k+1}$ were one of the minor premises.
	Then it would be of the form $\defns, \Gamma_k, \formulatwo\substihpi \vdash \ihf{Q}{\bar{s}}, \Delta_k$ for %
	a definitional rule $\forall \bar{y} : Q(\bar{s}) \rul \formulatwo$ in $\defn_i$.
	Then $\formulatwo\substihpi \in \Gamma_\omega$ and $\ihf{Q}{\bar{s}} \in \Delta_\omega$.
	By Lemma \ref{lem:counter-model-satisfies-Gamma_omega-but-not-Delta_Omega}, the latter entails that $\structure_\omega \not\models \ihf{Q}{\bar{s}}$, and by (\ref{eq:FO-description-of-P^J}), $[\bar{s}] \not\in Q^\structuretwo$.
	It now suffices to show that $\structuretwo \models \formulatwo$.
	Indeed, by \ref{it:J-prefixpoint}, $\structuretwo \models \formulatwo$ implies that $[\bar{s}] \in Q^\structuretwo$, yielding the desired contradiction.
	We prove that $\structuretwo \models \formulatwo$ by structural induction on $\formulatwo$, using the fact that $\formulatwo\substihpi \in \Gamma_\omega$.
	Note that $\formulatwo$ is composed of atoms, connected only by conjunctions and disjunctions.
	
	\textbf{Case}: $\formulatwo$ is of the form $t=s$.
	Then $(t=s)\doteq \formulatwo\substihpi \in \Gamma_\omega$.
	By rule \ref{it:base} of Definition \ref{def:eqrel}, this means that $t \eqrel s$.
	Since $\structuretwo$ and $\structure_\omega$ coincide on terms, $\structuretwo \models t=s$.
	
	\textbf{Case}: $\formulatwo$ is of the form $R(\bar{v})$.
	
	\textbf{Subcase}: $R$ is defined by $\defn_i$.
	Then $\formula_R[\bar{v}/\bar{x}_R] \doteq \ihf{R}{\bar{v}} \doteq \formulatwo\substihpi \in \Gamma_\omega$.
	By Lemma \ref{lem:counter-model-satisfies-Gamma_omega-but-not-Delta_Omega}, $\structure_\omega \models \formula_R[\bar{v}/\bar{x}_R]$.
	By (\ref{eq:FO-description-of-P^J}), $[\bar{v}] \in R^\structuretwo$.
	Therefore, $\structuretwo \models R(\bar{v})$.
	
	\textbf{Subcase}: $R$ is not defined by $\defn_i$.
	Then $R(\bar{v}) \doteq \formulatwo\substihpi \in \Gamma_\omega$.
	By Lemma \ref{lem:counter-model-satisfies-Gamma_omega-but-not-Delta_Omega}, $\structure_\omega \models R(\bar{v})$, and since $\structuretwo$ and $\structure_\omega$ differ only on $\defi{\defn_i}$, $\structuretwo \models R(\bar{v})$.

	\textbf{Case}: $\formulatwo$ is of the form $\formulatwo_1 \land \formulatwo_2$.
	Then $\formulatwo_1\substihpi \land \formulatwo_2\substihpi \doteq \formulatwo\substihpi \in \Gamma_\omega$.
	By %
	Definition \ref{def:limit-sequent}, 
	the set $\Gamma_\omega$ contains $\formulatwo_1\substihpi$ and $\formulatwo_2\substihpi$.
	By %
	induction hypothesis, $\structuretwo \models \formulatwo_1$ and $\structuretwo \models \formulatwo_1$.
	Hence, $\structuretwo \models \formulatwo_1 \land \formulatwo_2$.
	
	\textbf{Case}: $\formulatwo$ is of the form $\formulatwo_1 \lor \formulatwo_2$.
	Then $\formulatwo_1\substihpi \lor \formulatwo_2\substihpi \doteq \formulatwo\substihpi \in \Gamma_\omega$.
	By %
	Definition \ref{def:limit-sequent}, 
	the set $\Gamma_\omega$ contains $\formulatwo_1\substihpi$ or $\formulatwo_2\substihpi$.
	By %
	induction hypothesis, $\structuretwo \models \formulatwo_1$ or $\structuretwo \models \formulatwo_1$.
	Hence, $\structuretwo \models \formulatwo_1 \lor \formulatwo_2$.
\end{proof}

In conclusion, we have constructed a Henkin structure $(\structure_\omega, \henkinclass_\omega)$ (Lemma \ref{lem:h-omega-henkin-class}) %
such that $(\structure_\omega, \henkinclass_\omega) \modelsh \defn_i$ for all $i \in \{1, \dots, n\}$ (Lemma \ref{lem:henkin-structure-satisfies-defs}), $(\structure_\omega, \henkinclass_\omega) \modelsh \formula$ for all $\formula \in \Gamma_\omega$ %
and $(\structure_\omega, \henkinclass_\omega) \not\modelsh \formula$ for all $\formula \in \Delta_\omega$ (Lemma \ref{lem:counter-model-satisfies-Gamma_omega-but-not-Delta_Omega}).
Since $\Gamma \subseteq \Gamma_\omega$ and $\Delta \subseteq \Delta_\omega$, $(\structure_\omega, \henkinclass_\omega)$ is a counter-model of the original canonical sequent $\fixedcanseq$.
This finishes the proof of Theorem \ref{thm:brotherston}.

\subsection{Proof of Lemma \ref{lem:cfp-positive-rewriting}} \label{app:cfp-positive-rewriting}

We prove Lemma \ref{lem:cfp-positive-rewriting}:

\cfppositiverewriting*

We do so by first proving a few auxiliary lemmas:

\begin{lemma} \label{lem:CFP-substitution-negation}
	Let $\regularsequent$ be a regular sequent.
	Let $R$ and $S$ be two predicate symbols with the same arity, such that $R \notin \defi{\defn_i}$ for every $i \in \{1, \dots, n\}$.
	If $\regularsequent$ is cut-free provable, then so is the sequent $\defn_1\srsubst, \dots, \defn_n\srsubst, \Gamma\srsubst \vdash \Delta\srsubst$ obtained from $\regularsequent$ by replacing every atom of the form $R(\bar{t})$ with $\lnot S(\bar{t})$.
\end{lemma}

\begin{proof}
	Suppose that $\regularsequent$ admits a cut-free proof $\prf$.
	We show by induction on the height of $\prf$ that $\defn_1\srsubst, \dots, \defn_n\srsubst, \Gamma\srsubst \vdash \Delta\srsubst$ is cut-free provable.
	The most interesting cases are those corresponding to the rules (def R) and (def L).
	As an illustration, we also discuss the cases corresponding to  (ax), ($=$R) and ($\lnot$L).
	
	\textbf{Case}: $\regularsequent$ is derived by (ax).
	Then $\bot \in \Gamma$, $\top \in \Delta$, or there exists a formula $\formula$ that is both in $\Gamma$ and in $\Delta$.
	In the latter case, the formula $\formula\srsubst$, obtained from $\formula$ by replacing every atom of the form $R(\bar{t})$ by $\lnot S(\bar{t})$, is both in $\Gamma\srsubst$ and $\Delta\srsubst$.
	In any case, $\defn_1\srsubst, \dots, \defn_n\srsubst, \Gamma\srsubst \vdash \Delta\srsubst$ is also derivable by (ax).
	
	\textbf{Case}: $\regularsequent$ is derived by ($=$R).
	Then $\Delta$ contains a formula of the form $t=t$.
	Then so does $\Delta\srsubst$.
	Hence, $\defn_1\srsubst, \dots, \defn_n\srsubst, \Gamma\srsubst \vdash \Delta\srsubst$ is derivable by ($=$R).
	
	\textbf{Case}: $\regularsequent$ is derived by ($\lnot$L).
	Then the final step of $\prf$ is of the form 
	\begin{equation*}
		\begin{prooftree}
			\hypo{\defns, \Gamma' \vdash \formula, \Delta}
			\infer1[($\lnot$L)]{\defns, \Gamma', \lnot \formula \vdash \Delta}
		\end{prooftree}
	\end{equation*}
	where $\Gamma = \Gamma' \cup \{ \neg \formula \}$.
	By induction hypothesis, $\defn_1\srsubst, \dots, \defn_n\srsubst, \Gamma'\srsubst \vdash \formula\srsubst, \Delta\srsubst$ is cut-free provable.
	By application of ($\lnot$L), so is $\defn_1\srsubst, \dots, \defn_n\srsubst, \Gamma'\srsubst, \neg \formula\srsubst \vdash \Delta\srsubst$.
	Since $\Gamma\srsubst = \Gamma'\srsubst \cup \{ \neg \formula\srsubst \}$, this means that $\defn_1\srsubst, \dots, \defn_n\srsubst, \Gamma\srsubst \vdash \Delta\srsubst$ is cut-free provable.
	
	\textbf{Case}: $\regularsequent$ is derived by (def R).
	Then the final step of $\prf$ is of the form 
	\begin{equation*}
		\begin{prooftree}
			\hypo{\defns, \Gamma \vdash \formula[\bar{s}/\bar{x}], \Delta'}
			\infer1[(def R)]{\defns, \Gamma \vdash P(\bar{t}[\bar{s}/\bar{x}]), \Delta'}
		\end{prooftree}
	\end{equation*}
	for a definitional rule $\defrul$ in a definition ${\defn_i}$ and a tuple of terms $\bar{s}$ with the same length as $\bar{x}$.
	Here $\Delta = \Delta' \cup \{P(\bar{t}[\bar{s}/\bar{x}])\}$.
	By induction hypothesis, $\defn_1\srsubst, \dots, \defn_n\srsubst, \Gamma\srsubst \vdash \formula[\bar{s}/\bar{x}]\srsubst, \Delta'\srsubst$ is cut-free provable.
	Since $R \notin \defi{\defn_i}$, $R$ is distinct from $P$.
	Note that $\forall \bar{x}: P(\bar{t}) \rul \formula\srsubst$ is a definitional rule in $\defn_i\srsubst$, and that $\formula[\bar{s}/\bar{x}]\srsubst \doteq \formula\srsubst[\bar{s}/\bar{x}]$.
	By application of (def R), the sequent $\defn_1\srsubst, \dots, \defn_n\srsubst, \Gamma\srsubst \vdash P(\bar{t}[\bar{s}/\bar{x}]), \Delta'\srsubst$ is cut-free provable.
	Since $\Delta\srsubst=\Delta'\srsubst \cup \{P(\bar{t}[\bar{s}/\bar{x}])\}$,
	this means that the sequent $\defn_1\srsubst, \dots, \defn_n\srsubst, \Gamma\srsubst \vdash \Delta\srsubst$ is cut-free provable.
	
	\textbf{Case}: $\regularsequent$ is derived by (def L).
	Then the final step of $\prf$ is of the form 
	\begin{equation*}
		\begin{prooftree}
			\hypo{\text{minor premises}}
			\hypo{\defns, \Gamma', \ihf{P}{\bar{t}} \vdash \Delta}
			\infer2[(def L)]{\defns, \Gamma', P(\bar{t}) \vdash \Delta}
		\end{prooftree}
	\end{equation*}
	for a defined predicate $P$ of a definition $\defn_i$, a subset $\Pi$ of $\defi{\defn_i}$ containing $P$, and an induction hypothesis $\ihs{Q}$ for every $Q \in \Pi$.
	Here $\Gamma = \Gamma' \cup \{ P(\bar{t}) \}$.
	For every definitional rule $\forall \bar{y} : Q(\bar{s}) \rul \formulatwo$ in $\defn_i$ with $Q \in \Pi$, this rule %
	has a minor premise $\defns, \Gamma', \formulatwo\substposihpi \vdash \ihf{Q}{\bar{s}}, \Delta$.
	By induction hypothesis, the sequent $\defn_1\srsubst, \dots, \defn_n\srsubst, \Gamma'\srsubst, \formulatwo\substposihpi\srsubst \vdash \ihf{Q}{\bar{s}}\srsubst, \Delta\srsubst$ is cut-free provable for each definitional rule $\forall \bar{y} : Q(\bar{s}) \rul \formulatwo$ of $\defn_i$ with $Q \in \Pi$, and so is the sequent 
	$\defn_1\srsubst, \dots, \defn_n\srsubst, \Gamma'\srsubst,$ $\ihf{P}{\bar{t}}\srsubst \vdash \Delta\srsubst$.
	By application of (def L) with predicate set $\Pi$ and induction hypotheses $\ihs{Q}\srsubst \doteq \setexp{\bar{z}}{\ihf{Q}{\bar{z}}\srsubst}$ for every $Q \in \Pi$, the sequent $\defn_1\srsubst, \dots, \defn_n\srsubst, \Gamma'\srsubst, P(\bar{t}) \vdash \Delta\srsubst$ is cut-free provable.
	Here we used that $\formulatwo\substposihpi\srsubst \doteq \formulatwo\srsubst\substpos{\ihs{\Pi}\srsubst}{\Pi}$ for all rules $\forall \bar{y} : Q(\bar{s}) \rul \formulatwo$ in $\defn_i$ with $Q \in \Pi$, since $R \notin \defi{\defn_i}$ for every $i \in \{1, \dots, n\}$ and hence, $R \notin \Pi$.
	Since $\Gamma\srsubst = \Gamma'\srsubst \cup \{P(\bar{t})\}$, this means that $\defn_1\srsubst, \dots, \defn_n\srsubst, \Gamma\srsubst \vdash \Delta\srsubst$ is cut-free provable.
\end{proof}

\begin{lemma} \label{lem:CFP-inv-neg}
	Let $\Gamma$ and $\Delta$ be sets of \foid-formulas and $\formula$ an \foid-formula.
	If\/ $\Gamma \vdash \lnot \formula, \Delta$ is cut-free provable, then so is $\Gamma, \formula \vdash \Delta$.
\end{lemma}

\begin{proof}
	This is shown by induction on the height of a cut-free proof $\prf$ of $\Gamma \vdash \lnot \formula, \Delta$.
	We discuss the cases corresponding to the rules (ax) and ($\lnot$R).
	The other cases are straightforward.
	
	\textbf{Case}: $\Gamma \vdash \lnot \formula, \Delta$ is derived by (ax).
	If $\bot \in \Gamma$, $\top \in \Delta$, or $\Gamma \cap \Delta \neq \emptyset$, then $\Gamma, \formula \vdash \Delta$ is also derivable by (ax).
	Otherwise, $\lnot \formula \in \Gamma$.
	In that case, $\Gamma, \formula \vdash \Delta$ is derivable by (ax) and ($\lnot$L).
	
	\textbf{Case}: $\Gamma \vdash \lnot \formula, \Delta$ is derived by ($\lnot$R).
	
	\textbf{Subcase}: the last inference in $\prf$ is of the form
	\begin{equation*}
		\begin{prooftree}
			\hypo{\Gamma, \formula \vdash \Delta}
			\infer1[($\lnot$R)]{\Gamma \vdash \lnot \formula, \Delta}
		\end{prooftree}
	\end{equation*}
	Then the subproof of $\prf$ ending with $\Gamma, \formula \vdash \Delta$ is a cut-free proof of $\Gamma, \formula \vdash \Delta$.
	
	\textbf{Subcase}: the last inference in $\prf$ is of the form
	\begin{equation*}
		\begin{prooftree}
			\hypo{\Gamma, \formulatwo \vdash \lnot \formula, \Delta'}
			\infer1[($\lnot$R)]{\Gamma \vdash \lnot \formulatwo, \lnot \formula, \Delta'}
		\end{prooftree}
	\end{equation*}
	where $\Delta = \Delta' \cup \{\lnot \formulatwo\}$.
	By induction hypothesis, $\Gamma, \formulatwo, \formula \vdash \Delta'$ is cut-free provable.
	By application of ($\lnot$R), so is $\Gamma, \formula \vdash \lnot \formulatwo, \Delta'$.
\end{proof}

\begin{lemma} \label{lem:replace-equivs-by-disjs}
	Let $\Gamma$ and $\Delta$ be sets of \foid-formulas.
	Let $Q$ be a $k$-ary predicate symbol, $i \in \{0, \dots, k\}$, $\bar{t}$ an $i$-tuple of terms, and $\bar{y}$ a $(k-i)$-tuple of object symbols.
	If\/ $\Gamma, \forall \bar{y}: \lnot Q(\bar{t}, \bar{y}) \Leftrightarrow \lnot Q(\bar{t}, \bar{y}) \vdash \Delta$ is cut-free provable, then so is $\Gamma, \forall \bar{y}: Q(\bar{t}, \bar{y}) \lor \lnot Q(\bar{t}, \bar{y}) \vdash \Delta$.\footnote{
		Writing out $\bar{t} = (t_1, \dots, t_i)$ and $\bar{y} = (y_1, \dots, y_{k-i})$, $Q(\bar{t}, \bar{y})$ stands for $Q(t_1, \dots, t_i, y_1, \dots, y_{k-i})$.
	}
\end{lemma}

\begin{proof}
	We show the statement by induction on the height of a cut-free proof $\prf$ of a sequent of the form $\Gamma, \forall \bar{y}: \lnot Q(\bar{t}, \bar{y}) \Leftrightarrow \lnot Q(\bar{t}, \bar{y}) \vdash \Delta$. %
	We discuss the cases corresponding to the inference rules (ax), ($\Leftrightarrow$L), ($\forall$L).
	The other cases are straightforward.
	
	\textbf{Case}: $\Gamma, \forall \bar{y}: \lnot Q(\bar{t}, \bar{y}) \Leftrightarrow \lnot Q(\bar{t}, \bar{y}) \vdash \Delta$ is derived by (ax).
	If $\bot \in \Gamma$, $\top \in \Delta$, or $\Gamma \cap \Delta \neq \emptyset$, then $\Gamma, \forall \bar{y}: Q(\bar{t}, \bar{y}) \lor \lnot Q(\bar{t}, \bar{y}) \vdash \Delta$ is clearly also derivable by (ax).
	Otherwise, $\forall \bar{y}: \lnot Q(\bar{t}, \bar{y}) \Leftrightarrow \lnot Q(\bar{t}, \bar{y}) \in \Delta$.
	In this case, the fact that $\Gamma, \forall \bar{y}: Q(\bar{t}, \bar{y}) \lor \lnot Q(\bar{t}, \bar{y}) \vdash \Delta$ is cut-free provable follows from the fact that $\vdash \forall \bar{y}: \lnot Q(\bar{t}, \bar{y}) \Leftrightarrow \lnot Q(\bar{t}, \bar{y})$ is cut-free provable, which is straightforward to show.
	
	\textbf{Case}: $\Gamma, \forall \bar{y}: \lnot Q(\bar{t}, \bar{y}) \Leftrightarrow \lnot Q(\bar{t}, \bar{y}) \vdash \Delta$ is derived by ($\Leftrightarrow$L).
	
	\textbf{Subcase}: the last inference in $\prf$ is of the form
	\begin{equation*}
		\begin{prooftree}
			\hypo{\Gamma, \lnot Q(\bar{t}) \vdash \Delta}
			\hypo{\Gamma \vdash \lnot Q(\bar{t}), \Delta}
			\infer2[($\Leftrightarrow$L)]{\Gamma, \lnot Q(\bar{t}) \Leftrightarrow \lnot Q(\bar{t}), \vdash \Delta}
		\end{prooftree}
	\end{equation*}
	(This is only possible only if $i=0$.)
	Then $\Gamma, \lnot Q(\bar{t}) \vdash \Delta$ and $\Gamma \vdash \lnot Q(\bar{t}), \Delta$ are cut-free provable.
	By Lemma \ref{lem:CFP-inv-neg}, $\Gamma, Q(\bar{t}) \vdash \Delta$ is cut-free provable as well, and by application of ($\lor$L), so is $\Gamma, Q(\bar{t}) \lor \lnot Q(\bar{t}), \vdash \Delta$.
	
	\textbf{Subcase}: the last inference in $\prf$ is of the form
	\begin{equation*}
		\begin{prooftree}
			\hypo{\Gamma', \forall \bar{y}: \lnot Q(\bar{t}, \bar{y}) \Leftrightarrow \lnot Q(\bar{t}, \bar{y}), \formula, \formulatwo \vdash \Delta}
			\hypo{\Gamma', \forall \bar{y}: \lnot Q(\bar{t}, \bar{y}) \Leftrightarrow \lnot Q(\bar{t}, \bar{y}) \vdash \formula, \formulatwo, \Delta}
			\infer2[($\Leftrightarrow$L)]{\Gamma', \forall \bar{y}: \lnot Q(\bar{t}, \bar{y}) \Leftrightarrow \lnot Q(\bar{t}, \bar{y}), \formula \Leftrightarrow \formulatwo \vdash \Delta}
		\end{prooftree}
	\end{equation*}
	where $\Gamma = \Gamma' \cup \{\formula \Leftrightarrow \formulatwo\}$.
	By induction hypothesis, the sequents $\Gamma', \forall \bar{y}: Q(\bar{t}, \bar{y}) \lor \lnot Q(\bar{t}, \bar{y}), \formula, \formulatwo \vdash \Delta$ and $\Gamma', \forall \bar{y}: Q(\bar{t}, \bar{y}) \lor \lnot Q(\bar{t}, \bar{y}), \formula, \formulatwo \vdash \Delta$ are cut-free provable.
	Then, by application of ($\Leftrightarrow$L), so is $\Gamma', \forall \bar{y}: Q(\bar{t}, \bar{y}) \lor \lnot Q(\bar{t}, \bar{y}), \formula \Leftrightarrow \formulatwo \vdash \Delta$.

	\textbf{Case}: $\Gamma, \forall \bar{y}: \lnot Q(\bar{t}, \bar{y}) \Leftrightarrow \lnot Q(\bar{t}, \bar{y}) \vdash \Delta$ is derived by ($\forall$L).
	
	\textbf{Subcase}: the last inference in $\prf$ is of the form
	\begin{equation*}
		\begin{prooftree}
			\hypo{\Gamma, \forall \bar{y}': \lnot Q(\bar{t}, t, \bar{y}') \Leftrightarrow \lnot Q(\bar{t}, t, \bar{y}') \vdash \Delta}
			\infer1[($\forall$L)]{\Gamma, \forall \bar{y}: \lnot Q(\bar{t}, \bar{y}) \Leftrightarrow \lnot Q(\bar{t}, \bar{y}) \vdash \Delta}
		\end{prooftree}
	\end{equation*}
	where, writing out $\bar{y} = (y_1, y_2, \dots, y_i)$, $\bar{y}'$ denotes the tuple $(y_2, \dots, y_i)$.
	By induction hypothesis, $\Gamma, \forall \bar{y}': Q(\bar{t}, t, \bar{y}') \lor \lnot Q(\bar{t}, t, \bar{y}') \vdash \Delta$ is cut-free provable, and by application of ($\forall$L), so is $\Gamma, \forall \bar{y}: Q(\bar{t}, \bar{y}) \lor \lnot Q(\bar{t}, \bar{y}) \vdash \Delta$.
	
	\textbf{Subcase}: the last inference in $\prf$ is of the form
	\begin{equation*}
		\begin{prooftree}
			\hypo{\Gamma', \forall \bar{y}: \lnot Q(\bar{t}, \bar{y}) \Leftrightarrow \lnot Q(\bar{t}, \bar{y}), \formula[t/x] \vdash \Delta}
			\infer1[($\forall$L)]{\Gamma', \forall \bar{y}: \lnot Q(\bar{t}, \bar{y}) \Leftrightarrow \lnot Q(\bar{t}, \bar{y}), \forall x: \formula \vdash \Delta}
		\end{prooftree}
	\end{equation*}
	where $\Gamma = \Gamma' \cup \{\forall x: \formula\}$.
	By induction hypothesis, $\Gamma', \forall \bar{y}: Q(\bar{t}, \bar{y}) \lor \lnot Q(\bar{t}, \bar{y}), \formula[t/x] \vdash \Delta$ is cut-free provable, and by application of ($\forall$L), so is $\Gamma', \forall \bar{y}: Q(\bar{t}, \bar{y}) \lor \lnot Q(\bar{t}, \bar{y}), \forall x: \formula \vdash \Delta$.
	\end{proof}

We are now ready to prove Lemma \ref{lem:cfp-positive-rewriting}:

\begin{proof}[Proof of Lemma \ref{lem:cfp-positive-rewriting}]
	Suppose that $\positiverewriting$ is cut-free provable.
	Then, by repeated application of Lemma \ref{lem:CFP-substitution-negation} (with $R = \compl{Q}$ and $S = Q$ for each predicate symbol $Q$ appearing negatively in the body of a rule of a definition $\defn_i$), so is $\defns, \Gamma, \equivs' \vdash \Delta$, where $\equivs'$ consists of all formulas of the form $\forall \bar{y}: \lnot Q(\bar{y}) \Leftrightarrow \lnot Q(\bar{y})$ such that $Q$ appears negatively in the body of a rule of a definition $\defn_i$.
	Next, by repeated application of Lemma \ref{lem:replace-equivs-by-disjs} (with $i=0$), the sequent $\Gamma, \Omega'', {\defn_1}, \dots, {\defn_n} \vdash \Delta$ is cut-free provable, where $\Omega''$ consists of all formulas of the form $\forall \bar{y}: Q(\bar{y}) \lor \lnot Q(\bar{y})$ such that $Q$ appears negatively in the body of a rule of a definition $\defn_i$.
	It is straightforward to check that $\vdash \forall \bar{y}: Q(\bar{y}) \lor \lnot Q(\bar{y})$ is cut-free provable for every such formula.
	By repeated application of (wk) and (cut) on these formulas, we deduce that $\defns, \Gamma \vdash \Delta$ is EC-provable.
\end{proof}